\documentclass[12pt,a4paper]{elsarticle}
\usepackage[nottoc]{tocbibind}
\usepackage{graphicx}
\usepackage{slashed}
\usepackage{enumerate}
\usepackage{amsmath}
\usepackage{mathtools,cancel}
\usepackage[bottom]{footmisc}
\usepackage[scr=rsfso]{mathalfa}
\usepackage[hidelinks]{hyperref}
{

}
\usepackage{amsfonts}
\usepackage{psfrag}
\usepackage{color}
\usepackage{mathtools}
\usepackage{pgf,tikz}
\usepackage{mathrsfs}
\usetikzlibrary{arrows}
\usepackage{fancyhdr}
\usepackage{amssymb}
\usepackage{slashed}
\usepackage{enumitem}

\usepackage[T1]{fontenc}
\usepackage[utf8]{inputenc}

\newtheorem{remark}{Remark}

\newtheorem{proposition}{Proposition}[section]

\newtheorem{theorem}{Theorem}[section]
\newtheorem{corollary}{Corollary}[section]
\newtheorem{lemma}{Lemma}[section]

\numberwithin{equation}{section}
\allowdisplaybreaks[4]

\newenvironment{proof}{\smallskip\noindent\emph{Proof.}\hspace{1pt}}%
{\hspace{-5pt}{\nobreak\quad\nobreak\hfill\nobreak$\square$\vspace{8pt}%
		\par}\smallskip\goodbreak}

\newcommand{\Lbar}{\underline{L}}
\newcommand{\Hbar}{\underline{H}}

\newcommand{\omegabar}{\underline{\omega}}

\newcommand{\hnabla}{\widehat{\nabla}}

\newcommand{\chihat}{\widehat{\chi}}

\newcommand{\chibar}{\underline{\chi}}
\newcommand{\chibarhat}{\underline{\widehat{\chi}}}
\newcommand{\ubar}{\underline{u}}

\newcommand{\be}{\begin{equation}}
\newcommand{\ee}{\end{equation}}

\newcommand{\bm}{\begin{align*}}
\newcommand{\enm}{\end{align*}}

\newcommand{\bespeq}{\begin{equation}\begin{split}}
\newcommand{\espeq}{\end{split}\end{equation}}

\newcommand{\alphabar}{\underline{\alpha}}
\newcommand{\betabar}{\underline{\beta}}
\newcommand{\etabar}{\underline{\eta}}

\newcommand{\Hodge}[1]{\prescript{*}{}{#1}}

\def\a {\alpha}
\def\b {\beta}

\newcommand{\tr}{\mbox{tr}}

\newcommand{\m}{\mu}
\newcommand{\n}{\nu}
\newcommand\restri[2]{{
		\left.\kern-\nulldelimiterspace 
		#1 
		\right|_{#2} 
}}

\definecolor{ffqqqq}{rgb}{1.,0.,0.}
\definecolor{uuuuuu}{rgb}{0.26666666666666666,0.26666666666666666,0.26666666666666666}

\makeatletter
\def\ps@pprintTitle{%
  \let\@oddhead\@empty
  \let\@evenhead\@empty
  \let\@oddfoot\@empty
  \let\@evenfoot\@oddfoot
}
\def\@author#1{\g@addto@macro\elsauthors{\normalsize%
    \def\baselinestretch{1}%
    \upshape\authorsep#1\unskip\textsuperscript{%
      \ifx\@fnmark\@empty\else\unskip\sep\@fnmark\let\sep=,\fi
      \ifx\@corref\@empty\else\unskip\sep\@corref\let\sep=,\fi
      }%
    \def\authorsep{\unskip,\space}%
    \global\let\@fnmark\@empty
    \global\let\@corref\@empty  
    \global\let\sep\@empty}%
    \@eadauthor={#1}
}
\makeatother

\begin{document}
\begin{frontmatter}
	
\title{Einstein-Yang-Mills equations in the double null framework}
\author{Puskar Mondal\fnref{fn1,fn2}}
\ead{puskar_mondal@fas.harvard.edu}
\author{Shing-Tung Yau\fnref{fn1,fn2}}
\ead{yau@math.harvard.edu}
\fntext[fn1]{Centre of Mathematical Sciences and Applications, Harvard University}
\fntext[fn2]{Department of Mathematics, Harvard University}

\begin{abstract}
\vspace{5pt}
\par \noindent 	We prove a semi-global gauge-invariant estimate for the solutions of the characteristic initial value problem associated with the coupled Einstein-Yang-Mills equations. In particular, we prove the existence of \textit{a} future development of regular initial data on a pair of incoming and outgoing null hypersurfaces emanating from a spacelike topological $2$-sphere. This marks the first study of the characteristic initial value problem of Einstein's equations with  a non-linear source. 
\end{abstract}
\end{frontmatter}

\begin{abstract}
 
\end{abstract}

\medskip

\tableofcontents


\begin{section} {Introduction and motivation}
\noindent One of the important problems of modern general relativity is the dynamical formation of spacetime singularities and their stability properties. According to Penrose's \textit{weak cosmic censorship conjecture}, \cite{penrose1999question}, the singularities in a general relativistic system can not be accessed by a future observer. If the singularities were to occur, they had to be hidden behind a horizon and therefore are not accessible to an observer located in the domain of outer communication. In the original singularity theorem of Penrose, the formation of a future singularity was understood in terms of the null geodesic incompleteness and such an incompleteness required the formation of a \textit{trapped} surface in a spacetime with certain topological properties (such as the spacetime admits a non-compact Cauchy hypersurface) \cite{penrose1965gravitational}. Following Penrose's analysis, trapped surface formation implies geodesic incompleteness, and therefore at a formal level formation of a trapped surface corresponds to the formation of a black hole. However, the major challenge in the fully general relativistic setting (and possibly without symmetry) is the precise condition under which a trapped surface may form. A few years after Penrose's incompleteness theorem was published, Schoen and Yau \cite{schoen1983existence} proved that for an asymptotically flat initial data set with mass density large on a large region, there is a closed trapped surface in the initial data. The future evolution of such data would then generate a geodesically incomplete spacetime according to Penrose's theorem. Much Later, \cite{christodoulou2012formation} proved the formation of trapped surface in an evolutionary manner. More specifically, he showed that regular dispersed initial data that contains no trapped surface \textit{can} lead to the formation of a trapped surface under the Einsteinian evolution of vacuum spacetime. Later \cite{klainerman2012formation} presented a simplified proof of the formation of a trapped surface in vacuum gravity and enlarged the admissible set of initial data.

Moving one step further, one would like to understand the formation of black holes (trapped surfaces) including suitable sources. This is of course motivated by the fact that our universe contains the structure and such structure is expected to arise due to matter (or radiation)-gravity interaction. Therefore, it is important to couple Einstein's equations with suitable sources and subsequently study the coupled dynamics. There has been progress in studying the trapped surface formation in the context of source-coupled Einstein dynamics over the past few years. \cite{yu2011dynamical} proved the dynamical formation of a trapped surface by coupling electromagnetic field to Einstein's gravity without symmetry assumption. \cite{an2020trapped} established a trapped surface formation criterion for the Einstein-Maxwell-charged scalar field system under the assumption of spherical symmetry. There are several other studies including Vlasov matter source (\cite{andreasson2012black}), perfect fluid source \cite{burtscher2014formation}, and null dust source \cite{moschidis2020proof} and studies in the context of trapped surface formation by focusing incoming gravitational radiation from the past null infinity \cite{an2014trapped} as well. The first step towards studying a trapped surface formation is to establish a well-posedness result of the characteristic initial value problem.

Apart from the formation of singularities that are hidden behind a horizon and therefore inaccessible to the observers located in the domain of outer communication, naked singularities are of significant importance in general relativity. As we have mentioned in beginning, the existence of this type of singularity is ruled out by Penrose's weak cosmic censorship conjecture. In other words, the existence of such a singularity that is accessible by an observer at future null infinity would indicate a pathological breakdown of Einstein's theory.  
Christodoulou \cite{christodoulou1994examples} showed a possible formation of a naked singularity in the context of Einstein-scalar field dynamics right before the collapse to a black hole. The genericity of such singularity is known to be violated i.e., perturbations seem to destroy such singularity \cite{christodoulou1999instability, liu2018robust} and as such they appear to be rather an artifact of high symmetry of the spacetimes. Apart from the study of \cite{christodoulou1994examples}, recently \cite{rodnianski2018asymptotically,rodnianski2019naked} introduced a new type of geometric twisting phenomenon that contributes to the formation of a `naked' singularity in a self-similar vacuum setting. However, the genericity of such solutions remains to be studied. Recently, Yau, Chen, and Du \cite{Yau} constructed a remarkable family of spherically symmetric solutions of the Einstein-Yang-Mills equations that possess the property of being regular at the center of symmetry. However, the spacetime Riemann curvature (suitable invariant) is shown to blow up at the apparent horizon. Robust numerical studies suggest $C^{0}$-stability of such solutions in the class of spherical symmetry. Contrary to the Einstein-Maxwell system or Einstein-scalar field system, Einstein-Yang-Mills equations are tremendously rich even in spherical symmetry and exhibit non-trivial dynamics. The numerical result of Bartnik \cite{bartnik1988particlelike} first showed the existence of a countable family of soliton-type solutions that are globally regular. Later Yau, Wasserman, and Smoller \cite{smoller1991smooth} rigorously proved the existence of such soliton-like solutions. However, such solutions were proven to be unstable against perturbations \cite{straumann1990instability}. Later \cite{smoller1993existence} also proved the existence of an infinite family of black hole solutions with a regular event horizon. The existence of these nontrivial solutions essentially unfolds the rich characteristics of the Einstein-Yang-Mills system. Due to the non-linear characteristics of the Yang-Mills fields, the fully coupled Einstein-Yang-Mills system is dynamically flexible i.e., both the possibility of the existence of regular solutions and the formation of singularities are open. This is precisely due to the fact that the non-linearity of Yang-Mills fields can counterbalance the non-linearity of gravity and the formation of singularity or regularity of the solutions is essentially dictated by the dominating one which in turn depends on several additional conditions. Returning back to the EYM solution \cite{Yau} containing a naked singular horizon, one is compelled to ask the following question: can these solutions arise in an evolutionary manner? In other words, one would want to study an initial value problem where the initial data is assumed to be sufficiently regular and possesses a degree of genericity and investigate whether such data can yield these naked singular solutions in finite time. This is motivated by the weak cosmic censorship conjecture \cite{penrose1999question} that rules out the possibility of the existence of evolutionary naked singularity (arising from regular and generic data). In addition, one would also like to provide analytical arguments supporting the stability (instability) of these solutions.

Motivated by these fundamental problems, we initiate the study of the dynamics of the Einstein-Yang-Mills system in the setting of a characteristic initial value formulation. In particular, we want to explore the nonlinear interaction of gravity and the Yang-Mills field and study two problems in the potential future: deducing the criteria to form trapped surfaces and naked singularities. The first step towards proving a trapped surface formation result is to establish a semi-global existence property of the coupled system i.e., one needs to ensure that the spacetime exists for long enough to form a trapped surface. Since the null hypersurfaces \textit{are} the carrier of the gravitational and Yang-Mills radiation (both have the same characteristics), it is most natural to work in this \textit{double null} framework. In addition, the naked singular solution of \cite{Yau} does not arise at the origin but rather on a sphere of finite radius, and therefore the question of the stability of such solutions translates to an exterior stability problem. In other words, one would like to understand if one perturbs these solutions, can the energy of the perturbations escape through the outgoing null cones, or can they potentially focus to form a trapped surface thereby destroying the naked singularities. In order to address such a question, the double null framework seems to be the most natural one to adapt. 

The study of characteristic initial value problem for vacuum Einstein equation was initiated by Rendall \cite{rendall1992characteristic}. In particular, \cite{rendall1992characteristic} proved the existence of a solution to the characteristic initial value problem in a small enough neighborhood of the intersection of an outgoing and an incoming null hypersurface. This construction is not very useful in the context of studying trapped surface formation since in the latter one ought to evolve the initial data long enough along one of the null directions. Later Luk \cite{luk2012local} improved the time of existence along one of the null directions in the context of vacuum gravity in a fairly straightforward way. However, it turns out that coupling to the Yang-Mills source (or Maxwell for that matter) complicates the analysis, and as such the analysis of \cite{luk2012local} does not apply due to obstruction of closing the regularity argument. Roughly, the complication arises due to the presence of Yang-Mills source terms in the null Bianchi equations for the Weyl curvature. The appearance is such that one requires the Yang-Mills curvature components to have a regularity level one order higher than that of Weyl curvature components. In other words, if we work with $K$ ($K\geq 3$) angular derivatives of Weyl curvature in $L^{2}(H,  \Hbar)$ ($H$ and $  \Hbar$ denote the outgoing and incoming null hypersurfaces to be defined later), then from the null Bianchi equations for the Weyl curvature, one would need to control $K+1$ angular derivatives of the Yang-Mills curvature. This in turn would require control of the $K+1$ angular derivatives of the connection coefficients from the null Yang-Mills equations. However, this seems to be incompatible with the analysis of \cite{luk2012local} since the latter is compatible with controlling $K$ angular derivatives of the connection coefficients on the topological 2-spheres. To circumvent this issue, elliptic estimates become indispensable. Here, we choose to work with the optimal regularity level of \cite{klainerman2012formation} (or a higher-order regularity level consistent with the optimal regularity in a relative sense). In particular, we work with only $1$ angular derivative of the Weyl curvature bounded in $L^{2}(H,  \Hbar)$. This in turn requires control of $2$ angular derivatives of the Yang-Mills curvature and the space-time connection coefficients. This regularity argument can be closed by means of the aforementioned elliptic estimates and trace estimates. Of course, one can propagate these estimates to successive higher orders yielding estimates for a classical solution. In addition to the subtlety associated with the regularity level, an important difficulty arises in the choice of gauge. Since Yang-Mills theory is a gauge theory after all, one ought to work in a particular choice of gauge (or equivalently descend to the `orbit space' of the theory). Unfortunately, there does not exist a global gauge choice in Yang-Mills theory (in topological terms, one can not find a single chart to cover the entire orbit space \cite{singer1978some}). The traditional choice of Lorentz gauge is known to develop finite time coordinate singularities in non-abelian theory contrary, to the linear Maxwell theory where such a breakdown is absent.  In fact, the geometry of the orbit space of the theory (i.e., the space of connections modulo the bundle automorphisms) has a non-trivial effect in the matter of gauge choice \cite{babelon1981riemannian, narasimhan1979geometry}. The positivity of sectional curvature \cite{babelon1981riemannian} of the orbit space leads to the development of the so-called `Gribov horizon'  \cite{moncrief1979gribov} which essentially indicates the breakdown of the so-called `Coulomb' gauge. However, this gauge issue can be avoided since the Yang-Mills equations are manifestly hyperbolic in the double null framework if one works with the fully gauge covariant derivative instead of splitting it into the spacetime covariant derivative part and the pure gauge part \footnote{This gauge-invariant formalism is not known to have applied before}. Therefore, in the double null framework, we work with a manifestly hyperbolic system of coupled Einstein-Yang-Mills Bianchi equations supplemented by the constraints (of elliptic nature) and transport equations. In addition, our analysis does not require a smallness assumption on the initial data. 

The structure of the article is as follows. Starting from the null structure equations and a bootstrap assumption on the connection coefficients, we derive the necessary estimates for the connection coefficients which allows us to estimate the sectional curvature of the topological $2-$spheres throughout the spacetime slab of interest using the null Hamiltonian constraint. This in turn allows us to utilize the null Codazzi equations to obtain elliptic estimates. Utilizing these estimates, we then use a direct integration by parts argument and the null evolution equations for the Weyl curvature and the Yang-Mills curvature to obtain the energy estimates in terms of the initial data thereby closing the bootstrap argument. A few conclusions are drawn based on our result and we make a conjecture about the nonlinear exterior stability of the Minkowski space under Einstein-Yang-Mills perturbations.

\end{section}

\section{Preliminaries}
\subsection{Canonical double null foliation}
\noindent Let $M$ be a $C^{\infty}-$ manifold equipped with a Lorentzian metric $g$. The information contained in Einstein-Yang-Mills equations is captured here through the structure equations, null Bianchi equations, and the null Yang-Mills equations associated with a double null foliation $(u,\ubar)$. The constant $u$ and $\ubar$ hypersurfaces are outgoing and incoming null hypersurfaces, respectively and they intersect at a spacelike topological $2-$sphere that is denoted by $S_{u\ubar}$.  The null hypersurfaces $H$ and $  \Hbar$ of $(M,g)$ are described by the level sets of the optical functions $u$ and $\ubar$, respectively. Assume $u$ and $\ubar$ satisfy the Eikonal equations 
\begin{eqnarray}
g^{\mu\nu}\partial_{\mu}u\partial_{\nu}u=0,~g^{\mu\nu}\partial_{\mu}\ubar\partial_{\nu}\ubar=0,
\end{eqnarray}
where $g^{\mu\nu}:(g^{-1})^{\mu\nu}$.
Through the variation of $u$ and $\ubar$, we can foliate a spacetime slab $\mathcal{D}_{u,\ubar}$ by these two families of null hypersurfaces. The geodesic generators of the double null foliation are the vector fields $L$ and $\Lbar$ given by 
\begin{eqnarray}
L:=-g^{\mu\rho}\partial_{\rho}u\partial_{\mu},~\Lbar:=-g^{\mu\rho}\partial_{\rho}\ubar\partial_{\mu}
\end{eqnarray}
and manifestly they satisfy 
\begin{eqnarray}
\nabla_{L}L=0=\nabla_{\Lbar}\Lbar.
\end{eqnarray}
Whenever, we say $H$ (resp. $  \Hbar$) we will always mean $H_{u}$ (resp. $  \Hbar_{\ubar}$) i.e., level sets of $u$ (resp. $\ubar$). In this notation, $H_{0}$ and $  \Hbar_{0}$ are the two initial null hypersurfaces corresponding to $u=0$ and $\ubar=0$, respectively on which the data is to be provided. Intersection of $H$ and $  \Hbar$ is a topological $2-$ sphere $S_{u,\ubar}$. Evidently, the slab $\mathcal{D}_{u,\ubar}$ (see the picture) is the causal future of $S_{0,0}$ extended up to $u=\epsilon$ and $\ubar=J$ for sufficiently large $J>0$. A point in $(M,g)$ is denoted by $(u,\ubar,\theta^{1},\theta^{2})$. Let $(\theta^{1},\theta^{2})$ be a chart of an open set of the initial sphere $S_{0,0}$. These functions can be extended to define a coordinate chart in an open subset of spacetime. First define $(\theta^{1},\theta^{2})$ on $H_{0}$ by solving 
\begin{eqnarray}
L(\theta^{A})=0,~A=1,2.
\end{eqnarray}
Now we can obtain ($\theta^{1},\theta^{2}$) for $u>0$ by solving 
\begin{eqnarray}
\Lbar(\theta^{A})=0,~ ~A=1,2.   
\end{eqnarray}
In other words, we may first define a coordinate chart $\mathcal{A}$ on $S_{0,0}$ then drag it by the flow of the vector field $L$ along $H_{0}$ and then drag it by the flow of $\Lbar$ to fill out the entire slab $\mathcal{D}_{u,\ubar}$. For another chart $\mathcal{B}$ on $S_{0,0}$, one repeats the procedure. Provided $\mathcal{A}$ and $\mathcal{B}$ cover the sphere $S_{0,0}$, the constructed charts cover the desired open set of spacetime. The spacetime metric in the double null coordinates takes the following form (in a local chart $(u,\ubar,\theta^{1},\theta^{2})$) 
\begin{eqnarray}
g:=-2\Omega^{2}(du\otimes d\ubar+d\ubar\otimes du)+\gamma_{AB}(d\theta^{A}-b^{A}du)\otimes (d\theta^{B}-b^{B}du),
\end{eqnarray}
where $\Omega$ is the null lapse function and $b:=b^{A}\frac{\partial}{\partial \theta^{A}}$ is the null shift vector field. $\{\theta^{A}\}_{A=1}^{2}$ are the coordinates on the topological sphere $S_{u,\ubar}$. The induced metric on $S_{u,\ubar}$ is $\gamma_{AB}$. We can identify a normalized null frame $(e_{4},e_{3},e_{1},e_{2})$ such that $g(e_{4},e_{4})=g(e_{3},e_{3})=0=g(e_{A},e_{4})=g(e_{A},e_{3})=0$ and $g(e_{4},e_{3})=-2$, where $(e_{1},e_{2})$ is an arbitrary frame on $S_{u,\ubar}$. We may identify $e_{4}$ and $e_{3}$ as follows 
\begin{eqnarray}
e_{4}=\Omega^{-1}\frac{\partial}{\partial \ubar},~e_{3}=\Omega^{-1}(\frac{\partial}{\partial u}+b^{A}\frac{\partial}{\partial \theta^{A}}).
\end{eqnarray}
For more detailed information about the double null foliation of a spacetime, see \cite{klainerman2012formation,luk2012local, klainerman2012evolution}\\

\subsection{Yang-Mills Theory}
\noindent Now we define a Yang-Mills theory on $(M,g)$. We denote by $\mathfrak{P}$ a $C^{\infty}$ principal bundle with base a $3+1$ dimensional Lorentzian manifold $M$ and a Lie group $G$. We assume that $G$ is compact semi-simple (for physical purposes) and therefore admits a positive definite non-degenerate bi-invariant metric. Its Lie algebra $\mathfrak{g}$ by construction admits an adjoint invariant, positive definite scalar product denoted by  $\langle~,~\rangle$ which enjoys the property: for $A, B, C\in \mathfrak{g}$,
\begin{eqnarray}
\label{eq:adinvpp}
\langle[A,B],C\rangle=\langle A,[B,C]\rangle.
\end{eqnarray}
as a consequence of adjoint invariance.
A Yang-Mills connection is defined as a $1-$form $w$ on $\mathfrak{P}$ with values in $\mathfrak{g}$ endowed with compatibility properties. Its representative in a local trivialization of $\mathfrak{P}$ over $U\subset M$, 
\begin{eqnarray}
\varphi: p\mapsto (x,a),~p\in \mathfrak{P},~x\in U,~a\in G
\end{eqnarray}
is the $1-$form $s^{*}w$ on $U$, where $s$ is the local section of $\mathfrak{P}$ corresponding canonically to the local trivialization  $s(x)=\varphi^{-1}(x,e)$, called a \textit{gauge}. Let $A_{1}$ and $A_{2}$ be representatives of $w$ in gauges $s_{1}$ and $s_{2}$ over $U_{1}\subset M$ and $U_{2}\in M$. In $U_{1}\cap U_{2}$, one has 
\begin{eqnarray}
\label{eq:gauge}
A_{1}=Ad(u^{-1}_{12})A_{2}+u_{12}\Theta_{MC},
\end{eqnarray}
where $\Theta_{MC}$ is the Maurer-Cartan form on $G$, and $u_{12}:U_{1}\cap U_{2}\to G$ generates the transformation between the two local trivializations: \begin{eqnarray}
s_{1}=R_{u_{12}}s_{2},
\end{eqnarray}
$R_{u_{12}}$ is the right translation on $\mathfrak{P}$ by $u_{12}$. Given the principal bundle $\mathfrak{P}\to M$, a Yang-Mills potential $A$ on $M$ is a section of the fibered tensor product $T^{*}M\otimes_{M}\mathfrak{P}_{Affine,\mathfrak{g}}$ where $\mathfrak{P}_{Affine,\mathfrak{g}}$ is the affine bundle with base $M$ and typical fibre $\mathfrak{g}$, associated to $\mathfrak{P}$ via relation (\ref{eq:gauge}). If $\widehat{A}$ is another Yang-Mills potential on $M$, then $A-\widehat{A}$ is a section of the tensor product of vector bundles $T^{*}M\otimes_{M}P_{Ad,\mathfrak{g}}$, where $\mathfrak{P}_{Ad,\mathfrak{g}}:=\mathfrak{P}\times _{Ad}\mathfrak{g}$ is the vector bundle associated to $\mathfrak{P}$ by the adjoint representation of $G$ on $\mathfrak{g}$. There is an inner product in the fibers of $\mathfrak{P}_{Ad,\mathfrak{g}}$, induced from that on $\mathfrak{g}$. The curvature $\mathbf{C}$ of the connection $w$ considered as a $1-$ form on $\mathfrak{P}$ is a $\mathfrak{g}$-valued $2-$form on $\mathfrak{P}$. Its representative in a gauge where $w$ is represented by $A$ is given by 
\begin{eqnarray}
F:=dA+[A,A],
\end{eqnarray}
and relation between two representatives $F_{1}$ and $F_{2}$ on $U_{1}\cap U_{2}$ is $F_{1}=Ad(u^{-1}_{12})F_{2}$ and therefore $F$ is a section of the vector bundle $\Lambda^{2}T^{*}M\otimes _{M}\mathfrak{P}_{Ad,\mathfrak{g}}$. For a section $\mathfrak{O}$ of the vector bundle $\otimes^{k}T^{*}M\otimes _{M}\mathfrak{P}_{Ad,\mathfrak{g}}$, a natural covariant derivative is defined as follows 
\begin{eqnarray}
\label{eq:covariant}
\widehat{D}\mathfrak{O}:=D\mathfrak{O}+[A,\mathfrak{O}],
\end{eqnarray}
where $D$ is the usual covariant derivative induced by the Lorentzian structure of $M$ and by construction $\widehat{D}\mathfrak{O}$ is a section of the vector bundle $\otimes^{k+1}T^{*}M\otimes _{M}\mathfrak{P}_{Ad,\mathfrak{g}}$. The Yang-Mills coupling constant $g_{YM}$ is set to 1. If the space of connections in a particular Sobolev class is denoted by $\mathcal{A}$ and $\mathcal{G}$ is the automorphism group of the bundle $\mathfrak{P}$, then the true configuration space (or the orbit space) of the theory is $\mathcal{A}/\mathcal{G}$\footnote{$\mathcal{A}/\mathcal{G}$ is an infinite dimensional manifold modulo certain additional criteria}. 

\noindent The classical Yang-Mills equations (in the absence of sources)
correspond to setting the natural (spacetime and gauge as defined in \ref{eq:covariant}) covariant divergence of this curvature
two-form $F$ to zero. By virtue of its definition in terms of the connection, this curvature also satisfies
the Bianchi identity that asserts the vanishing of its gauge covariant exterior derivative. Taken
together, these equations provide a geometrically natural nonlinear generalization of Maxwell's
equations (when the latter are written in terms of a `vector potential') and of course, play a
fundamental role in modern elementary particle physics. If nontrivial bundles are considered
or nontrivial spacetime topologies are involved, then the foregoing so-called `local trivializations' of the bundles in question must be patched together to give global descriptions but, by virtue of the covariance of the formalism, there is a natural way of carrying out this patching procedure, at least over those regions of spacetime where the connections are well-defined. 

We choose a vector space $V$ and a matrix representation for the action of $G$ on $V$. For simplicity, let us confine our attention to real representations though, in fact, this restriction is inessential. We now consider vector bundles over spacetime (so-called ‘associated’ bundles) with standard fiber $\simeq V$. Cross sections of such bundles would represent, in physical terms, multiplets of Higgs fields. To formulate field equations for such Higgs fields that are naturally covariant with respect to automorphisms of the associated vector bundle (i.e., with respect to gauge transformations acting on the Higgs fields) one needs a covariant derivative operator $\widehat{D}$ or connection defined on this bundle. Such an object is naturally induced from a ‘fundamental’ connection on the principal $G$-bundle described above and in turn, induces (when expressed relative to a local trivialization) a one-form on (some local chart for) the base manifold with values in the chosen matrix representation for the Lie algebra $\mathfrak{g}$. Let us consider the dimension of the group $G$ to be $\text{dim}_{G}$ and since $\mathfrak{g}:=T_{e}G$, it has a natural vector space structure. Assume that the vector space $\mathfrak{g}$ has a basis $\{\chi_{A}\}_{A=1}^{\text{dim}_{G}}$ given by a set of $k\times k$ real valued matrices ($k$ being the dimension of the representation $V$ of the Lie algebra $\mathfrak{g}$). The connection $1-$form field is then defined to be 
\begin{eqnarray}
A:=A^{A}_{\mu}\chi_{A}dx^{\mu}=A^{A}_{\mu}(\chi_{A})^{P}_{Q}dx^{\mu}=A^{P}~_{Q\mu}dx^{\mu},~P,Q=1,2,3,...,k.
\end{eqnarray}
From now on by the connection 1-form field $A_{\mu}$, we will always mean $A^{P}~_{Q\mu}$. In the current setting $A\in \Omega^{1}(M;\text{End}(V))$, where $\text{End}(V)$ denotes the space of endomorphisms of the vector space $V$. The curvature of this connection is defined to be the Yang-Mills field $F\in \Omega^{2}(M;\text{End}(V))$
\begin{eqnarray}
F^{P}~_{Q\mu\nu}:=\partial_{\mu}A^{P}~_{Q\nu}-\partial_{\nu}A^{P}~_{Q\mu}+[A,A]^{P}~_{Q\mu\nu},
\end{eqnarray}
where the bracket is defined on the Lie algebra $\mathfrak{g}$ and given by the commutator of matrices under multiplication. The Yang-Mills coupling constant is set to unity. Since $G$ is compact, it admits a positive definite adjoint invariant metric on $\mathfrak{g}$. We choose a basis of $\mathfrak{g}$ such that this adjoint invariant metric takes the Cartesian form $\delta_{AB}$ and work with representations for which the bases satisfy 
\begin{eqnarray}
-\text{tr}(\chi_{A}\chi_{B})=(\chi_{A})^{P}_{Q}(\chi_{B})^{Q}_{P}=\delta_{AB}.
\end{eqnarray}

\noindent Under a gauge transformation by $\mathcal{U}$, the $\mathfrak{g}$-valued 1-form field $A$ transforms as 
\begin{eqnarray}
A_{\mu}\mapsto \mathcal{U}^{-1}A_{\mu}\mathcal{U}+\mathcal{U}\partial_{\mu}\mathcal{U}^{-1} 
\end{eqnarray}
and therefore $A_{\mu}$ is not a tensor in the sense that it is \textit{not} a $(1,1)$ section of the associated $V-$bundle over $\mathcal{D}_{u,\ubar}$. For any $\mathfrak{g}$-valued section $\mathcal{K}^{P}~_{Q\mu_{1}\mu_{2}\mu_{3}....\mu_{k}}$ of a vector bundle over $\mathcal{D}_{u,\ubar}$ that transforms as a tensor under the gauge transformation, the gauge covariant derivative is defined to be  
\begin{align}
\label{eq:gaugecov}
\nonumber &\widehat{D}_{\alpha}\mathcal{K}^{P}~_{Q\mu_{1}\mu_{2}\mu_{3}....\mu_{k}}
:= D_{\alpha}\mathcal{K}^{P}~_{Q\mu_{1}\mu_{2}\mu_{3}....\mu_{k}}+A^{P}~_{R\alpha}\mathcal{K}^{R}~_{Q\mu_{1}\mu_{2}\mu_{3}....\mu_{k}}-A^{R}~_{Q\alpha}\mathcal{K}^{P}~_{R\mu_{1}\mu_{2}\mu_{3}....\mu_{k}}\\
=& D_{\alpha}\mathcal{K}^{P}~_{Q\mu_{1}\mu_{2}\mu_{3}....\mu_{k}}+[A
,\mathcal{K}]^{P}~_{Q\mu_{1}\mu_{2}\mu_{3}....\mu_{k}},
\end{align}
where $D_{\alpha}$ is the ordinary spacetime covariant derivative with respect to a Lorentzian metric on $\mathcal{D}_{u,\ubar}$. Even though the connection of the gauge bundle appears in the definition of the gauge covariant derivative, we will never make explicit use of it in the current context but rather work with the fully gauge covariant derivative $\widehat{D}$. More specifically, in our analysis, we  will encounter the commutator of the fully gauge covariant derivative which yields Riemann curvature and Yang-Mills curvature components. In other words, using the fully gauge covariant derivative, we do not encounter any connection terms, allowing us to obtain estimates in a gauge-invariant way. The commutator of the fully gauge covariant derivative while acting on a $\mathfrak{g}$-valued section of a vector bundle $\mathcal{K}^{P}~_{Q\mu_{1}\mu_{2}\mu_{3}....\mu_{k}}$ (or a section of the mixed bundle) yields 
\begin{align} \nonumber 
[\widehat{D}_{\alpha},\widehat{D}_{\beta}]\mathcal{K}^{P}~_{Q\mu_{1}\mu_{2}\mu_{3}....\mu_{k}}=& F^{P}~_{R\alpha\beta}\mathcal{K}^{R}~_{Q\mu_{1}\mu_{2}\mu_{3}....\mu_{k}}-F^{R}~_{Q\alpha\beta}\mathcal{K}^{P}~_{R\mu_{1}\mu_{2}\mu_{3}....\mu_{k}}\\ -&\sum_{i}R^{\gamma}~_{\mu_{i}\alpha\beta}\mathcal{K}^{P}~_{Q\mu_{1}\mu_{2}\mu_{3}...\hat{\gamma}....\mu_{k}},
\end{align}
where $\hat{\gamma}$ indicates the removal of the index $\hat{\mu}_{i}$ and replacing by $\gamma$. Note that $\widehat{D}$ is compatible with both the metrics and therefore the commutator produces curvature of the mixed bundle (Yang-Mills curvature and spacetime curvature). The action of $\widehat{D}$ is only well-defined on \textit{sections} of the mixed bundle.

\begin{center}
\begin{figure}
\begin{center}
\includegraphics[width=13cm,height=60cm,keepaspectratio,keepaspectratio]{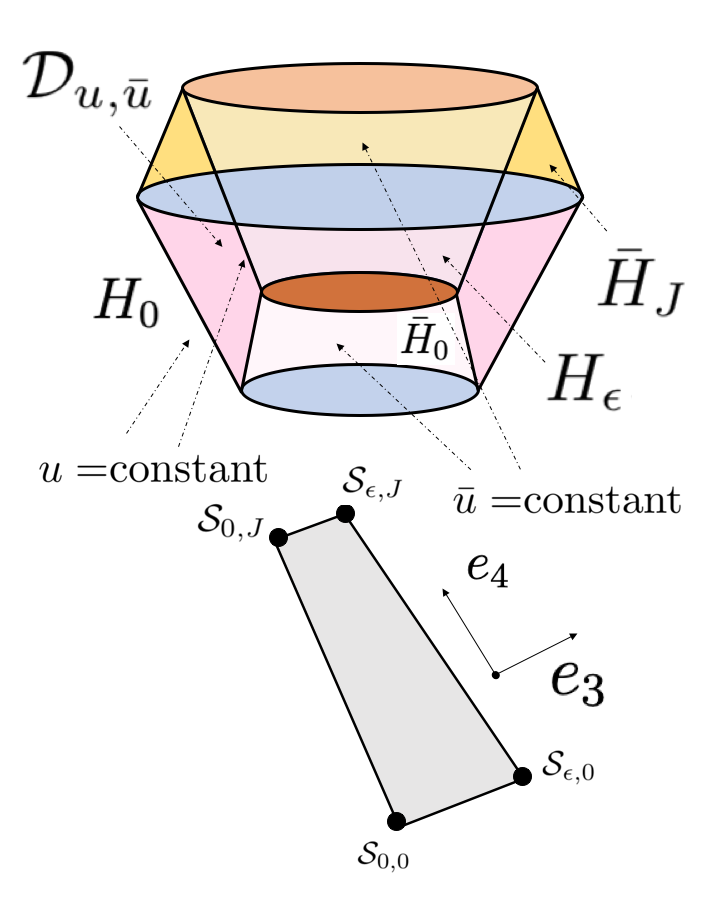}
\end{center}
\begin{center}
\caption{Figure depicting the double null foliation adapted to the current framework. The bottom picture is constructed after taking a formal quotient of the top figure by the topological $2-$spheres i.e., each point of the bottom figure denotes a topological $2-$sphere. The theorem of \cite{rendall1992characteristic} deals with the existence of a solution to the characteristic initial value problem in an $\epsilon$ neighborhood of $S_{0,0}$. Note $\mathcal{D}_{u,\ubar}$ denotes the bulk spacetime region that is causal future of $S_{0,0}$ up to $u=\epsilon, \ubar=J$.}
\label{fig:doublenull}
\end{center}
\end{figure}
\end{center}

\section{Field equations in the double null framework}
\subsection{Bianchi and Yang-Mills equations}
\noindent In this section, we explicitly define all the entities associated with the double null foliation and write down structure equations. The spacetime covariant derivative $D$ admits the following usual decomposition in terms of its components parallel and orthogonal to the topological $2-$sphere $S_{u,\ubar}$
\begin{eqnarray}
D_{e_{a}}e_{b}=\nabla_{e_{a}}e_{b}-\frac{1}{2}\langle\mathcal{D}_{e_{a}}e_{b},e_{3}\rangle e_{4}-\frac{1}{2}\langle\mathcal{D}_{e_{a}}e_{b},e_{4}\rangle e_{3},
\end{eqnarray}
where $\nabla$ is the $S_{u,\ubar}$-parallel covariant derivative. Similar to the spacetime covariant derivative, the spacetime gauge covariant derivative $\widehat{D}$ admits the following decomposition
\begin{eqnarray}
\widehat{D}_{e_{a}}e_{b}=\hnabla_{e_{a}}e_{b}-\frac{1}{2}\langle\widehat{D}_{e_{a}}e_{b},e_{3}\rangle e_{4}-\frac{1}{2}\langle\widehat{D}_{e_{a}}e_{b},e_{4}\rangle e_{3}\\\nonumber 
=\nabla_{e_{a}}e_{b}-\frac{1}{2}\langle\mathcal{D}_{e_{a}}e_{b},e_{3}\rangle e_{4}-\frac{1}{2}\langle\mathcal{D}_{e_{a}}e_{b},e_{4}\rangle e_{3}
\end{eqnarray}
simply because the basis $(e_{3},e_{4},e_{a})_{a=1}^{2}$ are not sections of the gauge bundle and therefore gauge covariant derivative acts as ordinary covariant derivative. Now recalling the following definitions of the outgoing and incoming null second fundamental form of $S_{u,v}$
\begin{eqnarray}
\chi_{ab}:=\langle D_{e_{a}}e_{4},e_{b}\rangle,~~~~\bar{\chi}_{ab}:=\langle D_{e_{a}}e_{3},e_{b}\rangle,
\end{eqnarray}
the spacetime covariant (also gauge covariant) derivative satisfies
\begin{eqnarray}
D_{e_{a}}e_{b}=\nabla_{e_{a}}e_{b}+\frac{1}{2}\bar{\chi}_{ab}e_{4}+\frac{1}{2}\chi_{ab}e_{3}.
\end{eqnarray}
Now let $K\in sections\{TS_{u,\ubar}\}$ and $Z$ be a $\mathfrak{g}-$valued frame vector field on $S_{u,\ubar}$ that transforms as a tensor under a gauge transformation, then using the definition of the gauge covariant derivative (\ref{eq:gauge}) and its metric compatibility property, the following holds for the gauge covariant derivative
\begin{eqnarray}
\widehat{D}_{K}Z=\hnabla_{K}Z+\frac{1}{2}\bar{\chi}(K,Z)e_{4}+\frac{1}{2}\chi(K,Z)e_{3},
\end{eqnarray}
where $\chi(K,Z):=\chi_{ab}K_{a}Z^{P}~_{Qb}$ and $\bar{\chi}(K,Z):=\chi_{ab}K_{a}Z^{P}~_{Qb}$.
Now we recall the definitions of the remaining connection coefficients adapted to the double null framework 
\begin{eqnarray}
\eta_{a}:=-\frac{1}{2}\langle D_{e_{3}}e_{a},e_{4}\rangle,~\omega:=-\frac{1}{4}\langle D_{e_{4}}e_{3},e_{4}\rangle=-\frac{1}{2}D_{e_{4}}\ln\Omega,\\\nonumber 
\etabar_{a}:=-\frac{1}{2}\langle D_{e_{4}}e_{a},e_{3}\rangle,~\omegabar:=-\frac{1}{4}\langle D_{e_{3}}e_{4},e_{3}\rangle=-\frac{1}{2}D_{e_{3}}\ln\Omega
\end{eqnarray}
and the torsion $\zeta_{a}:=\frac{1}{2}\langle D_{e_{a}}e_{4},e_{3}\rangle=\frac{1}{2}(\eta-\etabar)$.
Utilizing these definitions, let us write down the kinematical set of structural equations 
\begin{eqnarray}
D_{e_{a}}e_{3}=\bar{\chi}_{ab}e_{b}+\zeta_{a}e_{3},~D_{e_{3}}e_{a}=\nabla_{e_{3}}e_{a}+\eta_{a}e_{3},\\
D_{e_{4}}e_{a}=\nabla_{e_{4}}e_{a}+\etabar_{a}e_{4},~\nabla_{e_{3}}e_{3}=-2\omegabar e_{3},\\
D_{e_{4}}e_{4}=-2\omega e_{4},~\nabla_{e_{4}}e_{3}=2\omega e_{3}+2\etabar_{a}e_{a},\\
D_{e_{3}}e_{4}=2\omegabar e_{4}+2\eta_{a}e_{a},~
D_{e_{4}}e_{a}=\nabla_{e_{4}}e_{a}+\etabar_{a}e_{4},\\
D_{e_{a}}e_{4}=\chi_{ab}e_{b}-\zeta_{a}e_{4}.
\end{eqnarray}
Also note $\eta_{a}=\zeta_{a}+\nabla_{e_{a}}\ln\Omega,~\etabar_{a}=-\zeta_{a}+\nabla_{e_{a}}\ln\Omega$. Utilizing these kinematical structure equations, we obtain the dynamical set of structural equations suitable trace of which are nothing but Einstein's equations sourced by Yang-Mills stress-energy tensor and expressed in the double null framework. Before writing down such equations, let us recall the following decomposition of the spacetime Riemann curvature tensor
\begin{eqnarray}
\label{eq:weyl}
R_{\alpha\beta\gamma\delta}=W_{\alpha\beta\gamma\delta}+\frac{1}{2}(g_{\alpha\gamma}R_{\beta\delta}+g_{\beta\delta}R_{\alpha\gamma}-g_{\beta\gamma}R_{\alpha\delta}-g_{\alpha\delta}R_{\beta\gamma})+\frac{1}{6}R(g_{\alpha\delta}g_{\beta\gamma}-g_{\alpha\gamma}g_{\beta\delta}),
\end{eqnarray}
where $W$ is the Weyl curvature tensor describing pure gravitational degrees of freedom. $W$ is trace-free and enjoys the same algebraic symmetry as the Riemann curvature. Since the Ricci and the scalar curvatures are fixed by Einstein's field equations, we ought to write down the equation for the null components of the Weyl curvature. These equations will constitute the null Bianchi equations. The components of the Weyl curvature are defined as follows 
\begin{eqnarray}
\alpha_{ab}:=W(e_{a},e_{4},e_{b},e_{4}),~\bar{\alpha}_{ab}:=W(e_{a},e_{3},e_{b},e_{3}),\\
2\beta_{a}:=W(e_{4},e_{3},e_{4},e_{a}),~
2\bar{\beta}_{a}:=W(e_{3},e_{4},e_{3},e_{a}),\\
\rho:=\frac{1}{4}W(e_{4},e_{3},e_{4},e_{3}),
~\sigma:=\frac{1}{4}~^{*}W(e_{4},e_{3},e_{4},e_{3}),
\end{eqnarray}
where $^{*}W$ is the left Hodge dual of $W$ defined as follows 
\begin{eqnarray}
^{*}W_{\alpha\beta\gamma\delta}:=\frac{1}{2}\epsilon_{\alpha\beta\mu\nu}W^{\mu\nu}~_{\gamma\delta},
\end{eqnarray}
where $\epsilon_{\alpha\beta\mu\nu}$ is the volume form on the spacetime $M$. In addition to the Weyl field, we also have the Yang-Mills curvature $F:=\frac{1}{2}F^{P}~_{Q\mu\nu}dx^{\mu}\wedge dx^{\nu}\in \Omega^{2}(M;\text{End}(V)),~P, Q=1,2,3,....,\text{dim}(V)$. The components of the Yang-Mills curvature $F$ are defined as follows 
\begin{eqnarray}
\alpha^{F}_{a}:=F^{P}~_{Q}(e_{a},e_{4}),~\bar{\alpha}^{F}_{a}:=F^{P}~_{Q}(e_{a},e_{3}),~\rho^{F}:=\frac{1}{2}F^{P}~_{Q}(e_{3},e_{4}),\\\nonumber 
\sigma^{F}=\frac{1}{2}~^{*}F^{P}~_{Q}(e_{3},e_{4})=F^{P}~_{Q}(e_{1},e_{2}).
\end{eqnarray}
Also decompose the null second fundamental forms $\chi,\bar{\chi}$ into their trace and trace-free components
\begin{eqnarray}
\chi=\widehat{\chi}+\frac{1}{2}\ \tr\chi \gamma,~\bar{\chi}=\chibarhat+\frac{1}{2}\ \tr\chibar \gamma.
\end{eqnarray}
Einstein's equations (with the choice of unit $8\pi G=c=1$)
\begin{eqnarray}
R_{\mu\nu}-\frac{1}{2}Rg_{\mu\nu}=\mathfrak{T}_{\mu\nu}
\end{eqnarray}
in the double null framework reads 
\begin{eqnarray}
\label{eq:connection1}
\nabla_{4} \tr\chi+\frac{1}{2}( \tr\chi)^{2}=-|\widehat{\chi}|^{2}_{\gamma}-2\omega  \tr\chi-\mathfrak{T}_{44}\\
\nabla_{4}\widehat{\chi}+ \tr\chi \widehat{\chi}=-2\omega\widehat{\chi}-\alpha\\
\nabla_{3} \tr\chibar+\frac{1}{2}( \tr\chibar)^{2}=-|\chibarhat|^{2}_{\gamma}-2\omegabar \tr\chibar-\mathfrak{T}_{33}\\
\nabla_{3}\chibarhat+ \tr\chibar\chibarhat=-2\omegabar\chibarhat-\bar{\alpha}\\
\label{eq:eta}
\nabla_{4}\eta_{a}=-\chi\cdot(\eta-\etabar)-\beta-\frac{1}{2}\mathfrak{T}_{a4}\\
\nabla_{3}\etabar_{a}=-\bar{\chi}\cdot(\etabar-\eta)+\bar{\beta}+\frac{1}{2}\mathfrak{T}_{a3}\\
\nabla_{4}\omegabar=2\omega\omegabar+\frac{3}{4}|\eta-\etabar|^{2}-\frac{1}{4}(\eta-\etabar)\cdot(\eta+\etabar)-\frac{1}{8}|\eta+\etabar|^{2}\nonumber+\frac{1}{2}\rho+\frac{1}{4}\mathfrak{T}_{43}\\
\nabla_{3}\omega=2\omega\omegabar+\frac{3}{4}|\eta-\etabar|^{2}+\frac{1}{4}(\eta-\etabar)\cdot(\eta+\etabar)-\frac{1}{8}|\eta+\etabar|^{2}\nonumber+\frac{1}{2}\rho+\frac{1}{4}\mathfrak{T}_{43}\\
\nabla_{4} \tr\chibar+\frac{1}{2} \tr\chi  \tr\chibar=2\omega  \tr\chibar+2\text{div}\etabar+2|\etabar|^{2}_{\gamma}+2\rho-\widehat{\chi}\cdot\chibarhat\\
\nabla_{3} \tr\chi+\frac{1}{2} \tr\chibar \tr\chi=2\omegabar \tr\chi+2\text{div}\eta+2|\eta|^{2}+2\rho-\widehat{\chi}\cdot \chibarhat\\
\nabla_{4}\chibarhat+\frac{1}{2} \tr\chi\chibarhat=\nabla\hat{\otimes}\etabar+2\omega\chibarhat-\frac{1}{2} \tr\chibar\widehat{\chi}+\etabar\hat{\otimes}\etabar+\hat{\mathfrak{T}}_{ab}\\
\nabla_{3}\widehat{\chi}+\frac{1}{2} \tr\chibar\widehat{\chi}=\nabla\hat{\otimes}\eta+2\omegabar\widehat{\chi}-\frac{1}{2} \tr\chi\chibarhat+\eta\hat{\otimes}\eta+\hat{\mathfrak{T}}_{ab}\\
\label{eq:1}
\text{div}\widehat{\chi}=\frac{1}{2}\nabla  \tr\chi-\frac{1}{2}(\eta-\etabar)\cdot(\widehat{\chi}-\frac{1}{2} \tr\chi\gamma_{ab})-\beta+\frac{1}{2}\mathfrak{T}(e_{4},\cdot)\\
\text{div}\chibarhat=\frac{1}{2}\nabla  \tr\chibar-\frac{1}{2}(\etabar-\eta)\cdot(\chibarhat-\frac{1}{2} \tr\chibar\gamma_{ab})-\bar{\beta}+\frac{1}{2}\mathfrak{T}(e_{3},\cdot)\\
\text{curl}\eta=\chibarhat\wedge\widehat{\chi}+\sigma\epsilon=-\text{curl} \etabar\\
\label{eq:4}
K-\frac{1}{2}\widehat{\chi}\cdot \chibarhat+\frac{1}{4} \tr\chi  \tr\chibar=-\rho+\frac{1}{4}\mathfrak{T}_{43},
\end{eqnarray}
where $\mathfrak{T}_{\mu\nu}:=\frac{1}{2}\left(F^{P}~_{Q\mu\alpha}F^{Q}~_{P\nu}~^{\alpha}+^{*}F^{P}~_{Q\mu\alpha}~^{*}F^{Q}~_{P\nu}~^{\alpha}\right)$ is the Yang-Mills stress-energy tensor and $K$ is the sectional curvature of the topological $2-$sphere $S_{u,\ubar}$ (or a constant multiple of Gauss curvature). Here equations (\ref{eq:1}-\ref{eq:4}) are the null-constraint equations. Now recall the Bianchi identities 
\begin{eqnarray}
D_{\mu}R_{\alpha\beta\nu\lambda}+D_{\nu}R_{\alpha\beta\lambda\mu}+D_{\lambda}R_{\alpha\beta\mu\nu}=0
\end{eqnarray}
which together with the decomposition (\ref{eq:weyl}) and Einstein's equations yields the following Yang-Mills type equations for the Weyl curvature 
\begin{eqnarray}
D^{\alpha}W_{\alpha\beta\gamma\delta}=J[\mathfrak{T}]_{\beta\gamma\delta},~D_{[\mu}W_{\gamma\delta]\alpha\beta}=\frac{1}{3}\epsilon_{\nu\mu\gamma\delta}J[\mathfrak{T}]^{*\nu}~_{\alpha\beta}.
\end{eqnarray}
Here $J[\mathfrak{T}]$ is the source term determined fully by the Yang-Mills stress-energy tensor $\mathfrak{T}$. After elementary algebraic manipulations, the differential equations for the Weyl curvature may be cast into the following double-null form \footnote{$A\hat{\otimes}B:=(A\otimes B+B\otimes A-A\cdot B\gamma)$,~$(A\wedge B):=\epsilon^{ac}\gamma^{bd}A_{ab}B_{cd}$,~$(\text{curl} A)_{a_{1}\cdot\cdot\cdot a_{n}}:=\epsilon^{cd}\nabla_{c}A_{da_{1}\cdot\cdot\cdot a_{n}}$}
\begin{eqnarray}
\label{eq:bianchi1}
\nabla_{3}\alpha_{ab}+\frac{1}{2} \tr\chibar\alpha_{ab}=(\nabla\hat{\otimes}\beta)_{ab}+4\omegabar\alpha_{ab}-3(\widehat{\chi}\rho+~^{*}\widehat{\chi}\sigma)_{ab}+((\zeta+4\eta)\hat{\otimes}\beta)_{ab}\\\nonumber -D_4 R_{AB}+\frac{1}{2}(D_{3}R_{44}-D_{4}R_{43})\gamma_{ab}+\frac{1}{2}\left(D_b R_{4a} + D_a R_{4b}\right)\\
\nabla_{4}\beta_{a}+2 \tr\chi\beta_{a}=(\text{div}\alpha)_{a}-2\omega\beta_{a}+(\eta\cdot\alpha)_{a}-\frac{1}{2}(D_{a}R_{44}-D_{4}R_{4a})\\
\nabla_{3}\beta_{a}+ \tr\chibar\beta_{a}=\nabla_{a}\rho+~^{*}\nabla_{a}\sigma+2\omegabar\beta_{a}+2(\widehat{\chi}\cdot \bar{\beta})_{a}+3(\eta\rho+~^{*}\eta\sigma)_{a}\\\nonumber+\frac{1}{2}(D_{a}R_{34}-D_{4}R_{3a}),\\
\nabla_{4}\sigma+\frac{3}{2} \tr\chi\sigma=-\text{div}~^{*}\beta+\frac{1}{2}\chibarhat\cdot~^{*}\alpha-\zeta\cdot~^{*}\beta-2\etabar\cdot~^{*}\beta\\\nonumber-\frac{1}{4}(D_{\mu}R_{4\nu}-D_{\nu}R_{4\mu})\epsilon^{\mu\nu}~_{34}\\
\nabla_{3}\sigma+\frac{3}{2} \tr\chibar\sigma=-\text{div}~^{*}\bar{\beta}+\frac{1}{2}\widehat{\chi}\cdot~^{*}\bar{\alpha}-\zeta\cdot~^{*}\bar{\beta}-2\eta\cdot~^{*}\bar{\beta}\\\nonumber+\frac{1}{4}(D_{\mu}R_{3\nu}-D_{\nu}R_{3\mu})\epsilon^{\mu\nu}~_{34}\\
\nabla_{4}\rho+\frac{3}{2} \tr\chi\rho=\text{div}\beta-\frac{1}{2}\chibarhat\cdot\alpha+\zeta\cdot\beta+2\etabar\cdot\beta\\\nonumber -\frac{1}{4}(D_{3}R_{44}-D_{4}R_{34})\\
\nabla_{3}\rho+\frac{3}{2} \tr\chibar\rho=-\text{div}\bar{\beta}-\frac{1}{2}\widehat{\chi}\cdot\bar{\alpha}+\zeta\cdot\bar{\beta}-2\eta\cdot\bar{\beta}\\\nonumber+\frac{1}{4}(D_{3}R_{34}-D_{4}R_{33})\\
\nabla_{4}\bar{\beta}_{a}+ \tr\chi\bar{\beta}_{a}=-\nabla_{a}\rho+~^{*}\nabla_{a}\sigma+2\omega\bar{\beta}_{a}+2(\chibarhat\cdot\beta)_{a}-3(\etabar\rho-~^{*}\etabar\sigma)_{a}\\\nonumber-\frac{1}{2}(D_{a}R_{43}-D_{3}R_{4a})\\
\nabla_{3}\bar{\beta}_{a}+2 \tr\chibar\bar{\beta}_{a}=-(\text{div}\bar{\alpha})_{a}-2\omegabar\bar{\beta}_{a}+(\etabar\cdot\bar{\alpha})_{a}\\\nonumber +\frac{1}{2}(D_{a}R_{33}-D_{3}R_{3a})\\
\label{eq:bianchi2}
\nabla_{4}\bar{\alpha}_{ab}+\frac{1}{2} \tr\chi\bar{\alpha}=-(\nabla\hat{\otimes}\bar{\beta})_{ab}+4\omega\bar{\alpha}_{ab}-3(\chibarhat_{ab}\rho-~^{*}\chibarhat_{ab}\sigma)+((\zeta-4\etabar)\hat{\otimes}\bar{\beta})_{ab}\\\nonumber- D_3 R_{AB} +\frac{1}{2}(D_{4}R_{33}-D_{3}R_{34})\gamma_{ab}+\frac{1}{2}\left(D_a R_{3b} + D_b R_{3a}\right),
\end{eqnarray}
where $R_{\mu\nu}=\mathfrak{T}_{\mu\nu}$ due to the trace-free property of the Yang-Mills stress-energy tensor. The Yang-Mills equations 
\begin{eqnarray}
\widehat{D}_{\mu}F^{P}~_{Q\nu\lambda}+\widehat{D}_{\nu}F^{P}~_{Q\lambda\mu}+\widehat{D}_{\lambda}F^{P}~_{Q\mu\nu}=0,~g^{\alpha\beta}\widehat{D}_{\alpha}F^{P}~_{Q\beta\mu}=0
\end{eqnarray}
imply the following double null Yang-Mills equations 
\begin{eqnarray}
\label{eq:YM1}
\hnabla_{4}\bar{\alpha}^{F}+\frac{1}{2} \tr\chi\bar{\alpha}^{F}=-\hnabla\rho^{F}-~^{*}\hnabla\sigma^{F}-2~^{*}\etabar\sigma^{F}-2\etabar\rho^{F}+2\omega\bar{\alpha}^{F}-\chibarhat\cdot\alpha^{F}\\
\hnabla_{3}\alpha^{F}+\frac{1}{2} \tr\chibar\alpha^{F}=-\hnabla\rho^{F}+~^{*}\hnabla\sigma^{F}-2~^{*}\eta\sigma^{F}+2\eta\rho^{F}+2\omegabar\alpha^{F}-\widehat{\chi}\cdot \bar{\alpha}^{F}\\
\hnabla_{4}\rho^{F}=-\widehat{\text{div}} \alpha^{F}- \tr\chi\rho^{F}-(\eta-\etabar)\cdot\alpha^{F}\\
\hnabla_{4}\sigma^{F}=-\widehat{\text{curl}} \alpha^{F}- \tr\chi \sigma^{F}+(\eta-\etabar)\cdot ~^{*}\alpha^{F}\\
\hnabla_{3}\rho^{F}=-\widehat{\text{div}} \bar{\alpha}^{F}- \tr\chibar\rho^{F}+(\eta-\etabar)\cdot\bar{\alpha}^{F}\\
\label{eq:YM2}
\hnabla_{3}\sigma^{F}=-\widehat{\text{curl}} \bar{\alpha}^{F}- \tr\chibar\sigma^{F}+(\eta-\etabar)\cdot~^{*}\bar{\alpha}^{F},
\end{eqnarray}
where $\hnabla$ is the horizontal (tangential to $S_{u,\ubar}$) component of the spacetime gauge covariant derivative $\widehat{D}$. We have the following lemma regarding the properties of the null Bianchi and null Yang-Mills equations. \noindent We may write down explicitly the null components of the Yang-Mills stress-energy tensor
\begin{eqnarray}
\mathfrak{T}_{43}=\rho^{F}\cdot\rho^{F}+\sigma^{F}\cdot\sigma^{F},~\mathfrak{T}_{4A}=\alpha^{F}_{A}\cdot\rho^{F}+\epsilon_{A}~^{B}\alpha^{F}_{B}\cdot\sigma^{F},\\
\mathfrak{T}_{3A}=-\bar{\alpha}^{F}_{A}\cdot\rho^{F}+\epsilon_{A}~^{B}\bar{\alpha}^{F}_{B}\cdot\sigma^{F},~\mathfrak{T}_{44}=\alpha^{F}_{A}\cdot\alpha^{F}_{B}\gamma^{AB},~\mathfrak{T}_{33}=\bar{\alpha}^{F}_{A}\cdot\bar{\alpha}^{F}_{B}\gamma^{AB},\\
\mathfrak{T}_{ab}=\frac{1}{2}(\rho^{F}\cdot\rho^{F}+\sigma^{F}\cdot\sigma^{F})\gamma_{AB}-(\alpha^{F}_{A}\cdot \bar{\alpha}^{F}_{B}+\bar{\alpha}^{F}_{A}\cdot \alpha^{F}_{B}-\alpha^{F}_{C}\cdot\bar{\alpha}^{F}_{D}\gamma^{CD}\gamma_{AB}),
\end{eqnarray}
where $\cdot$ indicates the inner product on the fibers of the associated gauge bundle.\\ 
\begin{proposition}
\textit{The null Bianchi equations (\ref{eq:bianchi1}-\ref{eq:bianchi2}) are manifestly hyperbolic.}
\end{proposition}
\begin{proof}
The proof is a simple consequence of the existence of the Bel-Robinson tensor for the Weyl curvature
\begin{eqnarray}
Q_{\alpha\beta\gamma\delta}:=W_{\alpha\rho\gamma\sigma}W_{\beta}~^{\rho}~_{\delta}~^{\sigma}+~^{*}W_{\alpha\rho\gamma\sigma}~^{*}W_{\beta}~^{\rho}~_{\delta}~^{\sigma}
\end{eqnarray}
Indeed one may explicitly obtain energy identities for the Weyl curvature energy. For a future directed unit time-like vector field $n=\frac{1}{2}(e_{3}+e_{4})$, construct the current 
\begin{eqnarray}
\mathfrak{C}:=Q(n,n,n,\cdot).
\end{eqnarray}
Integrating the divergence of $\mathfrak{C}$ over the spacetime domain $\mathcal{D}_{u,\ubar}$ yields 
\begin{eqnarray}
\int_{\mathcal{D}_{u,\ubar}}\mathcal{D}_{\mu}\mathfrak{C}^{\mu}=\int_{H}\mathfrak{C}(e_{4})+\int_{  \Hbar_{\ubar}}\mathfrak{C}(e_{3})-\int_{H_{0}}\mathfrak{C}(e_{4})-\int_{  \Hbar_{\bar{0}}}\mathfrak{C}(e_{3}),
\end{eqnarray}
where $\mathcal{D}_{\mu}\mathfrak{C}^{\mu}$ is algebraic in the Weyl curvature $W$
\begin{eqnarray}
\mathcal{D}_{\mu}\mathfrak{C}^{\mu}=\mathcal{D}^{\mu}(Q_{\mu\beta\gamma\delta}n^{\beta}n^{\gamma}n^{\delta})=(W_{\beta}~^{\mu}~_{\delta}~^{\nu}J(\mathfrak{T})_{\mu\gamma\nu}+W_{\beta}~^{\mu}~_{\gamma}~^{\nu}J(\mathfrak{T})_{\mu\delta\nu}\\\nonumber 
+~^{*}W_{\beta}~^{\mu}~_{\delta}~^{\nu}~^{*}J(\mathfrak{T})_{\mu\gamma\nu}+~^{*}W_{\beta}~^{\mu}~_{\gamma}~^{\nu}~^{*}J(\mathfrak{T})_{\mu\delta\nu})+3Q_{\mu\beta\gamma\delta}\mathcal{D}^{\mu}n^{\beta}n^{\gamma}n^{\delta}.
\end{eqnarray}
Explicit calculations show 
\begin{eqnarray}
\mathcal{C}(e_{4})=Q(W)(n,n,n,e_{4})
\approx\left(|\alpha|^{2}_{\gamma}+|\bar{\beta}|^{2}_{\gamma}+|\beta|^{2}_{\gamma}+\rho^{2}+\sigma^{2}\right),\\
\mathcal{C}(e_{3})=Q(W)(n,n,n,e_{3})\approx\left(|\alpha|^{2}_{\gamma}+|\bar{\beta}|^{2}_{\gamma}+|\beta|^{2}_{\gamma}+\rho^{2}+\sigma^{2}\right),
\end{eqnarray}
where the involved constants are purely numerical positive constants. This completes the proof.
\end{proof}
\noindent For a gauge covariant system such as Yang-Mills, we can prove symmetric hyperbolic characteristics by means of integration by parts argument. First, we prove an integration by parts lemma. 
\begin{lemma}
\label{integration}
\textit{Let $f:\mathcal{M}\to\mathbb
R$ be a gauge-invariant object on spacetime. The following integration by parts identities holds for $f$:
\begin{eqnarray}
\int_{D_{u,\underline{u}}}\nabla_{4}f=\int_{\underline{H}_{\underline{u}}}f-\int_{\underline{H}_{\underline{u}_{0}}}f+\int_{D_{u,\underline{u}}}(2\omega-\tr\chi)f
\end{eqnarray}
and 
\begin{eqnarray}
\int_{D_{u,\underline{u}}}\nabla_{3}f=\int_{H_{u}}f-\int_{H_{u_{0}}}f+\int_{D_{u,\underline{u}}}(2\underline{\omega}-\tr\underline{\chi})f,
\end{eqnarray}
}
\end{lemma}
\begin{proof}
The proof is a simple consequence of the following integration by parts procedure:
\begin{equation}
\begin{split}
\int_{D_{u\underline{u}}}\nabla_{4}f=&\int_{u_{0}}^{u}du^{'}\int_{\underline{u}_{0}}^{\underline{u}}\left(\int_{S_{u^{'}\underline{u}^{'}}}\frac{\partial f}{\partial \underline{u}^{'}}\Omega \mu_{\gamma}\right)d\underline{u}^{'}\\\nonumber 
=&\int_{u_{0}}^{u}du^{'}\int_{\underline{u}_{0}}^{\underline{u}}\left\{\frac{d}{d\underline{u}^{'}}\int_{S_{u^{'}\underline{u}^{'}}}f\Omega \mu_{\gamma}-\int_{S_{u^{'}\underline{u}^{'}}}f(\frac{\partial\Omega}{\partial \underline{u}^{'}}+\frac{\Omega}{2}\tr_{\gamma}\partial_{u}\gamma)\mu_{\gamma}\right\}d\underline{u}^{'}\\\nonumber 
=&\int_{u_{0}}^{u}du^{'}\int_{S_{u^{'}\underline{u}}}f\Omega \mu_{\gamma}-\int_{u_{0}}^{u}\int_{S_{u^{'}\underline{u}_{0}}}f\Omega \mu_{\gamma}+\int_{D_{u\underline{u}}}f(2\omega-\tr\chi)\\\nonumber 
=&\int_{\underline{H}_{\underline{u}}}f-\int_{\underline{H}_{\underline{u}_{0}}}f+\int_{D_{u\underline{u}}}(2\omega-\tr\chi)f.
\end{split}
\end{equation}
The other part follows in a similar fashion.
\end{proof}
\begin{proposition}
\label{hyperbolic}
Null Yang-Mills equations are manifestly hyperbolic.
\end{proposition}
\begin{proof}
The integration lemma can be utilized to prove the manifestly hyperbolic characteristics of the Yang-Mills equations while expressed in the double null coordinates. First consider the null triple $(\underline{\alpha}^{F},\rho^{F},\sigma^{F})$ and recall their gauge covariant evolution equations 
\begin{equation}
    \hnabla_4 \alphabar^F +\frac{1}{2}\tr\chi\alphabar^F = - \hnabla \rho^F + \Hodge{\hnabla}\sigma^F -2 \Hodge{\etabar} \sigma^F - 2 \etabar \rho^F + 2 \omega \alphabar^F - \chibarhat \cdot \alpha^F,
\end{equation}
\begin{equation}
    \hnabla_3 \rho^F + \tr\chibar \rho^F = - \widehat{\text{div}} \alphabar^F + \left( \eta-\etabar \right)\cdot \alphabar^F,
\end{equation}
\begin{equation}
    \hnabla_3 \sigma^F + \tr\chibar \sigma^F = -\widehat{\text{curl}}\alphabar^F+\left(\eta-\etabar\right)\cdot\Hodge{\alphabar}^F.
\end{equation}
Now define $f_{1}:=|\underline{\alpha}^{F}|^{2}_{\gamma,\delta},~f_{2}=|\rho^{F}|^{2}_{\gamma,\delta}$, and $f_{3}:=|\sigma^{F}|^{2}_{\gamma,\delta}$. With these definitions in mind, let us apply the integration lemma to $f_{1}$, $f_{2}$, and $f_{3}$ to yield
\begin{eqnarray}
\int_{D_{u,\underline{u}}}\nabla_{4}f_{1}+\int_{D_{u,\underline{u}}}\nabla_{3}f_{2}+\int_{D_{u,\underline{u}}}\nabla_{3}f_{3}=\int_{\underline{H}_{\underline{u}}}f_{1}+\int_{H_{u}}f_{2}+\int_{H_{u}}f_{3}-\int_{\underline{H}_{\underline{u}_{0}}}f_{1}\\\nonumber 
-\int_{H_{u_{0}}}f_{2}-\int_{H_{u_{0}}}f_{3}+\int_{D_{u,\underline{u}}}(2\omega-\tr\chi)f_{1}+\int_{D_{u,\underline{u}}}(2\underline{\omega}-\tr\underline{\chi})(f_{2}+f_{3}).
\end{eqnarray}
In order for these equations to exhibit a hyperbolic characteristic, the left-hand side should simplify to terms that are algebraic in $\underline{\alpha}^{F}, \rho^{F},$ and $\sigma^{F}$ upon using the null evolution equations. Now we note the most important point: $f_{1}$, $f_{2}$, and $f_{3}$ are gauge-invariant objects and therefore we have the following as a consequence of the compatibility of the gauge covariant connection $\hat{\nabla}$ with the metrics $\gamma$ and $\delta$ of the fibers:
\begin{eqnarray}
\nabla_{4}f_{1}=2\langle \underline{\alpha}^{F},\widehat{\nabla}_{4}\underline{\alpha}^{F}\rangle_{\gamma,\delta},~\nabla_{3}f_{2}=2\langle \rho^{F},\widehat{\nabla}_{3}\rho^{F}\rangle_{\gamma,\delta},~\nabla_{3}f_{3}=2\langle \sigma^{F},\widehat{\nabla}_{3}\sigma^{F}\rangle_{\gamma,\delta}.
\end{eqnarray}
Now we only focus on the principal terms for  the hyperbolicity argument.

\begin{equation}
\begin{split}
\langle\underline{\alpha}^{F},\widehat{\nabla}_{4}\underline{\alpha}^{F}\rangle_{\gamma,\delta}&=\langle \underline{\alpha}^{F},-\hnabla \rho^F + \Hodge{\hnabla}\sigma^F+\cdot\cdot\cdot\cdot\rangle_{\gamma,\delta}\\
&=-\text{div}\langle\underline{\alpha}^{F},\rho^{F}\rangle_{\gamma,\delta}+\langle\widehat{\text{div}}\underline{\alpha}^{F},\rho^{F}\rangle_{\gamma,\delta}+\text{div}~^{*}\langle\underline{\alpha}^{F},\sigma^{F}\rangle_{\gamma,\delta}+\langle\widehat{\text{curl}}\underline{\alpha}^{F},\sigma^{F}\rangle +\text{l.o.t}
\end{split}
\end{equation} 

\begin{equation} 
\langle \rho^{F},\hnabla _{3}\rho^{F}\rangle_{\gamma,\delta}=\langle \rho^{F},-\widehat{\text{div}} \alphabar^F+\cdot\cdot\cdot\cdot\rangle_{\gamma,\delta}\end{equation} \begin{equation}
\langle \sigma^{F},\hnabla _{3}\sigma^{F}\rangle_{\gamma,\delta}=\langle \sigma^{F}, -\widehat{\text{curl}}\underline{\alpha}^{F}+\cdot\cdot\cdot\cdot\rangle_{\gamma,\delta}.
\end{equation}
Now after addition, we have 
\begin{equation}
\begin{split}
&\langle\underline{\alpha}^{F},\widehat{\nabla}_{4}\underline{\alpha}^{F}\rangle_{\gamma,\delta}+\langle \rho^{F},\hnabla _{3}\rho^{F}\rangle_{\gamma,\delta}+\langle \sigma^{F},\hnabla _{3}\sigma^{F}\rangle_{\gamma,\delta}\\ 
=&-\text{div}\langle\underline{\alpha}^{F},\rho^{F}\rangle_{\gamma,\delta}+\langle\widehat{\text{div}}\underline{\alpha}^{F},\rho^{F}\rangle_{\gamma,\delta}+\text{div}~^{*}\langle\underline{\alpha}^{F},\sigma^{F}\rangle_{\gamma,\delta}+\langle\widehat{\text{curl}}\underline{\alpha}^{F},\sigma^{F}\rangle\\ 
&\hspace{5cm}-\langle \rho^{F},\widehat{\text{div}} \alphabar^F\rangle_{\gamma,\delta}-\langle \sigma^{F}, \widehat{\text{curl}}\underline{\alpha}^{F}\rangle_{\gamma,\delta}+\text{l.o.t}\\ 
=&-\text{div}\langle\underline{\alpha}^{F},\rho^{F}\rangle_{\gamma,\delta}+\text{div}~^{*}\langle\underline{\alpha}^{F},\sigma^{F}\rangle_{\gamma,\delta}+\text{l.o.t}
\end{split}
\end{equation}
which, upon integration over the topological $2-$sphere, yields terms that are algebraic in $\underline{\alpha}^{F}, \rho^{F}$, and $\sigma^{F}$. Here $\langle\underline{\alpha}^{F},\rho^{F}\rangle_{\gamma,\delta}=(\underline{\alpha}^{F})^{P}~_{Qab}(\rho^{F})^{Q}~_{P}~^{b}$ and $~^{*}\langle \underline{\alpha}^{F},\sigma^{F}\rangle_{\gamma,\delta}=\epsilon^{ca}(\underline{\alpha}^{F})^{P}~_{Qab}\rho^{Q}~_{P}~^{b}$. The most vital property that is utilized here is the compatibility of the connection $\hnabla$ with the inner product $\langle~,~\rangle_{\gamma,\delta}$ induced by the fiber metrics $\gamma$ and $\delta$ together with the Hodge structure present in the null Yang-Mills equations. Notice that nowhere in the procedure did we require explicit information about the Yang-Mills connection $1-$form $A^{P}~_{Q\mu}dx^{\mu}$. The remaining Yang-Mills null evolution equations may be utilized in a  similar  manner to obtain energy identities associated with the triple $\alpha^{F},\rho^{F},\sigma^{F}$. This concludes the proof of the hyperbolic characteristics of the null Yang-Mills equations.
\end{proof}

\subsection{Commutation Formulae}
\noindent Let us now write down the commutator formulas for $\mathfrak{g}-$valued sections of vector bundles over $S_{u,\ubar}$ i.e., sections of the bundle $\otimes^{n}T^{*}M\otimes _{M}\mathfrak{P}_{Ad,\mathfrak{g}}$ that transform as tensors under gauge transformations
\begin{eqnarray}
[\hnabla_{4},\hnabla_{B}]\mathcal{G}^{P}~_{QA_{1}A_{2}A_{3}\cdot\cdot\cdot\cdot A_{n}}=[\widehat{D}_{4},\widehat{D}_{B}]\mathcal{G}^{P}~_{QA_{1}A_{2}A_{3}\cdot\cdot\cdot\cdot A_{n}}\nonumber+(\nabla_{B}\log\Omega)\hnabla_{4}\mathcal{G}^{P}~_{QA_{1}A_{2}A_{3}\cdot\cdot\cdot\cdot A_{n}}\\\nonumber 
-\gamma^{CD}\chi_{BD}\hnabla_{C}\mathcal{G}^{P}~_{QA_{1}A_{2}A_{3}\cdot\cdot\cdot\cdot A_{n}}-\sum_{i=1}^{n}\gamma^{CD}\chi_{BD}\etabar_{A_{i}}\mathcal{G}^{P}~_{QA_{1}A_{2}A_{3}\cdot\cdot\hat{A}_{i}C\cdot\cdot A_{n}}\\\nonumber 
+\sum_{i=1}^{n}\gamma^{CD}\chi_{A_{i}B}\etabar_{D}\mathcal{G}^{P}~_{QA_{1}A_{2}A_{3}\cdot\cdot\hat{A}_{i}C\cdot\cdot A_{n}}
\end{eqnarray}
and 
\begin{eqnarray}
[\widehat{D}_{4},\widehat{D}_{A}]\mathcal{G}^{P}~_{QA_{1}A_{2}\cdot\cdot\cdot\cdot A_{n}}=(\nabla_{A}\log\Omega)\hnabla_{4}\mathcal{G}^{P}~_{QA_{1}A_{2}\cdot\cdot\cdot\cdot A_{n}}-\sum_{i}R(e_{C},e_{A_{i}},e_{4},e_{A})\nonumber\mathcal{G}^{P}~_{QA_{1}\cdot\cdot\hat{A}_{i},\cdot\cdot A_{n}}\\\nonumber 
+F^{P}~_{R4A}\mathcal{G}^{R}~_{QA_{1}A_{2}\cdot\cdot\cdot A_{n}}-F^{R}~_{Q4A}\mathcal{G}^{P}~_{RA_{1}A_{2}\cdot\cdot\cdot A_{n}}
\end{eqnarray}
which through the algebraic Bianchi identity and the definition of the null curvature components may be written in the following schematic forms (indices are suppressed)
\begin{eqnarray}
[\widehat{D}_{4},\widehat{D}_{A}]\mathcal{G}\sim (\beta+\alpha^{F}\cdot(\rho^{F}+\sigma^{F}))\mathcal{G}+\alpha^{F}\mathcal{G}+(\eta+\etabar)\hnabla_{4}\mathcal{G}
\end{eqnarray}
and 
\begin{eqnarray}
[\hnabla_{4},\hnabla_{B}]\mathcal{G}\sim[\widehat{D}_{4},\widehat{D}_{B}]\mathcal{G} \nonumber+(\eta+\etabar)\hnabla_{4}\mathcal{G}-\chi\hnabla\mathcal{G}+\chi\etabar\mathcal{G}.
\end{eqnarray}
Similar schematic expressions hold for $[\hnabla_{3},\hnabla_{B}]\mathcal{G}$ 
\begin{eqnarray}
[\hnabla_{3},\hnabla_{B}]\mathcal{G}\sim[\widehat{D}_{3},\widehat{D}_{B}]\mathcal{G}+(\eta+\etabar)\hnabla_{3}\mathcal{G}-\bar{\chi}\nabla\mathcal{G}+\bar{\chi}\eta\mathcal{G}
\end{eqnarray}
and 
\begin{eqnarray}
[\widehat{D}_{3},\widehat{D}_{A}]\mathcal{G}\sim (\bar{\beta}+\bar{\alpha}^{F}\cdot(\rho^{F}+\sigma^{F}))\mathcal{G}+\bar{\alpha}^{F}\mathcal{G}+(\eta+\etabar)\hnabla_{3}\mathcal{G}.
\end{eqnarray}
For a function (or gauge invariant object) $f$ the following holds 
\begin{eqnarray}
[\hnabla_{4},\hnabla_{A}]f\sim (\eta+\etabar)\nabla_{4}f-\chi\nabla f,~[\hnabla_{3},\hnabla_{A}]f\sim (\eta+\etabar)\nabla_{3}f-\bar{\chi}\nabla f.
\end{eqnarray}
We will use these commutation formulas while deriving higher-order energy estimates. Note that we write the schematic form since we would not require the exact form while deriving estimates. Below is a lemma for a general commutation scheme
\begin{lemma} 
\label{commutation}
\textit{Suppose $\mathcal{G}$ is a section of the product vector bundle $~^{k}\otimes T^{*}\mathbb{S}^{2}\otimes P_{Ad,\rho}$, $k\geq 1$,  that satisfies $\hnabla_{4}\mathcal{G}=\mathcal{F}_{1}$ and $\hnabla_{4}\hnabla^{I}\mathcal{G}=\mathcal{F}^{I}_{1}$, then $\mathcal{F}^{I}_{1}$ verifies the following schematic expression:}
\begin{equation}
\begin{split}
\mathcal{F}^{I}_{1}\sim&\sum_{J_{1}+J_{2}+J_{3}+J_{4}=I-1}\nabla^{J_{1}}(\eta+\underline{\eta})^{J_{2}}\nabla^{J_{3}}\beta\hnabla^{J_{4}}\mathcal{G}\\+&\sum_{J_{1}+J_{2}+J_{3}+J_{4}+J_{5}=I-1}\nabla^{J_{1}}(\eta+\underline{\eta})^{J_{2}}\hnabla^{J_{3}}\alpha^{F}\hnabla^{J_{4}}(\rho^{F},\sigma^{F})\hnabla^{J_{5}}\mathcal{G}\\
+&\sum_{J_{1}+J_{2}+J_{3}+J_{4}=I-1}\nabla^{J_{1}}(\eta+\underline{\eta})^{J_{2}}\hnabla^{J_{3}}\alpha^{F}\hnabla^{J_{4}}\mathcal{G}+\sum_{J_{1}+J_{2}+J_{3}=I}\nabla^{J_{1}}(\eta+\underline{\eta})^{J_{2}}\hnabla^{J_{3}}\mathcal{F}_{1}\\+&\sum_{J_{1}+J_{2}+J_{3} +J_{4}=I}\nabla^{J_{1}}(\eta+\underline{\eta})^{J_{2}}\hat{\nabla}^{J_{3}}\chi\hnabla^{J_{4}}\mathcal{G}.
\end{split}
\end{equation}
Similarly, for $\hnabla_{3}\mathcal{G}=\mathcal{F}_{2}$, and $\hnabla_{3}\hnabla^{I}\mathcal{G}=\mathcal{F}^{I}_{2}$, 
\begin{equation}
\begin{split}
\mathcal{F}^{I}_{2}\sim &\sum_{J_{1}+J_{2}+J_{3}+J_{4}=I-1}\nabla^{J_{1}}(\eta+\underline{\eta})^{J_{2}}\nabla^{J_{3}}\underline{\beta}\hnabla^{J_{4}}\mathcal{G}\\+&\sum_{J_{1}+J_{2}+J_{3}+J_{4}+J_{5}=I-1}\nabla^{J_{1}}(\eta+\underline{\eta})^{J_{2}}\hnabla^{J_{3}}\underline{\alpha}^{F}\hnabla^{J_{4}}(\rho^{F},\sigma^{F})\hnabla^{J_{5}}\mathcal{G}\\
+&\sum_{J_{1}+J_{2}+J_{3}+J_{4}=I-1}\nabla^{J_{1}}(\eta+\underline{\eta})^{J_{2}}\hnabla^{J_{3}}\underline{\alpha}^{F}\hnabla^{J_{4}}\mathcal{G}+\sum_{J_{1}+J_{2}+J_{3}=I}\nabla^{J_{1}}(\eta+\underline{\eta})^{J_{2}}\hnabla^{J_{3}}\mathcal{F}_{2}\\+&\sum_{J_{1}+J_{2}+J_{3} +J_{4}=I}\nabla^{J_{1}}(\eta+\underline{\eta})^{J_{2}}\hat{\nabla}^{J_{3}}\underline{\chi}\hnabla^{J_{4}}\mathcal{G}.
\end{split}
\end{equation}
\end{lemma}
\begin{proof}
For $I=1$, this identity is clearly satisfied due to the calculations above. Assume, it holds for $J=I-1$ and show that it holds for $J=I$. We omit the proof and refer to \cite{AnAth}. 
\end{proof} 

\subsection{Integration and norms}
\noindent Now we define the norms that are adapted to the double null framework. We need to define the norms within the bulk spacetime $\mathcal{D}_{u,\ubar}$, on the null hypersurfaces ($H$ and $  \Hbar$), and on the topological $2-$ sphere $S_{u,\ubar}$. First, we need an integration measure. On $\mathcal{D}_{u,\ubar}$, we can use the canonical volume measure corresponding to the spacetime metric. On the null hypersurfaces, the metric is degenerate, and therefore no canonical choice of volume form is available. For a function $f$, we define its integration over $\mathcal{D}_{u,\ubar}, H,  \Hbar$, and $S_{u,\ubar}$ as follows 
\begin{eqnarray}
\int_{S_{u,\ubar}}f:=\int_{S_{u,\ubar}}f \mu_{\gamma},~
\int_{\mathcal{D}_{u,\ubar}}f:=\int_{0}^{u}\int_{0}^{\ubar}\int_{S}\Omega^{2}f\mu_{\gamma}dud\ubar,\\
\int_{H}f:=\int_{0}^{\ubar}\int_{S}f \Omega\mu_{\gamma}d\ubar^{'}, ~\int_{  \Hbar}f:=\int_{0}^{u}\int_{S}f \Omega\mu_{\gamma}du,
\end{eqnarray}
where $\mu_{\gamma}$ is the volume form associated with the metric $\gamma$ on $S_{u,\ubar}$. Now we define the norm of a $\mathfrak{g}-$valued section associated with $S_{u,\ubar},H$, and $  \Hbar$. Let $\mathcal{G}$ be a $\mathfrak{g}-$valued section of a vector bundle over $S_{u,\ubar}$ (we will also call $\mathcal{G}$ as a $\mathfrak{g}-$valued horizontal tensor field). Its $L^{p}$ norms ($1\leq p<\infty$) are defined as follows 
\begin{eqnarray}
||\mathcal{G}||^{p}_{L^{p}(S_{u,\ubar})}:=\int_{S_{u,\ubar}}(|\mathcal{G}|^{2}_{\gamma,\delta})^{\frac{p}{2}}\mu_{\gamma},\\
||\mathcal{G}||^{p}_{L^{p}(H)}:=\int_{0}^{\ubar}\int_{S}(|\mathcal{G}|^{2}_{\gamma,\delta})^{\frac{p}{2}} \Omega\mu_{\gamma}d\ubar,~
||\mathcal{G}||^{p}_{L^{p}(  \Hbar)}:=\int_{0}^{u}\int_{S}(|\mathcal{G}|^{2}_{\gamma,\delta})^{\frac{p}{2}} \Omega\mu_{\gamma}du,
\end{eqnarray}
where $|\mathcal{G}|^{2}_{\gamma,\delta}$ is defined as follows 
\begin{eqnarray}
|\mathcal{G}|^{2}_{\gamma,\delta}:=\mathcal{G}^{P}~_{QA_{1}A_{2}A_{3}\cdot\cdot\cdot\cdot A_{N}}\mathcal{G}^{Q}~_{PB_{1}B_{2}B_{3}\cdot\cdot\cdot\cdot B_{N}}\gamma^{A_{1}B_{1}}\gamma^{A_{2}B_{2}}\gamma^{A_{3}B_{3}}\cdot\cdot\cdot\cdot \gamma^{A_{N}B_{N}}.
\end{eqnarray}
$L^{\infty}$ norm over $S_{u,\ubar}$ is defined as $||\mathcal{G}||_{L^{\infty}(S_{u,\ubar})}:=\sup_{\theta^{1},\theta^{2}\in S_{u,\ubar}}\sqrt{|\mathcal{G}|^{2}_{\gamma,\delta}(\theta^{1},\theta^{2})}$. 
The norm of the initial data specified on the hypersurfaces $H_{0}$ and $\Hbar_{0}$ are defined as follows. We use $\mathcal{O}_{0},\mathcal{W}_{0},$ and $\mathcal{F}_{0}$ to denote the initial norms of the connection coefficients, the Weyl curvature, and the Yang-Mills curvature, respectively. The norm of the connection coefficients on the initial null hypersurfaces $H_{0}$ and $\Hbar_{0}$ is defined as follows
\begin{eqnarray}
\label{eq:connectioninitialnorm}
\nonumber\mathcal{O}_{0}:=\\
\sup_{S\subset H_{0},~S\in  \Hbar_{0}}\sup_{\varphi\in\{\widehat{\chi}, \tr\chi,\chibarhat, \tr\chibar,\omega,\omegabar,\eta,\etabar\}}\max\left(\sum_{I=0}^{2}||\nabla^{I}\varphi||_{L^{2}(S)},\nonumber\sum_{I=0}^{1}||\nabla^{I}\varphi||_{L^{4}(S)},||\varphi||_{L^{\infty}(S)}\right)\\+\sup_{S\subset H_{0},~S\in  \Hbar_{0}}\sup_{\varphi\in(\eta,\omegabar)} ||\nabla_{3}\varphi||_{L^{2}(S)}+\sup_{S\subset H_{0},~S\in  \Hbar_{0}}\sup_{\varphi\in(\etabar,\omega)}||\nabla_{4}\varphi||_{L^{2}(S)}\\\nonumber+\left(\sup_{u,\ubar} C|C^{-1}\gamma_{round}(X,X)\leq \gamma(X,X)\leq C\gamma_{round}(X,X),\gamma~\text{is~a~metric~on}\right.\\\nonumber\left. ~S_{u,0} ~or~S_{0,\ubar}, \forall X\in TS_{u,0}~or~TS_{0,\ubar}\right).
\end{eqnarray}
The norm of the Weyl curvature components on the initial null hypersurfaces $H_{0}$ and $\Hbar_{0}$ is defined as follows
\begin{eqnarray}
\label{eq:weylinitialnorm}
\mathcal{W}_{0}:=\nonumber\sup_{\Psi\in\{\alpha,\beta,\bar{\beta},\rho,\sigma\}}||\nabla\Psi||_{L^{2}(H_{0})}+\sup_{\Psi\in\{\bar{\alpha},\beta,\bar{\beta},\rho,\sigma\}}||\nabla\Psi||_{L^{2}(  \Hbar_{0})}\\ +\max\left(||\nabla_{4}\alpha||_{L^{2}(H_{0})},||\nabla_{3}\bar{\alpha}||_{L^{2}(  \Hbar_{0})},||\nabla_{4}\beta||_{L^{2}(  \Hbar_{0})},||\nabla_{3}\bar{\beta}||_{L^{2}(H_{0})}\right)
\\\nonumber+\sup_{S\subset H_{0},~S\in  \Hbar_{0}}\sup_{\Psi\in\{\alpha,\bar{\alpha},\beta,\bar{\beta},\rho,\sigma\}}\max\left(\sum_{I=0}^{1}||\nabla^{I}\Psi||_{L^{2}(S)},||\Psi||_{L^{4}(S)}\right).
\end{eqnarray}
Lastly, the norm of the Yang-Mills curvature components on the initial null hypersurfaces $H_{0}$ and $\Hbar_{0}$ is defined as follows
\begin{eqnarray}
\label{eq:yanginitialnorm}
\mathcal{F}_{0}:=\nonumber\sum_{I=0}^{2}\left(\sup_{\Phi\in\{\alpha^{F},\rho^{F},\sigma^{F}\}}||\hnabla^{I}\Phi||_{L^{2}(H_{0})}+\sup_{\Phi\in\{\bar{\alpha}^{F},\rho^{F},\sigma^{F}\}}||\hnabla^{I}\Phi||_{L^{2}(  \Hbar_{0})}\right.\\\nonumber\left.+\sup_{\Phi\in\{\rho^{F},\sigma^{F}\}}||\hnabla^{I}_{4}\Phi||_{L^{2}(  \Hbar_{0})}\right.\\
\left.+\sup_{\Phi\in\{\rho^{F},\sigma^{F}\}}||\hnabla^{I}_{3}\Phi||_{L^{2}(H_{0})}+||\hnabla^{I}_{4}\alpha^{F}||_{L^{2}(H_{0})}+||\hnabla^{I}_{3}\bar{\alpha}^{F}||_{L^{2}(H_{0})}\right.\\\nonumber \left.+\sup_{\Phi\in(\rho^{F},\sigma^{F})}||\hnabla_{4}\hnabla\Phi||_{L^{2}(  \Hbar_{0})}+\sup_{\Phi\in(\rho^{F},\sigma^{F})}||\hnabla_{3}\hnabla\Phi||_{L^{2}(H_{0})}\right)\\\nonumber+||\hnabla_{4}\hnabla\alpha^{F}||_{L^{2}(H_{0})}+||\hnabla_{3}\hnabla\bar{\alpha}^{F}||_{L^{2}(  \Hbar_{0})}
\\\nonumber +\sup_{S\subset H_{0},~S\subset  \Hbar_{0}}\sup_{\Phi\in\{\alpha^{F},\bar{\alpha}^{F},\rho^{F},\sigma^{F}\}}\max\left(\sum_{I=0}^{1}||\hnabla^{I}\Phi||_{L^{4}(S)}\right).
\end{eqnarray}
Here $\gamma_{round}$ is the standard round metric on a $2-$sphere and $J>0$ is sufficiently large but finite. In addition to the initial data norm, we also define the following curvature norms defined on $H_{u}$ and $\Hbar_{\ubar}$ 
\begin{eqnarray}
\label{eq:weylnorm}
\mathcal{W}:=\sum_{I=0}^{1}\left(\sup_{u}\sup_{\Psi\in\{\alpha,\beta,\bar{\beta},\rho,\sigma\}}||\nabla^{I}\Psi||_{L^{2}(H)}\nonumber+\sup_{\ubar}\sup_{\Psi\in\{\bar{\alpha},\beta,\bar{\beta},\rho,\sigma\}}||\nabla^{I}\Psi||_{L^{2}(  \Hbar)}\right.\\ 
\left.+\sup_{u}||\nabla_{4}\alpha||_{L^{2}(H)}+\sup_{\ubar}||\nabla_{3}\bar{\alpha}||_{L^{2}(  \Hbar)}\right),\\
\label{eq:yangnorm}
\mathcal{F}:=\sum_{I=0}^{2}\left(\sup_{u}\sup_{\Phi\in\{\alpha^{F},\rho^{F},\sigma^{F}\}}||\hnabla^{I}\Phi||_{L^{2}(H)}\nonumber+\sup_{\ubar}\sup_{\Phi\in\{\bar{\alpha}^{F},\rho^{F},\sigma^{F}\}}||\hnabla^{I}\Phi||_{L^{2}(  \Hbar)}\right)\\ +\sup_{u}||\hnabla_{4}\alpha^{F}||_{L^{2}(H)}+\sup_{\ubar}||\hnabla_{3}\bar{\alpha}^{F}||_{L^{2}(  \Hbar)}+\sup_{u}||\hnabla_{4}\hnabla_{4}\alpha^{F}||_{L^{2}(H)}\\\nonumber+\sup_{\ubar}||\hnabla_{3}\hnabla_{3}\bar{\alpha}^{F}||_{L^{2}(  \Hbar)}+\sup_{\ubar}||\hnabla_{3}\hnabla\bar{\alpha}^{F}||_{L^{2}(  \Hbar)}
+\sup_{u}||\hnabla_{4}\hnabla\alpha^{F}||_{L^{2}(H)}.
\end{eqnarray}
In addition to the $L^{2}(H_{u})$ (resp. $L^{2}(\Hbar_{\ubar}))$ norms, we also need to define the following norms defined on the topological $2-$sphere $S_{u,\ubar}$ i.e., $H_{u}\cap \Hbar_{\ubar}$
\begin{eqnarray}
    \mathcal{W}(S):=\sup_{u,\ubar}||\alpha,\bar{\alpha},\beta,\bar{\beta},\rho,\sigma||_{L^{4}(S_{u,\ubar})},\\
\mathcal{F}(S):=\sum_{i=0}^{1}\sup_{u,\ubar}||\hnabla^{i}(\alpha^{F},\bar{\alpha}^{F},\rho^{F},\sigma^{F})||_{L^{4}(S_{u,\ubar})},\\
\end{eqnarray}
Following the decomposition (\ref{eq:weyl}), the Weyl curvature and the Yang-Mills curvature completely determine the Riemann curvature of the spacetime. 
\section{Main theorem and idea of the proof}
\subsection{Main Theorem}
\noindent In this section, we describe the main result of the article and sketch the main argument behind the proof. As discussed previously, the Yang-Mills sourced null Bianchi equations (\ref{eq:bianchi1}-\ref{eq:bianchi2}) and null Yang-Mills equations are manifestly hyperbolic contrary to the equations $d\Gamma+\Gamma\Gamma=R$ (\ref{eq:connection1}-\ref{eq:4}) and $dA+[A,A]=F$ (definition of the Yang-Mills curvature in terms of the connection) which are manifestly non-hyperbolic, where $\Gamma$, $A$, $R$, and $F$ denote connection coefficients on the frame tangent bundle, connection on the principle $G-$bundle (or gauge bundle), the curvature of frame tangent bundle, and curvature of the principle $G-$bundle, respectively. Nevertheless, the equations $d\Gamma+\Gamma\Gamma=R$ (\ref{eq:connection1}-\ref{eq:4}) exhibit a special structure as we shall observe. In addition, we will never use the equations $dA+[A,A]=F$ but rather derive estimates with the full gauge covariant derivatives (since working at the level of connection would require a choice of Yang-Mills gauge). While proving a local existence in the Cauchy problem (this needs to be done at the end utilizing the estimates obtained from analysis in double null gauge) for Einstein-Yang-Mills equations, one ought to work with the connections directly. However, the connections can always be estimated in terms of the gauge invariant norms of the Yang-Mills curvature and a local existence theory can be obtained (e.g., Yang-Mills equations take the form of a symmetric hyperbolic system in temporal gauge and one can estimate the residual spacial connections in terms of the electric field that can be constructed by means of the null components of the Yang-Mills curvature). The primary factor behind the semi-global existence feature in the context of the characteristic initial value problem is the remarkable null structure associated with the nonlinearities of the Einstein-Yang-Mills equations while expressed in the double null gauge.  These null structures are more obvious if one writes down the gauge covariant wave equations for the Weyl curvature and the Yang-Mills curvature (see \cite{moncrief,klainerman2007kirchoff}). Essentially this null structure played a crucial role in establishing the non-linear stability of Minkowski space under vacuum and electromagnetic perturbations \cite{christodoulou1993global,zipser2000global}. We first write down the theorem and then briefly discuss the idea of the proof and in particular how the remarkable structure of the Einstein-Yang-Mills equations is conducive to proving a semi-global existence result. We omit the description of a characteristic initial value problem since we are only interested in obtaining estimates. For a precise formulation, the reader is referred to section 2.3 of \cite{luk2012local}\\

\noindent \begin{theorem}
    Let $(M,g)$ be a $3+1$ dimensional globally hyperbolic Lorentzian manifold and $F$ be the curvature of a principle $G-$bundle $\mathfrak{P}$ on $M$ such that the pair $(g,F)$ verify the Einstein-Yang-Mills equations 
\begin{gather}
	 \label{Ein}   R[g]_{\mu\nu} - \frac{1}{2}R[g]g_{\m\n}= \mathfrak{T}_{\m\n},\\
	 \label{Yan}
	 \mathfrak{T}_{\m\n} := \frac{1}{2} \left( {F^P}_{Q\m\alpha}{{F^Q}_{P\n}}^{\alpha} + {\Hodge{F}^P}_{Q\m\alpha}{{\Hodge{F}^Q}_{P\n}}^{\alpha} \right),\\
	    \widehat{D}_{\m} {{F^P}}_{Q\n \lambda} +\widehat{D}_{\n} {{F^P}}_{Q\lambda \m}+\widehat{D}_{\lambda} {{F}^P}_{Q\m\n}=0,  g^{\a\b}\widehat{D}_{\a}{F^P}_{Q\b\m}=0,
	\end{gather}
where by ${\Hodge{F}^{P}}_Q$ we denote the Hodge dual of ${F^P}_Q$. Let $H_{0}$ and $\Hbar_{0}$ be two null hypersurfaces and $S_{0,0}:=H_{0}\cap \Hbar_{0}$. The initial data of the connections $(\tr\chi,\tr\chibar,\chihat,\chibarhat,\eta,\etabar,\omega,\omegabar)$ and curvatures ($\alpha,\alphabar,\beta,\betabar,\rho,\sigma,\alpha^{F},\alphabar^{F},\rho^{F},\sigma^{F}$) are provided on the double-null hypersurfaces $H_{0}$ and $\Hbar_{0}$ such that these data verify the characteristic constraint equations \footnote{Given `free data' on the intersection sphere $S_{0,0}$, one can integrate the characteristic constraint equations to obtain the data on $H_{0}$ and $\Hbar_{0}$. Therefore, any data provided on $H_{0}$ and $\Hbar_{0}$ must verify the compatibility condition that they are the solution of characteristic constraint equations (ODE) on $H_{0}$ and $\Hbar_{0}$ with data on $S_{0,0}$. See \cite{luk2012local} for a detailed discussion} on the pair $(H_{0},\Hbar_{0})$ and their corresponding norms $\mathcal{O}_{0}, \mathcal{W}_{0}, \mathcal{F}_{0}$ (\ref{eq:connectioninitialnorm}-\ref{eq:yanginitialnorm}) verify $\mathcal{O}_{0}, \mathcal{W}_{0}, \mathcal{F}_{0}<\infty$.
Given a large but finite $J>0$, there exists a sufficiently small $\epsilon>0$ depending on the initial data $\mathcal{O}_{0}, \mathcal{W}_{0}, \mathcal{F}_{0}$ such that the norms $\mathcal{W}$ (\ref{eq:weylnorm}) and $\mathcal{F}$ (\ref{eq:yangnorm}) of the Weyl curvature and the Yang-Mills curvature remain bounded uniformly throughout the future causal domain $\mathcal{D}_{u,\ubar}$ of $S_{0,0}$ foliated by the two families of the incoming and outgoing null hypersurfaces $H_{u}$ and $\Hbar_{\ubar}$ for $u\in[0,\epsilon],\ubar\in[0, J]$ in terms of the initial data norm $\mathcal{O}_{0}, \mathcal{W}_{0}, \mathcal{F}_{0}<\infty$ i.e.,  
\begin{eqnarray}
\mathcal{W}<C(\mathcal{O}_{0}, \mathcal{W}_{0}, \mathcal{F}_{0}),~\mathcal{F}<C(\mathcal{O}_{0}, \mathcal{W}_{0}, \mathcal{F}_{0}).  
\end{eqnarray}
A solution to the coupled Einstein-Yang-Mills equations (\ref{Ein}) exists in the function space defined by the norm $\mathcal{W}$ and $\mathcal{F}$ in $\mathcal{D}_{u,\ubar}$.
\end{theorem} 

\begin{remark}
Once we obtain the gauge invariant estimates for the Weyl and Yang-Mills curvature, all the degrees of freedom are exhausted. One can solve the Cauchy problem for the quadruple $(\widehat{g},\widehat{K},E,B)$ satisfying constraint equations in the bulk $\mathcal{D}_{u,\ubar}$ in spacetime harmonic-temporal gauge given data on some initial slice lying in the bulk. In the light of current estimates (Weyl curvature is controlled in $H^{1}$ and Yang-Mills curvature is controlled in $H^{2}$), the regularity of such solution would be $H^{3}\times H^{2}\times H^{2}\times H^{2}$, where $\widehat{g}$, $\widehat{K}$, $E$, and $B$ are the induced Riemannian metric on a Cauchy hypersurface $\Sigma$ by the Lorentzian structure of $(M,g)$, its second-fundamental form, chromo-electric, and chromo-magnetic field associated to the Yang-Mills curvature $F$, respectively. We will discuss the existence issue in the appropriate section. Doing so would require invoking Randall's theorem on characteristic initial value problem \cite{rendall1992characteristic}. One may continue to prove these estimates to the successive higher orders and therefore one can establish the result for the classical solution.
\end{remark}

\subsection{Idea of the Proof}
\noindent  Let us now discuss the idea of the proof. The inclusion of Yang-Mills sources introduces additional difficulties that need to be addressed. We first assume that the connection coefficients $\varphi$ enjoy an upper bound (possibly large), namely the bootstrap constant. This upper bound allows us to estimate the ellipticity constant of the metric on the topological $2-$spheres in terms of the initial data (independent of the assumed upper bound of the connection coefficients if the $u$ direction is chosen sufficiently small e.g., $u\in[0,\epsilon]$) thereby allowing us to utilize the standard Sobolev inequalities. We start the main estimates with the connection coefficients assuming the finiteness of the curvature norms and $L^{2}(H,\Hbar)$ norm of second angular derivatives of the connection coefficients. The good connection coefficients $\varphi_{g}$ satisfy a $\nabla_{3}$ equation and therefore they gain a smallness factor $\epsilon$ through the integration of the transport equations. As a consequence, they can be bounded by the initial data alone. On the contrary, the bad connection coefficients $\varphi_{b}$ ($\{\varphi\}:=\{\varphi_{g}\}\cup \{\varphi_{b}\}$) satisfying $\nabla_{4}$ equations are estimated in terms of the curvature norms and the $L^{2}(H,  \Hbar)$ norm of second angular derivatives of $\varphi$. The remarkable point to note here is that in the structure equation for these bad connection coefficients $\varphi_{b}$ the terms $\varphi_{b}\varphi_{b}$ do not appear. This is precisely a consequence of the null structure of the non-linearities present. This allows us to employ Gr\"onwall's inequality to yield the desired estimates. Next, using the available estimates, we show that the undetermined connection norms $||\nabla^{2}\varphi||_{L^{2}(H,\Hbar}$ is determined by the curvature norms (Weyl and Yang-Mills) through a series of transport-elliptic estimates. These estimates are absolutely necessary to close the regularity argument. 

In the next step, we use simple integration by parts arguments to obtain the estimates of the Weyl and Yang-Mills curvatures (contrary to using the Bel-Robinson and Yang-Mills stress-energy tensors). This is a consequence of the fact that the Bianchi equations for the Weyl curvature and the gauge-covariant null Yang-Mills equations exhibit symmetric hyperbolic characteristics (the integration by parts argument for the Bianchi equations are standard, see proposition (\ref{hyperbolic}) for the argument for a gauge-covariant system such as Yang-Mills equations here). Once again the good curvature components ($\alpha,\beta,\bar{\beta},\rho,\sigma,\alpha^{F},\rho^{F},\sigma^{F}$) enjoy a gain of a smallness factor $\epsilon$ and therefore are innocuous. However, the bad curvature components $(\bar{\alpha},\bar{\alpha}^{F})$ do not gain such a small factor since they are integrated along $  \Hbar$. Therefore, in the energy estimate, they pose a potential obstruction. However, the remarkable null structure of non-linearities once again remedies the situation. In other words, the connection coefficients multiplying the terms $|\nabla\bar{\alpha}|^{2}$ and $|\hnabla\bar{\alpha}^{F}|^{2},|\hnabla^{2}\bar{\alpha}^{F}|^{2}$ are good connection coefficients (i.e., they satisfy the $\nabla_{3}$ transport equations) and therefore are estimated completely by means of the initial data. More precisely, for these curvature components (e.g., $\hnabla\alphabar^{F}$), we obtain inequalities such as 
\begin{eqnarray}
||\hnabla\alphabar^{F}||^{2}_{L^{2}(\Hbar_{\ubar})}\leq ||\hnabla\alphabar^{F}||^{2}_{L^{2}(\Hbar_{0})}+C(\mathcal{O}_{0})\int_{0}^{\ubar}||\hnabla\alphabar^{F}||^{2}_{L^{2}(\Hbar_{\ubar^{'}})}d\ubar^{'}    
\end{eqnarray}
Therefore, we can use a Gr\"onwall's inequality to obtain the final estimate purely in terms of the initial data $\mathcal{O}_{0}$. In addition, there are several occasions where the null structure of the Einstein-Yang-Mills equations plays a subtle role. Another important point to note here is that while working in optimal regularity characteristic initial value problem (one derivative of curvature), one ought to estimate $\nabla_{4}\alpha$ and $\nabla_{3}\alphabar$ separately. Most importantly, the null Bianchi equations (\ref{eq:bianchi1} and \ref{eq:bianchi2}) contain terms $\nabla_{4}R_{AB}$ and $\nabla_{3}R_{AB}$, respectively, which would produce terms $\hnabla_{4}\hnabla_{4}\alpha^{F}$ and   $\hnabla_{3}\hnabla_{3}\alphabar^{F}$. 
However, there are no direct estimates for $\nabla_{4}\alpha^{F}$ (or $\nabla_{3}\alphabar^{F}$) since $\alpha^{F}$ and $\alphabar^{F}$ only verifies $\hnabla_{3}$ and $\hnabla_{4}$ equations. This complicates the situation substantially since we need a separate energy norm for these terms and suitably constructed Bianchi-pair integration is utilized in a hierarchical manner to complete the argument. Once we have obtained the final estimates (that are independent of the initial bootstrap constant chosen and only depend on the initial data), we may choose the initial bootstrap constant accordingly to close the argument. A novelty in our proof of the estimates is the use of fully gauge-invariant norms of the Yang-Mills curvature components instead of working in a particular choice of gauge.     

\begin{remark}
We note that the gauge group for the Yang-Mills theory is compact allowing for a positive definite adjoint invariant inner product. Therefore, it is convenient to work with the gauge co-variant derivative since this derivative is compatible with the positive definite inner product on the Lie-algebra $\mathfrak{g}$. A similar strategy on the Einstein part would fail since the associated gauge group would be $SO(3,1)$ which is non-compact.    
\end{remark}

\subsection{Sketch of the existence proof}
\noindent The idea of the proof behind the existence of a solution of the coupled Einstein-Yang-Mills system throughout the domain $\mathcal{D}_{u,\ubar}$ given data on the two null hypersurfaces relies on Rendall's theorem \cite{rendall1992characteristic} and standard energy argument for the Cauchy problem. In doing so, we would utilize the estimates on the Weyl and Yang-Mills curvature. We have given the characteristic data on the two null hypersurfaces $H_{0}$ and $\Hbar_{0}$ for a quasi-linear wave equation, Rendall's theorem guarantees the existence of a unique solution in a sufficiently small neighborhood of the topological $2-$ sphere $S_{0,0}:=H_{0}\cap \Hbar_{0}$. To be more precise, we state Rendall's theorem 
\begin{theorem}\cite{rendall1992characteristic}
Let us consider a quasi-linear wave equation 
\begin{eqnarray}
\label{eq:quasi}
g^{\mu\nu}(\Phi)\partial_{\mu}\partial_{\nu}\Phi=\mathfrak{F}(\Phi,\partial \Phi),
\end{eqnarray}
where $g(\Phi)$ and $\mathfrak{F}(\Phi,\partial\Phi)$ are smooth in their respective arguments with smooth initial data prescribed on the two intersecting initial null hypersurfaces $H_{0}$ and $\Hbar_{0}$. Suppose $H_{0}$ and $\Hbar_{0}$ intersect at the topological $2-$sphere $S_{0,0}$ and all derivatives of $\Phi$ are continuous up to $S_{0,0}$, then there exists a small neighbourhood of $S_{0,0}$ in its future such that a unique solution to (\ref{eq:quasi}) exists within this neighbourhood.
\end{theorem}
Now it is well-known since the work of \cite{yvone} that Einstein's equations reduce to a system of quasi-linear wave equations in the space-time harmonic gauge. Therefore, in the space-time harmonic gauge, Einstein's equations (or reduced Einstein's equations to be precise) fall under the category (\ref{eq:quasi}), and Rendall's theorem applies. The final step is then to confirm that the gauge condition is verified throughout the domain of existence. This once again follows from the fact that the space-time gauge condition $\Gamma^{\mu}=0$ ($\Gamma^{\mu}$ is suitably contracted Connection coefficients of the spacetime metric $g$) verifies a wave type equation from which preservation of gauge condition follows in a straightforward way (see section 8 of chapter 6 in \cite{choquet2008general}). Our aim is to extend Rendall's theorem to the coupled Einstein-Yang-Mills system. We want to reduce the coupled system to a system of quasi-linear wave equations. Since Yang-Mills curvature appears as the source term in Einstein's equations and hence does not alter the principal symbol, we shall only focus on the Yang-Mills equations. Yang-Mills equations can be reduced to a symmetric hyperbolic system in the temporal gauge \cite{satah}. For the moment, if we consider the metric in the form (suitable for the Cauchy problem)
$~^{1+3}g=-N^{2}dt\otimes dt+g_{ij}dx^{i}\otimes dx^{j}$, where $\{x^{i}\}_{i=0}^{3}=(t,x^{1},x^{2},x^{3})$ is local chart, $N$ is the lapse function. Also, let $\mathbf{n}:=\frac{1}{N}\partial_{t}$ be the $t=$constant hypersurface orthogonal future directed unit normal field. In the framework of a Cauchy problem, the gauge covariant Yang-Mills equations read \footnote{These coupled equations are essentially the first-order formulation of the semi-linear wave equation for the Yang-Mills curvature 
$\widehat{D}^{\alpha}\widehat{D}_{\alpha}F^{\hat{a}}~_{\hat{b}\mu\nu}=2F^{\hat{a}}~_{\hat{c}\mu\beta}F^{\hat{c}}~_{\hat{b}\nu}~^{\beta}-2F^{\hat{a}}~_{\hat{c}\nu\beta}F^{\hat{c}}~_{\hat{b}\mu}~^{\beta}-R^{\gamma}~_{\beta\mu\nu}F^{\hat{a}}~_{\hat{b}\gamma}~^{\beta}-R^{\gamma}~_{\nu}F^{\hat{a}}~_{\hat{b}\gamma\mu}-R^{\gamma}~_{\mu}F^{\hat{a}}~_{\hat{b}\nu\gamma}$}
\begin{eqnarray}
\widehat{\mathcal{L}}_{{\partial}_{t}}\mathcal{E}_{i}=-N\epsilon_{i}~^{jk}\widehat{\nabla}_{j}H_{k}-2NK_{i}~^{j}\mathcal{E}_{j}+Ntr_{g}K\mathcal{E}_{i}-\epsilon_{i}~^{jk}\nabla_{j}N H_{k},\\
\widehat{\mathcal{L}}_{{\partial}_{t}}H_{i}=N\epsilon_{i}~^{jk}\widehat{\nabla}_{j}\mathcal{E}_{k}-2NK_{i}~^{j}H_{j}+Ntr_{g}KH_{i}+\epsilon_{i}~^{jk}\nabla_{j}N \mathcal{E}_{k},
\end{eqnarray}
where $\mathcal{E}_{i}:=F(\mathbf{n},\partial_{i}), H_{i}=~^{*}F(\mathbf{n},e_{i})$ are the chromo-electric and chromo-magnetic fields, respectively. Here $\widehat{\mathcal{L}}$ denotes the gauge-covariant Lie derivative operator. More explicitly, $\widehat{\mathcal{L}}_{\partial_{t}}:=\partial_{t}+[A_{0},\cdot]$. These system of equations are symmetric hyperbolic and one may obtain a local well-posedness result in the temporal gauge ($A_{0}=0$ and the spatial connections $A_{i}$ verify $A_{i}(t)=\int_{0}^{t}N\mathcal{E}_{i}dt^{'}$). Therefore, in the spacetime-harmonic-temporal gauge, the complete Einstein-Yang-Mills equations reduce to coupled quasi-linear-semi-linear wave equations (while Einstein's equations are quasi-linear, Yang-Mills equations are only semi-linear) for the spacetime metric and the Yang-Mills curvature. An application of  Renadall's theorem \cite{rendall1992characteristic} then yields a unique solution in a sufficiently small neighborhood of $S_{0,0}$ lying in the causal future of $S_{0,0}$ 
\begin{theorem}
\cite{rendall1992characteristic}
Given a regular initial data set, there exists a small neighborhood to the future of $S_{0,0}$ such that Einstein-Yang-Mills equations can be solved within it.  
\end{theorem}

\noindent The remaining task is to verify that the gauge conditions and constraints are propagated throughout the domain of existence. This is a rather straightforward calculation and as such a consequence of the Bianchi identities (see \cite{choquet2008general}). This is rather monotonous and therefore we do not proceed with the calculations here. Note that there are other gauges that would work equally well. The Spacetime Harmonic Lorentz gauge (SHL) is one such gauge choice among many others. In the SHL gauge, the complete system reduces to a coupled semi-linear wave equations for the connections (see Yvonne Choquet Bruhat's book \cite{choquet2008general} for the local well-posedness of Einstein-Yang-Mills equations in this gauge) 
\begin{eqnarray}
g^{\mu\nu}\partial_{\mu}\partial_{\nu}A_{\alpha}=\mathcal{G}_{\alpha}(A,\partial A),
\end{eqnarray}
where $\partial A$ appears linearly in $\mathcal{G}_{\alpha}$. Therefore, Rendall's theorem directly applies as well. Therefore, at this point, we are able to apply Rendall's theorem to the Einstein-Yang-Mills system both in spacetime harmonic temporal and spacetime harmonic Lorentz gauge.\\

\noindent Now we sketch the existence proof by utilizing the apriori estimates on the Weyl and Yang-Mills curvature together with Rendall's theorem. For zero Yang-Mills fields, Luk \cite{luk2012local} sketched an argument for the existence theorem. We follow \cite{luk2012local}. Let us consider the spacetime domain $\mathcal{D}_{u,\ubar}:\{u\leq \epsilon, \ubar\leq J\}$ and define a time function $t:=u+\ubar$. Now by virtue of the null characteristic of $L$ and $\Lbar$, $\nabla t$ is time-like. Denote the level sets of $t$ by $\Sigma_{t}$. By the theorem of \cite{rendall1992characteristic}, if the spacetime does not exist in the whole of $D_{u,\ubar}$ as a solution of the coupled Einstein-Yang-Mills equations, then there exists a $t^{*}\in (0,J+\epsilon)$ such that
\begin{eqnarray}
t^{*}=\sup\{t|\text{the~spacetime~exists~in}~D_{u,\ubar}\cap\cup_{\tau\in[0,t)}\Sigma_{\tau}\}.
\end{eqnarray}
Now on each Cauchy slice $\Sigma_{t}$, let $\widehat{g}_{t}, \widehat{k}_{t}, \mathcal{E}_{t}, H_{t}$ be the induced Riemannian metric, second fundamental form, the chromo-electric field, and the chromo-magnetic field, respectively. We shall argue that they converge to $(\widehat{g}_{t^{*}},\widehat{k}_{t^{*}},\mathcal{E}_{t^{*}},H_{t^{*}})$ in a smooth way and $(\widehat{g}_{t^{*}},\widehat{k}_{t^{*}},\mathcal{E}_{t^{*}},H_{t^{*}})$ verify the Einstien-Yang-Mills constraint equations. To this end, it suffices to show that all the derivatives of the metric $g$ and the Yang-Mills curvatures are uniformly bounded for $t<t^{*}$. Since the Weyl and Yang-Mills curvatures exhaust all the degrees of freedom, we want to control $||\nabla^{I_{1}}_{4}\nabla^{J_{1}}_{3}\nabla^{K_{1}}\Psi||_{L^{2}(\Sigma_{t})}$ and $||\hnabla^{I_{2}}_{4}\hnabla^{J_{2}}_{3}\hnabla^{K_{2}}\Phi||_{L^{2}(\Sigma_{t})}$ for all $I_{1,2},J_{1,2},$ and $K_{1,2}$, where $\Psi=\{\alpha,\beta,\betabar,\rho,\sigma,\alphabar\}$,~$\Phi:\{\alpha^{F},\rho^{F},\sigma^{F},\alphabar^{F}\}$. This is a consequence of the apriori estimate (the main estimates) that we obtain section 7 and 8. In other words, we have 
\begin{eqnarray}
\label{eq:apriori}
||\nabla^{I_{1}}_{4}\nabla^{J_{1}}_{3}\nabla^{K_{1}}\Psi||_{L^{2}(\Sigma_{t})}\leq C_{I_{1},J_{1},K_{1}},~||\hnabla^{I_{2}}_{4}\hnabla^{J_{2}}_{3}\hnabla^{K_{2}}\Phi||_{L^{2}(\Sigma_{t})}\leq C_{I_{2},J_{2},K_{2}}~\forall t<t^{*} 
\end{eqnarray}
and $C_{I_{1},J_{1},K_{1}},C_{I_{2},J_{2},K_{2}}$ are independent of $t^{*}$. This can be proven in an iterative manner. Now let us construct spacetime harmonic coordinates on $\Sigma_{t^{*}}$. Let $g_{ij}=(\widehat{g}_{t^{*}})_{ij}$, $g_{00}=-1,g_{0i}=0$ and $\partial_{t}g_{ij}=-2k_{ij}$, $i,j\in \{1,2,3\}$, and use the freedom on $\partial_{t}g_{0i}$ and $\partial_{t}g_{00}$ to satisfy $\Gamma^{\mu}:=g^{\alpha\beta}\Gamma[g]^{\mu}_{\alpha\beta}=0$. Similarly, since the Yang-Mills estimates are gauge invariant, we may set $A_{0}=0$ for temporal gauge or $\mathfrak{A}:=\partial_{\mu}A^{\mu}=0$ for Lorentz gauge. Due to the constraint equations on $\Sigma_{t^{*}}$ that are satisfied by continuity, we have $\partial_{t}\Gamma^{\mu}=0$, $\partial_{t}\mathfrak{A}=0$ on $\Sigma_{t^{*}}$ as well. Now, from the theory of quasi-linear hyperbolic equation, it is straightforward to show that there exists a time $T>t^{*}$ such that a unique solution of the coupled Einstein-Yang-Mills system exists in the time slab $t\in[t^{*},T]$. Since the gauge and constraints verify wave equations, if they are zero initially with zero speed, they continue to be zero in the domain of existence. The size of the future domain of existence $T-t^{*}$ depends on the size of the data on $\Sigma_{t^{*}}$ which by means of the estimates (\ref{eq:apriori}) (obtained due to the estimates in sections 7 and 8) is bounded from above uniformly. Therefore, $T-t^{*}$ is strictly greater than zero. This violates the maximality of the domain of existence and therefore $t^{*}=J+\epsilon$ i.e., the domain within which we proved the apriori estimates. Now the question remains how to perform a smooth change of coordinate from the spacetime harmonic gauge to the canonical double-null gauge. \\

\noindent In order to smoothly change the coordinates from the spacetime harmonic coordinates to the canonical double null coordinate, we solve the Eikonal equations 
\begin{eqnarray}
g^{\mu\nu}\partial_{\mu}u\partial_{\nu}u=0,~g^{\mu\nu}\partial_{\mu}\ubar\partial_{\nu}\ubar=0
\end{eqnarray}
and the transport equation for the angular variable on topological $2-$sphere $S_{u,\ubar}$
\begin{eqnarray}
\frac{\partial\theta^{A}}{\partial \ubar}=0,~A=\{1,2\}.
\end{eqnarray}
These equations can be solved in the neighbourhood of $\Sigma_{t^{*}}$ and doing so one obtains a smooth solution $(u,\ubar,\theta^{1},\theta^{2})$. By uniqueness, these solutions agree with the canonical double null coordinate functions in a neighborhood to the past of $\Sigma_{t^{*}}$. We can then change to
the $(u,\ubar,\theta^{1},\theta^{2})$ coordinates. This completes the sketch of the existence proof. 

\noindent 

\section{Important inequalities}
\noindent Throughout our analysis, we need to employ Sobolev inequalities on the topological 2-sphere $S$ at different stages. However, since the metric on $S$ is dynamical, we need to make appropriate bootstrap assumptions. Technically, one could define the norms and energies with respect to a background metric on $S$ and try to control the additional terms that arise. However, we will start with making a bootstrap assumption on the connection coefficients similar to \cite{luk2012local,klainerman2012formation}. Let us assume the following 
\begin{eqnarray}
\label{eq:bootstrapinitial}
\sup_{u,\ubar}||\varphi||_{L^{\infty}(S_{u,\ubar})}\leq \Delta,
\end{eqnarray}
where $\varphi:=( \tr\chibar,\chibarhat,\etabar,\omega, \tr\chi,\widehat{\chi},\eta,\omegabar)$ and $\Delta$ is possibly a large constant. Later, we will prove that one can choose $\Delta$ such that the bootstrap assumption (\ref{eq:bootstrapinitial}) is closed, in particular, the improved estimates do not depend on $\Delta$. Under this assumption, it is straightforward to prove that the null lapse function $\Omega$, the ellipticity constant of the dynamical metric $\gamma$, and the shift vector field $b^{A}\frac{\partial}{\partial \theta^{A}}$ are bounded in the spacetime slab $D_{u,\bar
u}$ ($u\in [0,\epsilon],~\ubar\in [0,J]$) in terms of the initial data $\mathcal{O}_{0}$ (see \cite{klainerman2012formation,luk2012local}). For example, consider the null-lapse function $\Omega$. It verifies the following equation by the definition of the connection coefficient $\omegabar$ 
\begin{eqnarray}
  \omegabar=-\frac{1}{2}\Omega^{-1}\nabla_{3}\Omega=\frac{1}{2}\frac{d\Omega^{-1}}{du}  
\end{eqnarray}
which upon integration yields 
\begin{eqnarray}
 ||\Omega^{-1}-2||_{L^{\infty}(S_{u,\ubar})}\leq C \int_{0}^{u}||\omegabar||_{L^{\infty}({S_{u^{'},\ubar}})}\lesssim \epsilon \delta  
\end{eqnarray}
which implies for sufficiently small $\epsilon$, $\frac{1}{4}\leq \Omega\leq 4$.  Similarly one may estimate the metric coefficients $\gamma_{AB}$ induced on the $2-$ sphere $S_{u,\ubar}$. These estimates provide us with estimates on the area of $S_{u,\ubar}$ (\cite{klainerman2012formation,luk2012local}). Once we have the metric components under control in this double null gauge, we may write down the following set of inequalities that will be useful throughout.\\

\begin{proposition}
\label{sobolev}
\textit{There exists $\epsilon_{0}=\epsilon_{0}(\mathcal{O}_{0},\Delta)$ such that the following gauge invariant Sobolev inequalities hold for any horizontal gauge field strength in the spacetime slab $\mathcal{D}_{u,\ubar}$ ($u\in[0,\epsilon],\ubar\in[0,J]$), $\epsilon\leq \epsilon_{0}$
\begin{eqnarray}
\label{eq:inequal1}
||\mathcal{G}||_{L^{4}(S_{u,\ubar}})\leq C(\mathcal{O}_{0})\sum_{I=0}^{1}||\hnabla^{I}\mathcal{G}||_{L^{2}(S_{u,\ubar})},\\
||\mathcal{G}||_{L^{\infty}(S_{u,\ubar})}\leq C(\mathcal{O}_{0})(||\mathcal{G}||_{L^{4}(S_{u,\ubar})}+||\hnabla\mathcal{G}||_{L^{4}(S_{u,\ubar})}),\\
\label{eq:sobolev4}
||\mathcal{G}||_{L^{\infty}(S_{u,\ubar}})\leq C(\mathcal{O}_{0})\sum_{I=0}^{2}||\hnabla^{I}\mathcal{G}||_{L^{2}(S_{u,\ubar})}.
\end{eqnarray}
}
\end{proposition}
\begin{proof} Let $f=\left(\mathcal{G}^{P}~_{QA_{1}A_{2}\cdot\cdot\cdot\cdot A_{n}}\mathcal{G}^{Q}~_{PB_{1}B_{2}\cdot\cdot\cdot\cdot B_{n}}\gamma^{A_{1}B_{1}}\gamma^{A_{2}B_{2}}\cdot\cdot \cdot\cdot\gamma^{A_{n}B_{n}}+\delta\right)^{1/2}$ and apply the standard Sobolev embedding for scalars
\begin{eqnarray}
||f||_{L^{4}(S)}\leq C(\mathcal{O}_{0})\left(||f||_{L^{2}(S)}+||\frac{\langle\mathcal{G},\hnabla\mathcal{G}\rangle}{(|\mathcal{G}|^{2}+\delta)^{\frac{1}{2}}}||_{L^{2}(S)}\right)\leq C\left(||\mathcal{G}||_{L^{2}(S)}\nonumber+||\hnabla\mathcal{G}||_{L^{2}(S)}\right)
\end{eqnarray}
after letting $\delta\to0$.
Note that $\nabla f^{2}=2\mathcal{G}^{P}~_{QA_{1}A_{2}\cdot\cdot\cdot\cdot A_{n}}\hnabla\mathcal{G}^{Q}~_{PB_{1}B_{2}\cdot\cdot\cdot\cdot B_{n}}\gamma^{A_{1}B_{1}}\gamma^{A_{1}B_{1}}\cdot\cdot\cdot\cdot \gamma^{A_{n}B_{n}}$ since $f$ is gauge invariant and $\hnabla$ is metrics compatible. The second inequality follows in a similar way. The last inequality is a consequence of the first two.
\end{proof}

\begin{proposition}
\label{transport}
Given $\sup_{u,\ubar}|| \tr\chi, \tr\chibar||_{L^{\infty}(S_{u,\ubar})}\leq C(\mathcal{O}_{0})$, the following inequalities hold throughout $\mathcal{D}_{u,\ubar}$ ($u\in[0,\epsilon],\ubar\in[0,J]$) for a sufficiently small $\epsilon$
\begin{eqnarray}
||\mathcal{G}||_{L^{p}(S_{u,\ubar})}\leq C(\mathcal{O}_{0})\left(||\mathcal{G}||_{L^{p}(S_{u,\ubar^{'}})}+\int_{\ubar^{'}}^{\ubar}||\hnabla_{4}\mathcal{G}||_{L^{p}(S_{u,\ubar^{''}})}d\ubar^{''}\right)\\
||\mathcal{G}||_{L^{p}(S_{u,\ubar})}\leq C\left(||\mathcal{G}||_{L^{p}(S_{u^{'},\ubar})}+\int_{u^{'}}^{u}||\hnabla_{3}\mathcal{G}||_{L^{p}(S_{u^{''},\ubar)}}du^{''}\right)
\end{eqnarray}
for $1\leq p\leq \infty$. Here $\mathcal{G}$ can be a section of the mixed bundle and $\hnabla$ is the gauge covariant connection compatible with the metrics of both fibers.
\end{proposition}
\begin{proof}
Follows in an exactly similar way as that of the first one \ref{sobolev}.    
\end{proof}

\begin{proposition}
The following inequalities hold for any horizontal gauge field strength $\mathcal{G}$ throughout $\mathcal{D}_{u,\ubar}$ ($u\in[0,\epsilon],\ubar\in[0,J]$ under the bootstrap assumption (\ref{eq:bootstrapinitial})) (which in turn controls the metric on $S_{u,\ubar}$)
\begin{eqnarray}
\label{eq:gagliardo1}
||\mathcal{G}||_{L^{4}(S_{u,\ubar})}\leq C(\mathcal{O}_{0})\left(||\mathcal{G}||^{\frac{1}{2}}_{L^{2}(S_{u,\ubar})}||\hnabla\mathcal{G}||^{\frac{1}{2}}_{L^{2}(S_{u,\ubar})}+||\mathcal{G}||_{L^{2}(S_{u,\ubar})}\right),\\
\label{eq:gagliardo2}
||\mathcal{G}||_{L^{\infty}(S_{u,\ubar})}\leq C(\mathcal{O}_{0})\left(||\mathcal{G}||^{\frac{1}{2}}_{L^{4}(S_{u,\ubar})}||\hnabla\mathcal{G}||^{\frac{1}{2}}_{L^{4}(S_{u,\ubar})}+||\mathcal{G}||_{L^{4}(S_{u,\ubar})}\right).
\end{eqnarray}
\end{proposition}
\begin{proof}
Similar to the proof of the proposition \ref{sobolev}.    
\end{proof}

\begin{proposition}
Under the bootstrap assumption (\ref{eq:bootstrapinitial}), following holds for any horizontal gauge field strength $\mathcal{G}$ throughout $\mathcal{D}_{u,\ubar}$ ($u\in[0,\epsilon],\ubar\in[0,J]$)
\begin{eqnarray}
||\mathcal{G}||_{L^{4}(S_{u,\ubar})}\leq C(\mathcal{O}_{0})\left(||\mathcal{G}||_{L^{4}(S_{0,\ubar})}+||\mathcal{G}||^{1/2}_{L^{2}(H)}||\hnabla_{4}\mathcal{G}||^{1/4}_{L^{2}(H)}(||\mathcal{G}||_{L^{2}(H)}+||\hnabla\mathcal{G}||_{L^{2}(H)})^{1/4}\right),\\
||\mathcal{G}||_{L^{4}(S_{u,\ubar})}\leq C\left(||\mathcal{G}||_{L^{4}(S_{0,\ubar})}+||\mathcal{G}||^{1/2}_{L^{2}(  \Hbar)}||\hnabla_{3}\mathcal{G}||^{1/4}_{L^{2}(  \Hbar)}(||\mathcal{G}||_{L^{2}(  \Hbar)}+||\hnabla\mathcal{G}||_{L^{2}(  \Hbar)})^{1/4}\right).
\end{eqnarray}
\end{proposition}
\begin{proof}
Similar to the proof of the proposition \ref{sobolev}.    
\end{proof}

\noindent All of these inequalities hold true for sections of tangent bundles of $S_{u,\ubar}$ (i.e., the non-gauge field strengths) and in such case, $\hnabla$ and $\nabla$ coincide. These will be the main inequalities that we will use throughout.

\section{Estimate of the connection and curvature components in terms of the initial data and curvature energy}
\noindent We divide the connection coefficients into two classes depending on the transport equations they satisfy. Notice that each element of the set $( \tr\chibar,\chibarhat,\etabar,\omega, \tr\chi,\widehat{\chi})$ satisfies a $\nabla_{3}$ equation, where as $\eta$ and $\omegabar$ only satisfy a $\nabla_{4}$ equation. We denote each element of $( \tr\chibar,\chibarhat,\etabar,\omega, \tr\chi,\widehat{\chi})$ by $\varphi_{g}$ and an element of $(\eta,\omegabar)$ by $\varphi_{b}$ and element of the whole set $( \tr\chibar,\chibarhat,\etabar,\omega, \tr\chi,\widehat{\chi},\eta,\omegabar)$ by $\varphi$. By $C$, we will always mean a universal constant while a constant that depends on the other entities will be denoted so. We first define the curvature norms (both Weyl and Yang-Mills) that we intend to estimate 
   
\begin{lemma}
\label{1}
Assume that $||\nabla^{2}\varphi_{g}||_{L^{2}(H,  \Hbar)},||\nabla^{2}\varphi_{b}||_{L^{2}(H,  \Hbar)},\mathcal{W},\mathcal{W}(S),\mathcal{F},\mathcal{F}(S)<\infty$. Then the connection coefficients satisfy the following point-wise estimates
\begin{eqnarray}
\sup_{u,\ubar}||\varphi_{g}||_{L^{\infty}(S_{u,\ubar})}\leq C\mathcal{O}_{0},~\sup_{u,\ubar}||\varphi_{b}||_{L^{\infty}(S_{u,\ubar})}\leq C(\mathcal{O}_{0},\mathcal{W},\mathcal{W}(S),\mathcal{F},\mathcal{F}(S),||\nabla^{2}\varphi||_{L^{2}(H)}).
\end{eqnarray}
\end{lemma} 
\begin{proof}
The proof relies on the direct transport inequalities (proposition \ref{transport}) 
\begin{eqnarray}
||\varphi||_{L^{\infty}(S_{u,\ubar})}\leq C(\mathcal{O}_{0})\left(||\varphi||_{L^{\infty}(S_{u,\ubar^{'}})}+\sup_{S_{u,\ubar}}\int_{\ubar^{'}}^{\ubar}|\nabla_{4}\varphi|d\ubar^{''}\right)\\
||\varphi||_{L^{\infty}(S_{u,\ubar})}\leq C\left(||\varphi||_{L^{\infty}(S_{u^{'},\ubar})}+\sup_{S_{u,\ubar}}\int_{u^{'}}^{u}|\nabla_{3}\varphi|du^{''}\right)
\end{eqnarray}
and delicate trace estimates. First, assume the bootstrap assumption of the good connection coefficients $\varphi_{g}$ 
\begin{eqnarray}
\label{eq:boot1}
\sup_{u\in[0,\epsilon],\ubar\in[0,J]}||\varphi_{g}||_{L^{\infty}(S_{u,\ubar})}\leq 2C \mathcal{O}_{0}.
\end{eqnarray}
Using the transport equations and trace estimates, we will improve this bootstrap estimate thereby completing the proof. First consider the bad connection coefficients $\varphi_{g}:=(\eta,\omegabar)$\footnote{note that one may obtain an equation for $\nabla_{3}\eta$ by means of the relation $\nabla_{3}\eta=-\nabla_{3}\etabar+2\nabla_{3}\nabla\ln\Omega$}
\begin{eqnarray}
||\eta||_{L^{\infty}(S_{u,\ubar})}\leq C||\eta||_{L^{\infty}(S_{u,0})}+C\sup_{S_{u,\ubar}}\int_{0}^{\ubar}|-\chi\cdot(\eta-\etabar)-\beta-\frac{1}{2}(\alpha^{F}\rho^{F}-\alpha^{F}\sigma^{F})|d\ubar^{''} 
\end{eqnarray}
which under the boot-strap, Sobolev embedding (\ref{eq:sobolev4}), and the definition of the entity $\mathcal{F}(S)$ yields
\begin{eqnarray}
||\eta||_{L^{\infty}(S_{u,\ubar})}\leq C||\eta||_{L^{\infty}(S_{u,0})}+C(\mathcal{O}_{0})\int_{0}^{\ubar}||\eta||_{L^{\infty}(S)}d\ubar^{'}+C\mathcal{F}(S)+C\sup_{S_{u,\ubar}}(\int_{0}^{\ubar}|\beta|^{2}d\ubar^{'})^{\frac{1}{2}}.
\end{eqnarray}
Now recall that $\sup_{S_{u,\ubar}}(\int_{0}^{\ubar}|\beta|^{2}d\ubar^{'})^{\frac{1}{2}}$ is nothing but the trace norm $||\beta||_{tr(H)}$ defined in (\cite{klainerman2012formation}). We may estimate this term by means of the following trace inequality (\cite{klainerman2012formation})
\begin{eqnarray}
\label{eq:trace2}
||\nabla_{4}A||_{tr(H)}\leq \left(||\nabla^{2}_{4}A||_{L^2(H)}+||A||_{L^{2}(H)}\right)^{1/2}||\nabla^{2}A||^{1/2}_{L^{2}(H)}+||\nabla_{4}\nabla A||_{L^{2}(H)}+||\nabla A||_{L^{2}(H)}
\end{eqnarray}
for a section $A$ of $TS_{u,\ubar}$. Now in order to estimate $||\beta||_{tr(H)}$, we write the following 
\begin{eqnarray}
\beta=-\nabla_{4}\eta\underbrace{-\chi\cdot(\eta-\etabar)-\frac{1}{2}(\alpha^{F}\rho^{F}-\alpha^{F}\sigma^{F})}_{I}.
\end{eqnarray}
The term $I$ is harmless and its tr norm can be estimated by $C(\mathcal{O},\mathcal{F},\mathcal{F}(S))$. Therefore we focus on the $||\cdot||_{tr(H)}$ of $\nabla_{4}\eta$. The trace inequality (\ref{eq:trace2}) yields 
\begin{eqnarray}
||\nabla_{4}\eta||_{tr(H)}\leq \left(||\nabla^{2}_{4}\eta||_{L^2(H)}+||\nabla_{4}\eta||_{L^{2}(H)}\right)^{1/2}||\nabla^{2}\eta||^{1/2}_{L^{2}(H)}\nonumber+||\nabla_{4}\nabla \eta||_{L^{2}(H)}+||\nabla \eta||_{L^{2}(H)}.
\end{eqnarray}
Each term on the right-hand side may be estimated as follows. Using the equation of motion, $\nabla_{4}\eta$ is estimated by the curvature norm 
\begin{eqnarray}
||\nabla_{4}\eta||_{L^{2}(H)}\leq C(\mathcal{O}_{0},\mathcal{W},\mathcal{F}(S)).
\end{eqnarray}
In order to estimate $||\nabla^{2}_{4}\eta||_{L^{2}(H)}$, we act $\nabla_{4}$ on the transport equation for $\eta$
\begin{eqnarray}
\nabla^{2}_{4}\eta=-(\nabla_{4}\widehat{\chi}+\frac{1}{2}\nabla_{4} \tr\chi \gamma)\cdot(\eta-\etabar)-\chi\cdot(\nabla_{4}\eta-\underbrace{\nabla_{4}\etabar}_{II})-\nabla_{4}\beta-\frac{1}{2}(\underbrace{\nabla_{4}\alpha^{F}}_{III}(\rho^{F}-\sigma^{F})\\\nonumber+\alpha^{F}(\nabla_{4}\rho^{F}-\nabla_{4}\sigma^{F})).
\end{eqnarray}
Now notice the terms $II$ and $III$ do not have explicit expressions in terms of the transport equations. However, term $III$ can be estimated by $\mathcal{F}$ and $(\eta-\etabar)$ verifies $\nabla_{4}$ transport equations. Using the evolution equations for $\widehat{\chi}$, $ \tr\chi$, and $\beta$ we obtain 
\begin{eqnarray}
\nabla^{2}_{4}\eta_{a}=-\left((- \tr\chi \widehat{\chi}-2\omega\widehat{\chi}-\alpha+\frac{1}{2}(-\frac{1}{2}( \tr\chi)^{2}-|\widehat{\chi}|^{2}_{\gamma}-2\omega  \tr\chi-\mathfrak{T}_{44}) \gamma)\cdot(\eta-\etabar)\right)_{a}\\\nonumber-\underbrace{(\chi\cdot(\nabla\omega+\frac{1}{2}\chi\cdot(\etabar-\zeta)+\omega(\etabar+\zeta)-\frac{1}{2}\beta -\frac{1}{2}~^{*}\sigma^{F}\cdot\alpha^{F}\nonumber-\frac{1}{2}\rho^{F}\cdot\alpha^{F})_{a}}_{V}\\\nonumber 
-(-2 \tr\chi\beta_{a}+(div\alpha)_{a}-2\omega\beta_{a}+(\eta\cdot\alpha)_{a}-\frac{1}{2}(\underbrace{\mathcal{D}_{a}R_{44}-\mathcal{D}_{4}R_{4a}}_{IV}))-\frac{1}{2}(\underbrace{\nabla_{4}\alpha^{F}_{a}}_{VI}(\rho^{F}-\sigma^{F})\\\nonumber+
\alpha^{F}_{a}(-\hat{div} \alpha^{F}- \tr\chi\rho^{F}-(\eta-\etabar)\cdot\alpha^{F}-(-\hat{curl} \alpha^{F}- \tr\chi \sigma^{F}+(\eta-\etabar)\cdot ~^{*}\alpha^{F})).
\end{eqnarray}
Explicit computation using Einstein's equation with Yang-Mills source i.e., $R_{\mu\nu}=\mathfrak{T}_{\mu\nu}$ yields 
\begin{eqnarray}
IV=\langle\alpha^{F},\hnabla_{b}\alpha^{F}\rangle-\chi_{bc}\mathfrak{T}_{c4}+\eta_{b}\mathfrak{T}_{44}-2\omega \mathfrak{T}_{4b}-\hnabla_{4}(\alpha^{F}_{b}\cdot\rho^{F}-\alpha^{F}_{b}\cdot\sigma^{F})
\end{eqnarray}
and therefore one more use of the evolution equation confirms that $||IV||_{L^{2}(H)}$ is controlled by $\mathcal{F}$ and $\mathcal{O}_{0}$. Similarly $||V||_{L^{2}(H)}$ and $||VI||_{L^{2}(H)}$ are controlled by $\mathcal{W}$, $\mathcal{F}$, and $\mathcal{O}_{0}$ under the bootstrap assumption (\ref{eq:boot1}). Luckily we have a $\nabla_{4}$ equation for $\eta-\etabar$ (the term $V$ arises due to $\nabla_{4}$ acting on $\eta-\etabar$). Also note that $||\Phi^{F}||_{L^{\infty}(S)}:=||(\alpha^{F},\bar{\alpha}^{F},\rho^{F},\sigma^{F})||_{L^{\infty}(S))}$ can be controlled by $||\Phi^{F}||_{L^{4}(S)}$ and $||\hnabla\Phi^{F}||_{L^{4}(S)}$ or equivalently $\mathcal{F}(S)$. Collecting all the terms and applying Sobolev embedding, we have
\begin{eqnarray}
\left(||\nabla^{2}_{4}\eta||_{L^2(H)}+||\nabla_{4}\eta||_{L^{2}(H)}\right)^{1/2}||\nabla^{2}\eta||^{1/2}_{L^{2}(H)}\leq C(\mathcal{O}_{0},\mathcal{W},\mathcal{F},\mathcal{F}(S),||\nabla^{2}\varphi||_{L^{2}(H)}),
\end{eqnarray}
where $\varphi$ can be any element of the set $( \tr\chibar,\chibarhat,\etabar,\omega, \tr\chi,\widehat{\chi},\eta,\omegabar)$.
The remaining is to estimate $||\nabla_{4}\nabla\eta||_{L^{2}(H)}$. We can do this by means of the evolution equation for $\eta$. Commuting the evolution equation for $\eta$ with the horizontal derivative $\nabla$ yields (schematically)
\begin{eqnarray}
\nabla_{4}\nabla\eta=-\nabla\chi\cdot (\eta-\etabar)-\chi\cdot(\nabla\eta-\nabla\etabar)-\nabla\beta-\frac{1}{2}\hnabla\alpha^{F}(\rho^{F}-\sigma^{F})-\frac{1}{2}\alpha^{F}(\hnabla\rho^{F}-\hnabla\sigma^{F})\\\nonumber 
+(\beta+\alpha^{F}(\rho^{F}-\sigma^{F})+\alpha^{F})\eta+(\eta+\etabar)(-\chi(\eta-\etabar)-\beta-\frac{1}{2}\alpha^{F}(\rho^{F}-\sigma^{F}))-\chi\nabla\eta+\chi\etabar\eta
\end{eqnarray}
which utilizing Sobolev embedding on $S_{u,\ubar}$ (to handle the term $||\eta||_{L^{4}(S_{u,v})}$) can be estimated as 
\begin{eqnarray}
||\nabla_{4}\nabla\eta||_{L^{2}(H)}\leq C(\mathcal{O}_{0},\mathcal{W},\mathcal{F},||\nabla^{2}\varphi||_{L^{2}(H)}).
\end{eqnarray}
Collecting all the terms together, we obtain 
\begin{eqnarray}
||\eta||_{L^{\infty}(S_{u,\ubar})}\leq C||\eta||_{L^{\infty}(S_{u,0})}+C(\mathcal{O}_{0},\mathcal{W},\mathcal{F},\mathcal{F}(S))+C(\mathcal{O}_{0})\int_{0}^{\ubar}||\eta||_{L^{\infty}(S)}d\ubar^{'}
\end{eqnarray}
which upon using Gr\"onwall yields 
\begin{eqnarray}
\label{eq:es1}
||\eta||_{L^{\infty}}(S_{u,\ubar})\leq C(\mathcal{O}_{0},\mathcal{W},\mathcal{F},\mathcal{F}(S),||\nabla^{2}\eta||_{L^{2}(H)}).
\end{eqnarray}
Later, we shall see that $\mathcal{F}(S)$ is actually dominated by $\mathcal{F}$ and $||\nabla^{2}\eta||_{L^{2}(H)}$ is dominated by $\mathcal{W}$. Now we want to estimate $\sup_{u,\ubar}||\omegabar||_{L^{\infty}(S_{u,\ubar})}$. Once again, the use of transport inequality (proposition \ref{transport}) yields 
\begin{eqnarray}
||\omegabar||_{L^{\infty}(S_{u,\ubar})}\leq C\left(||\omegabar||_{L^{\infty}(S_{u,\ubar^{'}})}+\sup_{S_{u,\ubar}}\int_{0}^{\ubar}|2\omega\omegabar+\frac{3}{4}|\eta-\etabar|^{2}-\frac{1}{4}(\eta-\etabar)\cdot(\eta+\etabar)\right.\\\nonumber 
\left.-\frac{1}{8}|\eta+\etabar|^{2}\nonumber+\frac{1}{2}\rho+\frac{1}{4}\mathfrak{T}_{43}|d\ubar^{''}\right).
\end{eqnarray}
Once again, we observe that all the terms except $\rho$ enjoy estimates that are controlled by $\mathcal{O}_{0},\mathcal{W},\mathcal{F}$, and  $\mathcal{F}(S)$. Therefore, we focus on the term $\rho$. The previous inequality reduces to 
\begin{eqnarray}
||\omegabar||_{L^{\infty}(S_{u,\ubar})}\leq\\\nonumber  C||\omegabar||_{L^{\infty}(S_{u,\ubar^{'}})}+C(\mathcal{O}_{0})\int_{0}^{\ubar}||\omegabar||_{L^{\infty}(S_{u,\ubar})}d\ubar^{'}+\mathcal{C}(\mathcal{O}_{0},\mathcal{F}(S))+C\sup_{S_{u,\ubar}}(\int_{0}^{\ubar}\rho^{2}d\ubar^{'})^{\frac{1}{2}}.
\end{eqnarray}
Now notice $\sup_{S_{u,\ubar}}(\int_{0}^{\ubar}\rho^{2}d\ubar^{'})^{\frac{1}{2}}=||\rho||_{tr(H)}$. Now we follow the same procedure as before i.e., utilize the trace inequality (\ref{eq:trace2}). Since $\omega$ does not satisfy a $\nabla_{4}$ equation, we want to get rid of $\omega$. Write $\rho$ as follows after the re-scaling $\omegabar=\Omega\Tilde{\omegabar}$
\begin{eqnarray}
\rho=\Omega\nabla_{4}\tilde{\omegabar}-\underbrace{\frac{3}{4}|\eta+\etabar|^{2}+\frac{1}{4}(\eta-\etabar)\cdot(\eta+\etabar)+\frac{1}{8}|\eta+\etabar|^{2}\nonumber-\frac{1}{4}((\rho^{F})^{2}+(\sigma^{F})^{2})}_{V}.
\end{eqnarray}
Here the term $V$ is harmless following the bootstrap (\ref{eq:boot1}) and the previous estimate (\ref{eq:es1}). Therefore, we only focus on the term $\Omega\nabla_{4}\tilde{\omegabar}$ and noting $||\Omega||_{L^{\infty}(S_{u,\ubar})}\leq C(\mathcal{O}_{0})$ under bootstrap, estimating $||\nabla_{4}\tilde{\omegabar}||_{tr(H)}$ suffices. Now we use the trace inequality to estimate $||\nabla_{4}\tilde{\omegabar}||_{tr(H)}$
\begin{eqnarray}
||\nabla_{4}\tilde{\omegabar}||_{tr(H)}\leq \left(||\nabla^{2}_{4}\tilde{\omegabar}||_{L^2(H)}+||\nabla_{4}\tilde{\omegabar}||_{L^{2}(H)}\right)^{1/2}||\nabla^{2}\tilde{\omegabar}||^{1/2}_{L^{2}(H)}\nonumber+||\nabla_{4}\nabla \tilde{\omegabar}||_{L^{2}(H)}+||\nabla \tilde{\omegabar}||_{L^{2}(H)}.
\end{eqnarray}
Following the same procedure as before and controlling $||\nabla_{4}\etabar||_{L^{2}(H)}$ by $||\nabla^{2}\etabar||_{L^{2}(H)}$ and $\mathcal{W}$. Collecting all the terms together, one obtains 
\begin{eqnarray}
||\omegabar||_{L^{\infty}(S_{u,\ubar})}\leq C||\omegabar||_{L^{\infty}(S_{u,0})}+\mathcal{C}(\mathcal{O}_{0},\mathcal{W},\mathcal{F},\mathcal{F}(S),||\nabla^{2}\varphi||_{L^{2}(H)})\nonumber+C(\mathcal{O}_{0})\int_{0}^{\ubar}||\omegabar||_{L^{\infty}(S_{u,\ubar})}d\ubar^{'}
\end{eqnarray}
which upon utilizing Gr\"onwall's inequality yields 
\begin{eqnarray}
||\omegabar||_{L^{\infty}(S_{u,\ubar})}\leq \mathcal{C}(\mathcal{O}_{0},\mathcal{W},\mathcal{F},\mathcal{F}(S),||\nabla^{2}\varphi||_{L^{2}(H)}). 
\end{eqnarray}
Now we estimate the good connection coefficients $\varphi_{g}$ i.e., the ones satisfying $\nabla_{3}$ equations. We use transport inequality (proposition \ref{transport}). First consider $ \tr\chibar$ and use that fact that $||\hnabla\bar{\alpha}^{F}||_{L^{4}(S)}$ is controlled by $\mathcal{F}$ together with Sobolev embedding (\ref{eq:sobolev4}) 
\begin{eqnarray}
|| \tr\chibar||_{L^{\infty}(S_{u,\ubar})}\leq\\\nonumber  \frac{C}{2}\left(|| \tr\chibar||_{L^{\infty}(S_{0,\ubar})}+\int_{0}^{u}||-\frac{1}{2}( \tr\chibar)^{2}-|\chibarhat|^{2}_{\gamma}-2\omegabar \tr\chibar-\mathfrak{T}_{33}||_{L^{\infty}(S_{u,\ubar})}du^{''}\right)\\\nonumber 
\leq \frac{C\mathcal{O}_{0}}{2}+\epsilon C(\mathcal{W},\mathcal{F},\mathcal{F}(S))\leq C\mathcal{O}_{0}
\end{eqnarray}
if we choose $\epsilon>0$ sufficiently small. Now consider $\chibarhat$ 
\begin{eqnarray}
||\chibarhat||_{L^{\infty}(S_{u,\ubar})}\leq \frac{C}{2}\left(||\chibarhat||_{L^{\infty}(S_{0,\ubar})}+\sup_{S_{u,\ubar}}\int_{0}^{u}|\nabla_{3}\chibarhat|du^{'}\right)\\\nonumber 
\leq \frac{C\mathcal{O}_{0}}{2}+\frac{C}{2}\int_{0}^{u}||- \tr\chibar\chibarhat-2\omegabar\chibarhat||_{L^{\infty}(S)}du^{'}+\frac{C}{2}\epsilon^{\frac{1}{2}}\sup_{S_{u,\ubar}}(\int_{0}^{\ubar}\bar{\alpha}^{2}d\ubar^{'})^{\frac{1}{2}}.
\end{eqnarray}
Now we need to estimate $\sup_{S_{u,\ubar}}(\int_{0}^{\ubar}\bar{\alpha}^{2}d\ubar^{'})^{\frac{1}{2}}=||\bar{\alpha}||_{tr(  \Hbar)}$. Proceeding exactly the similar way as before i.e., write $\alpha$ as follows 
\begin{eqnarray}
\bar{\alpha}=-\Omega\nabla_{3}\tilde{\chibarhat}- \tr\chibar\chibarhat
\end{eqnarray}
where $\tilde{\chibarhat}=\frac{1}{\Omega}\chibarhat$ since we do not have a $\nabla_{3}$ equation for $\omegabar$ and $||\Omega||_{L^{\infty}(S)}\leq C(\mathcal{O}_{0})$ under the bootstrap assumption \ref{eq:bootstrapinitial}. Proceeding exactly the same way and collecting all the terms we obtain
\begin{eqnarray}
||\chibarhat||_{L^{\infty}(S_{u,\ubar})}&\leq& \frac{C\mathcal{O}_{0}}{2}\nonumber+\epsilon C(\mathcal{O}_{0},\mathcal{W},\mathcal{F},\mathcal{F}(S),||\nabla^{2}(\etabar,\eta)||_{L^{2}(H)})\\\nonumber &&+\epsilon^{\frac{1}{2}} C(\mathcal{O}_{0},\mathcal{W},\mathcal{F},\mathcal{F}(S),||\nabla^{2}(\etabar,\eta)||_{L^{2}(H)})\\\nonumber 
&\leq& C\mathcal{O}_{0}
\end{eqnarray}
after choosing sufficiently small $\epsilon$. Since in the trace estimate we need $||\nabla^{2}_{3}\tilde{\chibarhat}||_{L^{2}(  \Hbar)}$, we will need $||\nabla_{3}\bar{\alpha}||_{L^{2}(  \Hbar)}$ and therefore we include $||\nabla_{3}\bar{\alpha}||_{L^{2}(  \Hbar)}$ in the definition of $\mathcal{W}$ from the beginning. Exact similar procedure yields 
\begin{eqnarray}
||\etabar||_{L^{\infty}(S)}\leq C\mathcal{O}_{0}.
\end{eqnarray}
Now for $ \tr\chi$ and $\widehat{\chi}$, we will encounter terms that can be estimated by $||\nabla^{2}(\eta,\etabar)||_{L^{2}(H,  \Hbar)}, \mathcal{W},\mathcal{F}$, and $\mathcal{F}(S)$. The transport inequality (proposition \ref{transport}) yields
\begin{eqnarray}
|| \tr\chi||_{L^{\infty}(S_{u,\ubar})}\leq C|| \tr\chi||_{L^{\infty}(S_{0,\ubar})}+C\int_{0}^{u}||-\frac{1}{2} \tr\chibar \tr\chi+2\omegabar \tr\chi\nonumber+2|\eta|^{2}-\widehat{\chi}\cdot \chibarhat||_{L^{\infty}(S)}\\\nonumber +C\epsilon^{\frac{1}{2}}\sup_{S_{u,\ubar}}(\int_{0}^{u}(\nabla\eta)^{2}du^{'})^{\frac{1}{2}}
+C\epsilon^{\frac{1}{2}}\sup_{S_{u,\ubar}}(\int_{0}^{u}\rho^{2}du^{'})^{\frac{1}{2}}.
\end{eqnarray}
Therefore notice that we need to estimate two terms $\sup_{S_{u,\ubar}}(\int_{0}^{u}(\nabla\eta)^{2}du^{'})^{\frac{1}{2}}=||\nabla\eta||_{tr(  \Hbar)}$ and $\sup_{S_{u,\ubar}}(\int_{0}^{u}\rho^{2}du^{'})^{\frac{1}{2}}:=||\rho||_{tr(  \Hbar)}$. We have the following estimate 
\begin{eqnarray}
|| \tr\chi||_{L^{\infty}(S_{u,\ubar})}\leq \frac{C\mathcal{O}_{0}}{2}+\epsilon C(\mathcal{O}_{0},\mathcal{W},\mathcal{F})+C\epsilon^{\frac{1}{2}}||\nabla\eta||_{tr(  \Hbar)}+C\epsilon^{\frac{1}{2}}||\rho||_{tr(  \Hbar)}\leq C\mathcal{O}_{0}.
\end{eqnarray}
Similarly, for $\widehat{\chi}$, we obtain through the use of the transport inequality (proposition \ref{transport})  
\begin{eqnarray}
||\widehat{\chi}||_{L^{\infty}(S_{u,\ubar})}\leq \frac{C\mathcal{O}_{0}}{2}+\epsilon C(\mathcal{O}_{0},\mathcal{W},\mathcal{F})+C\epsilon^{\frac{1}{2}}||\nabla\eta||_{tr(  \Hbar)}\leq C\mathcal{O}_{0},
\end{eqnarray}
where we can estimate $||\rho||_{tr(  \Hbar)}$ and $||\nabla\eta||_{tr(  \Hbar)}$ in terms of $||\nabla^{2}\varphi||_{L^{2}(  \Hbar)}$ and $(\mathcal{W}(S),\mathcal{F}(S),\mathcal{W},\mathcal{F})$ in a similar way as the other entities such as $\omegabar$ and $\eta$ and therefore we omit the calculations. The most important point to note here is the presence of a smallness factor $\epsilon^{\frac{1}{2}}$ with $tr(  \Hbar)$ norm of $\rho$ and $\nabla\eta$. This concludes the $L^{\infty}$ estimates for the connection coefficients.
\end{proof}
\begin{lemma}
\label{2}
Assume $\mathcal{W},\mathcal{F},\mathcal{F}(S),||\nabla^{2}\varphi||_{L^{2}(H,  \Hbar)}<\infty$. Then $||\varphi_{g}||_{L^{4}(S)}\leq C\mathcal{O}_{0}$ and $||\varphi_{b}||_{L^{4}(S)}\leq C(\mathcal{O}_{0},\mathcal{W},\mathcal{F},\mathcal{F}(S),||\nabla^{2}\varphi||_{L^{2}(H,  \Hbar)})$, where $\varphi_{g}:=(\widehat{\chi},\chibarhat,  \tr\chi,  \tr\chibar,\etabar,\omega)$ and $\varphi_{b}:=(\eta,\omegabar)$.
\end{lemma} 
\begin{proof}
 We prove it under the assumption 
\begin{eqnarray}
\label{eq:boot2}
||\varphi_{g}||_{L^{4}(S)}\leq C\mathcal{O}_{0}.
\end{eqnarray}
and later try to improve it. We start with $\eta$ since it satisfies a $\nabla_{4}$ equation. A direct application of the transport inequality (proposition \ref{transport}) yields 
\begin{eqnarray}
 ||\eta||_{L^{4}(S_{u,\ubar})}\leq C\left(||\eta||_{L^{4}(S_{u,0})}+\int_{0}^{\ubar}||\nabla_{4}\eta||_{L^{4}(S)}d\ubar^{'}\right)\\\nonumber 
 =C\left(||\eta||_{L^{4}(S_{u,0})}+\int_{0}^{\ubar}||-\chi\cdot(\eta-\etabar)-\beta-\frac{1}{2}\alpha^{F}\cdot(\rho^{F}-\sigma^{F})||_{L^{4}(S)}d\ubar^{'}\right)\\\nonumber 
 \leq C(\mathcal{O}_{0},\mathcal{W}(S),\mathcal{F}(S))+C(\mathcal{O}_{0})\int_{0}^{u}||\eta||_{L^{4}(S)}d\ubar^{'}.
 \end{eqnarray}
 An application of Gr\"onwall's inequality leads to 
 the desired estimate for $\eta$
\begin{eqnarray}
||\eta||_{L^{4}(S)}\leq C(\mathcal{O}_{0},\mathcal{W}(S),\mathcal{F}(S))e^{C(\mathcal{O}_{0})\ubar}\leq  C(\mathcal{O}_{0},\mathcal{W}(S),\mathcal{F}(S))
\end{eqnarray}
since $\ubar\leq J$. Now we repeat the same procedure for $\omegabar$ which also satisfies a $\nabla_{4}$ transport equation 
\begin{eqnarray}
||\omegabar||_{L^{4}(S_{u,\ubar})}\leq C\left(||\omegabar||_{L^{4}(S_{u,0})}\nonumber+\int_{0}^{\ubar}||\nabla_{4}\omegabar||_{L^{4}(S)}d\ubar^{'}\right)\\\nonumber 
=C\left(||\omegabar||_{L^{4}(S_{u,0})}+\int_{0}^{\ubar}||2\omega\omegabar+\frac{3}{4}|\eta-\etabar|^{2}-\frac{1}{4}(\eta-\etabar)\cdot(\eta+\etabar)-\frac{1}{8}|\eta+\etabar|^{2}\right.\\\nonumber
\left.+\frac{1}{2}\rho+\rho^{F}\cdot\rho^{F}+\sigma^{F}\cdot\sigma^{F}||_{L^{4}(S)}d\ubar^{'}\right)\\\nonumber 
\leq C(\mathcal{O}_{0},\mathcal{W},\mathcal{F},\mathcal{W}(S),\mathcal{F}(S),||\nabla^{2}\varphi||_{L^{2}(H)})+C(\mathcal{O}_{0})\int_{0}^{\ubar}||\omegabar||_{L^{4}(S)}d\ubar^{'}\\\nonumber 
\leq C(\mathcal{O}_{0},\mathcal{W},\mathcal{F},\mathcal{W}(S),\mathcal{F}(S),||\nabla^{2}\varphi||_{L^{2}(H)}) e^{C(\mathcal{O}_{0}J)}\\\nonumber 
\leq  C(\mathcal{O}_{0},\mathcal{W},\mathcal{F},\mathcal{W}(S),\mathcal{F}(S),||\nabla^{2}\varphi||_{L^{2}(H)}),
\end{eqnarray}
where the last step follows from an application of the Gr\"onwall's inequality and $\ubar\leq J$.
Notice an extremely important fact: $||\eta||_{L^{4}(S)}$only depends on $\mathcal{O}_{0},\mathcal{W}(S),$ and $\mathcal{F}(S)$, whereas $||\omegabar||_{L^{4}(S_{u,\ubar})}$ depends on $\mathcal{O}_{0},\mathcal{W},\mathcal{F},\mathcal{W}(S),\mathcal{F}(S),||\nabla^{2}\varphi||_{L^{2}(H)}$. This will be vital when we show $||\nabla^{2}\varphi||_{L^{2}(H/  \Hbar)}<\infty$ given $\mathcal{W},\mathcal{W}(S),\mathcal{F},\mathcal{F}(S)<\infty$. Now we move on to estimating $\varphi_{g}$ i.e., the connection coefficients that satisfy $\nabla_{3}$ equation. We start with $\etabar$. Using the assumption (\ref{eq:boot2}) together with the previous estimate (lemma 6.2) and estimates for $||\eta||_{L^{4}(S)}$, we obtain through the transport inequality (proposition \ref{transport}) 
\begin{eqnarray}
||\etabar||_{L^{4}(S)}\leq \frac{C}{2}\left(||\etabar||_{L^{4}(S_{0,\ubar})}+\int_{0}^{u}||\nabla_{3}\etabar||_{L^{4}(S)}du^{'}\right)\\\nonumber 
\leq \frac{C\mathcal{O}_{0}}{2}+\epsilon C(\mathcal{O}_{0},\mathcal{W}(S),\mathcal{F}(S))\leq C\mathcal{O}_{0}
\end{eqnarray}
for sufficiently small $\epsilon$. Estimate for $||\tr \chi||$ (and similar others involving $\nabla\eta/\etabar$) follows in a similar way but now we need to use $||\nabla\varphi||_{L^{4}(S)}\lesssim ||\nabla^{2}\varphi||_{L^{2}(S)}+||\nabla\varphi||_{L^{2}(S)}+||\varphi||_{L^{2}(S)}$ Sobolev embedding to close the estimate. transport inequality (proposition \ref{transport}) yields 
\begin{eqnarray}
||\tr\chi||_{L^{4}(S)}\leq \frac{C}{2}\nonumber\left(|| \tr\chi||_{L^{4}(S_{0,\ubar})}+\int_{0}^{u}||\nabla_{3} \tr\chi||_{L^{4}(S)}du^{'}\right)\\\nonumber 
=\frac{C}{2}\left(|| \tr\chi||_{L^{4}(S_{0,\ubar})}+\int_{0}^{u}||-\frac{1}{2} \tr\chibar \tr\chi+2\omegabar \tr\chi+2div\eta+2|\eta|^{2}+2\rho-\widehat{\chi}\cdot \chibarhat||_{L^{4}(S)}du^{'}\right)\\\nonumber 
\leq \frac{C\mathcal{O}_{0}}{2}+\epsilon C(\mathcal{O}_{0},\mathcal{W},\mathcal{F},\mathcal{W}(S),\mathcal{F}(S),||\nabla^{2}\varphi||_{L^{2}(H)})+\epsilon^{\frac{1}{2}}C(||\nabla^{2}\eta||_{L^{2}(  \Hbar)})\leq C\mathcal{O}_{0}
\end{eqnarray}
for sufficiently small $\epsilon$. Through a similar argument, we obtain the improved estimates for the remaining connection coefficients that is we prove
\begin{eqnarray}
||\varphi_{g}||_{L^{4}(S)}\leq C\mathcal{O}_{0}.
\end{eqnarray}
This concludes the proof of the lemma.
\end{proof}
These estimates will be extremely crucial in proving the following lemma as well as estimating $||\nabla^{2}\varphi||_{L^{2}(H/  \Hbar)}$ in terms of $\mathcal{W},\mathcal{F},\mathcal{W}(S),$ and $\mathcal{F}(S)$ in near future (lemma \ref{lemma5}). Notice another important point: since $||\eta||_{L^{4}(S)}\leq C(\mathcal{O}_{0},\mathcal{W},\mathcal{F},\mathcal{W}(S),\mathcal{F}(S))$, we have $||\omegabar||_{L^{2}(S)}\leq C(\mathcal{O}_{0},\mathcal{W},\mathcal{F},\mathcal{W}(S),\mathcal{F}(S))$ i.e., $L^{2}(S)$ of $\omegabar$ does not depend on $||\nabla^{2}\varphi||_{L^{2}(H)}$.

\begin{lemma}
\label{3}
\textit{Assume $\mathcal{W},\mathcal{F},\mathcal{F}(S),||\nabla^{2}\varphi||_{L^{2}(H,  \Hbar)}<\infty$. Then $||\nabla\varphi_{g}||_{L^{2}(S)}\leq C\mathcal{O}_{0}$ and $||\nabla\varphi_{b}||_{L^{2}(S)}\leq C(\mathcal{O}_{0},\mathcal{W},\mathcal{F},\mathcal{F}(S),||\nabla^{2}\varphi||_{L^{2}(H,  \Hbar)})$, where $\varphi_{g}:=(\widehat{\chi},\chibarhat,  \tr\chi,  \tr\chibar,\etabar,\omega)$ and $\varphi_{b}:=(\eta,\omegabar)$}.
\end{lemma} 
\begin{proof} We will prove this under the assumption 
\begin{eqnarray}
\label{eq:boot}
||\nabla\varphi_{g}||_{L^{2}(S)}\leq 2C\mathcal{O}_{0}.
\end{eqnarray}
 We will obtain a better estimate, therefore, closing the argument. As usual, we first start with $\eta$ since it satisfies a $\nabla_{4}$ equation. We commute $\nabla$ with the transport equation (\ref{eq:eta}) satisfied by $\eta$ to obtain (we write it in a schematic way) 
 \begin{eqnarray}
 \nabla_{4}\nabla\eta=-\nabla\chi(\eta-\etabar)-\chi(\nabla\eta-\nabla\etabar)-\nabla\beta-\frac{1}{2}\nabla\alpha^{F}\cdot(\rho^{F}-\sigma^{F})-\frac{1}{2}\alpha^{F}(\nabla\rho^{F}-\nabla\sigma^{F})\\\nonumber 
 +(\beta+\alpha^{F}\cdot(\rho^{F}-\sigma^{F}))\eta+(\eta+\etabar)(-\chi(\eta-\etabar)-\beta-\frac{1}{2}\alpha^{F}(\rho^{F}-\sigma^{F})).
 \end{eqnarray}
 The transport inequality (proposition \ref{transport}) for $\nabla\eta$ reads 
 \begin{eqnarray}
 ||\nabla\eta||_{L^{2}(S_{u,\ubar})}\leq C\left(||\nabla\eta||_{L^{2}(S_{u,0})}+\int_{0}^{\ubar}||\nabla_{4}\nabla\eta||_{L^{2}(S)}d\ubar^{'}\right).
 \end{eqnarray}
 Under the assumption (\ref{eq:boot}) and the  previous estimates (lemma \ref{1}-\ref{3}) we obtain
 \begin{eqnarray}
 ||\nabla_{4}\nabla\eta||_{L^{2}(S)}\leq C(\mathcal{O}_{0},\mathcal{W},\mathcal{F},\mathcal{F}(S),||\nabla^{2}\varphi||_{L^{2}(H)})+C(\mathcal{O}_{0})\int_{0}^{\ubar}||\nabla\eta||_{L^{2}(S)}d\ubar^{'}
 \end{eqnarray}
 which after using Gr\"onwall's inequality yields 
 \begin{eqnarray}
 \label{eq:nablaeta}
 ||\nabla\eta||_{L^{2}(S)}\leq C(\mathcal{O}_{0},\mathcal{W},\mathcal{F},\mathcal{F}(S),||\nabla^{2}\varphi||_{L^{2}(H)})e^{C(\mathcal{O}_{0})\ubar}\leq\\\nonumber  C(\mathcal{O}_{0},\mathcal{W},\mathcal{F},\mathcal{F}(S),||\nabla^{2}\varphi||_{L^{2}(H)})
 \end{eqnarray}
 since $\ubar\leq J$.
 A similar argument for $\omegabar$ yields 
 \begin{eqnarray}
 ||\nabla\omegabar||_{L^{2}(S)}\leq C(\mathcal{O}_{0},\mathcal{W},\mathcal{F},\mathcal{F}(S),||\nabla^{2}\varphi||_{L^{2}(H)}).
 \end{eqnarray}
 Now we want to estimate the good connection coefficients i.e., the ones that satisfy $\nabla_{3}$ equation. Let us start with $\etabar$. Commuting the transport equation of $\etabar$ with $\nabla$ yields
 \begin{eqnarray}
 \nabla_{3}\nabla\etabar=\\\nonumber -\nabla\bar{\chi}(\etabar-\eta)-\chi(\nabla\eta-\nabla\etabar)+\nabla\bar{\beta}+\frac{1}{2}\hnabla\bar{\alpha}^{F}\cdot(\rho^{F}-\sigma^{F})+\frac{1}{2}\bar{\alpha}^{F}\cdot(\nabla\rho^{F}-\nabla\sigma^{F})\\\nonumber+(\bar{\beta}+\bar{\alpha}^{F}\cdot(\rho^{F}-\sigma^{F}))\etabar+(\eta+\etabar)(-\bar{\chi}(\etabar-\eta)+\bar{\beta}+\frac{1}{2}\bar{\alpha}^{F}\cdot(\rho^{F}-\sigma^{F})).
 \end{eqnarray}
Similarly, we may apply the transport inequality (proposition \ref{transport}) for $||\nabla\etabar||_{L^{2}(S)}$ to obtain 
\begin{eqnarray}
||\nabla\etabar||_{L^{2}(S)}\leq \frac{C}{2}\left(||\nabla\etabar||_{L^{2}(S_{0,u})}+\int_{0}^{u}||\nabla_{3}\nabla\etabar||_{L^{2}(S)}du^{'}\right)\\\nonumber 
\leq \frac{C\mathcal{O}_{0}}{2}+\epsilon C(\mathcal{O}_{0},\mathcal{F},\mathcal{W}(S),\mathcal{F}(S),||\nabla^{2}\varphi||_{L^{2}(H)})+\epsilon^{\frac{1}{2}}C(\mathcal{W})
\end{eqnarray}
which can be made to satisfy the following estimate after choosing sufficiently small $\epsilon$ yields 
\begin{eqnarray}
||\nabla\etabar||_{L^{2}(S)}\leq C\mathcal{O}_{0}.
\end{eqnarray}
Now we estimate $ \tr\chi$. Commuting $\nabla$ with the transport equation of $ \tr\chi$ yields the following equation (schematic)
\begin{eqnarray}
\nabla_{3}\nabla  \tr\chi+\frac{1}{2}\nabla  \tr\chibar \tr\chi+\frac{1}{2} \tr\chibar\nabla  \tr\chi=\\\nonumber 2\nabla\omegabar \tr\chi+2\omegabar\nabla  \tr\chi+2\nabla^{2}\eta+2\eta\nabla\eta+2\nabla\rho-\nabla\widehat{\chi}\chibarhat-\widehat{\chi}\nabla\chibarhat.  
\end{eqnarray}
The transport inequality (proposition \ref{transport}) provides the following estimate for $\nabla  \tr\chi$ under the assumption (\ref{eq:boot}) together with the previous estimates for $\eta$ and $\omegabar$
\begin{eqnarray}
||\nabla  \tr\chi||_{L^{2}(S)}\leq \frac{C}{2}\left(||\nabla  \tr\chi||_{L^{2}(S_{0,\ubar})}+\int_{0}^{u}||2\nabla\omegabar \tr\chi+2\omegabar\nabla  \tr\chi+2\nabla^{2}\eta\right.\\\nonumber 
\left.+2\eta\nabla\eta\nonumber+2\nabla\rho-\nabla\widehat{\chi}\chibarhat-\widehat{\chi}\nabla\chibarhat||_{L^{2}(S)}du^{'}\right.\\\nonumber 
\left.\int_{0}^{u}||\frac{1}{2}\nabla  \tr\chibar \tr\chi+\frac{1}{2} \tr\chibar\nabla  \tr\chi||_{L^{2}(S)}du^{'}\right)\\\nonumber 
\leq \frac{C\mathcal{O}_{0}}{2}+\epsilon C(\mathcal{O}_{0},\mathcal{F},\mathcal{W},\mathcal{F}(S),||\nabla^{2}\varphi||_{L^{2}(H)})+\epsilon^{\frac{1}{2}}C(||\nabla^{2}\varphi||_{L^{2}(  \Hbar)},\mathcal{W})\leq C\mathcal{O}_{0}
\end{eqnarray}
for sufficiently small $\epsilon$. Notice that here we needed $||\nabla^{2}\varphi||_{L^{2}(  \Hbar)}$ or more precisely $||\nabla^{2}\eta||_{L^{2}(  \Hbar)}$. The remaining connection coefficients are estimated in an exactly similar way. We sketch the estimates below  
\begin{eqnarray}
||\nabla\chibarhat||_{L^{2}(S_{u,\ubar})}\leq \frac{C}{2}\left(||\nabla\chibarhat||_{L^{2}(S_{u^{'},\ubar})}+\int_{u^{'}}^{u}||\nabla_{3}\nabla\chibarhat||_{L^{2}(S_{u^{''},\ubar)}}du^{''}\right)\\\nonumber 
\leq \frac{C\mathcal{O}_{0}}{2}+\epsilon^{1/2}C(\mathcal{O}_{0},\mathcal{W},\mathcal{W}(S)\mathcal{F},\mathcal{F}(S))\leq C\mathcal{O}_{0},\\
||\nabla\etabar||_{L^{2}(S_{u,\ubar})}\leq C\left(||\nabla\hat{\etabar}||_{L^{2}(S_{u^{'},\ubar})}+\int_{u^{'}}^{u}||\nabla_{3}\nabla\hat{\etabar}||_{L^{2}(S_{u^{''},\ubar)}}du^{''}\right)\\\nonumber 
\leq \frac{C\mathcal{O}_{0}}{2}+\epsilon^{1/2}C(\mathcal{O}_{0},\mathcal{W},\mathcal{W}(S)\mathcal{F},\mathcal{F}(S))\leq C\mathcal{O}_{0},\\
||\nabla\omega||_{L^{2}(S_{u,\ubar})}\leq C\left(||\nabla\omega||_{L^{2}(S_{u^{'},\ubar})}+\int_{u^{'}}^{u}||\nabla_{3}\nabla\omega||_{L^{2}(S_{u^{''},\ubar)}}du^{''}\right)\\\nonumber 
\leq \frac{C\mathcal{O}_{0}}{2}+\epsilon^{1/2}C(\mathcal{O}_{0},\mathcal{W},\mathcal{W}(S)\mathcal{F},\mathcal{F}(S))\leq C\mathcal{O}_{0},\\ 
||\nabla  \tr\chi||_{L^{2}(S_{u,\ubar})}\leq C\left(||\nabla  \tr\chi||_{L^{2}(S_{u^{'},\ubar})}\nonumber+\int_{u^{'}}^{u}||\nabla_{3}\nabla  \tr\chi||_{L^{2}(S_{u^{''},\ubar)}}du^{''}\right)\\\nonumber 
\leq \frac{C\mathcal{O}_{0}}{2}+\epsilon^{1/2}C(\mathcal{O}_{0},\mathcal{W},\mathcal{W}(S)\mathcal{F},\mathcal{F}(S),||\nabla^{2}\eta||_{L^{2}(  \Hbar)})\leq C\mathcal{O}_{0},\\
||\nabla\widehat{\chi}||_{L^{2}(S_{u,\ubar})}\leq C\left(||\nabla\widehat{\chi}||_{L^{2}(S_{u^{'},\ubar})}\nonumber+\int_{u^{'}}^{u}||\nabla_{3}\nabla\widehat{\chi}||_{L^{2}(S_{u^{''},\ubar)}}du^{''}\right)\\\nonumber 
\leq \frac{C\mathcal{O}_{0}}{2}+\epsilon^{1/2}C(\mathcal{O}_{0},\mathcal{W},\mathcal{W}(S)\mathcal{F},\mathcal{F}(S),||\nabla^{2}\eta||_{L^{2}(  \Hbar)})\leq C\mathcal{O}_{0}
\end{eqnarray}
after choosing $\epsilon$ sufficiently small.
In this case, we will hit $\mathfrak{T}_{ab}\sim \rho^{F}\rho^{F}+\alpha^{F}\bar{\alpha}^{F}+\sigma^{F}\sigma^{F}$ by $\hnabla_{3}$ derivative and therefore we must incorporate $||\hnabla_{3}\bar{\alpha}^{F}||_{L^{2}(  \Hbar)}$ in the $\mathcal{F}$ norm. This completes the proof of the lemma.
\end{proof}
\begin{corollary}
\label{C1}
The Gauss curvature $K$ of the topological $2$-sphere $S_{u,\ubar}$ satisfies 
\begin{eqnarray}
||K||_{L^{4}(S_{u,\ubar})}\leq C(\mathcal{O}_{0},\mathcal{W}(S),\mathcal{F}(S)),\\
||\nabla K||_{L^{2}(S_{u,\ubar})}\leq C(\mathcal{O}_{0},\mathcal{F}(S))+||\nabla\rho||_{L^{2}(S_{u,\ubar})}.
\end{eqnarray}
\end{corollary}
\textbf{Proof:} A direct consequence of the null Hamiltonian constraint (\ref{eq:4}), lemma (\ref{1}), (\ref{2}), (\ref{3}), and the definitions of $\mathcal{W}(S)$ and $\mathcal{F}(S)$.

\noindent \begin{lemma}
\label{lemma5} Let $\varphi_{g}:=( \tr\chibar,\chibarhat,\etabar,\omega, \tr\chi,\widehat{\chi})$ and $\varphi_{b}:=(\eta,\omegabar)$ and $\mathcal{W},\mathcal{F},\mathcal{F}(S)<\infty$. Then $||\nabla^{2}\varphi_{g}||_{L^{2}(H,  \Hbar)}\leq C(\mathcal{O}_{0},\mathcal{W},\mathcal{F},\mathcal{W}(S),\mathcal{F}(S))$ and $||\nabla^{2}\varphi_{b}||_{L^{2}(H,  \Hbar)}\leq C(\mathcal{O}_{0},\mathcal{W},\mathcal{F},\mathcal{W}(S),\mathcal{F}(S))$.
\end{lemma}
\begin{proof}
Following \cite{klainerman2012formation}, we prove this lemma by means of constructing a transport-Hodge system.
The basic idea is the following. We construct a set of new entities $\Xi:=\{\nabla  \tr\chi,\mu,\bar{\xi}\}$ and $\bar{\Xi}:=\{\nabla  \tr\chibar,\bar{\mu},\xi\}$. We first obtain their transport equations. We proceed exactly the same way as \cite{klainerman2012formation}, only keeping track of the additional Yang-Mills curvature terms. Yang-Mills curvature terms are harmless in this context since they have one order higher regularity than the Weyl curvature. We define $\mu,\bar{\mu}$ and $\xi,\bar{\xi}$ as follows 
\begin{eqnarray}
\mu:=-\text{div} \eta-\rho,~\bar{\mu}:-\text{div} \etabar-\rho, ~\xi:=\nabla\omega+~^{*}\nabla\omega^{\dag}-\frac{1}{2}\beta,\\
\bar{\xi}:=-\nabla\omegabar+~^{*}\nabla\omegabar^{\dag}-\frac{1}{2}\bar{\beta},
\end{eqnarray}
where $\omega^{\dag}$ and $\omegabar^{\dag}$ are the auxiliary variables that satisfy the following boundary-valued equations
\begin{eqnarray}
\nabla_{3}\omega^{\dag}=\frac{1}{2}\sigma,~\omega^{\dag}=0~on~H_{0},~\nabla_{4}\omegabar^{\dag}=\frac{1}{2}\sigma,~\omegabar^{\dag}=0~on~  \Hbar_{0}.
\end{eqnarray}
By definition, we have the Hodge system 
\begin{eqnarray}
\label{eq:hodge2}
^{*}\mathcal{D}_{1}\langle\omegabar\rangle=\bar{\xi}+\frac{1}{2}\bar{\beta},~~^{*}\mathcal{D}_{1}\langle\omega\rangle=\xi+\frac{1}{2}\beta,\\
\label{eq:Hodge1}
\text{div}\eta=-\mu-\rho,~\text{curl} \eta=\chibarhat\wedge\widehat{\chi}+\sigma,\\
\text{div}\widehat{\chi}=\frac{1}{2}\nabla  \tr\chi-\frac{1}{2}(\eta-\etabar)\cdot(\widehat{\chi}-\frac{1}{2} \tr\chi\delta_{ab})-\beta+\frac{1}{2}\mathfrak{T}(e_{4},\cdot)\\
\label{eq:hodge4}
\text{div}\chibarhat=\frac{1}{2}\nabla  \tr\chibar-\frac{1}{2}(\etabar-\eta)\cdot(\chibarhat-\frac{1}{2} \tr\chibar\delta_{ab})-\bar{\beta}+\frac{1}{2}\mathfrak{T}(e_{3},\cdot),
\end{eqnarray}
where $^{*}\mathcal{D}_{1}\langle\omega\rangle:=\nabla\omega+~^{*}\nabla\omega^{\dag}$ and $^{*}\mathcal{D}_{1}\langle\omegabar\rangle:=-\nabla\omegabar+~^{*}\nabla\omegabar^{\dag}$.   $\Xi=\mu$ satisfies the following type of transport equation (schematically)
\begin{eqnarray}
\nonumber\nabla_{4}\mu=\\\nonumber
\underbrace{-\text{div}(-\chi\cdot(\eta-\etabar))}_{A}-\underbrace{\text{div}\beta}_{-ID}-\underbrace{\text{div}(\frac{1}{2}\alpha^{F}(\rho^{F}-\sigma^{F}))
-(\underbrace{-\text{div}\beta}_{ID}-\frac{3}{2} \tr\chi\rho\nonumber-\frac{1}{2}\chibarhat\cdot\alpha +\zeta\cdot\beta+2\etabar\cdot\beta}_{IA}\\\nonumber-\underbrace{\frac{1}{4}(\alpha^{F}\cdot\hnabla_{3}\alpha^{F}-\omegabar|\alpha^{F}|^{2}+4\eta\alpha^{F}\cdot(\rho^{F}-\sigma^{F})
-\rho^{F}\cdot\hnabla_{4}\rho^{F}+\sigma^{F}\cdot\hnabla_{4}\sigma^{F}-2\etabar\alpha^{F}\cdot(\rho^{F}+\sigma^{F}))}_{IB}\\\nonumber 
+\underbrace{(\beta+(\alpha^{F}\cdot(\rho^{F}-\sigma^{F})+\alpha^{F}))\eta+(\eta+\etabar)(-\chi\cdot(\eta-\etabar)-\beta-\frac{1}{2}\alpha^{F}\cdot(\rho^{F}-\sigma^{F}))}_{IC} -\chi\nabla\eta+\chi\etabar\eta
\end{eqnarray}
where all the terms involving Yang-Mills curvature can be controlled by means of null Yang-Mills equations. The most important point is to note that $\text{div}\beta$ terms cancel each other in a point-wise way. This is the purpose of constructing the new function $\mu$ (and similar others) so that the derivative of the Weyl curvature does not appear since that would obstruct closing the regularity argument. Now observe an extremely important point: term $A$ would contain $\text{div}(\chi)(\eta-\etabar)$ and we can need to estimate this term in $L^{2}(S)$ in order to estimate $||\mu||_{L^{2}(S)}$ using the transport inequality (proposition \ref{transport}). Notice the following calculations 
\begin{eqnarray}
||A||_{L^{2}(S)}=||-\text{div}(-\chi\cdot(\eta-\etabar))||_{L^{2}(S)}\leq ||\chi||_{L^{\infty}(S)}(||\nabla\eta||_{L^{2}(S)}+||\nabla\etabar||_{L^{2}(S)})\\\nonumber+||\nabla\chi||_{L^{4}}(||\eta||_{L^{4}(S)}+||\etabar||_{L^{4}(S)})
\end{eqnarray}
Now since $\widehat{\chi}$ and $ \tr\chi$ both satisfy $\nabla_{3}$ equation, we may estimate $||\nabla\chi||_{L^{4}(S)}$ solely by means of the initial data $\mathcal{O}_{0}$ using codimension-1 trace inequality 
\begin{eqnarray}
||\nabla\varphi||_{L^{4}(S_{u,\ubar})}\leq \\\nonumber C\left(||\nabla\varphi||_{L^{4}(S_{0,\ubar})}+||\nabla\varphi||^{1/2}_{L^{2}(  \Hbar)}||\hnabla_{3}\nabla\varphi||^{1/4}_{L^{2}(  \Hbar)}(||\nabla\varphi||_{L^{2}(  \Hbar)}+||\nabla\nabla\varphi||_{L^{2}(  \Hbar)})^{1/4}\right)
\end{eqnarray}
given $\mathcal{W},\mathcal{W}(S),\mathcal{F},\mathcal{F}(S),||\nabla^{2}\varphi||_{L^{2}(  \Hbar)}<\infty$. This is because, we gain $\epsilon^{\frac{1}{4}}$ from the term $||\nabla\varphi||^{\frac{1}{2}}_{L^{2}(  \Hbar)}$. From lemma (\ref{2}), we have the estimates for $||\eta||_{L^{4}(S)}$ in terms of $\mathcal{W},\mathcal{F},\mathcal{W}(S),\mathcal{F}(S)$ and for $||\etabar||_{L^{4}(S)}$ in terms of $\mathcal{O}_{0}$. Therefore after choosing $\epsilon$ sufficiently small, we have 
\begin{eqnarray}
||A||_{L^{2}(S)}\leq C(\mathcal{O}_{0})||\nabla\eta||_{L^{2}(S)}+C(\mathcal{O}_{0},\mathcal{W},\mathcal{F},\mathcal{W}(S),\mathcal{F}(S))
\end{eqnarray}

\noindent Now an application of transport inequality (proposition \ref{transport}) 
\begin{eqnarray}
||\mu||_{L^{2}(S_{u,\ubar})}\leq C(\mathcal{O}_{0})\left(||\mu||_{L^{2}(S_{u,0})}+\int_{0}^{\ubar}||\nabla_{4}\mu||_{L^{2}(S)}d\ubar^{''}\right)
\end{eqnarray}
yields 
\begin{eqnarray}
\label{eq:mu1}
||\mu||_{L^{2}(S_{u,\ubar})}\leq C(\mathcal{O}_{0},\mathcal{W}_{0},\mathcal{W},\mathcal{F},\mathcal{W}(S),\mathcal{F}(S))+C(\mathcal{O}_{0})\int_{0}^{\ubar}||\nabla\eta||_{L^{2}(S)}d\ubar^{'}\\\nonumber+C\int_{0}^{\ubar}(||IA,IB,IC||_{L^{2}(S)})d\ubar^{'}
\end{eqnarray}
and from the Hodge system (\ref{eq:Hodge1})
\begin{eqnarray}
||\nabla\eta||_{L^{2}(S)}\leq C(\mathcal{O}_{0})(||K||_{L^{2}(S)}+||\rho||_{L^{2}(S)}+||\sigma||_{L^{2}(S)})+||\mu||_{L^{2}(S)}.
\end{eqnarray}
Therefore we are left to estimate $||IA||_{L^{2}(S)},||IB||_{L^{2}(S)},$ and $||IC||_{L^{2}(S)}$. Using the null Yang-Mills  equations and elementary inequality such as Holder's inequality, we obtain
\begin{eqnarray}
||IA||_{L^{2}(S)}+||IB||_{L^{2}(S)}+||IC||_{L^{2}(S)}\leq\\\nonumber  C(\mathcal{O}_{0},\mathcal{W},\mathcal{F}(S))+C(\mathcal{O}_{0})(||\alpha||_{L^{2}(S)}+||\beta||_{L^{2}(S)}\nonumber+||\rho||_{L^{2}(S)})
\end{eqnarray}
substitution of which in (\ref{eq:mu1}), an application of Gr\"onwall, and integration in $\ubar$ yield
\begin{eqnarray}
||\mu||_{L^{2}(S)}\leq C(\mathcal{O}_{0},\mathcal{W},\mathcal{F},\mathcal{F}(S))
\end{eqnarray}
which in turn yields 
\begin{eqnarray}
||\nabla\eta||_{L^{2}(S)}\leq C(\mathcal{O}_{0},\mathcal{W},\mathcal{F},\mathcal{F}(S))
\end{eqnarray}
since $||\alpha||_{L^{2}(S)},||\beta||_{L^{2}(S)},||\rho||_{L^{2}(S)}\leq C(\mathcal{W}(S))$. In this process, we obtain $||\nabla\eta||_{L^{2}(S)}$ independent of $||\nabla^{2}\varphi||_{L^{2}(H,\bar{H)}}$ and therefore improve (\ref{eq:nablaeta}). 

\noindent Similarly, analysing the pair $(\bar{\xi},\langle\omegabar\rangle)$, we can estimate $||\nabla\omegabar||_{L^{2}(S)}$ by means of \\\nonumber $C(\mathcal{O}_{0},\mathcal{W},\mathcal{F},\mathcal{W}(S),\mathcal{F}(S))$ utilizing the estimate for $||\nabla\eta||_{L^{2}(S)}$ which is now independent of $||\nabla^{2}\varphi||_{L^{2}(H,\bar{H)}}$. This is the whole point of obtaining estimates in a hierarchical way i.e., start with $\eta$ and estimate $\omegabar$ by using the estimate for $\eta$ and then continue to do so for the remaining connection coefficients.   
An extremely important point is that $\bar{\alpha}$ does not appear due to the special structure of the null-Bianchi equations (recall $\bar{\alpha}$ can not be controlled on $H$). Now we want to estimate $||\nabla^{2}\eta||_{L^{2}(H)}$. Recall, we have constructed the entities $\mu,\bar{\mu},\xi,\bar{\xi}$ which adding $\nabla  \tr\chi$ and $\nabla  \tr\chibar$ constructs the set $\Xi:=\{\nabla  \tr\chi,\mu,\bar{\xi}\}$ and $\bar{\Xi}:=\{\nabla  \tr\chibar,\bar{\mu},\xi\}$. We obtain a set of following transport equations presented in a schematic way for $\nabla\Xi$ and $\nabla\bar{\Xi}$
\begin{eqnarray}
\label{eq:F1}
\nabla_{4}\nabla\Xi=\varphi\nabla^{2}\varphi+\nabla\varphi \nabla\varphi+\varphi\nabla\Psi+\Psi\nabla\varphi+\Phi^{F}\cdot\hnabla^{2}\Phi^{F}\\\nonumber +\hnabla\Phi^{F}\cdot\hnabla\Phi^{F}+\varphi\Phi^{F}\cdot\hnabla\Phi^{F}+\Phi^{F}\cdot\Phi^{F}\nabla\varphi\\\nonumber+\varphi\nabla\varphi+\varphi_{g}\varphi\nabla\varphi+\varphi\varphi\nabla\varphi_{g}+\varphi_{g}\nabla\Xi
+\varphi\nabla_{4}\Xi+\varphi\varphi_{g}\Xi:=\mathcal{F}_{1}
\end{eqnarray}
where $\Psi$ can consist of any of the `good' Weyl curvature components i.e., $(\alpha,\beta,\bar{\beta},\rho,\sigma)$ or more precisely it does not contain $\alpha$. Similarly, $\Phi^{F}$ can contain all the Yang-Mills curvature components $(\alpha^{F},\bar{\alpha}^{F},\rho^{F},\sigma^{F})$. However, there is no term involving $\hnabla^{2}\bar{\alpha}^{F}$ i.e., the topmost derivative operator does not act on $\bar{\alpha}^{F}$. This is once again a consequence of the special structure of the Yang-Mills equations. Of course, Yang-Mills equations do satisfy a null condition. In addition, a good $\varphi$ ($\varphi_{g}$) always appears multiplied with the top derivative of $\varphi$. Similarly, we obtain the following transport equation for $\bar{\Xi}$
\begin{eqnarray}
\label{eq:F2}
\nabla_{3}\nabla\bar{\Xi}=\varphi\nabla^{2}\varphi+\nabla\varphi \nabla\varphi+\varphi\nabla\Psi+\Psi\nabla\varphi+\Phi^{F}\cdot\hnabla^{2}\Phi^{F}\\\nonumber +\hnabla\Phi^{F}\cdot\hnabla\Phi^{F}+\varphi\Phi^{F}\cdot\hnabla\Phi^{F}+\varphi\nabla\varphi\\\nonumber+\Phi^{F}\cdot\Phi^{F}\nabla\varphi+\varphi_{g}\varphi\nabla\varphi+\varphi\varphi\nabla\varphi_{g}+\varphi_{g}\nabla\bar{\Xi}+\varphi\nabla_{3}\bar{\Xi}+\varphi\varphi_{g}\bar{\Xi}:=\mathcal{F}_{2}.
\end{eqnarray}
Similarly, $\Psi$ represents the Weyl curvature components belonging to the set $(\bar{\alpha},\beta,\bar{\beta},\rho,\sigma)$. $\Phi^{F}$ can contain all the Yang-Mills curvature components $(\alpha^{F},\bar{\alpha}^{F},\rho^{F},\sigma^{F})$. However, we do not have $\hnabla^{2}\alpha^{F}$ term. This is favorable to us since we do not have control of $\hnabla^{2}\alpha^{F}$ on $  \Hbar$. Now we may apply the direct transport inequalities (proposition \ref{transport}) to obtain estimates for $\nabla\Xi$ and $\nabla\bar{\Xi}$. Also notice $\nabla_{4}\Xi$ satisfies equation of the following type 
\begin{eqnarray}
\nabla_{4}\Xi=\varphi \nabla\varphi+\varphi \Psi+\Phi^{F}\hnabla\Phi^{F}+\varphi \Phi^{F}\Phi^{F}+\varphi_{g}\varphi\varphi+\varphi\varphi,
\end{eqnarray}
where $\Psi$ contains the Weyl curvature components $(\alpha,\beta,\bar{\beta},\rho,\sigma)$ i.e., the ones that can be controlled over $H$. Similarly, $\nabla_{3}\bar{\Xi}$ satisfies equation of the following type  
\begin{eqnarray}
\nabla_{3}\bar{\Xi}=\varphi \nabla\varphi+\nabla\varphi_{g}\varphi+\varphi \Psi+\Phi^{F}\hnabla\Phi^{F}+\varphi \Phi^{F}\Phi^{F}+\varphi_{g}\varphi\varphi+\varphi\varphi.
\end{eqnarray}
Once again $\Psi$ appearing in $\nabla_{3}\bar{\Xi}$ equation contains the Weyl curvature components that can be controlled over $  \Hbar$. Remarkably note that equations for $\nabla_{4}\Xi$ and $\nabla_{3}\bar{\Xi}$ do not contain derivatives of the Weyl curvature. As mentioned previously, this is vital to close the regularity. Also $L^{2}(S)$ and $\dot{H}^{1}(S)$ of $\varphi_{g}$ can be controlled only by means of the initial data $\mathcal{O}_{0}$. We utilize the transport inequalities (proposition \ref{transport}) to estimate $\nabla\Xi$ and $\nabla\bar{\Xi}$
\begin{eqnarray}
||\nabla\Xi||_{L^{2}(S_{u,\ubar})}\leq C(\mathcal{O}_{0})\left(||\nabla\Xi||_{L^{2}(S_{u,0})}+\int_{0}^{\ubar}||\nabla_{4}\nabla\Xi||_{L^{2}(S)}d\ubar^{''}\right),\\
||\nabla\bar{\Xi}||_{L^{2}(S_{u,\ubar})}\leq C\left(||\nabla\bar{\Xi}||_{L^{2}(S_{u,0})}+\int_{0}^{\ubar}||\nabla_{3}\nabla\bar{\Xi}||_{L^{2}(S)}d\ubar^{''}\right).
\end{eqnarray}
Therefore we will need to estimate $\mathcal{F}_{1}$ and $\mathcal{F}_{2}$ in $L^{2}(S)$. We first estimate different elements of $\mathcal{F}_{1}$ and $\mathcal{F}_{2}$ (\ref{eq:F1}-\ref{eq:F2}) using estimates derived in lemma \ref{1}-\ref{3}\\ 
\textbf{$\mathcal{F}_{1},\mathcal{F}_{2}$:} 
\begin{eqnarray}
||\varphi\nabla^{2}\varphi||_{L^{2}(S)}\leq ||\varphi||_{L^{\infty}(S}\nonumber||\nabla^{2}\varphi||_{L^{2}(S)}\leq C(\mathcal{O}_{0},\mathcal{W},\mathcal{F},\mathcal{W}(S),\mathcal{F}(S))||\nabla^{2}\varphi||_{L^{2}(S)},\\
||\nabla\varphi\nabla\varphi||_{L^{2}(S)}\leq ||\nabla\varphi||^{2}_{L^{4}(S)}\leq C \nonumber(||\nabla\varphi||^{\frac{1}{2}}_{L^{2}(S)}||\nabla^{2}\varphi||^{\frac{1}{2}}_{L^{2}(S)}+||\nabla\varphi||_{L^{2}})^{2}\leq C||\nabla^{2}\varphi||_{L^{2}(S)},\\
||\varphi \nabla\Psi||_{L^{2}(S)}\leq\nonumber ||\varphi||_{L^{\infty}(S)}||\nabla\Psi||_{L^{2}(S)}\leq C (||\varphi||_{L^{4}(S)}+||\nabla\varphi||_{L^{4}(S)})||\nabla\Psi||_{L^{2}(S)}\\\nonumber 
\leq C(||\nabla\varphi||^{1/2}_{L^{2}(S)}||\nabla^{2}\varphi||^{1/2}_{L^{2}(S)}+||\nabla\varphi||_{L^{2}(S)})||\nabla\Psi||_{L^{2}(S)}\\\nonumber\leq C(\mathcal{O}_{0},\mathcal{W},\mathcal{F})(||\nabla^{2}\varphi||_{L^{2}(S)}+||\nabla\Psi||^{2}_{L^{2}(S)}).
\end{eqnarray}
The remaining terms are estimated as follows
\begin{eqnarray}
||\Psi\nabla\varphi||_{L^{2}(S)}\leq||\Psi||_{L^{4}(S)}||\nabla\varphi||_{L^{4}(S)}\\\nonumber \lesssim ||\Psi||_{L^{4}(S)}(||\nabla\varphi||^{1/2}_{L^{2}(S)}||\nabla^{2}\varphi||^{1/2}_{L^{2}(S)}+||\nabla\varphi||_{L^{2}(S)})\\\nonumber 
\lesssim C(\mathcal{O}_{0},\mathcal{W},\mathcal{F})||\Psi||_{L^{4}(S)}||\nabla^{2}\varphi||^{1/2}_{L^{2}(S)},
||\Phi^{F}\cdot\hnabla^{2}\Phi^{F}||_{L^{2}(S)}\\\nonumber \leq ||\Phi^{F}||_{L^{\infty}(S)}||\hnabla^{2}\Phi^{F}||_{L^{2}(S)}\leq C(\mathcal{F}(S))||\hnabla^{2}\Phi^{F}||_{L^{2}(S)},\\\nonumber ~||\hnabla\Phi^{F}\cdot\hnabla\Phi^{F}||_{L^{2}(S)}\leq ||\hnabla\Phi^{F}||^{2}_{L^{4}(S)}\leq C(\mathcal{F}(S)),||\varphi\nabla\varphi||_{L^{2}(S)}\\\nonumber \leq ||\varphi||_{L^{4}}||\nabla\varphi||_{L^{4}(S)}\lesssim ||\nabla\varphi||_{L^{2}(S)}||\nabla^{2}\varphi||_{L^{2}(S)}\\\nonumber 
\lesssim ||\nabla^{2}\varphi||_{L^{2}(S)},
\end{eqnarray}
and
\begin{eqnarray}
||\varphi\Phi^{F}\cdot\hnabla\Phi^{F}||_{L^{2}(S)}\lesssim ||\varphi||_{L^{4}(S)}||\Phi^{F}||_{L^{\infty}(S)}||\hnabla\Phi^{F}||_{L^{4}(S)}\leq C(\mathcal{F}(S))\\\nonumber (||\varphi||^{1/2}_{L^{2}(S)}||\nabla\varphi||^{1/2}_{L^{2}(S)}+||\varphi||_{L^{2}(S)}),~||\varphi\nabla\varphi||_{L^{2}(S)}\\\nonumber \lesssim ||\varphi||_{L^{4}(S)}||\nabla\varphi||_{L^{4}(S)}\lesssim (||\nabla\varphi||^{1/2}_{L^{2}(S)}||\nabla^{2}\varphi||^{1/2}_{L^{2}(S)}+||\nabla\varphi||_{L^{2}(S)})\\\nonumber 
\lesssim C(\mathcal{O}_{0},\mathcal{W},\mathcal{F})+C(\mathcal{O}_{0},\mathcal{W},\mathcal{F})||\nabla^{2}\varphi||_{L^{2}(S)},
\end{eqnarray}
\begin{eqnarray}
||\varphi_{g}\varphi\nabla\varphi||_{L^{2}(S)}\leq ||\varphi_{g}||_{L^{\infty}(S)}||\varphi||_{L^{4}(S)}||\nabla\varphi||_{L^{4}(S)}\lesssim ||\nabla^{2}\varphi||_{L^{2}(S)},\\\nonumber 
||\varphi\varphi\nabla\varphi_{g}||_{L^{2}(S)}\lesssim ||\varphi||^{2}_{L^{\infty}(S)}||\nabla\varphi_{g}||_{L^{2}(S)}\leq C(\mathcal{O}_{0},\mathcal{W},\mathcal{F},\mathcal{F}(S)),\\\nonumber ||\varphi\nabla\Xi||_{L^{2}(S)}=||\nabla^{2}\varphi+\nabla\Psi||_{L^{2}(S)}\\\nonumber 
\leq ||\nabla^{2}\varphi||_{L^{2}(S)}+||\nabla\Psi||_{L^{2}(S)},~||\varphi_{g}\nabla\Xi||_{L^{2}(S)}\\\nonumber \leq||\varphi_{g}||_{L^{\infty}(S)}||\nabla\Xi||_{L^{2}(S)}\lesssim||\nabla^{2}\varphi||_{L^{2}(S)}+||\nabla\Psi||_{L^{2}(S)}.
\end{eqnarray}
The last type of terms are estimated as follows
\begin{eqnarray}
||\varphi\nabla_{4}\Xi||_{L^{2}(S)}=||\varphi(\varphi \nabla\varphi+\varphi \Psi+\Phi^{F}\hnabla\Phi^{F}+\varphi \Phi^{F}\Phi^{F}+\varphi_{g}\varphi\varphi+\varphi\varphi)||_{L^{2}(S)}\\\nonumber \lesssim ||\varphi||_{L^{\infty}}||\varphi||_{L^{4}(S)}||\nabla\varphi||_{L^{4}(S)}\\\nonumber 
+||\varphi||_{L^{\infty}(S)}||\varphi||_{L^{4}(S)}||\hnabla\Phi^{F}||_{L^{4}(S)}+||\varphi||^{2}_{L^{4}(S)}||\Phi^{F}||^{2}_{L^{\infty}(S)}\\\nonumber +||\varphi_{g}||_{L^{\infty}(S)}||\varphi||^{2}_{L^{4}(S)}||\varphi||_{L^{\infty}(S)}+||\varphi||_{L^{\infty}}||\varphi||_{L^{4}(S)}||\Psi||_{L^{4}(S)}\\\nonumber +||\varphi||_{L^{\infty}}||\varphi||^{2}_{L^{4}(S)}
\leq C(\mathcal{O}_{0},\mathcal{W},\mathcal{F},\mathcal{F}(S))+C(\mathcal{O}_{0},\mathcal{W}(S),\mathcal{F},\mathcal{F}(S))||\nabla^{2}\varphi||_{L^{2}(S)},\\\nonumber 
||\varphi\nabla_{3}\bar{\Xi}||_{L^{2}(S)}\lesssim C(\mathcal{O}_{0},\mathcal{W},\mathcal{F},\mathcal{F}(S))+C(\mathcal{O}_{0},\mathcal{W}(S),\mathcal{F},\mathcal{F}(S))||\nabla^{2}\varphi||_{L^{2}(S)},\\\nonumber
||\varphi\varphi_{g}\Xi||_{L^{2}(S)}\lesssim C||\nabla^{2}\varphi||_{L^{2}(S)},~||\varphi\varphi_{g}\bar{\Xi}||_{L^{2}(S)}\lesssim C||\nabla^{2}\varphi||_{L^{2}(S)},
\end{eqnarray}
where we extensively used the inequalities (\ref{eq:inequal1}-\ref{eq:gagliardo2}). Notice that we once again proceed in a hierarchical fashion i.e., start $(\mu,\eta)$ and then use that result to obtain estimates for $(\bar{\xi},\omegabar)$ and so on. This is one of the the main reasons why $\mathcal{F}_{1}$ and $\mathcal{F}_{2}$ are estimated by means of $C(\mathcal{O}_{0},\mathcal{W},\mathcal{F},\mathcal{W}(S),\mathcal{F}(S))$ and $||\nabla^{2}\varphi||_{L^{2}(S)}$ (linearly in the latter). Collecting all the terms together, we obtain the following two inequalities satisfied by $\nabla\Xi$ and $\nabla\bar{\Xi}$ 
\begin{eqnarray}
\label{eq:gronwal1}
||\nabla\Xi||_{L^{2}(S)}\leq C(\mathcal{O}_{0},\mathcal{W},\mathcal{F},\mathcal{W}(S),\mathcal{F}(S))(1+\int_{0}^{\ubar}||\nabla^{2}\varphi||_{L^{2}(S)}d\ubar^{'}),\\
\label{eq:gronwal2}
||\nabla\bar{\Xi}||_{L^{2}(S)}\leq C(\mathcal{O}_{0},\mathcal{W},\mathcal{F},\mathcal{W}(S),\mathcal{F}(S))(1+\int_{0}^{u}||\nabla^{2}\varphi||_{L^{2}(S)}du^{'})
\end{eqnarray}
Now we use the elliptic estimates resulting from the Hodge system. 
After an application of $\mathcal{D}^{*}$, the Hodge system (\ref{eq:hodge2}-\ref{eq:hodge4}) reduces to the following second order elliptic equation 
\begin{eqnarray}
\mathcal{D}^{*}\mathcal{D}\varphi=K\varphi+\nabla(\Xi/\bar{\Xi})+\nabla\Psi+\varphi\nabla\varphi+\nabla\varphi+\Phi^{F}\cdot\hnabla\Phi^{F}
\end{eqnarray}
which yields an estimate of the type 
\begin{eqnarray}
||\nabla^{2}\varphi||_{L^{2}(S)}\lesssim ||K\varphi||_{L^{2}(S)}+||\nabla (\Xi/\bar{\Xi})||_{L^{2}(S)}\nonumber+||\nabla\Psi||_{L^{2}(S)}+||\varphi\nabla\varphi||_{L^{2}(S)}\\\nonumber +||\nabla\varphi||_{L^{2}(S)}+||\Phi^{F}\cdot\hnabla\Phi^{F}||_{L^{2}(S)}\\
\label{eq:elliptic}
\lesssim C(\mathcal{O}_{0},\mathcal{W},\mathcal{F},\mathcal{W}(S),\mathcal{F}(S))+||\nabla (\Xi/\bar{\Xi})||_{L^{2}(S)}+||\nabla\Psi||_{L^{2}(S)}
\end{eqnarray}
substitution of which in the previous inequalities (\ref{eq:gronwal1}-\ref{eq:gronwal2}) and an application of Gr\"onwall's inequality yields 
\begin{eqnarray}
\label{eq:L2}
||\nabla\Xi||_{L^{2}(S)}\leq C(\mathcal{O}_{0},\mathcal{W},\mathcal{F},\mathcal{W}(S),\mathcal{F}(S)),~||\nabla\bar{\Xi}||_{L^{2}(S)}\leq C(\mathcal{O}_{0},\mathcal{W},\mathcal{F},\mathcal{W}(S),\mathcal{F}(S)).
\end{eqnarray}
Substitution of (\ref{eq:L2}) into the elliptic estimate (\ref{eq:elliptic}) yields 
\begin{eqnarray}
\label{eq:elliptic4}
||\nabla^{2}\varphi||_{L^{2}(S)}\lesssim C(\mathcal{O}_{0},\mathcal{W},\mathcal{F},\mathcal{W}(S),\mathcal{F}(S))+||\nabla\Psi||_{L^{2}(S)},
\end{eqnarray}
for $\varphi\in (\widehat{\chi},\chibarhat, \tr\chi, \tr\chibar,\eta,\etabar,\omega,\omegabar)$ and the most important point is that $\Psi$ in the elliptic estimate (\ref{eq:elliptic4}) does not contain $\alpha,\bar{\alpha}$. Therefore integrating over $H$ and $  \Hbar$, we obtain 
\begin{eqnarray}
||\nabla^{2}\varphi||_{L^{2}(H)},||\nabla^{2}\varphi||_{L^{2}(\Hbar)}\leq C(\mathcal{O}_{0},\mathcal{W},\mathcal{F},\mathcal{W}(S),\mathcal{F}(S)).
\end{eqnarray}
This concludes the proof of the lemma.
\end{proof}
\begin{corollary}
\label{C2}
\textit{$||\nabla\varphi_{g}||_{L^{4}(S)}\leq C(\mathcal{O}_{0}),~||\nabla\varphi_{b}||_{L^{4}(S)}\leq C(\mathcal{O}_{0},\mathcal{W},\mathcal{F},\mathcal{W}(S),\mathcal{F}(S))$.}
\end{corollary} 
\begin{proof} The proof relies on the co-dimension 1 trace inequalities 
\begin{eqnarray}
||\nabla\varphi||_{L^{4}(S_{u,\ubar})}\leq C\left(||\nabla\varphi||_{L^{4}(S_{u,0})}+||\nabla\varphi||^{1/2}_{L^{2}(H)}||\nabla_{4}\nabla\varphi||^{1/4}_{L^{2}(H)}(||\nabla\varphi||_{L^{2}(H)}+||\nabla\nabla\varphi||_{L^{2}(H)})^{1/4}\right),\\
||\nabla\varphi||_{L^{4}(S_{u,\ubar})}\leq C\left(||\nabla\varphi||_{L^{4}(S_{0,\ubar})}+||\nabla\varphi||^{1/2}_{L^{2}(  \Hbar)}||\hnabla_{3}\nabla\varphi||^{1/4}_{L^{2}(  \Hbar)}(||\nabla\varphi||_{L^{2}(  \Hbar)}+||\nabla\nabla\varphi||_{L^{2}(  \Hbar)})^{1/4}\right),
\end{eqnarray}
the null evolution equations and lemma (4) and (5). For the good connection coefficients $\varphi_{g}$ i.e., the ones satisfying $\nabla_{3}$ equations, we easily observe the following using lemma (4)
\begin{eqnarray}
||\nabla\varphi_{g}||^{2}_{L^{2}(  \Hbar)}=\int_{0}^{u}||\nabla\varphi_{g}||^{2}_{L^{2}(S)}du^{'}\leq \epsilon C(\mathcal{O}_{0}).
\end{eqnarray}
In addition using the commuted null transport equation, $||\hnabla_{3}\nabla\varphi_{g}||_{L^{2}(  \Hbar}$ is controlled by $||\nabla^{2}\varphi_{g}||_{L^{2}(  \Hbar)}$ which in turn is controlled by $\mathcal{C}(\mathcal{O}_{0},\mathcal{W},\mathcal{F},\mathcal{W}(S),\mathcal{F}(S))$ by lemma (5). For example, if we look at the commuted $\nabla_{3}$ equation for $\widehat{\chi}$
\begin{eqnarray}
\nabla_{3}\nabla\widehat{\chi}\sim-\frac{1}{2}\nabla( \tr\chibar)\widehat{\chi}-\frac{1}{2} \tr\chibar\nabla\widehat{\chi}+\nabla^{2}\eta+2\nabla(\omegabar\widehat{\chi})-\frac{1}{2}\nabla( \tr\chi\chibarhat)\\\nonumber +\eta\nabla\eta+\nabla(\rho^{F}\cdot\rho^{F}\nonumber+\alpha^{F}\cdot \bar{\alpha}^{F}+\sigma^{F}\cdot\sigma^{F})\\\nonumber 
+(\Bar{\beta}+\Bar{\alpha}^{F}\cdot(\rho^{F}+\sigma^{F}))\widehat{\chi}+(\eta+\etabar)\nabla_{3}\widehat{\chi}-\Bar{\chi}\nabla\widehat{\chi}+\Bar{\chi}\eta\widehat{\chi},
\end{eqnarray}
we observe every term at the right-hand side can be estimated in $L^{2}(  \Hbar)$ by $C(\mathcal{O}_{0},\mathcal{W},\mathcal{F},\mathcal{W}(S),\mathcal{F}(S))$. Similar results hold for other good connection coefficients. Therefore we gain an overall factor of $\epsilon^{\frac{1}{4}}$ arising from the integration over $  \Hbar$ i.e., 
\begin{eqnarray}
||\nabla\varphi_{g}||_{L^{4}(S_{u,\ubar})}\leq C(\mathcal{O}_{0})+\epsilon^{\frac{1}{4}}\mathcal{C}(\mathcal{O}_{0},\mathcal{W},\mathcal{F},\mathcal{W}(S),\mathcal{F}(S))
\end{eqnarray}
which for sufficiently small $\epsilon$ yields 
\begin{eqnarray}
||\nabla\varphi_{g}||_{L^{4}(S_{u,\ubar})}\leq C(\mathcal{O}_{0}).
\end{eqnarray}
Now for the bad connection coefficients $\varphi_{g}$ that is the ones satisfying $\nabla_{4}$ equations, we would not have $||\nabla\varphi_{b}||_{L^{4}(S)}$ determined solely in terms of the initial data rather by\\ $C(\mathcal{O}_{0},\mathcal{W},\mathcal{F},\mathcal{W}(S),\mathcal{F}(S))$. This is because we do not gain a factor $\epsilon$ from the integral over $H$. Therefore putting everything together, using lemma (4) and (5) along with the $\nabla$ commuted $\nabla_{4}$ equations for $\varphi_{g}$, we obtain 
\begin{eqnarray}
||\nabla\varphi_{b}||_{L^{4}(S_{u,\ubar})}\leq C(\mathcal{O}_{0},\mathcal{W},\mathcal{F},\mathcal{W}(S),\mathcal{F}(S)).
\end{eqnarray}
This completes the proof. 
\end{proof}

\noindent In the spirit of the previous lemma, we will prove a similar estimate for mixed derivatives of connection coefficients that can not be estimated directly using their evolution equations. We need to set up a transport-div-curl system. More precisely, we need to utilize transport inequalities (proposition \ref{transport}) together with the elliptic estimates. We do so in the following lemma.\\
\begin{lemma}
\label{5}
\textit{Let $\varphi_{g}:=( \tr\chibar,\chibarhat,\etabar,\omega, \tr\chi,\widehat{\chi})$ and $\varphi_{b}:=(\eta,\omegabar)$ and $\mathcal{W},\mathcal{W}(S),\mathcal{F},\mathcal{F}(S)<\infty$. Then $||\nabla\nabla_{3}\eta||_{L^{2}(H)},||\nabla\nabla_{3}\omegabar||_{L^{2}(H)},~||\nabla\nabla_{4}\etabar||_{L^{2}(  \Hbar)},||\nabla\nabla_{4}\omega||_{L^{2}(  \Hbar)}\leq C(\mathcal{O}_{0},\mathcal{W},\mathcal{F},\mathcal{W}(S),\mathcal{F}(S))$}
\end{lemma}
\begin{proof} The proof follows in a similar fashion as that of \ref{lemma5}. Given the estimates obtained in the previous lemmas, we will only prove it for one connection coefficient $\eta$. The remaining estimates can be obtained in an exactly similar fashion. Let us recall the Hodge transport system for the pair $(\eta,\mu)$ (schematically)
\begin{eqnarray}
\text{div}\eta=-\mu-\rho,~\text{curl} \eta=\chibarhat\wedge\widehat{\chi}+\sigma
\end{eqnarray}
\begin{eqnarray}
\nabla_{4}\mu=-\text{div}(-\chi\cdot(\eta-\etabar))-\text{div}(\frac{1}{2}\alpha^{F}(\rho^{F}-\sigma^{F}))
-(-\frac{3}{2} \tr\chi\rho\nonumber-\frac{1}{2}\chibarhat\cdot\alpha +\zeta\cdot\beta+2\etabar\cdot\beta\\\nonumber-\frac{1}{2}(\alpha^{F}\cdot\hnabla_{3}\alpha^{F}-\omegabar|\alpha^{F}|^{2}+4\eta\alpha^{F}\cdot(\rho^{F}-\sigma^{F})
-\rho^{F}\cdot\hnabla_{4}\rho^{F}+\sigma^{F}\cdot\hnabla_{4}\sigma^{F}-2\etabar\alpha^{F}\cdot(\rho^{F}+\sigma^{F}))\\\nonumber 
+((\alpha^{F}\cdot(\rho^{F}-\sigma^{F})+\alpha^{F}))\eta+(\eta+\etabar)(-\chi\cdot(\eta-\etabar)-\frac{1}{2}\alpha^{F}\cdot(\rho^{F}-\sigma^{F}))-\etabar\beta-\chi\nabla\eta+\chi\etabar\eta.
\end{eqnarray}
If we commute the transport equation for $\mu$ with $\nabla_{3}$ we obtain an equation of the following type (we keep the potentially dangerous terms in exact form and write the harmless terms in a schematic way)
\begin{eqnarray}
\nonumber \nabla_{4}\nabla_{3}\mu\sim\chi\nabla\nabla_{3}\eta+\underbrace{\nabla\chi\nabla_{3}\eta}_{IIA}+\widehat{\text{div}}(\hnabla_{3}\alpha^{F})\cdot(\rho^{F}-\sigma^{F})-\underbrace{\frac{1}{2}\alpha^{F}\cdot(\hnabla\hnabla_{3}\rho^{F}-~^{*}\hnabla\hnabla_{3}\sigma^{F})}_{IIC1}\\\nonumber 
+\widehat{\text{div}}\alpha^{F}\cdot\hnabla_{3}(\rho^{F}-\sigma^{F})+\hnabla_{3}\alpha^{F}\cdot \widehat{\text{div}}(\rho^{F}-\sigma^{F})+\tr\chi\nabla_{3}\rho+\rho\nabla_{3} \tr\chi+\nabla_{3}\chibarhat\alpha+\chibarhat\nabla_{3}\alpha\\\nonumber 
+\underbrace{\beta\nabla_{3}\zeta}_{IIB}+\zeta\nabla_{3}\beta+\nabla_{3}\etabar\beta+\etabar\nabla_{3}\beta+\hnabla_{3}\alpha^{F}\cdot\hnabla_{3}\alpha^{F}-\frac{1}{2}\underbrace{\alpha^{F}\cdot\hnabla^{2}_{3}\alpha^{F}}_{IIC2}+\underbrace{\nabla_{3}\omegabar|\alpha^{F}|^{2}}+\omegabar\alpha^{F}\cdot\hnabla_{3}\alpha^{F}\\\nonumber 
+\underbrace{\nabla_{3}\eta \alpha^{F}\cdot(\rho^{F}}_{IID}-\sigma^{F})+\eta \hnabla_{3}(\alpha^{F})\cdot(\rho^{F}-\sigma^{F})+\eta\alpha^{F}\cdot\hnabla_{3}(\rho^{F}-\sigma^{F})+\hnabla_{3}\rho^{F}\cdot\hnabla_{4}\rho^{F}+\underbrace{\rho^{F}\cdot\hnabla_{3}\hnabla_{4}\rho^{F}}_{IIE}\\\nonumber 
+\hnabla_{3}\sigma^{F}\cdot\hnabla_{4}\sigma^{F}+\underbrace{\sigma^{F}\cdot\hnabla_{3}\hnabla_{4}\sigma^{F}}_{IIF}+\nabla_{3}\etabar \alpha^{F}\cdot(\rho^{F}+\sigma^{F})+\etabar \hnabla_{3}(\alpha^{F})\cdot(\rho^{F}+\sigma^{F})+\etabar\alpha^{F}\cdot\hnabla_{3}(\rho^{F}+\sigma^{F})\\\nonumber 
+\underbrace{\chi\eta\nabla_{3}\eta}_{IIG}+\chi\etabar\nabla_{3}\etabar+(|\eta|^{2}+|\etabar|^{2})\nabla_{3}\chi+\eta\etabar\nabla_{3}\chi+\chi\eta\nabla_{3}\etabar+\underbrace{\chi\etabar\nabla_{3}\eta}_{IIH}+\underbrace{\nabla_{3}\chi\nabla\eta}_{IIK}\\\nonumber 
+\omega\nabla_{3}\mu+\omegabar\nabla_{4}\mu+(\eta-\etabar)\nabla\mu.
\end{eqnarray}
Let us identify the terms that do need care. The term $IIC1$ is extremely dangerous since $\hnabla\hnabla_{3}\rho^{F}$ and $\hnabla\hnabla_{3}\sigma^{F}$ contain terms of the type $\hnabla^{2}\bar{\alpha}^{F}$ which can not be controlled on $H$. Now, the previous equation is a $\nabla_{4}$ transport equation for $\nabla_{3}\mu$ and therefore after using the transport inequality (proposition \ref{transport}), $|\hnabla^{2}\bar{\alpha}^{F}|^{2}$ gets integrated over $H$ which is absolutely not under control. However, we also do have the term $IIC2$ which contains terms that cancel the dangerous terms of $IIC1$ in a point-wise manner thereby allowing us to close the argument. This is an extremely important point to note about the special structure of the Einstein-Yang-Mills equations. Without this cancellation, we would not have an obstruction to a potential blow up of bad norms (that are not under control) in finite time. We first show this cancellation. Write down the expression for $\hnabla^{2}_{3}\alpha^{F}$ using the $\hnabla_{3}$ transport equation for $\alpha^{F}$ (once again we write down the most important term exactly and the remaining terms are written in an schematic way)
\begin{eqnarray}
\hnabla^{2}_{3}\alpha^{F}\sim \nabla_{3} \tr\chibar\alpha^{F}+ \tr\chibar\hnabla_{3}\alpha^{F}-\hnabla(-\widehat{\text{{div}}} \bar{\alpha}^{F}+ \tr\chibar\rho^{F}+(\eta-\etabar)\cdot\bar{\alpha}^{F})\\\nonumber +~^{*}\hnabla(-\widehat{\text{{curl}}} \bar{\alpha}^{F}- \tr\chibar\sigma^{F}\nonumber+(\eta-\etabar)\cdot~^{*}\bar{\alpha}^{F})\\\nonumber 
+\nabla_{3}\eta\sigma^{F}+\eta\nabla_{3}\sigma^{F}+\nabla_{3}\eta\rho^{F}+\eta\hnabla_{3}\rho^{F}+ 
\nabla_{3}\omegabar\alpha^{F}+\omegabar\hnabla_{3}\alpha^{F}\\\nonumber -(\nabla_{3}\widehat{\chi})\bar{\alpha}^{F}+\widehat{\chi}\hnabla_{3}\bar{\alpha}^{F}+(\bar{\beta}+\bar{\alpha}^{F}+\bar{\alpha}(\rho^{F}-\sigma^{F}))\\\nonumber 
(\rho^{F}+\sigma^{F})+(\eta+\etabar)\hnabla_{3}(\rho^{F}+\sigma^{F})
+\bar{\chi}\hnabla\rho^{F}+\bar{\chi}\eta(\rho^{F}+\sigma^{F}).
\end{eqnarray}

\begin{eqnarray}
IIC1+IIC2=-\frac{1}{2}\alpha^{F}\cdot(\hnabla\hnabla_{3}\rho^{F}-~^{*}\hnabla\hnabla_{3}\sigma^{F})-\frac{1}{2}\alpha^{F}\cdot\hnabla^{2}_{3}\alpha^{F}\\\nonumber 
=-\frac{1}{2}\alpha^{F}\cdot\hnabla(-\widehat{\text{div}} \bar{\alpha}^{F})+\frac{1}{2}\alpha^{F}\cdot \hnabla(-\widehat{\text{{curl}}} \bar{\alpha}^{F})-\frac{1}{2}\alpha^{F}\cdot\hnabla(\widehat{\text{div}}\bar{\alpha}^{F})+\frac{1}{2}\alpha^{F}\cdot~^{*}\hnabla\widehat{\text{curl}}\bar{\alpha}^{F}\\\nonumber +\widehat{\chi}\alpha^{F}\cdot\hnabla_{3}\bar{\alpha}^{F}
+\nabla_{3}\omegabar\alpha^{F}\cdot\alpha^{F}+III
=\widehat{\chi}\alpha^{F}\cdot\hnabla_{3}\bar{\alpha}^{F}+\nabla_{3}\omegabar\alpha^{F}\cdot \alpha^{F}+III,
\end{eqnarray}
where $III$ denotes the collection of terms $L^{2}(S)$ norm of which can be easily controlled by the available estimates. $\nabla_{3}\omegabar$ can be estimated by commuting $\nabla_{4}$ transport equation of $\omegabar$ with $hnabla_{3}$ and the direct use of transport inequalities (proposition \ref{transport}). Using the $\hnabla_{3}$ evolution equation for $\rho^{F}$, we can easily estimate the terms $\eta\hnabla_{3}\rho^{F}$ and therefore both $\eta\hnabla_{3}\rho^{F}$ and $ 
\nabla_{3}\omegabar\alpha^{F}$ are under control. Similar to $\nabla_{3}\omegabar$, there are other entities that arise which can not be reduced by means of transport equations. This of course includes $\nabla_{3}\eta$. Therefore let us obtain estimates for $\nabla_{3}\eta$ and the estimates for similar entities such as $\nabla_{3}\omegabar$ will follow in the exact same way. First commute the transport equation for $\eta$ with $\nabla_{3}$ to yield (schematically)
\begin{eqnarray}
\label{eq:mixed}
\nabla_{4}\nabla_{3}\eta\sim\nabla_{3}\chi(\eta-\etabar)+\chi(\nabla_{3}\eta-\nabla_{3}\etabar)+\nabla_{3}\beta+(\hnabla_{3}\alpha^{F})\cdot(\rho^{F}-\sigma^{F})\\\nonumber +\alpha^{F}\cdot(\hnabla_{3}\rho^{F}-\hnabla_{3}\sigma^{F})+\omega\nabla_{3}\eta
+\omegabar\nabla_{4}\eta+(\eta-\etabar)\nabla\eta+\sigma\eta+\Phi^{F}\cdot\Phi^{F}\eta.
\end{eqnarray}
Once again, we note that $\nabla_{4}\eta$ terms may be eliminated by means of the $\nabla_{4}$ transport equation for $\eta$. Similarly we have \begin{eqnarray}
\nabla_{3}\beta\sim  \tr\chibar\beta+\nabla\rho+~^{*}\nabla\sigma+2\omegabar\beta+\widehat{\chi}\cdot \bar{\beta}+(\eta\rho+~^{*}\eta\sigma)\nonumber+\frac{1}{2}(\hnabla(\rho^{F}\cdot \rho^{F}+\sigma^{F}\cdot\sigma^{F})\nonumber\\\nonumber+\bar{\chi}\alpha^{F}\cdot(\rho^{F}-\sigma^{F}) -\hnabla_{4}(\bar{\alpha}^{F}\cdot\rho^{F}-\bar{\alpha}^{F}\cdot\sigma^{F})+\omega\bar{\alpha}^{F}(\sigma^{F}-\rho^{F})
+2\etabar_{a}(\rho^{F}\cdot\rho^{F}+\alpha^{F}\cdot\bar{\alpha}^{F}+\sigma^{F}\cdot\sigma^{F})\\\nonumber 
+\etabar(\rho^{F}\cdot\rho^{F}+\sigma^{F}\cdot\sigma^{F})),
\end{eqnarray}
\begin{eqnarray}
\hnabla_{3}\alpha^{F}\sim  \tr\chibar\alpha^{F}-\hnabla\rho^{F}+~^{*}\hnabla\sigma^{F}-~^{*}\eta\sigma^{F}+\eta\rho^{F}+\omegabar\alpha^{F}-\widehat{\chi}\cdot \bar{\alpha}^{F},\\
\hnabla_{4}\bar{\alpha}^{F}\sim  \tr\chi\bar{\alpha}^{F}+\hnabla\rho^{F}-~^{*}\hnabla\sigma^{F}-~^{*}\etabar\sigma^{F}-\etabar\rho^{F}+\omega\bar{\alpha}^{F}-\chibarhat\cdot\alpha^{F},\\
\hnabla_{3}\rho^{F}\sim\widehat{\text{div}} \bar{\alpha}^{F}+ \tr\chibar\rho^{F}+(\eta-\etabar)\cdot\bar{\alpha}^{F},\\
\hnabla_{3}\sigma^{F}\sim\widehat{\text{curl}} \bar{\alpha}^{F}- \tr\chibar\sigma^{F}+(\eta-\etabar)\cdot~^{*}\bar{\alpha}^{F},
\end{eqnarray}
Utilizing these evolution equations, the definitions of $\mathcal{W},\mathcal{F}$ and $\mathcal{F}(S)$, the $L^{\infty}$ estimate of the connection coefficients, and lemma (2), we note the following
\begin{eqnarray}
\int_{0}^{\ubar}||\nabla_{3}\beta||_{L^{2}(S)}d\ubar^{'}\leq C(\mathcal{O}_{0},\mathcal{W},\mathcal{F},\mathcal{W}(S),\mathcal{F}(S)).
\end{eqnarray}
Going back to (\ref{eq:mixed}) and utilizing available transport equations, we observe 
\begin{eqnarray}
\int_{0}^{\ubar}||\nabla_{4}\nabla_{3}\eta||_{L^{2}(S)}d\ubar^{'}\leq C(\mathcal{O}_{0},\mathcal{W},\mathcal{F},\mathcal{W}(S),\mathcal{F}(S))+C(\mathcal{O}_{0})\int_{0}^{\ubar}||\nabla_{3}\eta||_{L^{2}(S)}d\ubar^{'}.
\end{eqnarray}
The transport inequality (proposition \ref{transport}) applied to (\ref{eq:mixed}) yields
\begin{eqnarray}
||\nabla_{3}\eta||_{L^{2}(S)}\leq C(\mathcal{O}_{0},\mathcal{W},\mathcal{F},\mathcal{W}(S),\mathcal{F}(S))+C(\mathcal{O}_{0})\int_{0}^{\ubar}||\nabla_{3}\eta||_{L^{2}(S)}d\ubar^{'}
\end{eqnarray}
which through Gr\"onwall's inequality leads to 
\begin{eqnarray}
\label{eq:etaL2}
||\nabla_{3}\eta||_{L^{2}(S)}\leq C(\mathcal{O}_{0},\mathcal{W},\mathcal{F},\mathcal{W}(S),\mathcal{F}(S))e^{C(\mathcal{O}_{0})\ubar}\leq C(\mathcal{O}_{0},\mathcal{W},\mathcal{F},\mathcal{W}(S),\mathcal{F}(S)).
\end{eqnarray}
This inequality of course yields the estimate 
\begin{eqnarray}
\label{eq:middle}
||\nabla_{3}\eta||_{L^{2}(H)}\leq C(\mathcal{O}_{0},\mathcal{W},\mathcal{F},\mathcal{W}(S),\mathcal{F}(S)).
\end{eqnarray}
This later estimate (\ref{eq:middle}) will be used in proving the next lemma. In an exactly similar way, we obtain 
\begin{eqnarray}
||\nabla_{3}\omegabar||_{L^{2}(S)}\leq  C(\mathcal{O}_{0},\mathcal{W},\mathcal{F},\mathcal{W}(S),\mathcal{F}(S)),\\
||\nabla_{4}\etabar||_{L^{2}(S)}\leq  C(\mathcal{O}_{0},\mathcal{W},\mathcal{F},\mathcal{W}(S),\mathcal{F}(S)),\\
\label{eq:end}
||\nabla_{4}\omega||_{L^{2}(S)}\leq  C(\mathcal{O}_{0},\mathcal{W},\mathcal{F},\mathcal{W}(S),\mathcal{F}(S)).
\end{eqnarray}
These estimates will be extremely useful in the future.
Now we may estimate $\hnabla_{3}\bar{\alpha}^{F}$ through the codimension-1 trace inequality (\ref{eq:need1})
\begin{eqnarray}
||\hnabla_{3}\bar{\alpha}^{F}||_{L^{4}(S)}\leq C(\mathcal{O}_{0},\mathcal{F}_{0},\mathcal{F}).
\end{eqnarray}
Now we estimate $L^{2}(S)$ norm of the underlined entities in the expression of $\nabla_{4}\nabla_{3}\mu$ as follows 
\begin{eqnarray}
||IIA||_{L^{2}(S)}=||\nabla\chi\nabla_{3}\eta||_{L^{2}(S)}\leq ||\nabla\chi||_{L^{4}(S)}||\nabla_{3}\eta||_{L^{4}(S)}\leq C(\mathcal{O}_{0})||\nabla\nabla_{3}\eta||_{L^{2}(S)},\\
||IIB||_{L^{2}(S)}=||\beta \nabla_{3}\zeta||_{L^{2}(S)}=||\beta||_{L^{4}(S)}||-\nabla\omegabar-\frac{1}{2}\bar{\chi}\cdot(\eta+\zeta)+\omegabar(\zeta-\eta)\\\nonumber -\frac{1}{2}\bar{\beta}+\frac{1}{2}~^{*}\sigma^{F}\bar{\alpha}^{F}||_{L^{4}(S)}
\leq C(\mathcal{O}_{0},\mathcal{W},\mathcal{F},\mathcal{W}(S),\mathcal{F}(S)),\\
||IID||_{L^{2}(S)}=||\nabla_{3}\eta \alpha^{F}\cdot(\rho^{F}-\sigma^{F})||_{L^{2}(S)}\leq ||\nabla_{3}\eta||_{L^{2}(S)}||\alpha^{F}||_{L^{\infty}(S)}\\\nonumber (||\rho^{F}||_{L^{\infty}(S)}+||\sigma^{F}||_{L^{\infty}(S)})\\\nonumber 
\leq C(\mathcal{O}_{0},\mathcal{W},\mathcal{F},\mathcal{W}(S),\mathcal{F}(S)),\\
||IIE||_{L^{2}(S)}\leq C(\mathcal{O}_{0},\mathcal{W},\mathcal{F},\mathcal{W}(S),\mathcal{F}(S))(||\hnabla^{2}\rho^{F}||_{L^{2}(S)}+||\hnabla^{2}\sigma^{F}||_{L^{2}(S)})\\\nonumber +C(\mathcal{O}_{0},\mathcal{W},\mathcal{F},\mathcal{W}(S),\mathcal{F}(S)),\\
||IIF||_{L^{2}(S)}\leq C(\mathcal{O}_{0},\mathcal{W},\mathcal{F},\mathcal{W}(S),\mathcal{F}(S))(||\hnabla^{2}\rho^{F}||_{L^{2}(S)}+||\hnabla^{2}\sigma^{F}||_{L^{2}(S)})\\\nonumber +C(\mathcal{O}_{0},\mathcal{W},\mathcal{F},\mathcal{W}(S),\mathcal{F}(S)),\\
||IIG,IIH||_{L^{2}(S)}\leq C(\mathcal{O}_{0},\mathcal{W},\mathcal{F},\mathcal{W}(S),\mathcal{F}(S)),\\\nonumber 
||IIK||_{L^{2}(S)}\leq C(\mathcal{O}_{0},\mathcal{W},\mathcal{F},\mathcal{W}(S),\mathcal{F}(S))||\nabla\eta||^{2}_{L^{4}(S)}\leq C(\mathcal{O}_{0},\mathcal{W},\mathcal{F},\mathcal{W}(S),\mathcal{F}(S)). 
\end{eqnarray}
Putting together all the estimates and use of the transport inequality (proposition \ref{transport}) 
\begin{eqnarray}
||\nabla_{3}\mu||_{L^{2}(S)}\leq C\left(||\nabla_{3}\mu||_{L^{2}(S_{u,0})}+\int_{0}^{\ubar}||\nabla_{4}\nabla_{3}\mu||_{L^{2}(S)}d\ubar^{'}\right)
\end{eqnarray}
yields 
\begin{eqnarray}
\label{eq:muestimate}
||\nabla_{3}\mu||_{L^{2}(S)}\leq C(\mathcal{O}_{0},\mathcal{W}_{0},\mathcal{F}_{0},\mathcal{W},\mathcal{F},\mathcal{W}(S),\mathcal{F}(S))\\\nonumber +C(\mathcal{O}_{0},\mathcal{W},\mathcal{F},\mathcal{W}(S),\mathcal{F}(S))\int_{0}^{\ubar}||\nabla\nabla_{3}\eta||_{L^{2}(S)}d\ubar^{'}.
\end{eqnarray}
Now we go back to the definition of the Hodge system (\ref{eq:hodge2}-\ref{eq:hodge4}) to obtain 
\begin{eqnarray}
\label{eq:hodgenew}
\text{div}(\nabla_{3}\eta)=-\nabla_{3}\mu-\nabla_{3}\rho+[\text{div},\nabla_{3}]\eta,~\text{curl}(\nabla_{3}\eta)=\nabla_{3}\widehat{\chi}\wedge\chibarhat+\widehat{\chi}\wedge\chibarhat+\nabla_{3}\sigma.
\end{eqnarray}
This Hodge system yields the following elliptic estimate since $K$ is under control 
\begin{eqnarray}
\label{eq:transport1}
||\nabla\nabla_{3}\eta||_{L^{2}(S)}\leq C( ||\nabla_{3}\mu||_{L^{2}(S)}+||\nabla\Psi||_{L^{2}(S)}+C(\mathcal{O}_{0},\mathcal{W},\mathcal{F},\mathcal{W}(S),\mathcal{F}(S))
\end{eqnarray}
where we have used the $\nabla_{3}$ transport equations for the entities present in the right-hand side of the Hodge system (\ref{eq:hodgenew}). Substitution of (\ref{eq:transport1}) into the estimate (\ref{eq:muestimate}) yields 
\begin{eqnarray}
||\nabla_{3}\mu||_{L^{2}(S)}\leq C(\mathcal{O}_{0},\mathcal{W},\mathcal{F},\mathcal{W}(S),\mathcal{F}(S))+C(\mathcal{O}_{0})\int_{0}^{\ubar}||\nabla_{3}\mu||_{L^{2}(S)}d\ubar^{'}\\\nonumber 
+C(\mathcal{O}_{0},\mathcal{W},\mathcal{F},\mathcal{W}(S),\mathcal{F}(S))\int_{0}^{\ubar}||\nabla\Psi||_{L^{2}(S)}d\ubar^{'})
\end{eqnarray}
and since $\Psi$ is not $\bar{\alpha}$ or $\alpha$, we can write 
\begin{eqnarray}
||\nabla_{3}\mu||_{L^{2}(S)}\leq C(\mathcal{O}_{0},\mathcal{W},\mathcal{F},\mathcal{W}(S),\mathcal{F}(S))+C(\mathcal{O}_{0})\int_{0}^{\ubar}||\nabla_{3}\mu||_{L^{2}(S)}d\ubar^{'}
\end{eqnarray}
and therefore an application of Gr\"onwall's inequality yields 
\begin{eqnarray}
||\nabla_{3}\mu||_{L^{2}(S)}\leq C(\mathcal{O}_{0},\mathcal{W},\mathcal{F},\mathcal{W}(S),\mathcal{F}(S)) e^{C(\mathcal{O}_{0}\ubar)}\leq C(\mathcal{O}_{0},\mathcal{W},\mathcal{F},\mathcal{W}(S),\mathcal{F}(S))
\end{eqnarray}
since $\ubar\leq J$. Therefore after integrating the elliptic estimate (\ref{eq:hodgenew}), we have 
\begin{eqnarray}
||\nabla\nabla_{3}\eta||_{L^{2}(H)}\leq C(\mathcal{O}_{0},\mathcal{W},\mathcal{F},\mathcal{W}(S),\mathcal{F}(S))
\end{eqnarray}
Proceeding in an exactly similar way we obtain 
\begin{eqnarray}
||\nabla\nabla_{3}\omegabar||_{L^{2}(H)},~||\nabla\nabla_{4}\etabar||_{L^{2}(  \Hbar)},||\nabla\nabla_{4}\omega||_{L^{2}(  \Hbar)}\leq C(\mathcal{O}_{0},\mathcal{W},\mathcal{F},\mathcal{W}(S),\mathcal{F}(S)).
\end{eqnarray}
The estimates for the rest of the connection coefficients can be estimated directly through their transport equations since the remaining connection coefficients satisfy both $\nabla_{4}$ and $\nabla_{3}$ transport equations.
\end{proof}
In the following lemma we prove the $L^{4}(S)$ estimates of the difficult derivatives (i.e., can not be estimated directly from null transport equations)  of the connection coefficients.\\ 
\begin{lemma} 
\label{6}
Let $\mathcal{W},\mathcal{F}\mathcal{W}(S),\mathcal{F}(S)<\infty$, then the following estimate holds 
\begin{eqnarray}
||\nabla_{3}\eta||_{L^{4}(S)}\leq C(\mathcal{O}_{0},\mathcal{W},\mathcal{F}\mathcal{W}(S),\mathcal{F}(S)),~||\nabla_{3}\omegabar||_{L^{4}(S)}\leq C(\mathcal{O}_{0},\mathcal{W},\mathcal{F}\mathcal{W}(S),\mathcal{F}(S)),\\
||\nabla_{4}\etabar||_{L^{4}(S)}\leq C(\mathcal{O}_{0},\mathcal{W},\mathcal{F}\mathcal{W}(S),\mathcal{F}(S)),~||\nabla_{4}\omega||_{L^{4}(S)}\leq C(\mathcal{O}_{0},\mathcal{W},\mathcal{F}\mathcal{W}(S),\mathcal{F}(S))
\end{eqnarray}
\end{lemma}
\begin{proof} In order to prove these estimates, we use the co-dimension 1 trace inequalities for any field $\varphi$ (be it a section of the gauge bundle or tangent bundle or mixed)
\begin{eqnarray}
||\varphi||_{L^{4}(S)}\leq C\left(||\varphi||_{L^{4}(\mathcal{\bar{S}}^{'})}+||\varphi||^{1/2}_{L^{2}(  \Hbar)}||\hnabla_{3}\varphi||^{1/4}_{L^{2}(  \Hbar)}(||\varphi||_{L^{2}(  \Hbar)}\nonumber+||\hnabla\varphi||_{L^{2}(  \Hbar)})^{1/4}\right),\\
||\varphi||_{L^{4}(S)}\leq C\left(||\varphi||_{L^{4}(\mathcal{\bar{S}}^{'})}+||\varphi||^{1/2}_{L^{2}(H)}||\hnabla_{4}\varphi||^{1/4}_{L^{2}(H)}(||\varphi||_{L^{2}(H)}\nonumber+||\hnabla\varphi||_{L^{2}(H)})^{1/4}\right).
\end{eqnarray}
We prove one of the connection coefficients. The rest of the connection coefficients can be handled exactly similar way. Let's consider $\varphi=\nabla_{3}\eta$ and write
\begin{eqnarray}
\label{eq:L4eta}
||\nabla_{3}\eta||_{L^{4}(S)}\leq C\left(||\nabla_{3}\eta||_{L^{4}(\mathcal{\bar{S}}^{'})}+||\nabla_{3}\eta||^{1/2}_{L^{2}(H)}||\nabla_{4}\nabla_{3}\eta||^{1/4}_{L^{2}(H)}(||\nabla_{3}\eta||_{L^{2}(H)}+||\hnabla\nabla_{3}\eta||_{L^{2}(H)})^{1/4}\right).
\end{eqnarray}
Note that every term except $\nabla_{4}\nabla_{3}\eta$ on the right-hand side is estimated. In order to estimate this term we can differentiate the $\nabla_{4}$ transport equation for $\eta$ with respect to $e_{3}$. Such an operation yields (schematically) 
\begin{eqnarray}
\nabla_{4}\nabla_{3}\eta=-\nabla_{3}\chi(\eta-\etabar)-\chi(\nabla_{3}\eta-\nabla_{3}\etabar)-\nabla_{3}\beta-(\hnabla_{3}\alpha^{F})(\rho^{F}+\sigma^{F})\\\nonumber-\alpha^{F}(\hnabla_{3}\rho^{F}+\hnabla_{3}\sigma^{F}) 
+\omega\nabla_{3}\eta+\omegabar\nabla_{4}\eta+(\eta-\etabar)\nabla\eta+\sigma\eta+(\rho^{F}\cdot\rho^{F}+\sigma^{F}\cdot\sigma^{F})\eta.
\end{eqnarray}
Using the estimates from the previous lemma (6.2-6.6) and the transport equations we obtain 
\begin{eqnarray}
||\nabla_{4}\nabla_{3}\eta||_{L^{2}(S)}\leq \mathcal{C}(\mathcal{O}_{0},\mathcal{W},\mathcal{F},\mathcal{W}(S),\mathcal{F}(S))+||\nabla\Psi||_{L^{2}(S)}. 
\end{eqnarray}
Here $\Psi$ does not contain $\bar{\alpha}$ and the Yang-Mills curvature components are estimated by means of $\mathcal{F}$ and $\mathcal{F}(S)$ since they enjoy one order higher regularity than the Weyl curvature. Therefore, we obtain 
\begin{eqnarray}
||\nabla_{4}\nabla_{3}\eta||_{L^{2}(H)}\leq \mathcal{C}(\mathcal{O}_{0},\mathcal{W},\mathcal{F},\mathcal{W}(S),\mathcal{F}(S)).
\end{eqnarray}
Now if we plug this estimate into (\ref{eq:L4eta}), we obtain 
\begin{eqnarray}
||\nabla_{3}\eta||_{L^{4}(S)}\leq \mathcal{C}(\mathcal{O}_{0},\mathcal{W},\mathcal{F},\mathcal{W}(S),\mathcal{F}(S)).
\end{eqnarray}
Proceeding exactly a similar way, we obtain the remaining estimates 
\begin{eqnarray}
||\nabla_{3}\omegabar||_{L^{4}(S)}\leq C(\mathcal{O}_{0},\mathcal{W},\mathcal{F}\mathcal{W}(S),\mathcal{F}(S)),
||\nabla_{4}\etabar||_{L^{4}(S)}\leq C(\mathcal{O}_{0},\mathcal{W},\mathcal{F}\mathcal{W}(S),\mathcal{F}(S)),\\\nonumber 
||\nabla_{4}\omega||_{L^{4}(S)}\leq C(\mathcal{O}_{0},\mathcal{W},\mathcal{F}\mathcal{W}(S),\mathcal{F}(S)).
\end{eqnarray}
This concludes the proof of this lemma.
\end{proof}
\begin{lemma}
\label{7}
Let $\varphi$ be any connection coefficients belonging to the set ($ \tr\chi, \tr\chibar,\widehat{\chi},\chibarhat,\eta,\etabar,\omega,\omegabar$), then the following estimates hold
\begin{eqnarray}
||\nabla_{3}\varphi||_{L^{4}(S)}\leq C(\mathcal{O}_{0},\mathcal{W},\mathcal{F},\mathcal{W}(S),\mathcal{F}(S)),~||\nabla_{4}\varphi||_{L^{4}(S)}\leq C(\mathcal{O}_{0},\mathcal{W},\mathcal{F},\mathcal{W}(S),\mathcal{F}(S))
\end{eqnarray}
\end{lemma}
\begin{proof} The proof is a straightforward consequence of the previous lemma (7) and the null evolution equations. The connection coefficients that satisfy the $\nabla_{4}$ and $\nabla_{3}$ null transport equations, we can directly estimate their $||\nabla_{4}\varphi||_{L^{4}(S)}$ since the terms on the right-hand side of such equations satisfy $L^{4}(S)$ estimate. The connection coefficients $(\eta,\omegabar,\etabar,\omega)$) that do satisfy only one of the $\nabla_{4},\nabla_{3}$ transport equations, the previous lemma yields the result. 
\end{proof}
\begin{lemma}
\label{L4alphabar}
$\mathcal{W}(S)\leq C(\mathcal{O}_{0},\mathcal{W}_{0},\mathcal{W},\mathcal{F}),~\mathcal{F}(S)\leq C(\mathcal{O}_{0},\mathcal{F}_{0},\mathcal{W},\mathcal{F})$ \textit{and also} $||\bar{\alpha}||_{L^{4}(S)}\leq C(\mathcal{O}_{0},\mathcal{W},\mathcal{F}), ||\hnabla\bar{\alpha}^{F}||_{L^{4}(S)}\leq C(\mathcal{O}_{0},\mathcal{F})$.
\end{lemma}
\begin{proof} Recall the definitions of $\mathcal{W}(S)$ and $\mathcal{F}(S)$
\begin{eqnarray}
\mathcal{W}(S):=\sup_{u,\ubar}||(\alpha,\beta,\bar{\beta},\rho,\sigma)||_{L^{4}(S)},~
\mathcal{F}(S):=\sum_{I=0}^{1}\sup_{u,\ubar}||\hnabla^{I}(\alpha^{F},\rho^{F},\sigma^{F})||_{L^{4}(S)}.
\end{eqnarray}
First, we prove the estimate for $\mathcal{W}(S)$. Recall the codimension-1 trace inequalities (remember $\hnabla$ acts as a usual covariant derivative on fields that are not sections of gauge bundle)
\begin{eqnarray}
\label{eq:codim1}
||\Psi||_{L^{4}(S_{u,\ubar})}\leq C\left(||\Psi||_{L^{4}(S_{u,0})}+||\Psi||^{1/2}_{L^{2}(H)}||\hnabla_{4}\Psi||^{1/4}_{L^{2}(H)}(||\Psi||_{L^{2}(H)}+||\hnabla\Psi||_{L^{2}(H)})^{1/4}\right),\\
\label{eq:codim2}
||\Psi||_{L^{4}(S_{u,\ubar})}\leq C\left(||\Psi||_{L^{4}(S_{0,\ubar})}+||\Psi||^{1/2}_{L^{2}(  \Hbar)}||\hnabla_{3}\Psi||^{1/4}_{L^{2}(  \Hbar)}(||\Psi||_{L^{2}(  \Hbar)}+||\hnabla\Psi||_{L^{2}(  \Hbar)})^{1/4}\right),
\end{eqnarray}
where the constants may depend on the initial data $\mathcal{O}_{0}$. We shall observe that $||(\beta,\bar{\beta},\rho,\sigma)||_{L^{4}(S)}$ and $||\hnabla^{I}(\rho^{F},\sigma^{F})||_{L^{4}(S)}$ can be completely determined by the initial data $\mathcal{O}_{0}$. However, this would not hold true for $\alpha$ and $\alpha^{F}$. However, this would not matter since we may use the trace inequalities (\ref{eq:codim1}-\ref{eq:codim2}).
We start with the Weyl curvature components. A direct application of (\ref{eq:codim1}-\ref{eq:codim2}) applied to $\alpha$ and $\bar{\alpha}$ yields 
\begin{eqnarray}
||\alpha||_{L^{4}(S)}\leq C(\mathcal{O}_{0},\mathcal{W}_{0},\mathcal{W}),~||\bar{\alpha}||_{L^{4}(S)}\leq C(\mathcal{O}_{0},\mathcal{W}_{0},\mathcal{W})
\end{eqnarray}
since $||\nabla_{4}\alpha||_{L^{2}(H)}$ and $||\nabla_{3}\bar{\alpha}||_{L^{2}(  \Hbar)}$ are dominated by $\mathcal{W}$. Now since $((\beta,\bar{\beta},\rho,\sigma))$ satisfy $\nabla_{3}$ equations, we may write 
\begin{eqnarray}
||\beta||_{L^{2}}\leq C\left(||\beta||_{L^{2}(S^{'})}\nonumber+\int_{0}^{u}||\nabla_{3}\beta||_{L^{2}}du^{'}\right)\\
\sim C\left(||\beta||_{L^{2}(S^{'})}+\int_{0}^{u}||- \tr\chibar\beta+D\rho\nonumber+~^{*}D\sigma+2\omegabar\beta+2\widehat{\chi}\cdot \bar{\beta}\right.\\\nonumber\left.+3(\eta\rho+~^{*}\eta\sigma)\nonumber+\frac{1}{2}(\hnabla(|\rho^{F}|^{2}+|\sigma^{F}|^{2})\nonumber-\bar{\chi}(\alpha^{F}\cdot\rho^{F}+\alpha^{F}\cdot\sigma^{F})-\hnabla_{4}(\bar{\alpha}^{F}\cdot\rho^{F}+\bar{\alpha}^{F}\cdot\sigma^{F})\right.\\\nonumber 
\left.+\omega(\bar{\alpha}^{F}\cdot\rho^{F}+\bar{\alpha}^{F}\cdot\sigma^{F})+2\etabar(|\rho^{F}|^{2}+\alpha^{F}\cdot \Bar{\alpha}^{F}+|\sigma^{F}|^{2})+\etabar(|\rho^{F}|^{2}+|\sigma^{F}|^{2})||_{L^{2}(S)}du^{'}\right).
\end{eqnarray}
Notice that all of these terms are estimated by $\mathcal{W}(S),\mathcal{W},\mathcal{F}$ and $\mathcal{F}(S)$ with a factor of $\epsilon^{1/2}$ or $\epsilon$ in front. Therefore 
\begin{eqnarray}
||\beta||_{L^{2}}\leq C(\mathcal{O}_{0},\mathcal{W}_{0}).
\end{eqnarray}
after choosing a sufficiently small $\epsilon$. 
Now we repeat exact similar procedure for the remaining $\bar{\beta},\rho,$ and $\sigma$ to yield
\begin{eqnarray}
||\bar{\beta},\rho,\sigma||_{L^{2}(S)}\leq C(\mathcal{O}_{0},\mathcal{W}_{0}).
\end{eqnarray}
Now we want to use the trace inequality 
\begin{eqnarray}
||\Psi||_{L^{4}(S_{u,\ubar})}\leq C\left(||\Psi||_{L^{4}(S_{0,\ubar})}+||\Psi||^{1/2}_{L^{2}(  \Hbar)}||\nabla_{3}\Psi||^{1/4}_{L^{2}(  \Hbar)}(||\Psi||_{L^{2}(  \Hbar)}+||\nabla\Psi||_{L^{2}(  \Hbar)})^{1/4}\right)
\end{eqnarray}
for $\Psi=(\beta,\bar{\beta},\rho,\sigma)$.  Observe the following
\begin{eqnarray}
||\nabla\Psi||^{2}_{L^{2}(  \Hbar)}=\int_{0}^{u}\int_{S}|\nabla\Psi|^{2}\mu_{\gamma}du\leq C\mathcal{W}^{2}
\end{eqnarray}
and obtain
\begin{eqnarray}
||\Psi||_{L^{4}(S)}\leq C\left(||\Psi||_{L^{4}(S^{'})}+||\Psi||^{1/2}_{L^{2}(  \Hbar)}||\nabla_{3}\Psi||^{1/4}_{L^{2}(  \Hbar)}\mathcal{W}^{1/2}\right).
\end{eqnarray}
Now 
\begin{eqnarray}
||\nabla_{3}\Psi||_{L^{2}(  \Hbar)}\leq C(\mathcal{O}_{0},\mathcal{W},\mathcal{F}),
\end{eqnarray}
and 
\begin{eqnarray}
||\Psi||^{2}_{L^{2}(  \Hbar)}=\int_{0}^{u}\int_{S}|\Psi|^{2}\mu_{\gamma}du^{'}\leq \epsilon C(O_{0},\mathcal{W},\mathcal{W}(S),\mathcal{F},\mathcal{F}(S)).
\end{eqnarray}
due to the uniform $L^{2}(S)$ estimates for $\Psi=(\beta,\bar{\beta},\rho,\sigma)$. 
Putting everything together
\begin{eqnarray}
||(\beta,\bar{\beta},\rho,\sigma||_{L^{4}(S)}\leq C(\mathcal{O}_{0},\mathcal{W}_{0})+\epsilon^{\frac{1}{4}} C(O_{0},\mathcal{W},\mathcal{W}(S),\mathcal{F},\mathcal{F}(S))\leq C(\mathcal{O}_{0},\mathcal{W}_{0})
\end{eqnarray}
 due to the smallness of $\epsilon$.

\noindent Now we estimate $\mathcal{F}(S)$ in terms of $\mathcal{F}$. A direct application of the trace inequalities (\ref{eq:codim1}-\ref{eq:codim2}) yields 
\begin{eqnarray}
||\alpha^{F}||_{L^{4}(S)}\leq C(\mathcal{O}_{0},\mathcal{F}_{0},\mathcal{F}),~||\bar{\alpha}^{F}||_{L^{4}(S)}\leq C(\mathcal{O}_{0},\mathcal{F}_{0},\mathcal{F}).
\end{eqnarray}
Now recall the $\hnabla_{3}$ equation satisfied by $\rho^{F}$ and commute it with $\hnabla$ to obtain 
\begin{eqnarray}
\label{eq:temp}
\hnabla_{3}\hnabla\rho^{F}=-\hnabla\hat{div} \bar{\alpha}^{F}+\nabla( \tr\chibar)\rho^{F}+ \tr\chibar\hnabla\rho^{F}+\nabla(\eta-\etabar)\bar{\alpha}^{F}+(\eta-\etabar)\hnabla\bar{\alpha}^{F}\\\nonumber 
+(\bar{\beta}+\bar{\alpha}^{F}(\rho^{F}-\sigma^{F})+\bar{\alpha}^{F})\rho^{F}+(\eta+\etabar)\hnabla_{3}\rho^{F}-\bar{\chi}\hnabla\rho^{F}+\bar{\chi}\eta\rho^{F}.
\end{eqnarray}
Now an application of the transport inequality (proposition \ref{transport}) yields 
\begin{eqnarray}
||\hnabla\rho^{F}||_{L^{2}(S_{u,\ubar})}\leq C(||\hnabla\rho^{F}||_{L^{2}(S_{0,\ubar})}+\int_{0}^{u}||\hnabla_{3}\hnabla\rho^{F}||_{L^{2}(S)}du^{'}).
\end{eqnarray}
Now notice each term on the right-hand side of (\ref{eq:temp}) can be estimated in $L^{2}$ by definition of $\mathcal{W}$, $\mathcal{F}$ and utilizing the estimates for the connections (lemma )
\begin{eqnarray}
||\hnabla\rho^{F}||_{L^{2}(S_{u,\ubar})}\leq C||\hnabla\rho^{F}||_{L^{2}(S_{0,\ubar})}+\epsilon C(\mathcal{O}_{0},\mathcal{W}(S),\mathcal{W},\mathcal{F},\mathcal{F}(S))
\end{eqnarray}
Since $u\leq \epsilon$. We may therefore estimate $||\hnabla\rho^{F}||_{L^{2}(S)}$ by means of the initial data i.e., 
\begin{eqnarray}
\label{eq:adhoc1}
||\hnabla\rho^{F}||_{L^{2}(S_{u,\ubar})}\leq C(\mathcal{O}_{0},\mathcal{F}_{0}).
\end{eqnarray}
We obtain a similar estimate for $\sigma^{F}$ since it verifies a $\hnabla_{3}$ equation
\begin{eqnarray}
\label{eq:adhoc2}
||\hnabla\sigma^{F}||_{L^{2}(S_{u,\ubar})}\leq C(\mathcal{O}_{0},\mathcal{F}_{0}).
\end{eqnarray}
Of course, using the integration inequality, we also have 
\begin{eqnarray}
||\rho^{F},\sigma^{F}||_{L^{2}(S)}\leq C(\mathcal{O}_{0},\mathcal{F}_{0}).
\end{eqnarray}
Now using the trace inequality we obtain 
\begin{eqnarray}
||\hnabla(\rho^{F},\sigma^{F})||_{L^{4}(S_{u,\ubar})}\leq C\left(||\hnabla(\rho^{F},\sigma^{F})||_{L^{4}(S_{0,\ubar})}+||\hnabla(\rho^{F},\sigma^{F})||^{1/2}_{L^{2}(  \Hbar)}||\hnabla_{3}\hnabla(\rho^{F},\sigma^{F})||^{1/4}_{L^{2}(  \Hbar)}\right.\\\nonumber
\left.(||\hnabla(\rho^{F},\sigma^{F})||_{L^{2}(  \Hbar)}\nonumber+||\hnabla^{2}(\rho^{F},\sigma^{F})||_{L^{2}(  \Hbar)})^{1/4}\right)
\end{eqnarray}
where each term on the right-hand side is estimated in terms of $\mathcal{F},\mathcal{F}(S),\mathcal{W},$ and $\mathcal{W}(S)$ and in addition we gain a factor of $\epsilon^{\frac{1}{2}}$ from $||\hnabla(\rho^{F},\sigma^{F})||^{1/2}_{L^{2}(  \Hbar)}$ due to the estimates (\ref{eq:adhoc1}-\ref{eq:adhoc2}) i.e., 
\begin{eqnarray}
||\hnabla(\rho^{F},\sigma^{F})||_{L^{4}(S_{u,\ubar})}\leq C(\mathcal{O}_{0},\mathcal{F}_{0})+\epsilon^{\frac{1}{4}}\mathcal{C}(O_{0},\mathcal{F}_{0},\mathcal{W},\mathcal{F},\mathcal{F}(S),\mathcal{W}(S))
\end{eqnarray}
and therefore 
\begin{eqnarray}
||\hnabla(\rho^{F},\sigma^{F})||_{L^{4}(S_{u,\ubar})}\leq C(\mathcal{O}_{0},\mathcal{F}_{0})
\end{eqnarray}
after choosing sufficiently small $\epsilon$. However, for $\alpha^{F}$ we can not make use of the $\hnabla_{3}$ equation since it is not controlled on $  \Hbar$. Instead, we can make use of the $\hnabla_{4}$ trace inequality (\ref{eq:codim1})
\begin{eqnarray}
||\hnabla\alpha^{F}||_{L^{4}(S_{u,\ubar})}\nonumber\leq \\\nonumber C\left(||\hnabla\alpha^{F}||_{L^{4}(S_{u,0})}+||\hnabla\alpha^{F}||^{1/2}_{L^{2}(H)}||\hnabla_{4}\hnabla\alpha^{F}||^{1/4}_{L^{2}(H)}(||\hnabla\alpha^{F}||_{L^{2}(H)}\nonumber+||\hnabla\hnabla\alpha^{F}||_{L^{2}(H)})^{1/4}\right)
\end{eqnarray}
where every term on the right-hand side is under control 
\begin{eqnarray}
||\hnabla\alpha^{F}||_{L^{4}(S_{u,\ubar})}\leq C(\mathcal{O}_{0},\mathcal{W}_{0},\mathcal{F},\mathcal{F}_{0}).
\end{eqnarray}
An exact similar argument for $\bar{\alpha}^{F}$ yields 
\begin{eqnarray}
||\hnabla\bar{\alpha}^{F}||_{L^{4}(S_{u,\ubar})}\leq C(\mathcal{O}_{0},\mathcal{W}_{0},\mathcal{F},\mathcal{F}_{0}).
\end{eqnarray}
This concludes the proof of the lemma.
\end{proof}
Notice that we also have 
\begin{eqnarray}
\label{eq:need1}
||\hnabla_{4}\alpha^{F}||_{L^{4}(S_{u,\ubar})}\leq C(\mathcal{O}_{0},\mathcal{W}_{0},\mathcal{F},\mathcal{F}_{0}),~
||\hnabla_{3}\bar{\alpha}^{F}||_{L^{4}(S_{u,\ubar})}\leq C(\mathcal{O}_{0},\mathcal{W}_{0},\mathcal{F},\mathcal{F}_{0})
\end{eqnarray}
through a direct use of the trace inequalities (\ref{eq:codim1}-\ref{eq:codim2}).

\section{Energy estimates for the Weyl and Yang-Mills curvature components}
\noindent In this section we estimate the energy associated with the Weyl and Yang-Mills curvature components. This would complete the proof of the main theorem. Once the connection coefficients are estimated, it is straightforward to estimate the curvatures only we need to keep track of bad components ($\bar{\alpha}$ and $\bar{\alpha}^{F}$) since they controlled only on an incoming null hypersurface $  \Hbar$. However as we have mentioned previously, the connection coefficients that are associated with these bad terms in the curvature estimate are controlled solely by the initial data $\mathcal{O}_{0}$. Therefore, we can safely utilize the Gr\"onwall's inequality to close the energy argument. There are several ways to obtain the energy estimate for the curvatures. One of the ways is to utilize the stress-energy tensors and associated divergence identities. Even though we do not have a canonical stress-energy tensor for gravity due to the equivalence principle, we can still use the Bel-Robinson tensor. Yang-Mills field, on the other hand, is equipped with its own canonical stress-energy tensor. Despite the fact that the utilization of the stress-energy tensor provides us with a direct physical insight into energy propagation (and possible concentration), we will not use it in the current context. Instead, we take a direct approach (the approach that exists in the traditional PDE literature). First, we write down the integration identities that are useful in the current context.\\
For a gauge invariant object $f: M\to\mathbb{R}$, the following integration by parts identities hold
\begin{eqnarray}
\label{eq:IBP1}
\int_{D_{u},\ubar}\nabla_{3}f=\int_{H_{0}}f-\int_{H_{0}(0,\ubar)}f+\int_{\mathcal{D}_{u,\ubar}}f(2\omegabar- \tr\chibar),\\
\label{eq:IBP2}
\int_{D_{u},\ubar}\nabla_{4}f=\int_{  \Hbar_{\ubar}(0,u)}f-\int_{  \Hbar_{0}(0,u)}f+\int_{\mathcal{D}_{u,\ubar}}f(2\omega- \tr\chi).
\end{eqnarray}
See proposition (\ref{integration}) for a proof. 
We start the estimates with components of the Weyl curvature. Notice that we need to identify a suitable collection of the curvature components whose energy estimates generate a cancellation of principal terms. This should of course be possible since we have already established the manifestly hyperbolic characteristics of the Yang-Mills sourced Bianchi equations.\\
\subsection{Energy Estimates for Weyl curvature}
\begin{lemma}
\label{71}
The null components of the Weyl curvature satisfy the following $L^{2}$ energy estimates
\begin{eqnarray}
\int_{H}|\alpha|^{2}+2\int_{  \Hbar_{\ubar}}|\beta|^{2}\leq C(\mathcal{W}_{0})+\epsilon C(\mathcal{O}_{0},\mathcal{W},\mathcal{F})\\
\int_{H_{u}}|\beta|^{2}+\int_{  \Hbar_{\ubar}}|\rho|^{2}+\int_{  \Hbar_{\ubar}}|\sigma|^{2}\leq C(\mathcal{W}_{0})+\epsilon C(\mathcal{O}_{0},\mathcal{W},\mathcal{F}),\\
\int_{  \Hbar_{u}}|\bar{\beta}|^{2}+\int_{H_{u}}|\rho|^{2}+\int_{H_{u}}|\sigma|^{2}\leq C(\mathcal{W}_{0})+\epsilon C(\mathcal{O}_{0},\mathcal{W},\mathcal{F}),\\
2\int_{H_{u}}|\bar{\beta}|^{2}+\int_{  \Hbar_{\ubar}}|\bar{\alpha}|^{2}\leq C(\mathcal{W}_{0})+\epsilon C(\mathcal{O}_{0},\mathcal{W},\mathcal{F})+\epsilon^{\frac{1}{2}}C(\mathcal{O}_{0},\mathcal{W},\mathcal{F})\\\nonumber +C(\mathcal{O}_{0})\int_{0}^{\ubar}||\bar{\alpha}||^{2}_{L^{2}(  \Hbar)}
\end{eqnarray}
\end{lemma}
\begin{proof} The proof relies on the direct use of the integration by parts inequalities (\ref{eq:IBP1}-\ref{eq:IBP2}). First consider $(\alpha,\beta)$. Direct use of (\ref{eq:IBP1}-\ref{eq:IBP2}) yields 
\begin{eqnarray}
2\int_{\mathcal{D}_{u,\ubar}}\left(\langle\alpha,\nabla_{3}\alpha\rangle\nonumber+2\langle\beta,\nabla_{4}\beta\rangle\right)\\\nonumber 
=\int_{H_{u}}|\alpha|^{2}+2\int_{  \Hbar_{\ubar}}|\beta|^{2}-\int_{H_{0}(0,\ubar)}|\alpha|^{2}-2\int_{  \Hbar_{0}(u,0)}|\beta|^{2}+\int_{\mathcal{D}_{u,\ubar}}|\alpha|^{2}(2\omegabar-\frac{1}{2} \tr\chibar)\\\nonumber 
+\int_{\mathcal{D}_{u,\ubar}}2|\beta|^{2}(2\omega-\frac{1}{2} \tr\chi)
\end{eqnarray}
and therefore 
\begin{eqnarray}
\int_{H_{u}}|\alpha|^{2}+2\int_{  \Hbar_{\ubar}}|\beta|^{2}=\int_{H_{0}(0,\ubar)}|\alpha|^{2}+\int_{  \Hbar_{0}(u,0)}2|\beta|^{2}\\\nonumber\underbrace{-\int_{\mathcal{D}_{u,\ubar}}|\alpha|^{2}(2\omegabar-\frac{1}{2} \tr\chibar)
-\int_{\mathcal{D}_{u,\ubar}}2|\beta|^{2}(2\omega-\frac{1}{2} \tr\chi)}_{ER1}\\\nonumber 
-2\underbrace{\int_{\mathcal{D}_{u,\ubar}}\left(\langle\alpha,\nabla_{3}\alpha\rangle\nonumber+\langle\beta,\nabla_{4}\beta\rangle\right)}_{ER2}
\end{eqnarray}
We have to estimate the error terms $ER_{1}$ and $ER_{2}$. Since $\alpha$ is innocuous in the current context (recall we are trying to prove a semi-global existence result where the size of the domain transverse to $H$ is small), $ER_{1}$ and $ER_{2}$ can be easily estimated in terms of the initial data and smallness of $u$. Nevertheless, we write down the estimates explicitly. Let us recall the definition of integration over the spacetime slab $\mathcal{D}_{u,\ubar}$
\begin{eqnarray}
\int_{\mathcal{D}_{u,\ubar}}:=\int_{0}^{u}\int_{0}^{\ubar}\int_{S_{u,\ubar}}\Omega^{2}\mu_{\gamma}d\ubar du
\end{eqnarray}
and therefore in the error term $ER_{1}$, we can estimate the connection coefficients in $\sup_{u,\ubar}L^{\infty}(S_{u,\ubar})$ and estimate $\alpha,\beta$ by their $L^{2}$ norm on $H$. In doing so we gain a factor of $u$ ($\leq \epsilon$). Since $||\Omega||_{L^{\infty}(\mathcal{S)}}\lesssim 1$, we have 
\begin{eqnarray}
|ER_{1}|\leq
\sup_{u,\ubar}(||\omegabar||_{L^{\infty}(S_{u,\ubar})}+|| \tr\chibar||_{L^{\infty}(S_{u,\ubar})})\int_{0}^{u}||\alpha||^{2}_{L^{2}(H)}du\\\nonumber +\sup_{u,\ubar}(||\omega||_{L^{\infty}(S_{u,\ubar})}+|| \tr\chi||_{L^{\infty}(S_{u,\ubar})})\int_{0}^{u}||\beta||^{2}_{L^{2}(H)}du\\\nonumber 
\leq \epsilon C(\mathcal{O}_{0},\mathcal{W},\mathcal{F}).
\end{eqnarray}
Now we have to estimate the error term $ER_{2}$. We do so using the Yang-Mills sourced Bianchi equations
\begin{eqnarray}
ER_{2}=-2\int_{\mathcal{D}_{u,\ubar}}\left(\langle\alpha,\nabla_{3}\alpha\rangle\nonumber+2\langle\beta,\nabla_{4}\beta\rangle\right)\\\nonumber 
=-2\int_{\mathcal{D}_{u,\ubar}}\left(\underbrace{\langle\alpha,\nabla\hat{\otimes}\beta}_{PI1}-\frac{1}{2} \tr\chibar\alpha\nonumber+4\omegabar\alpha-3(\widehat{\chi}\rho+~^{*}\widehat{\chi}\sigma)+(\zeta+4\eta)\hat{\otimes}\beta\right.\\\nonumber 
\left.+\frac{1}{2}(D_{3}R_{44}-D_{4}R_{43})\gamma\rangle+2\underbrace{\langle\beta,div\alpha}_{PI2}-2 \tr\chi\beta-2\omega\beta^{W}+\eta\cdot\alpha-\frac{1}{2}(D_{b}R_{44}-D_{4}R_{4b})\rangle\right)\\\nonumber 
\sim\int_{\mathcal{D}_{u,\ubar}}\left(\langle\alpha,-\frac{1}{2} \tr\chibar\alpha\nonumber+4\omegabar\alpha-3(\widehat{\chi}\rho+~^{*}\widehat{\chi}\sigma)+(\zeta+4\eta)\hat{\otimes}\beta\rangle\right.\\\nonumber 
\left.+\langle\beta,-2 \tr\chi\beta-2\omega\beta+\eta\cdot\alpha-\frac{1}{2}(D_{b}R_{44}-D_{4}R_{4b})\rangle\right)+\underbrace{\int_{\mathcal{D}_{u,\ubar}}\langle\alpha,(\eta+\etabar)\beta\rangle}_{IBP1},
\end{eqnarray}
where we note that the principal terms $PI1$ and $PI2$ cancel each other producing the extra term ($IBP1$) due to the integration by parts procedure. In addition due to $\gamma-$traceless property of $\alpha$, $\langle\alpha,(D_{3}R_{44}-D_{4}R_{43})\gamma\rangle=0$. In order to estimate $ER_{2}$, we need to use the expression for $D_{b}R_{44}-D_{4}R_{4b}$. Since there is no delicate cancellation here, we will write this term schematically 
\begin{eqnarray}
D_{b}R_{44}-D_{4}R_{4b}\sim \langle\alpha^{F},\hnabla_{b}\alpha^{F}\rangle-\chi_{bc}\mathfrak{T}_{c4}+\eta_{b}\mathfrak{T}_{44}-2\omega \mathfrak{T}_{4b}\\\nonumber -(\hnabla_{4}\alpha^{F}_{b}\cdot(\rho^{F}\nonumber+\sigma^{F})-\alpha^{F}_{b}\cdot(\hnabla_{4}\rho^{F}+\hnabla_{4}\sigma^{F})).
\end{eqnarray}
Now note an extremely important fact: since we do not have a $\hnabla_{4}$ equation for $\alpha^{F}$, we can only estimate this term using the definition of $\mathcal{F}$ (this is one of the reasons we need to include $\hnabla_{4}\alpha^{F}$ term in the definition of $\mathcal{F}$). Noting $\mathfrak{T}_{44}\sim |\alpha^{F}|^{2}, \mathfrak{T}_{4b}\sim \alpha^{F}(\rho^{F}+\sigma^{F})$ we obtain
\begin{eqnarray}
|\int_{\mathcal{D}_{u,\ubar}}\langle\beta,D_{b}R_{44}-D_{4}R_{4b}\rangle|\leq C(\mathcal{O}_{0},\mathcal{W},\mathcal{F})||\beta||_{L^{2}(\mathcal{D})}||\hnabla_{4}\alpha^{F}||_{L^{2}(D)}\leq \epsilon C(\mathcal{O}_{0},\mathcal{W},\mathcal{F}).
\end{eqnarray}
Using the estimates of the connection coefficients, we may estimate the remaining terms of $ER_{2}$. Finally we obtain 
\begin{eqnarray}
|ER_{2}|\leq \epsilon C(\mathcal{O}_{0},\mathcal{W},\mathcal{F}).
\end{eqnarray}
The most important point here is that we gain a factor of $\epsilon$ from the bulk integral (i.e., the integral over $D$). Therefore, even though some of the connection coefficients (such as $\eta$ and $\omegabar$) are dependent on $\mathcal{W}$ and $\mathcal{F}$, they always appear as multiplied by a factor of $\epsilon$ and therefore under control. We have the first estimate 
\begin{eqnarray}
\int_{H_{u}}|\alpha|^{2}+\int_{  \Hbar_{\ubar}}|\beta|^{2}=\int_{H_{0}(0,\ubar)}|\alpha|^{2}+\int_{  \Hbar_{0}(u,0)}|\beta|^{2}+\epsilon C(\mathcal{O}_{0},\mathcal{W},\mathcal{F})
\end{eqnarray}
or 
\begin{eqnarray}
\int_{H_{u}}|\alpha|^{2}+\int_{  \Hbar_{\ubar}}|\beta|^{2}\leq C(\mathcal{W}_{0})+\epsilon C(\mathcal{O}_{0},\mathcal{W},\mathcal{F}).
\end{eqnarray}
We proceed in an exactly similar manner in order to estimate the remaining curvature components. Next we identify the triple $(\beta,\rho,\sigma)$ and write down the integral equations 
\begin{eqnarray}
2\int_{\mathcal{D}_{u,\ubar}}\left(\langle\beta,\nabla_{3}\beta\rangle+\langle\sigma,\nabla_{4}\sigma\rangle+\langle\rho,\nabla_{4}\rho\rangle\right)=\int_{H_{u}}|\beta|^{2}+\int_{  \Hbar_{\ubar}}|\rho|^{2}+\int_{  \Hbar_{\ubar}}|\sigma|^{2}\\\nonumber 
-\int_{H_{0}}|\beta|^{2}-\int_{  \Hbar_{0}}|\rho|^{2}-\int_{  \Hbar_{0}}|\sigma|^{2}+\int_{\mathcal{D}_{u,\ubar}}|\beta|^{2}(2\omegabar-\frac{1}{2} \tr\chibar)
+\int_{\mathcal{D}_{u,\ubar}}|\rho|^{2}(2\omega-\frac{1}{2} \tr\chi)\\\nonumber 
+\int_{\mathcal{D}_{u,\ubar}}|\sigma|^{2}(2\omega-\frac{1}{2} \tr\chi),
\end{eqnarray}
that is, 
\begin{eqnarray}
\int_{H_{u}}|\beta|^{2}+\int_{  \Hbar_{\ubar}}|\rho|^{2}+\int_{  \Hbar_{\ubar}}|\sigma|^{2}=\int_{H_{0}}|\beta|^{2}+\int_{  \Hbar_{0}}|\rho|^{2}+\int_{  \Hbar_{0}}|\sigma|^{2}\\\nonumber 
+\underbrace{\int_{\mathcal{D}_{u,\ubar}}|\beta|^{2}(2\omegabar-\frac{1}{2} \tr\chibar)
+\int_{\mathcal{D}_{u,\ubar}}|\rho|^{2}(2\omega-\frac{1}{2} \tr\chi)+\int_{\mathcal{D}_{u,\ubar}}|\sigma|^{2}(2\omega-\frac{1}{2} \tr\chi)}_{ER3}\\\nonumber 
-2\underbrace{\int_{\mathcal{D}_{u,\ubar}}\left(\langle\beta,\nabla_{3}\beta\rangle+\langle\sigma,\nabla_{4}\sigma\rangle+\langle\rho,\nabla_{4}\rho\rangle\right)}_{ER4}.
\end{eqnarray}
Once again, using the estimates of the connection coefficients from lemma (2), we estimate the error term $ER_{3}$
\begin{eqnarray}
|ER_{3}|\leq \epsilon C(\mathcal{O}_{0},\mathcal{W},\mathcal{F}).
\end{eqnarray}
In order to estimate $ER_{4}$, we utilize the null Bianchi equations (\ref{eq:bianchi1}-\ref{eq:bianchi2}). Explicit calculations yield 
\begin{eqnarray}
ER_{4}=2\int_{\mathcal{D}_{u,\ubar}}\left(\langle\beta,\nabla_{3}\beta\rangle+\langle\sigma,\nabla_{4}\sigma\rangle+\langle\rho,\nabla_{4}\rho\rangle\right)\\\nonumber 
=2\int_{\mathcal{D}_{u,\ubar}}\left(\langle\beta,- \tr\chibar\beta+\nabla\rho+~^{*}\nabla\sigma+2\omegabar\beta+2\widehat{\chi}\cdot \bar{\beta}+3(\eta\rho+~^{*}\eta\sigma)\nonumber+\frac{1}{2}(D_{b}R_{34}-D_{4}R_{3b})\rangle\right)\\\nonumber 
+2\int_{\mathcal{D}_{u,\ubar}}\left(\langle\sigma,-\frac{3}{2} \tr\chi \sigma-div~^{*}\beta+\frac{1}{2}\chibarhat\cdot~^{*}\alpha-\zeta\cdot~^{*}\beta-2\etabar\cdot~^{*}\beta\nonumber-\frac{1}{4}(D_{\mu}R_{4\nu}-D_{\nu}R_{4\mu})\epsilon^{\mu\nu}~_{34}\rangle\right)\\\nonumber 
+2\int_{\mathcal{D}_{u,\ubar}}\left(\langle\rho,-\frac{3}{2} \tr\chi\rho+div\beta-\frac{1}{2}\chibarhat\cdot\alpha+\zeta\cdot\beta+2\etabar\cdot\beta-\frac{1}{4}(D_{3}R_{44}-D_{4}R_{34})\rangle\right)\\\nonumber 
\sim\int_{\mathcal{D}_{u,\ubar}}\left(\langle\beta,- \tr\chibar\beta+2\omegabar\beta+2\widehat{\chi}\cdot \bar{\beta}+3(\eta\rho+~^{*}\eta\sigma)\nonumber+\frac{1}{2}(D_{b}R_{34}-D_{4}R_{3b})\rangle\right)\\\nonumber 
+\int_{\mathcal{D}_{u,\ubar}}\left(\langle\sigma,-\frac{3}{2} \tr\chi \sigma+\frac{1}{2}\chibarhat\cdot~^{*}\alpha-\zeta\cdot~^{*}\beta-2\etabar\cdot~^{*}\beta\nonumber-\frac{1}{4}(D_{\mu}R_{4\nu}-D_{\nu}R_{4\mu})\epsilon^{\mu\nu}~_{34}\rangle\right)\\\nonumber 
+\int_{\mathcal{D}_{u,\ubar}}\left(\langle\rho,-\frac{3}{2} \tr\chi\rho-\frac{1}{2}\chibarhat\cdot\alpha+\zeta\cdot\beta+2\etabar\cdot\beta-\frac{1}{4}(D_{3}R_{44}-D_{4}R_{34})\rangle\right)\\\nonumber +\int_{\mathcal{D}_{u,\ubar}}(\langle(\eta+\etabar)\sigma,~^{*}\beta\rangle-\langle(\eta+\etabar)\rho,\beta\rangle),
\end{eqnarray}
where the principal terms cancel point-wise through the integration by parts procedure. In order to estimate $ER_{4}$, we need the expressions for the source terms involving Yang-Mills curvature. Using the Einstein's equations $R_{\mu\nu}=\mathfrak{T}_{\mu\nu}$, schematically 
\begin{eqnarray}
D_{b}R_{34}-D_{4}R_{3b}\sim (\langle\hnabla\rho^{F},\rho^{F}\rangle+\langle\hnabla\sigma^{F},\sigma^{F}\rangle)\nonumber-\bar{\chi}\alpha^{F}(\sigma^{F}+\rho^{F})-\hnabla_{4}\bar{\alpha}^{F}\cdot(\rho^{F}+\sigma^{F})\\\nonumber -\bar{\alpha}^{F}\cdot(\hnabla_{4}\rho^{F}+\hnabla_{4}\sigma^{F})\\\nonumber 
+2\omega\bar{\alpha}^{F}(\rho^{F}+\sigma^{F})+2\etabar(\rho^{F}\cdot\rho^{F}+\alpha^{F}\cdot\bar{\alpha^{F}}+\sigma^{F}\cdot\sigma^{F})+\etabar(\rho^{F}\cdot\rho^{F}+\sigma^{F}\cdot\sigma^{F}),\\
(D_{\mu}R_{4\nu}-D_{\nu}R_{4\mu})\epsilon^{\mu\nu}~_{34}\sim \hnabla(\alpha^{F}\cdot\rho^{F}-\alpha^{F}\cdot\sigma^{F})-\chi(\rho^{F}\cdot\rho^{F}+\alpha^{F}\cdot\bar{\alpha^{F}}+\sigma^{F}\cdot\sigma^{F})\\\nonumber+(\eta-\etabar)\alpha^{F}\cdot(\rho^{F}+\sigma^{F}),\\\nonumber
D_{3}R_{44}-D_{4}R_{34}\sim\langle\alpha^{F},\hnabla_{3}\alpha^{F}\rangle+\omegabar|\alpha^{F}|^{2}+\eta_{A}(\alpha^{F}\cdot\sigma^{F}+\alpha^{F}\cdot\rho^{F})+\langle\rho^{F},\hnabla_{4}\rho^{F}\rangle\nonumber\\\nonumber+\langle\sigma^{F},\hnabla_{4}\sigma^{F}\rangle +\etabar(\rho^{F}\bar{\alpha}^{F}+\rho^{F}\cdot\sigma^{F})
\end{eqnarray}
where we may use the null evolution equations for the Yang-Mills curvature components
\begin{eqnarray}
\hnabla_{4}\bar{\alpha}^{F}+\frac{1}{2} \tr\chi\bar{\alpha}^{F}=-\hnabla\rho^{F}-~^{*}\hnabla\sigma^{F}-2~^{*}\etabar\sigma^{F}-2\etabar\rho^{F}+2\omega\bar{\alpha}^{F}-\chibarhat\cdot\alpha^{F}\\
\hnabla_{3}\alpha^{F}+\frac{1}{2} \tr\chibar\alpha^{F}=-\hnabla\rho^{F}+~^{*}\hnabla\sigma^{F}-2~^{*}\eta\sigma^{F}+2\eta\rho^{F}+2\omegabar\alpha^{F}-\widehat{\chi}\cdot \bar{\alpha}^{F}\\
\hnabla_{4}\rho^{F}=-\hat{div} \alpha^{F}- \tr\chi\rho^{F}-(\eta-\etabar)\cdot\alpha^{F}\\
\hnabla_{4}\sigma^{F}=-\hat{curl} \alpha^{F}- \tr\chi \sigma^{F}+(\eta-\etabar)\cdot ~^{*}\alpha^{F}\\
\hnabla_{3}\rho^{F}=-\hat{div} \bar{\alpha}^{F}+ \tr\chibar\rho^{F}+(\eta-\etabar)\cdot\bar{\alpha}^{F}\\
\label{eq:YM2}
\hnabla_{3}\sigma^{F}=-\hat{curl} \bar{\alpha}^{F}- \tr\chibar\sigma^{F}+(\eta-\etabar)\cdot~^{*}\bar{\alpha}^{F},
\end{eqnarray}

Note that even though $\bar{\alpha}^{F}$ appears in the source terms, we can estimate its point-wise norm on $S$ by means of $||\hnabla\bar{\alpha}^{F}||_{L^{4}(S)}$ through Sobolev embedding. We only write down the estimates involving the potentially problematic source terms since the remaining terms are harmless 
\begin{eqnarray}
|\int_{\mathcal{D}_{u,\ubar}}\langle\Psi,\bar{\alpha}^{F}\hnabla\Phi^{F}\rangle|\leq \int_{0}^{u}||\Psi||_{L^{2}(H)}||\hnabla\Phi^{F}||_{L^{2}(H)}\sup_{\ubar}||\nabla\bar{\alpha}^{F}||_{L^{4}(S)}du^{'}\leq \epsilon C(\mathcal{O}_{0},\mathcal{W},\mathcal{F}),
\end{eqnarray}
where $\Psi:=(\beta,\rho,\sigma)$ and $\Phi^{F}:=(\alpha^{F},\rho^{F},\sigma^{F})$. Therefore the error term $ER_{4}$ is estimated as follows 
\begin{eqnarray}
|ER_{4}|\leq \epsilon C(\mathcal{O}_{0},\mathcal{W},\mathcal{F}).
\end{eqnarray}
Putting everything together, we obtain 
\begin{eqnarray}
\int_{H_{u}}|\beta|^{2}+\int_{  \Hbar_{\ubar}}|\rho|^{2}+\int_{  \Hbar_{\ubar}}|\sigma|^{2}=\int_{H_{0}}|\beta|^{2}+\int_{  \Hbar_{0}}|\rho|^{2}+\int_{  \Hbar_{0}}|\sigma|^{2}+\epsilon C(\mathcal{O}_{0},\mathcal{W},\mathcal{F})\\\nonumber \leq C(\mathcal{W}_{0})+\epsilon C(\mathcal{O}_{0},\mathcal{W},\mathcal{F}).
\end{eqnarray}
Now we collect the triple $(\bar{\beta},\rho,\sigma)$ and write down the associated integral equation
\begin{eqnarray}
\int_{  \Hbar_{\ubar}}|\bar{\beta}|^{2}+\int_{H_{u}}|\rho|^{2}+\int_{H_{u}}|\sigma|^{2}=\int_{  \Hbar_{0}}|\bar{\beta}|^{2}+\int_{H_{0}}|\rho|^{2}+\int_{H_{0}}|\sigma|^{2}\\\nonumber 
+\underbrace{\int_{\mathcal{D}_{u,\ubar}}|\bar{\beta}|^{2}(2\omega-\frac{1}{2} \tr\chi)
+\int_{\mathcal{D}_{u,\ubar}}|\rho|^{2}(2\omegabar-\frac{1}{2} \tr\chibar)+\int_{\mathcal{D}_{u,\ubar}}|\sigma|^{2}(2\omegabar-\frac{1}{2} \tr\chibar)}_{ER5}\\\nonumber 
-2\underbrace{\int_{\mathcal{D}_{u,\ubar}}\left(\langle\bar{\beta},\nabla_{4}\bar{\beta}\rangle+\langle\sigma,\nabla_{3}\sigma\rangle+\langle\rho,\nabla_{3}\rho\rangle\right)}_{ER6}.
\end{eqnarray}
Using the estimates for the connections (lemma), $ER_{5}$ is under control
\begin{eqnarray}
|ER_{5}|\leq \epsilon C(\mathcal{O}_{0},\mathcal{W},\mathcal{F}).
\end{eqnarray}
In order to estimate $ER_{6}$, we utilize the null evolution equations
\begin{eqnarray}
ER_{6}=2\int_{\mathcal{D}_{u,\ubar}}\left(\langle\bar{\beta},\nabla_{4}\bar{\beta}\rangle+\langle\sigma,\nabla_{3}\sigma\rangle+\langle\rho,\nabla_{3}\rho\rangle\right)\\\nonumber 
=2\int_{\mathcal{D}_{u,\ubar}}\left(\langle\bar{\beta},- \tr\chi\bar{\beta}-\nabla\rho+~^{*}\nabla\sigma+2\omega\bar{\beta}+2\chibarhat\cdot\beta\nonumber-3(\etabar\rho-~^{*}\etabar\sigma)-\frac{1}{2}(D_{b}R_{43}-D_{3}R_{4b})\rangle\right)\\\nonumber 
+2\int_{{D_{u},\ubar}}\left(\langle\sigma,-\frac{3}{2} \tr\chibar\sigma-div~^{*}\bar{\beta}+\frac{1}{2}\widehat{\chi}\cdot~^{*}\bar{\alpha}-\zeta\cdot~^{*}\bar{\beta}-2\eta\cdot~^{*}\bar{\beta}\nonumber+\frac{1}{4}(D_{\mu}R_{3\nu}-D_{\nu}R_{3\mu})\epsilon^{\mu\nu}~_{34}\rangle\right)\\\nonumber 
+2\int_{\mathcal{D}_{u,\ubar}}\left(\langle\rho,-\frac{3}{2} \tr\chibar\rho-div\bar{\beta}-\frac{1}{2}\widehat{\chi}\cdot\bar{\alpha}+\zeta\cdot\bar{\beta}-2\eta\cdot\bar{\beta}+\frac{1}{4}(D_{3}R_{34}-D_{4}R_{33})\rangle\right)\\\nonumber 
\sim\int_{\mathcal{D}_{u,\ubar}}\left(\langle\bar{\beta},- \tr\chi\bar{\beta}+2\omega\bar{\beta}+2\chibarhat\cdot\beta\nonumber-3(\etabar\rho-~^{*}\etabar\sigma)-\frac{1}{2}(D_{b}R_{43}-D_{3}R_{4b})\rangle\right)\\\nonumber 
+\int_{{D_{u},\ubar}}\left(\langle\sigma,-\frac{3}{2} \tr\chibar\sigma+\frac{1}{2}\widehat{\chi}\cdot~^{*}\bar{\alpha}-\zeta\cdot~^{*}\bar{\beta}-2\eta\cdot~^{*}\bar{\beta}\nonumber+\frac{1}{4}(D_{\mu}R_{3\nu}-D_{\nu}R_{3\mu})\epsilon^{\mu\nu}~_{34}\rangle\right)\\\nonumber 
+\int_{\mathcal{D}_{u,\ubar}}\left(\langle\rho,-\frac{3}{2} \tr\chibar\rho-\frac{1}{2}\widehat{\chi}\cdot\bar{\alpha}+\zeta\cdot\bar{\beta}-2\eta\cdot\bar{\beta}+\frac{1}{4}(D_{3}R_{34}-D_{4}R_{33})\rangle\right)+\int_{\mathcal{D}_{u,\ubar}}\langle(\eta+\etabar)\rho,\bar{\beta}\rangle\\\nonumber 
+\int_{\mathcal{D}_{u,\ubar}}\langle(\eta+\etabar)\sigma,~^{*}\bar{\beta}\rangle.
\end{eqnarray}
Once again, we need the source terms explicitly in terms of the Yang-Mills curvature components
\begin{eqnarray}
D_{b}R_{43}-D_{3}R_{4b}\sim\langle\rho^{F},\hnabla\rho^{F}\rangle+\langle\sigma^{F},\hnabla\sigma^{F}\rangle-\bar{\chi}\alpha^{F}(\rho^{F}+\sigma^{F})+\chi\bar{\alpha}^{F}(\rho^{F}+\sigma^{F}),\\
(D_{\mu}R_{3\nu}-D_{\nu}R_{3\mu})\epsilon^{\mu\nu}~_{34}\sim \hnabla(\bar{\alpha}^{F}\rho^{F}-\bar{\alpha}^{F}\sigma^{F})-\bar{\chi}(\rho^{F}\rho^{F}+\bar{\alpha}^{F}\alpha^{F}+\sigma^{F}\sigma^{F})\\\nonumber 
+(\eta+\etabar)\bar{\alpha}^{F}(\rho^{F}+\sigma^{F}),\\
D_{3}R_{43}-D_{4}R_{33}\sim 2\langle\rho^{F},-\hat{div} \bar{\alpha}^{F}+ \tr\chibar\rho^{F}+(\eta-\etabar)\bar{\alpha}^{F}\rangle+2\langle\sigma_{F},-\hat{curl} \bar{\alpha}^{F}- \tr\chibar\sigma^{F}\\\nonumber +(\eta-\etabar)~^{*}\bar{\alpha}^{F}\rangle\\\nonumber 
-2\eta(\rho^{F}\bar{\alpha}^{F}+\rho^{F}\sigma^{F})-\langle \bar{\alpha}^{F},-\frac{1}{2} \tr\chi\bar{\alpha}^{F}-\hnabla\rho^{F}-~^{*}\hnabla\sigma^{F}-2~^{*}\etabar\cdot\sigma^{F}-2\etabar\cdot\rho^{F}\\\nonumber+2\omega\bar{\alpha}^{F}-\chibarhat\cdot\alpha^{F}\rangle+2\omega|\bar{\alpha}^{F}|^{2}-4\etabar_{a}(\bar{\alpha}^{F}\cdot\rho^{F}-\bar{\alpha}^{F}\cdot\sigma^{F}).
\end{eqnarray}
Using these explicit expressions, we will obtain estimates of the following types 
\begin{eqnarray}
|\int_{\mathcal{D}_{u,\ubar}}\langle\Psi, \hnabla\bar{\alpha}^{F}\Phi^{F}\rangle|\leq |\int_{0}^{u}\int_{0}^{\ubar}||\Psi||_{L^{2}(S)}||\hnabla\bar{\alpha}^{F}||_{L^{4}(S)}||\Phi^{F}||_{L^{4}(S)}d\ubar^{'}du^{'}|\leq \epsilon C(\mathcal{O}_{0},\mathcal{W},\mathcal{F}),\\
|\int_{D}\langle\Psi,\bar{\alpha}^{F}\hnabla\Phi^{F}\rangle|\leq |\int_{0}^{u}\int_{0}^{\ubar}||\Psi||_{L^{2}(S)}||\hnabla\Phi^{F}||_{L^{4}(S)}||\bar{\alpha}^{F}||_{L^{4}(S)}d\ubar^{'}du^{'}|\leq \epsilon C(\mathcal{O}_{0},\mathcal{W},\mathcal{F}),\\
|\int_{D}\langle\Psi,\bar{\alpha}\varphi\rangle|\leq |\int_{0}^{u}\int_{0}^{\ubar}||\Psi||_{L^{2}(S)}||\bar{\alpha}||_{L^{2}(S)}||\varphi||_{L^{\infty}(S)}d\ubar^{'}du|\leq \epsilon C(\mathcal{O}_{0},\mathcal{W},\mathcal{F})
\end{eqnarray}
where $\Psi:=(\bar{\beta},\rho,\sigma)$ and $\Phi^{F}:=(\alpha^{F},\rho^{F},\sigma^{F})$ and $\varphi$ denotes any connection coefficients. Collecting all the terms, we obtain 
\begin{eqnarray}
|ER_{6}|\leq \epsilon C(\mathcal{O}_{0},\mathcal{W},\mathcal{F})
\end{eqnarray}
and therefore 
\begin{eqnarray}
\int_{  \Hbar_{\ubar}}|\bar{\beta}|^{2}+\int_{H_{u}}|\rho|^{2}+\int_{H_{u}}|\sigma|^{2}=\int_{  \Hbar_{0}}|\bar{\beta}|^{2}+\int_{H_{0}}|\rho|^{2}+\int_{H_{0}}|\sigma|^{2}+\epsilon C(\mathcal{O}_{0},\mathcal{W},\mathcal{F})\\\nonumber \leq C(\mathcal{W}_{0})+\epsilon C(\mathcal{O}_{0},\mathcal{W},\mathcal{F}).
\end{eqnarray}
Now we move onto the last estimate. This one is the most delicate one and crucially depends on the special structure of the Einstein-Yang-Mills equations. Collecting the pair $(\bar{\alpha},\bar{\beta})$ and applying the integration identities, we obtain 
\begin{eqnarray}
2\int_{H_{u}}|\bar{\beta}|^{2}+\int_{  \Hbar_{\ubar}}|\bar{\alpha}|^{2}=2\int_{H_{0}}|\bar{\beta}|^{2}+\int_{  \Hbar_{0}}|\bar{\alpha}|^{2}+\underbrace{\int_{\mathcal{D}_{u,\ubar}}2|\bar{\beta}|^{2}(2\omegabar-\frac{1}{2} \tr\chibar)+\int_{\mathcal{D}_{u,\ubar}}|\bar{\alpha}|^{2}(2\omega-\frac{1}{2} \tr\chi)}_{ER7}\\\nonumber 
-2\underbrace{\int_{\mathcal{D}_{u,\ubar}}\left(2\langle\bar{\beta},\nabla_{3}\bar{\beta}\rangle+\langle\bar{\alpha},\nabla_{4}\bar{\alpha}\rangle\right)}_{ER8}.
\end{eqnarray}
Now notice the term $\int_{\mathcal{D}_{u,\ubar}}|\bar{\alpha}|^{2}(2\omega-\frac{1}{2} \tr\chi)$ in $ER_{7}$. Technically we can estimate $\bar{\alpha}$ in the $L^{4}(S)$ norm since in this current settings, we do have $||\bar{\alpha}||_{L^{4}(S)}$ under control (lemma 
\ref{L4alphabar}). However, we want to avoid using such an estimate since this issue will reappear at the level of higher order regularity as well (and in that case we will not have $||\nabla\bar{\alpha}||_{L^{4}(S)}$ under control). Therefore we want to emphasize the special structure of this term. The connection coefficients $\omega$ and $tr \chi$ that are multiplying $|\bar{\alpha}|^{2}$ in the bulk integral of $ER_{7}$ satisfy $\nabla_{3}$ transport equation and as a consequence, we could estimate these solely in terms of the initial data $\mathcal{O}_{0}$. In other words, we have 
\begin{eqnarray}
|\int_{\mathcal{D}_{u,\ubar}}|\bar{\alpha}|^{2}(2\omega-\frac{1}{2} \tr\chi)|\leq C(\mathcal{O}_{0})\int_{0}^{u}\int_{S}|\bar{\alpha}|^{2}\Omega\mu_{\gamma}du\leq C(\mathcal{O}_{0})||\bar{\alpha}||^{2}_{L^{2}(  \Hbar)}
\end{eqnarray}
or 
\begin{eqnarray}
|ER_{7}|\leq \epsilon C(\mathcal{O}_{0},\mathcal{W},\mathcal{F})+C(\mathcal{O}_{0})\int_{0}^{\ubar}||\bar{\alpha}||^{2}_{L^{2}(  \Hbar)}d\ubar^{'}
\end{eqnarray}
and therefore we can safely use Gr\"onwall's inequality to estimate $2\int_{H_{u}}|\bar{\beta}|^{2}+\int_{  \Hbar_{\ubar}}|\bar{\alpha}|^{2}$ provided we can control the bulk term $ER_{8}$. We do so now utilizing the null evolution equations
\begin{eqnarray}
ER_{8}\nonumber=2\int_{\mathcal{D}_{u,\ubar}}\left(2\langle\bar{\beta},\nabla_{3}\bar{\beta}\rangle+\langle\bar{\alpha},\nabla_{4}\bar{\alpha}\rangle\right)\\\nonumber 
=2\int_{\mathcal{D}_{u,\ubar}}\left(2\langle\bar{\beta},-2 \tr\chibar\bar{\beta}-div\bar{\alpha}-2\omegabar\omegabar+\etabar\cdot\bar{\alpha}+\frac{1}{2}(D_{b}R_{33}-D_{3}R_{3b})\rangle\right)\\\nonumber 
+2\int_{\mathcal{D}_{u,\ubar}}\left(\langle\bar{\alpha},-\frac{1}{2} \tr\chi\bar{\alpha}-\nabla\hat{\otimes}\bar{\beta}+4\omega\bar{\alpha}-3(\chibarhat\rho-~^{*}\chibarhat\sigma)+(\zeta-4\etabar)\hat{\otimes}\bar{\beta}\nonumber+\frac{1}{2}(D_{4}R_{33}-D_{3}R_{34})\gamma\rangle\right)\\\nonumber 
\sim 2\int_{\mathcal{D}_{u,\ubar}}\left(\langle\bar{\beta},-2 \tr\chibar\bar{\beta}-2\omegabar\omegabar+\etabar\cdot\bar{\alpha}+\frac{1}{2}(D_{b}R_{33}-D_{3}R_{3b})\rangle\right)\\\nonumber 
+2\int_{\mathcal{D}_{u,\ubar}}\left(\langle\underbrace{\bar{\alpha},-\frac{1}{2} \tr\chi\bar{\alpha}+4\omega\bar{\alpha}}_{A}-3(\chibarhat\rho-~^{*}\chibarhat\sigma)+(\zeta-4\etabar)\hat{\otimes}\bar{\beta}\nonumber+\frac{1}{2}(D_{4}R_{33}-D_{3}R_{34})\gamma\rangle\right)\\\nonumber 
+2\int_{D}\langle(\eta+\etabar)\bar{\beta},\bar{\alpha}\rangle)
\end{eqnarray}
where notice that the principal terms once again cancel each other in a point-wise manner after integration by parts procedure. An extremely important point to note here is that due to $\gamma-$trace-less property of $\bar{\alpha}$, the term $\langle\bar{\alpha},(D_{4}R_{33}-D_{3}R_{34})\gamma\rangle$ vanishes. This character persists in the higher order as well. Once again, we need to explicitly evaluate the source term to estimate this bulk integral
\begin{eqnarray}
D_{b}R_{33}-D_{3}R_{3b}\sim \langle\bar{\alpha}^{F},\hnabla\bar{\alpha}^{F}\rangle-\bar{\chi} \bar{\alpha}^{F}\cdot(\rho^{F}+\sigma^{F})+\etabar|\bar{\alpha}^{F}|^{2}-2\omegabar \bar{\alpha}^{F}\cdot(\rho^{F}+\sigma^{F})\\\nonumber-\hnabla_{3}(\bar{\alpha}^{F}\cdot\rho^{F}+\bar{\alpha}^{F}\cdot\sigma^{F}).
\end{eqnarray}
We note that we encounter a term $\hnabla_{3}\bar{\alpha}^{F}$ which does not arise from a null evolution equation and this is precisely why we needed to add this term in the definition of Yang-Mills curvature energy. Observe that in the most dangerous term $A$, the connection coefficients $ \tr\chi$ and $\omega$ satisfy $\nabla_{3}$ evolution equations and therefore using lemma \ref{1}, they are controlled by means of initial data $\mathcal{O}_{0}$ only (even though we could have utilized the $L^{4}(S)$ bound of $\bar{\alpha}$ at this level, we do not do so since in higher order we would not have the same privilege). The estimates for the bulk integral involves 
\begin{eqnarray}
|\int_{D}\langle\Psi,\varphi \bar{\alpha}\rangle|\leq \epsilon^{\frac{1}{2}}C(\mathcal{O}_{0},\mathcal{W},\mathcal{F}),~|\int_{D}\langle\Psi,\varphi\Psi\rangle|\leq \epsilon C(\mathcal{O}_{0},\mathcal{W},\mathcal{F}),\\\nonumber 
|\int_{D}\langle\Psi,\bar{\alpha}^{F}\hnabla\bar{\alpha}^{F}\rangle|\leq \epsilon C(\mathcal{F}_{0},\mathcal{W},\mathcal{F}),~|\int_{D}\langle\Psi,\bar{\alpha}^{F}\hnabla\bar{\alpha}^{F}\rangle|\leq \epsilon C(\mathcal{F}_{0},\mathcal{W},\mathcal{F}),\\
|\int_{D}\langle\Psi, \bar{\alpha}^{F}\hnabla\Phi^{F}\rangle|\leq \epsilon C(\mathcal{F}_{0},\mathcal{W},\mathcal{F}),~
|\int_{D}\langle\Psi,\hnabla_{3}\bar{\alpha}^{F}\rho^{F}\rangle|\leq \epsilon^{\frac{1}{2}}C(\mathcal{F}_{0},\mathcal{W},\mathcal{F}),\\
|\int_{D}\langle\bar{\alpha},-\frac{1}{2} \tr\chi \bar{\alpha}+4\omega \bar{\alpha}\rangle|\leq C(\mathcal{O}_{0})\int_{0}^{\ubar}||\bar{\alpha}||^{2}_{L^{2}(  \Hbar)}d\ubar^{'},
\end{eqnarray}
where $\Psi:=(\bar{\beta},\rho,\sigma),~\Phi^{F}:=(\rho^{F},\sigma^{F})$ and we have utilized $L^{4}(S)$ estimate of $\hnabla\bar{\alpha}^{F}$ and $\Phi^{F}$ (lemma \ref{L4alphabar}) together with the definitions of $\mathcal{W}$ and $\mathcal{F}$. Putting everything together, we obtain 
\begin{eqnarray}
\int_{H_{u}}|\bar{\beta}|^{2}+\int_{  \Hbar_{\ubar}}|\bar{\alpha}|^{2}=\int_{H_{0}}|\bar{\beta}|^{2}+\int_{  \Hbar_{0}}|\bar{\alpha}|^{2}+C(\mathcal{O}_{0})||\bar{\alpha}||^{2}_{L^{2}(  \Hbar)}\nonumber+\epsilon C(\mathcal{O}_{0},\mathcal{W},\mathcal{F})+\epsilon^{\frac{1}{2}}C(\mathcal{O}_{0},\mathcal{F}_{0},\mathcal{W},\mathcal{F})\\\nonumber 
\leq C(\mathcal{O}_{0},\mathcal{W}_{0},\mathcal{F}_{0})+C(\mathcal{O}_{0})\int_{0}^{\ubar}||\bar{\alpha}||^{2}_{L^{2}(  \Hbar)}.
\end{eqnarray}
This concludes the proof of the lemma.
\end{proof}

\noindent Now we need to obtain estimates for the higher-order energy associated with the Weyl curvature. Once we have completed the estimates for the Weyl curvature, we will move on to estimating the Yang-Mills curvature.\\  
\begin{lemma}
\label{72}
The horizontal derivatives of the null components of the Weyl curvature satisfy the following $L^{2}$ energy estimates
\begin{eqnarray}
\int_{H_{u}}|\nabla\alpha|^{2}+\int_{  \Hbar_{\ubar}}2|\nabla\beta|^{2}\leq C(\mathcal{W}_{0})+\epsilon C(\mathcal{O}_{0},\mathcal{W},\mathcal{F}),\\
\int_{H_{u}}|\nabla\beta|^{2}+\int_{  \Hbar_{\ubar}}|\nabla\rho|^{2}+\int_{  \Hbar_{\ubar}}|\nabla\sigma|^{2}\leq C(\mathcal{W}_{0})+\epsilon C(\mathcal{O}_{0},\mathcal{W},\mathcal{F}),\\
\int_{  \Hbar_{u}}|\nabla\bar{\beta}|^{2}+\int_{H_{u}}|\nabla\rho|^{2}+\int_{H_{u}}|\nabla\sigma|^{2}\leq C(\mathcal{W}_{0})+\epsilon C(\mathcal{O}_{0},\mathcal{W},\mathcal{F}),\\
\int_{H_{u}}2|\nabla\bar{\beta}|^{2}+\int_{  \Hbar_{\ubar}}|\nabla\bar{\alpha}|^{2}\leq C(\mathcal{W}_{0})+\epsilon^{\frac{1}{2}} C(\mathcal{O}_{0},\mathcal{W},\mathcal{F})+\epsilon C(\mathcal{O}_{0},\mathcal{W},\mathcal{F})\\\nonumber +C(\mathcal{O}_{0})\int_{0}^{\ubar}||\nabla\bar{\alpha}||^{2}_{L^{2}(  \Hbar)}d\ubar^{'}.
\end{eqnarray}
\end{lemma}
\begin{proof} The proof relies on the fact that the structure of the Yang-Mills sourced null Bianchi equation is preserved after commuting with horizontal derivatives. We proceed in an exactly similar manner. First we consider the pair $(\alpha,\beta)$
\begin{eqnarray}
\int_{H_{u}}|\nabla\alpha|^{2}+\int_{  \Hbar_{\ubar}}2|\nabla\beta|^{2}=\int_{H_{0}}|\nabla\alpha|^{2}+\int_{  \Hbar_{0}}2|\nabla\beta|^{2}\\\nonumber+\underbrace{\int_{\mathcal{D}_{u,\ubar}}|\nabla\alpha|^{2}(2\omegabar-\frac{1}{2} \tr\chibar)+\int_{\mathcal{D}_{u,\ubar}}2|\nabla\beta|^{2}(2\omega-\frac{1}{2} \tr\chi)}_{ER9}-2\underbrace{\int_{\mathcal{D}_{u,\ubar}}\left(\langle\nabla\alpha,\nabla_{3}\nabla\alpha\rangle+2\langle\nabla\beta,\nabla_{4}\nabla\beta\rangle\right)}_{ER10}.
\end{eqnarray}
Similar to the previous lemma \ref{71}, $ER9$ is controlled by using the estimates of the connection coefficients from lemma \ref{1}
\begin{eqnarray}
|ER9|\leq \epsilon C(\mathcal{O}_{0},\mathcal{W},\mathcal{F}).
\end{eqnarray}
In order to estimate $ER10$, we use the commuted null evolution equations
\begin{eqnarray}
ER10=2\int_{\mathcal{D}_{u,\ubar}}\left(\langle\nabla\alpha,\nabla_{3}\nabla\alpha\rangle+2\langle\nabla\beta,\nabla_{4}\nabla\beta\rangle\right)\\\nonumber 
=2\int_{\mathcal{D}_{u,\ubar}}\left(\langle\nabla\alpha,-\frac{1}{2}\nabla( \tr\chibar\alpha)+\nabla\hat{\otimes}\nabla\beta\nonumber+4\nabla(\omegabar\alpha)-3\nabla(\widehat{\chi}\rho+~^{*}\widehat{\chi}\sigma)+\nabla((\zeta+4\eta)\hat{\otimes}\beta)\right.\\\nonumber 
\left.+\frac{1}{2}\nabla(D_{3}R_{44}-D_{4}R_{43})\gamma+[\nabla_{3},\nabla]\alpha+[\nabla,\nabla]\hat{\otimes}\beta\rangle\right.\\\nonumber 
\left.+2\int_{\mathcal{D}_{u,\ubar}}2\langle\nabla\beta,-2\nabla( \tr\chi\beta)+div\nabla\alpha-2\nabla(\omega\beta)\nonumber+\nabla(\eta\cdot\alpha) 
-\frac{1}{2}\nabla(D_{b}R_{44}-D_{4}R_{4b})\right.\\\nonumber
\left.+[\nabla_{4},\nabla]\beta+[\nabla,div]\alpha\right)\\\nonumber 
\sim \int_{\mathcal{D}_{u,\ubar}}\left(\langle\nabla\alpha,-\frac{1}{2}\nabla( \tr\chibar\alpha)\nonumber+4\nabla(\omegabar\alpha)-3\nabla(\widehat{\chi}\rho+~^{*}\widehat{\chi}\sigma)+\nabla((\zeta+4\eta)\hat{\otimes}\beta)\right.\\\nonumber 
\left.+[\nabla_{3},\nabla]\alpha+[\nabla,\nabla]\hat{\otimes}\beta\rangle\right.\\\nonumber 
\left.+\int_{\mathcal{D}_{u,\ubar}}\langle\nabla\beta,-2\nabla( \tr\chi\beta)-2\nabla(\omega\beta)\nonumber+\nabla(\eta\cdot\alpha) 
-\frac{1}{2}\nabla(D_{b}R_{44}-D_{4}R_{4b})\right.\\\nonumber
\left.+[\nabla_{4},\nabla]\beta+[\nabla,div]\alpha\right)+\int_{D}\langle(\eta+\etabar\nabla\beta,\nabla\alpha\rangle.
\end{eqnarray}
Where once again $\langle\nabla\alpha,\nabla(D_{3}R_{44}-D_{4}R_{43})\gamma$ vanishes due to the $\gamma-$trace-less property of $\alpha$ and since $\gamma$ is metric. This error term may be estimated by means of the following estimates for the individual terms 
\begin{eqnarray}
|\int_{\mathcal{D}_{u,\ubar}}\langle\nabla\alpha,[\nabla_{3},\nabla]\alpha\rangle|
\sim|\int_{\mathcal{D}_{u,\ubar}}\langle\nabla\alpha, \bar{\beta}\alpha+\bar{\alpha}^{F}\alpha+\bar{\alpha}^{F}(\rho^{F}+\sigma^{F})\alpha\nonumber+(\eta+\etabar)\nabla_{3}\alpha-\bar{\chi}\nabla\alpha+\bar{\chi}\eta\alpha\rangle|\\\nonumber 
\leq \epsilon C(\mathcal{O}_{0},\mathcal{W},\mathcal{F})+\epsilon^{1/2}C(\mathcal{O}_{0},\mathcal{W},\mathcal{F}),\\
 |\int_{\mathcal{D}_{u,\ubar}}\langle\nabla\alpha,\nabla( \tr\chibar\alpha)\rangle|=|\int_{\mathcal{D}_{u,\ubar}}\langle\nabla\alpha,\nabla( \tr\chibar)\alpha\nonumber+tr{\bar{\chi}}\nabla\alpha\rangle|
 \\\nonumber \leq \int_{0}^{u}\sup_{\ubar}||\nabla  \tr\chibar||_{L^{4}(S_{u,\ubar})}\int_{0}^{\ubar}||\nabla\alpha||_{L^{2}(S_{u,\ubar)}}||\alpha||_{L^{4}(S_{u,\ubar})}\\\nonumber+\int_{0}^{u}(\sup_{\ubar}||\nabla  \tr\chibar||_{L^{4}(S)}||\alpha||_{L^{4}(S)})||\nabla\alpha||^{2}_{L^{2}(H)}du^{'}\leq \epsilon C(\mathcal{O}_{0},\mathcal{W}),\\\nonumber
|\int_{\mathcal{D}_{u,\ubar}}\langle\nabla\alpha,[\nabla,\nabla]\beta\rangle|\sim |\int_{\mathcal{D}_{u,\ubar}}\langle\nabla\alpha,K\beta\rangle|
\\\nonumber \leq \int_{0}^{u}\sup_{\ubar}\nonumber||K||_{L^{4}(S_{u,\ubar})}||\beta||_{L^{4}(S_{u,\ubar})}\int_{0}^{\ubar}||\nabla\alpha||_{L^{2}(S_{u,\ubar})}\leq \epsilon C(\mathcal{O}_{0},\mathcal{W}),
\end{eqnarray}
where we have estimated the Gauss curvature of $S$ in $L^{4}(S)$ through the null Hamiltonian constraint $K=\frac{1}{2}\chibarhat\widehat{\chi}-\frac{1}{4} \tr\chi  \tr\chibar-\rho+|\rho^{F}|^{2}+|\sigma^{F}|^{2}$
and the estimates from corollary (1) i.e., $||K||_{L^{4}(S)}\leq C(\mathcal{O}_{0},\mathcal{W}_{0},\mathcal{F}_{0},\mathcal{W},\mathcal{F})$. Moving on with the other terms 
\begin{eqnarray}
|\int_{\mathcal{D}_{u,\ubar}}\langle\nabla\beta,[\nabla_{4},\nabla]\beta\rangle|\sim|\nonumber \int_{\mathcal{D}_{u,\ubar}}\langle\nabla\beta,\beta\beta+\alpha^{F}\beta+\alpha^{F}(\rho^{F}-\sigma^{F})\beta+(\eta+\etabar)\nabla_{4}\beta-\chi\nabla\beta+\chi\etabar\beta\rangle|\\\nonumber 
 \leq \int_{0}^{u}\int_{0}^{\ubar}(||\nabla\beta||_{L^{2}(S)}||\beta||^{2}_{L^{4}(S)}+||\nabla\beta||_{L^{2}(S)}||\alpha^{F}||_{L^{4}(S)}||\beta||_{L^{4}(S)}+||\alpha^{F}||_{L^{2}(S)}||\rho^{F}||_{L^{4}(S)}||\beta||_{L^{4}(S)}\\\nonumber 
 +\int_{0}^{u}(||\nabla\beta||^{2}_{L^{2}(H)}||\chi||_{L^{\infty}}+||\chi||_{L^{\infty}}||\etabar||_{L^{\infty}}||\beta||_{L^{2}}||\nabla\beta||_{L^{2}})\\\nonumber +\int_{0}^{u}\int_{0}^{\ubar}||\eta+\etabar||_{L^{\infty}}||\nabla\beta||_{L^{2}}||\nabla_{4}\beta||_{L^{2}}
\end{eqnarray}
Now use the equation of motion for $\beta$
 \begin{eqnarray}
 \nabla_{4}\beta+2 \tr\chi\beta=div\alpha-2\omega\beta+\eta\cdot\alpha-\frac{1}{2}(D_{b}R_{44}-D_{4}R_{4b})
 \end{eqnarray}
 and 
 \begin{eqnarray}
 (D_{b}R_{44}-D_{4}R_{4b})\sim\langle\alpha^{F},\hnabla_{b}\alpha^{F}\rangle-\chi_{bc}\mathfrak{T}_{c4}+\eta_{b}\mathfrak{T}_{44}\nonumber-2\omega \mathfrak{T}_{4b}-\hnabla_{4}(\alpha^{F}_{b}\cdot\rho^{F}+\alpha^{F}_{b}\cdot\sigma^{F})
 \end{eqnarray}
 This is why we need $\hnabla_{4}\alpha^{F}$ in the derivative estimate. $\hnabla_{4}\rho^{F}$ and $\hnabla_{4}\sigma^{F}$ may be controlled by the respective  null evolution equations
 \begin{eqnarray}
\hnabla_{4}\rho^{F}=-\hat{div} \alpha^{F}- \tr\chi\rho^{F}-(\eta-\etabar)\cdot\alpha^{F},\\
\hnabla_{4}\sigma^{F}=-\hat{curl} \alpha^{F}- \tr\chi \sigma^{F}+(\eta-\etabar)\cdot ~^{*}\alpha^{F}.
 \end{eqnarray}
 Now recall $||\varphi||_{L^{\infty}(S)}\leq C||\nabla\varphi||_{L^{4}(S)}$ and therefore
 \begin{eqnarray}
 \int_{0}^{u}\int_{0}^{\ubar}||\nabla\beta||_{L^{2}}||\hnabla_{4}\alpha^{F}\rho^{F}||_{L^{2}}\leq C\int_{0}^{u}\int_{0}^{\ubar}||\nabla\beta||_{L^{2}}||\hnabla_{4}\alpha^{F}||_{L^{2}}||\nabla\rho^{F}||_{L^{4}}\\\nonumber 
 \leq \epsilon C(\mathcal{O}_{0},\mathcal{W},\mathcal{F}),\\
 \int_{0}^{u}\int_{0}^{\ubar}||\nabla\beta||_{L^{2}}||\alpha^{F}\hnabla\alpha^{F}||_{L^{2}}\leq \int_{0}^{u}\int_{0}^{\ubar}||\nabla\beta||_{L^{2}}||\alpha^{F}||_{L^{4}}||\hnabla\alpha^{F}||_{L^{4}}\\\nonumber 
 \leq \epsilon C(\mathcal{O}_{0},\mathcal{W},\mathcal{F}).
 \end{eqnarray}
 Therefore collecting all the terms together 
 \begin{eqnarray}
 \int_{\mathcal{D}_{u,\ubar}}\langle\nabla\beta,[\nabla_{4},\nabla]\beta\rangle\leq \epsilon C(\mathcal{O}_{0},\mathcal{W},\mathcal{F})
 \end{eqnarray}
which indicates that this is a good term since we gain a factor of $\epsilon$. Among other non-trivial terms, we have the following
\begin{eqnarray}
\int_{\mathcal{D}_{u,\ubar}}\langle\nabla\beta,\nabla(D_{b}R_{44}-D_{4}R_{4b})\rangle\\\nonumber \sim\int_{\mathcal{D}_{u,\ubar}}\langle\nabla\beta,\nabla(\langle\alpha^{F},\hnabla_{b}\alpha^{F}\rangle-\chi_{bc}\mathfrak{T}_{c4}+\eta_{b}\mathfrak{T}_{44}-2\omega \mathfrak{T}_{4b}-\hnabla_{4}(\alpha^{F}_{b}\cdot\rho^{F}-\alpha^{F}_{b}\cdot\sigma^{F}))\rangle.
 \end{eqnarray}
 Here we will encounter several terms of which the potentially problematic terms are estimated as follows
 \begin{eqnarray}
 |\int_{\mathcal{D}_{u,\ubar}}\langle\nabla\beta,\alpha^{F}\hnabla^{2}\alpha^{F}\rangle|\leq \int_{\mathcal{D}_{u,\ubar}}||\nabla\beta||_{L^{2}}||\alpha^{F}||_{L^{\infty}}||\hnabla^{2}\alpha^{F}||_{L^{2}}\\\nonumber 
 \leq C\int_{\mathcal{D}_{u,\ubar}}||\nabla\beta||_{L^{2}}||\hnabla\alpha^{F}||_{L^{4}}||\hnabla^{2}\alpha^{F}||_{L^{2}}\leq \epsilon C(O_{0},\mathcal{W},\mathcal{F})
 \end{eqnarray}

 and 
 \begin{eqnarray}
  |\int_{\mathcal{D}_{u,\ubar}}\langle\nabla\beta,\hnabla\alpha^{F}\hnabla\alpha^{F}\rangle|\leq \int_{\mathcal{D}_{u,\ubar}}||\nabla\beta||_{L^{2}(S)}||\hnabla\alpha^{F}||_{L^{4}(S)}||\hnabla\alpha^{F}||_{L^{4}(S)}\nonumber\leq \epsilon C(\mathcal{O}_{0},\mathcal{W},\mathcal{F}).
 \end{eqnarray}
 Therefore these are essentially good estimates since we gain a factor of $\epsilon$. Using the $L^{4}(S)$ estimates of the derivatives of the connection coefficients (corollary 2) and putting everything together, we have 
 \begin{eqnarray}
 |ER10|\leq \epsilon C(\mathcal{O}_{0},\mathcal{W},\mathcal{F})
 \end{eqnarray}
 and therefore 
 \begin{eqnarray}
 \int_{H_{u}}|\nabla\alpha|^{2}+\int_{  \Hbar_{\ubar}}2|\nabla\beta|^{2}=\int_{H_{0}}|\nabla\alpha|^{2}+\int_{  \Hbar_{0}}2|\nabla\beta|^{2}+\epsilon C(\mathcal{O}_{0},\mathcal{W},\mathcal{F})
 \end{eqnarray}
yielding 
\begin{eqnarray}
\int_{H_{u}}|\nabla\alpha|^{2}+\int_{  \Hbar_{\ubar}}2|\nabla\beta|^{2}\leq C(\mathcal{W}_{0})+\epsilon C(\mathcal{O}_{0},\mathcal{W},\mathcal{F}).
\end{eqnarray}
Now consider the triple $(\beta,\sigma,\rho)$ and write using (\ref{eq:IBP1}-\ref{eq:IBP2})
\begin{eqnarray}
\int_{H_{u}}|\nabla\beta|^{2}+\int_{  \Hbar_{\ubar}}|\nabla\rho|^{2}+\int_{  \Hbar_{\ubar}}|\nabla\sigma|^{2}= 
\int_{H_{0}}|\nabla\beta|^{2}+\int_{  \Hbar_{0}}|\nabla\rho|^{2}+\int_{  \Hbar_{0}}|\nabla\sigma|^{2}\\\nonumber+\underbrace{\int_{\mathcal{D}_{u,\ubar}}|\nabla\beta|^{2}(2\omegabar-\frac{1}{2} \tr\chibar)
+\int_{\mathcal{D}_{u,\ubar}}|\nabla\rho|^{2}(2\omega-\frac{1}{2} \tr\chi)
+\int_{\mathcal{D}_{u,\ubar}}|\nabla\sigma|^{2}(2\omega-\frac{1}{2} \tr\chi)}_{ER11}\\\nonumber -2\underbrace{\int_{\mathcal{D}_{u,\ubar}}\left(\langle\nabla\beta,\nabla_{3}\nabla\beta\rangle+\langle\nabla\sigma,\nabla_{4}\nabla\sigma\rangle+\langle\nabla\rho,\nabla_{4}\nabla\rho\rangle\right)}_{ER12}.
\end{eqnarray}
Once again, utilizing the estimates for the connection coefficients from lemma \ref{1}, we obtain 
\begin{eqnarray}
|ER11||\leq \epsilon C(\mathcal{O}_{0},\mathcal{W},\mathcal{F}).
\end{eqnarray}
Now we estimate $ER12$ using the null evolution equations
\begin{eqnarray}
-ER_{12}=2\int_{\mathcal{D}_{u,\ubar}}\left(\langle\nabla\beta,\nabla_{3}\nabla\beta\rangle+\langle\nabla\sigma,\nabla_{4}\nabla\sigma\rangle+\langle\nabla\rho,\nabla_{4}\nabla\rho\rangle\right)\\\nonumber 
=\int_{\mathcal{D}_{u,\ubar}}\left(\langle\nabla\beta,-\nabla( \tr\chibar\beta)+\nabla\nabla\rho\nonumber+~^{*}\nabla\nabla\sigma+2\nabla(\omegabar\beta)+2\nabla(\widehat{\chi}\cdot \bar{\beta})+3\nabla(\eta\rho+~^{*}\eta\sigma)\right.\\\nonumber
 \left.+\frac{1}{2}\nabla(D_{b}R_{34}-D_{4}R_{3b})+[\nabla_{3},\nabla]\beta+[\nabla,\nabla]\rho+[^{*}\nabla,\nabla]\sigma\rangle\right)\\\nonumber 
+\int_{\mathcal{D}_{u,\ubar}}\left(\langle\nabla\sigma,-\frac{3}{2}\nabla( \tr\chi\sigma)-div~^{*}\nabla\beta+\frac{1}{2}\nabla(\chibarhat\cdot~^{*}\alpha)-\nabla(\zeta\cdot~^{*}\beta)-2\nabla(\etabar\cdot~^{*}\beta)\right.\\\nonumber\left.-\frac{1}{4}\nabla(D_{\mu}R_{4\nu}-D_{\nu}R_{4\mu})\epsilon^{\mu\nu}~_{34}+[\nabla_{4},\nabla]\sigma-[\nabla,div]~^{*}\beta\rangle\right)\\\nonumber 
+\int_{\mathcal{D}_{u,\ubar}}\left(\langle\nabla\rho,-\frac{3}{2}\nabla( \tr\chi\rho)+div\nabla\beta-\frac{1}{2}\nabla(\chibarhat\cdot\alpha)+\nabla(\zeta\cdot\beta)\nonumber+2\nabla(\etabar\cdot\beta)\right.\\\nonumber 
\left.-\frac{1}{4}\nabla(D_{3}R_{44}-D_{4}R_{34})+[\nabla_{4}\nabla]\rho+[\nabla,div]\beta\rangle\right)\\\nonumber 
\sim\int_{\mathcal{D}_{u,\ubar}}\left(\langle\nabla\beta,-\nabla( \tr\chibar\beta)\nonumber+2\nabla(\omegabar\beta)+2\nabla(\widehat{\chi}\cdot \bar{\beta})+3\nabla(\eta\rho+~^{*}\eta\sigma)\right.\\\nonumber
\left.+\frac{1}{2}\nabla(D_{b}R_{34}-D_{4}R_{3b})+[\nabla_{3},\nabla]\beta+[\nabla,\nabla]\rho+[^{*}\nabla,\nabla]\sigma\rangle\right)\\\nonumber 
+ \int_{\mathcal{D}_{u,\ubar}}\left(\langle\nabla\sigma,-\frac{3}{2}\nabla( \tr\chi\sigma)+\frac{1}{2}\nabla(\chibarhat\cdot~^{*}\alpha)-\nabla(\zeta\cdot~^{*}\beta)-2\nabla(\etabar\cdot~^{*}\beta)\right.\\\nonumber\left.-\frac{1}{4}\nabla(D_{\mu}R_{4\nu}-D_{\nu}R_{4\mu})\epsilon^{\mu\nu}~_{34}+[\nabla_{4},\nabla]\sigma-[\nabla,div]~^{*}\beta\rangle\right)\\\nonumber 
+\int_{\mathcal{D}_{u,\ubar}}\left(\langle\nabla\rho,-\frac{3}{2}\nabla( \tr\chi\rho)-\frac{1}{2}\nabla(\chibarhat\cdot\alpha)+\nabla(\zeta\cdot\beta)\nonumber+2\nabla(\etabar\cdot\beta)\right.\\\nonumber 
\left.-\frac{1}{4}\nabla(D_{3}R_{44}-D_{4}R_{34})+[\nabla_{4}\nabla]\rho+[\nabla,div]\beta\rangle\right)+\int_{\mathcal{D}_{u,\ubar}}\langle(\eta+\etabar)\nabla\rho,\nabla\beta\rangle+\int_{\mathcal{D}_{u,\ubar}}\langle(\eta+\etabar)\nabla\sigma,\nabla\beta).
\end{eqnarray}
Once again we control the potentially problematic  terms that may arise from the coupled sector
\begin{eqnarray}
\int_{\mathcal{D}_{u,\ubar}}\langle\nabla\beta,\nabla(D_{b}R_{34}-D_{4}R_{3b})\rangle\\\nonumber 
\sim\int_{\mathcal{D}_{u,\ubar}}\langle\nabla\beta,\hnabla(\hnabla_{b}(|\rho^{F}|^{2}+|\sigma^{F}|^{2})\nonumber-\bar{\chi}_{ba}\mathfrak{T}_{a4}-\hnabla_{4}(\bar{\alpha}^{F}\cdot\rho^{F}+\bar{\alpha}^{F}\cdot\sigma^{F})+2\omega\mathfrak{T}_{3b}\\\nonumber
+2\etabar_{a}\mathfrak{T}_{ab}+\etabar_{b}\mathfrak{T}_{34})\rangle
\end{eqnarray}
Collect the possible dangerous terms and estimate them by means of elementary inequalities and the available lemmas \ref{1}-\ref{L4alphabar}
 \begin{eqnarray}
|\int_{\mathcal{D}_{u,\ubar}}\langle\nabla\beta,\bar{\alpha}^{F}\hnabla\hnabla_{4}\rho^{F}\rangle|\leq \nonumber\int_{\mathcal{D}_{u,\ubar}}\nonumber||\nabla\beta||_{L^{2}}||\bar{\alpha}^{F}||_{L^{\infty}}||\hnabla^{2}\alpha^{F}||_{L^{2}}\leq \epsilon C(\mathcal{O}_{0},\mathcal{W},\mathcal{F})\\\nonumber
|\int_{\mathcal{D}_{u,\ubar}}\langle\nabla\beta,\rho^{F}\hnabla^{2}\rho^{F}\rangle|\leq \epsilon C(\mathcal{O}_{0},\mathcal{W},\mathcal{F}) \\
|\int_{\mathcal{D}_{u,\ubar}}\langle\nabla\beta,\hnabla\bar{\alpha}^{F}\hnabla_{4}\rho^{F}\rangle|\leq \int_{\mathcal{D}_{u,\ubar}}||\nabla\beta||_{L^{2}}\nonumber||\hnabla\Bar{\alpha}^{F}||_{L^{4}}||\hnabla_{4}\rho^{F}||_{L^{4}}\leq \epsilon C(\mathcal{O}_{0},\mathcal{W},\mathcal{F})\\\nonumber
|\int_{\mathcal{D}_{u,\ubar}}\langle\nabla\beta,\bar{\alpha}^{F}\hnabla^{2}\alpha^{F}\rangle|\leq \epsilon C(\mathcal{O}_{0},\mathcal{W},\mathcal{F}),\\\nonumber
|\int_{\mathcal{D}_{u,\ubar}}\langle\nabla\sigma,\nabla(D_{\mu}R_{4\nu}-D_{\nu}R_{4\mu})\epsilon^{\mu\nu}~_{34}\rangle|\\\nonumber 
\sim |\int_{\mathcal{D}_{u,\ubar}}\langle\nabla\sigma, \nabla(\hnabla(\alpha^{F}\cdot\rho^{F}-\alpha^{F}\cdot\sigma^{F})-\chi(\rho^{F}\rho^{F}+\alpha^{F}\bar{\alpha}^{F}+\sigma^{F}\sigma^{F}+(\eta-\etabar)\alpha^{F}(\sigma^{F}+\rho^{F}))\rangle|\\\nonumber \leq \epsilon C(\mathcal{O}_{0},\mathcal{W},\mathcal{F}),\\
\int_{\mathcal{D}_{u,\ubar}}\langle\nabla\rho,\nabla(D_{3}R_{44}-D_{4}R_{34})\rangle
\sim\int_{\mathcal{D}_{u,\ubar}}\langle\nabla\rho,\hnabla(\alpha^{F}( \tr\chibar\alpha^{F}-\hnabla\rho^{F}+~^{*}\hnabla\sigma^{F}\nonumber-2~^{*}\eta\sigma^{F}+2\eta \rho^{F}\\\nonumber+2\omegabar\alpha^{F}-\widehat{\chi}\bar{\alpha}^{F})\\\nonumber 
 -4\omegabar|\alpha^{F}|^{2}-4\eta\alpha^{F}(\rho^{F}-\sigma^{F})+\rho^{F}\hnabla_{4}\rho^{F}+2\sigma^{F}\hnabla_{4}\sigma^{F}-2\etabar\alpha^{F}(\rho^{F}+\sigma^{F}))\rangle\\\nonumber 
 \leq \epsilon C(\mathcal{O}_{0},\mathcal{W},\mathcal{F})
 \end{eqnarray}
 where $\nabla\rho$ and the $\hnabla^{2}(\alpha^{F},\rho^{F},\sigma^{F})$ have to be bounded in $L^{2}(H)$, the remaining algebraic terms can be bounded in $L^{\infty}(S)$ and then apply Sobolev embedding for the $L^{\infty}(S)$ factor i.e., $||\alpha^{F},\rho^{F},\sigma^{F}||_{L^{\infty}(S)}\leq C(\mathcal{O}_{0})||\hnabla(\alpha^{F},\rho^{F},\sigma^{F})||_{L^{4}(S)}$.
 There will be term of type $|\int_{\mathcal{D}_{u,\ubar}}\nabla\Psi\bar{\alpha}^{F}\hnabla^{2}\bar{\alpha}^{F}|$ which may be a cause of concern. However, since we have $L^{4}(S)$ estimate for $\hnabla\bar{\alpha}^{F}$, we can control it in $L^{\infty}(S)$ through Sobolev embedding. Therefore we obtain
 \begin{eqnarray}
 |\int_{\mathcal{D}_{u,\ubar}}\nabla\Psi\bar{\alpha}^{F}\hnabla^{2}\bar{\alpha}^{F}|\nonumber=|\int_{0}^{u}\int_{0}^{\ubar}\int_{S}\nabla\Psi\bar{\alpha}^{F}\hnabla^{2}\bar{\alpha}^{F}|\\\nonumber 
 \leq \sup_{u,\ubar}||\bar{\alpha}^{F}||_{L^{\infty}}||\nabla\Psi||_{L^{2}(\mathcal{D}_{u,\ubar})}||\hnabla^{2}\bar{\alpha}^{F}||_{L^{2}(\mathcal{D}_{u,\ubar})}\leq \epsilon C(\mathcal{O}_{0},\mathcal{W},\mathcal{F})
 \end{eqnarray}
 where $\Psi$ is the Weyl curvature component belonging to the set $(\alpha,\beta,\rho,\sigma)$. In other words, as long as we do not have bad terms like $\nabla\bar{\alpha}\bar{\alpha}^{F}\hnabla^{2}\bar{\alpha}^{F}$, we will always have a gain of the factor $\epsilon$. Thus 
 \begin{eqnarray}
 \int_{\mathcal{D}_{u,\ubar}}\langle\nabla\sigma,\nabla(D_{\mu}R_{3\nu}-D_{\nu}R_{3\mu})\epsilon^{\mu\nu}~_{34})\rangle\leq \epsilon C(\mathcal{O}_{0},\mathcal{W},\mathcal{F}),\\
 \int_{\mathcal{D}_{u,\ubar}}\langle\nabla\rho,\nabla(D_{3}R_{44}-D_{4}R_{33})\rangle\leq \epsilon C(\mathcal{O}_{0},\mathcal{W},\mathcal{F}),
 \end{eqnarray}
 and 
 \begin{eqnarray}
 \int_{\mathcal{D}_{u,\ubar}}\langle\nabla\bar{\beta},\nabla(D_{b}R_{43}-D_{3}R_{4b})\rangle\leq \epsilon C(\mathcal{O}_{0},\mathcal{W},\mathcal{F}).
 \end{eqnarray}
 Putting all the terms together, we obtain 
\begin{eqnarray}
|ER12|\leq \epsilon C(\mathcal{O}_{0},\mathcal{W}_{0},\mathcal{F}_{0},\mathcal{W},\mathcal{F})
\end{eqnarray} 
and therefore 
\begin{eqnarray}
\int_{H_{u}}|\nabla\beta|^{2}+\int_{  \Hbar_{\ubar}}|\nabla\rho|^{2}+\int_{  \Hbar_{\ubar}}|\nabla\sigma|^{2}= 
\int_{H_{0}}|\nabla\beta|^{2}+\int_{  \Hbar_{0}}|\nabla\rho|^{2}+\int_{  \Hbar_{0}}|\nabla\sigma|^{2}\\\nonumber 
+\epsilon C(\mathcal{O}_{0},\mathcal{W}_{0},\mathcal{F}_{0},\mathcal{W},\mathcal{F})\leq C(\mathcal{W}_{0})+\epsilon C(\mathcal{O}_{0},\mathcal{W},\mathcal{F})
\end{eqnarray}
Proceeding in an exactly similar way, we obtain 
\begin{eqnarray}
\int_{  \Hbar_{u}}|\nabla\bar{\beta}|^{2}+\int_{H_{u}}|\nabla\rho|^{2}+\int_{H_{u}}|\nabla\sigma|^{2}\leq C(\mathcal{W}_{0})+\epsilon C(\mathcal{O}_{0},\mathcal{W},\mathcal{F}).
\end{eqnarray}
Now we focus on the most dangerous term which is the term arising from the remaining pair $(\bar{\alpha},\bar{\beta})$
\begin{eqnarray}
2\int_{H_{u}}|\nabla\bar{\beta}|^{2}+\int_{  \Hbar_{\ubar}}|\nabla\bar{\alpha}|^{2}=2\int_{H_{0}}|\nabla\bar{\beta}|^{2}\nonumber+\int_{  \Hbar_{0}}|\nabla\bar{\alpha}|^{2}\\\nonumber +\underbrace{\int_{\mathcal{D}_{u,\ubar}}2|\nabla\bar{\beta}|^{2}(2\omegabar-\frac{1}{2} \tr\chibar)+\int_{\mathcal{D}_{u,\ubar}}|\nabla\bar{\alpha}|^{2}(2\omega-\frac{1}{2} \tr\chi)}_{ER13}\\\nonumber 
-\underbrace{2\int_{\mathcal{D}_{u,\ubar}}\left(2\langle\nabla\bar{\beta},\nabla_{3}\nabla\bar{\beta}\rangle+\langle\nabla\bar{\alpha},\nabla_{4}\nabla\bar{\alpha}\rangle\right)}_{ER14}.
\end{eqnarray}
Now note the most important fact about the structure of the Yang-Mills sourced null Bianchi equations. In the term $ER_{13}$ we encounter the term $\int_{\mathcal{D}_{u,\ubar}}|\nabla\bar{\alpha}|^{2}(2\omega-\frac{1}{2} \tr\chi)$ which could potentially be dangerous since we can not touch $|\nabla\bar{\alpha}|^{2}$ term. But luckily, the connection coefficients $(\omega, \tr\chi)$ that multiply $|\nabla\bar{\alpha}|^{2}$ satisfy $\nabla_{3}$ evolution equation and therefore are solely estimated by means of the initial data $\mathcal{O}_{0}$. This crucial fact will allow us to make use of the Gr\"onwall's inequality. $ER13$ is estimated as 
\begin{eqnarray}
|ER13|\leq \epsilon C(\mathcal{O}_{0},\mathcal{W},\mathcal{F})+C(\mathcal{O}_{0})\int_{0}^{\ubar}||\nabla\bar{\alpha}||^{2}_{L^{2}(  \Hbar)}d\ubar^{'}. 
\end{eqnarray}
For the term $ER14$, we need caution since $\nabla\bar{\alpha}$ is involved. We explicitly evaluate $ER14$ using the commuted evolution equations 
\begin{eqnarray}
2\int_{\mathcal{D}_{u,\ubar}}\langle\nabla\bar{\beta},\nabla_{3}\nabla\bar{\beta}\rangle+\int_{\mathcal{D}_{u,\ubar}}\langle\nabla\bar{\alpha},\nabla_{4}\nabla\bar{\alpha}\rangle\\\nonumber
 =\int_{\mathcal{D}_{u,\ubar}}2\langle\nabla\bar{\beta},-2\nabla( \tr\chibar\bar{\beta})-div\nabla\bar{\alpha}-2\nabla(\omegabar\bar{\beta})+\nabla(\etabar\cdot\bar{\alpha})+\frac{1}{2}\nabla(D_{b}R_{33}\nonumber-D_{3}R_{3b})\\
+[\nabla_{3},\nabla]\bar{\beta}-[\nabla,div]\bar{\alpha}\rangle\\\nonumber 
+\int_{\mathcal{D}_{u,\ubar}}\langle\nabla\bar{\alpha},-\frac{1}{2}\nabla( \tr\chi\bar{\alpha})-\nabla\hat{\otimes}\nabla\bar{\beta}+4\nabla(\omega\bar{\alpha})\nonumber-3\nabla(\chibarhat\rho-~^{*}\chibarhat\sigma)+\nabla((\zeta-4\etabar)\hat{\otimes}\bar{\beta})\\\nonumber
+\frac{1}{2}\nabla(D_{4}R_{33}-D_{3}R_{34})\gamma+[\nabla_{4},\nabla]\bar{\alpha}+[\nabla,\nabla]\hat{\otimes}\bar{\beta}\rangle\\\nonumber 
\sim\int_{\mathcal{D}_{u,\ubar}}\langle\nabla\bar{\beta},-2\nabla( \tr\chibar\bar{\beta})-2\nabla(\omegabar\bar{\beta})+\nabla(\etabar\cdot\bar{\alpha})+\frac{1}{2}\nabla(D_{b}R_{33}\nonumber-D_{3}R_{3b})\\
+(\bar{\beta}+\bar{\alpha}^{F}(\sigma^{F}+\rho^{F}))\bar{\beta}+(\eta+\etabar)\nabla_{3}\bar{\beta}-\bar{\chi}\nabla\bar{\beta}+\chi\etabar\bar{\beta}-K\bar{\alpha}\rangle\\\nonumber 
+\int_{\mathcal{D}_{u,\ubar}}\langle\nabla\bar{\alpha},-\frac{1}{2}\nabla( \tr\chi\bar{\alpha})+4\nabla(\omega\bar{\alpha})\nonumber-3\nabla(\chibarhat\rho-~^{*}\chibarhat\sigma)+\nabla((\zeta-4\etabar)\hat{\otimes}\bar{\beta})\\\nonumber
+\frac{1}{2}\nabla(D_{4}R_{33}-D_{3}R_{34})\gamma+(\beta+\alpha^{F}(\rho^{F}+\sigma^{F}))\bar{\alpha}+(\eta+\etabar)\nabla_{4}\bar{\alpha}-\chi\nabla\bar{\alpha}+\chi\etabar\bar{\alpha}+K\bar{\beta}\rangle\\\nonumber 
+\int_{\mathcal{D}_{u,\ubar}}\langle (\eta+\etabar)\nabla\bar{\beta},\nabla\bar{\alpha}\rangle,
\end{eqnarray}
where once again the principal terms cancel point-wise after an integration by parts procedure. Apart from the cancellation of the principal terms, note that the $\gamma-$trace-less property of $\bar{\alpha}$ removes the term $\langle\nabla\bar{\alpha},\nabla(D_{4}R_{33}-D_{3}R_{34})\gamma\rangle$. This essentially indicates the fact that $\bar{\alpha}$ does not interact with the source stress-energy tensor at the level of topmost derivative. In order to estimate the remaining terms, we need to explicitly write down the Yang-Mills source term
\begin{eqnarray}
\nabla(D_{b}R_{33}-D_{3}R_{3b})\sim\nabla(\langle\bar{\alpha}^{F},\hnabla\bar{\alpha}^{F}\rangle-\bar{\chi} \bar{\alpha}^{F}\cdot(\rho^{F}\nonumber+\sigma^{F})+\etabar|\bar{\alpha}^{F}|^{2}-2\omegabar \bar{\alpha}^{F}\cdot(\rho^{F}+\sigma^{F})\\\nonumber-\hnabla_{3}(\bar{\alpha}^{F}\cdot\rho^{F}+\bar{\alpha}^{F}\cdot\sigma^{F}))\\\nonumber 
\sim\langle\hnabla\bar{\alpha}^{F},\hnabla\bar{\alpha}^{F}\rangle+\langle\bar{\alpha}^{F},\hnabla\hnabla\bar{\alpha}^{F}\rangle-\nabla(\bar{\chi}\bar{\alpha}^{F}(\rho^{F}+\sigma^{F}))+\etabar\langle\hnabla\bar{\alpha}^{F},\bar{\alpha}^{F}\rangle-\nabla(\omegabar\bar{\alpha}^{F}(\rho^{F}+\sigma^{F}))\\\nonumber-\hnabla\hnabla_{3}\bar{\alpha}^{F}(\rho^{F}+\sigma^{F}) -\hnabla_{3}\bar{\alpha}^{F}(\hnabla\rho^{F}+\hnabla\sigma^{F}),
\end{eqnarray}
where we notice that the potentially dangerous terms $\hnabla\hnabla_{3}\bar{\alpha}^{F}$ and $\hnabla_{3}\bar{\alpha}^{F}$ appear. However, these terms can be controlled by $\mathcal{F}$. We have $||\nabla_{3}\alpha^{F}||_{L^{4}(S)}$ estimated in terms of $\mathcal{F}$ by means of the co-dimension 1 trace inequality (\ref{eq:need1}) and $||\hnabla\hnabla_{3}\bar{\alpha}^{F}||_{L^{2}(  \Hbar)}$ is estimated in terms of  $\mathcal{F}$. We first write down how to control these two terms \begin{eqnarray}
|\int_{\mathcal{D}_{u,\ubar}}\langle\nabla\bar{\beta},\hnabla\hnabla_{3}\bar{\alpha}^{F}(\rho^{F}+\sigma^{F})\rangle|\leq ||\nabla\bar{\beta}||_{L^{2}(\mathcal{D}_{u,\ubar})}||\nabla\hnabla_{3}\bar{\alpha}^{F}(\rho^{F}+\sigma^{F})||_{L^{2}(\mathcal{D}_{u,\ubar})}\\\nonumber 
\leq \epsilon^{\frac{1}{2}}\sup_{u}||\nabla\bar{\beta}||_{L^{2}(H)}\sup_{u,\ubar}(||\hnabla\rho^{F}||_{L^{4}(S)}+||\hnabla\sigma^{F}||_{L^{4}(S)})\sup_{\ubar}||\hnabla\hnabla_{3}\bar{\alpha}^{F}||_{L^{2}(  \Hbar)}\leq \epsilon^{\frac{1}{2}}C(\mathcal{O}_{0},\mathcal{W},\mathcal{F}),\\
|\int_{\mathcal{D}_{u,\ubar}}\langle\nabla\bar{\beta},\hnabla\bar{\alpha}^{F}\hnabla(\rho^{F},\sigma^{F})|\leq \epsilon^{\frac{1}{2}}\sup_{u}||\nabla\bar{\beta}||_{L^{2}(H)}\sup_{u,\ubar}||\hnabla_{3}\bar{\alpha}^{F}||_{L^{4}(S)}\sup_{u,\ubar}||\hnabla(\rho^{F},\sigma^{F})||_{L^{4}(S)}\\\nonumber \leq\nonumber \epsilon^{\frac{1}{2}}C(\mathcal{O}_{0},\mathcal{W},\mathcal{F})
\end{eqnarray}
where we have utilized the fact that $\nabla\bar{\beta}$ is controlled on $H$ (in addition to $  \Hbar$). In addition, we also have terms of the following types
\begin{eqnarray}
|\int_{\mathcal{D}_{u,\ubar}}\langle\nabla\bar{\beta},\hnabla^{2}\bar{\alpha}^{F}\bar{\alpha}^{F}\rangle|\leq \epsilon^{\frac{1}{2}}\sup_{u}||\nabla\bar{\beta}||_{L^{2}(H)}\sup_{u,\ubar}||\hnabla\alpha||_{L^{4}(S)}\sup_{\ubar}||\hnabla^{2}\bar{\alpha}^{F}||_{L^{2}(  \Hbar)}\\\nonumber \leq \epsilon^{\frac{1}{2}}C(\mathcal{O}_{0},\mathcal{W},\mathcal{F}),\\
|\int_{\mathcal{D}_{u,\ubar}}\langle\nabla\bar{\beta},\hnabla\bar{\alpha}^{F}\hnabla\bar{\alpha}^{F}\rangle|\leq \epsilon   \sup_{u}||\nabla\bar{\beta}||_{L^{2}(H)} \sup_{u,\ubar}||\hnabla\bar{\alpha}^{F}||^{2}_{L^{4}(S)}\\\nonumber \leq \epsilon C(\mathcal{O}_{0},\mathcal{W},\mathcal{F}),\\
|\int_{\mathcal{D}_{u,\ubar}}\langle\nabla\bar{\beta},\hnabla\bar{\alpha}^{F}\varphi\rangle|\leq \epsilon \sup_{u}||\nabla\bar{\beta}||_{L^{2}(H)}\sup_{u,\ubar}||\hnabla\bar{\alpha}^{F}||_{L^{4}(S)}\sup_{u,\ubar}||\varphi||_{L^{\infty}(S)}\\\nonumber \leq \epsilon C(\mathcal{O}_{0},\mathcal{W},\mathcal{F}),\\
|\int_{\mathcal{D}_{u,\ubar}}\langle\nabla\bar{\beta},\hnabla_{3}\bar{\alpha}^{F}\varphi\rangle|\leq \epsilon C(\mathcal{O}_{0},\mathcal{W},\mathcal{F}),\\
|\int_{\mathcal{D}_{u,\ubar}}\langle\nabla\bar{\alpha},\alpha^{F}(\rho^{F},\sigma^{F})\bar{\alpha}\rangle|\leq \epsilon^{\frac{1}{2}}||\nabla\bar{\alpha}||_{L^{2}(  \Hbar)}\sup_{u,\ubar}||\alpha^{F}||_{L^{4}(S)}||\hnabla(\rho^{F},\sigma^{F})||_{L^{4}(S)}||\bar{\alpha}||_{L^{4}(S)}\leq \\\nonumber \epsilon^{\frac{1}{2}}C(\mathcal{O}_{0},\mathcal{W},\mathcal{F})
\end{eqnarray}
Therefore the coupled Yang-Mills-Einstein terms are under control. Now we control the remaining terms using the estimates on the connection coefficients (lemma \ref{1}-\ref{7}) as follows
\begin{eqnarray}
|\int_{\mathcal{D}_{u,\ubar}}\langle\nabla\bar{\beta},\bar{\beta}\nabla\varphi\rangle|\leq\epsilon \sup_{u}||\nabla\bar{\beta}||_{L^{2}(H)}\sup_{u,\ubar}||\bar{\beta}||_{L^{4}(S)}||\nabla\varphi||_{L^{4}(S)}\leq \epsilon C(\mathcal{O}_{0},\mathcal{W},\mathcal{F}),\\
|\int_{\mathcal{D}_{u,\ubar}}\langle\nabla\bar{\beta},\nabla\bar{\beta}\varphi\rangle|\leq \epsilon \sup_{u}||\nabla\bar{\beta}||^{2}_{L^{2}(H)}\sup_{u,\ubar}||\varphi||_{L^{\infty}(S)}\leq \epsilon C(\mathcal{O}_{0},\mathcal{W},\mathcal{F}),\\
|\int_{\mathcal{D}_{u,\ubar}}\langle\nabla\bar{\beta},\nabla\bar{\alpha}\varphi\rangle|\leq \epsilon^{\frac{1}{2}} \sup_{u}||\nabla\bar{\beta}||_{L^{2}(H)}\sup_{\ubar}||\nabla\bar{\alpha}||_{L^{2}(  \Hbar)}\sup_{u,\ubar}||\varphi||_{L^{\infty}(S)}\leq \epsilon^{\frac{1}{2}} C(\mathcal{O}_{0},\mathcal{W},\mathcal{F}),\\
|\int_{\mathcal{D}_{u,\ubar}}\langle\nabla\bar{\beta},K\bar{\alpha}\rangle|\leq \epsilon \sup_{u}||\nabla\bar{\beta}||_{L^{2}(H)} \sup_{u,\ubar}||K||_{L^{4}(S)}||\bar{\alpha}||_{L^{4}(S)}\leq \epsilon C(\mathcal{O}_{0},\mathcal{W},\mathcal{F}).
\end{eqnarray}
Now we estimate the most dangerous terms 
\begin{eqnarray}
|\int_{\mathcal{D}_{u,\ubar}}\langle\nabla\bar{\alpha},\nabla\bar{\alpha} \tr\chi\rangle|\leq \sup_{u,\ubar}|| \tr\chi||_{L^{\infty}(S)}\int_{0}^{\ubar}||\nabla\bar{\alpha}||^{2}_{L^{2}(  \Hbar)}d\ubar^{'}\leq C(\mathcal{O}_{0})\int_{0}^{\ubar}||\nabla\bar{\alpha}||^{2}_{L^{2}(  \Hbar)}d\ubar^{'},\\
|\int_{\mathcal{D}_{u,\ubar}}\langle\nabla\bar{\alpha},\nabla\bar{\alpha}\omega\rangle|\leq \sup_{u,\ubar}||\omega||_{L^{\infty}(S)}\int_{0}^{\ubar}||\nabla\bar{\alpha}||^{2}_{L^{2}(  \Hbar)}d\ubar^{'}\leq C(\mathcal{O}_{0})\int_{0}^{\ubar}||\nabla\bar{\alpha}||^{2}_{L^{2}(  \Hbar)}d\ubar^{'},
\end{eqnarray}
where we notice that the connection coefficients that appear multiplied with the top derivative of $\bar{\alpha}$ satisfy $\nabla_{3}$ equation and therefore solely determined by the initial data $\mathcal{O}_{0}$. This allows us to use the Gr\"onwall's inequality to complete the energy estimate. The remaining terms are harmless 
\begin{eqnarray}
|\int_{\mathcal{D}_{u,\ubar}}\langle\nabla\bar{\alpha},\nabla\Psi\varphi\rangle|\leq \epsilon^{\frac{1}{2}}C(\mathcal{O}_{0},\mathcal{W},\mathcal{F}),~|\int_{\mathcal{D}_{u,\ubar}}\langle\nabla\bar{\alpha},\nabla\varphi \Psi\rangle|\leq \epsilon^{\frac{1}{2}}C(\mathcal{O}_{0},\mathcal{W},\mathcal{F}),\\
|\int_{\mathcal{D}_{u,\ubar}}\langle(\eta+\etabar)\nabla\bar{\beta},\nabla{\bar{\alpha}}\rangle|\leq \epsilon^{\frac{1}{2}}C(\mathcal{O}_{0},\mathcal{W},\mathcal{F}),
\end{eqnarray}
where $\Psi\in(\beta,\bar{\beta},\rho,\sigma,\alpha)$ and $\varphi$ is any connection coefficients belonging to the set $(\widehat{\chi},\chibarhat,  \tr\chi,  \tr\chibar,\etabar,\omega,\eta,\omegabar)$. Collecting all the terms, we may now estimate $ER14$ as follows 
\begin{eqnarray}
|ER14|\leq \epsilon^{\frac{1}{2}} C(\mathcal{O}_{0},\mathcal{W},\mathcal{F})+\epsilon C(\mathcal{O}_{0},\mathcal{W},\mathcal{F})+C(\mathcal{O}_{0})\int_{0}^{\ubar}||\nabla\bar{\alpha}||^{2}_{L^{2}(  \Hbar)}d\ubar^{'}.
\end{eqnarray}
Therefore collecting all the terms we obtain 
\begin{eqnarray}
2\int_{H_{u}}|\nabla\bar{\beta}|^{2}+\int_{  \Hbar_{\ubar}}|\nabla\bar{\alpha}|^{2}\leq 2\int_{H_{0}}|\nabla\bar{\beta}|^{2}\nonumber+\int_{  \Hbar_{0}}|\nabla\bar{\alpha}|^{2}\\\nonumber+\epsilon^{\frac{1}{2}} C(\mathcal{O}_{0},\mathcal{W},\mathcal{F})+\epsilon C(\mathcal{O}_{0},\mathcal{W},\mathcal{F})+C(\mathcal{O}_{0})\int_{0}^{\ubar}||\nabla\bar{\alpha}||^{2}_{L^{2}(  \Hbar)}d\ubar^{'}\\\nonumber 
\leq C(\mathcal{F}_{0})+\epsilon^{\frac{1}{2}} C(\mathcal{O}_{0},\mathcal{W},\mathcal{F})+\epsilon C(\mathcal{O}_{0},\mathcal{W},\mathcal{F})+C(\mathcal{O}_{0})\int_{0}^{\ubar}||\nabla\bar{\alpha}||^{2}_{L^{2}(  \Hbar)}d\ubar^{'}
\end{eqnarray}
This completes the proof of the lemma.
\end{proof}
\subsection{Energy estimates for the Yang-Mills curvature}
\noindent Now that we have completed the energy estimates for the Weyl curvature, we move on to estimating the energy associated with the Yang-Mills curvature. For this purpose, we could use the canonical stress-energy tensor associated with the Yang-Mills theory. However, we will follow the direct integration by parts procedure using the null Yang-Mills equations. In the Yang-Mills case, we will not do the energy estimates separately for each order but only for the topmost derivatives. If the estimates at the top order close, so do the lower orders since the hyperbolic structure of the null Yang-Mills equations is preserved after commuting with derivatives.\\
\begin{lemma} 
\label{73}
The horizontal derivatives of the null Yang-Mills curvature satisfy the following $L^{2}$ energy estimates 
\begin{eqnarray}
\int_{H_{u}}|\hnabla^{I}\alpha^{F}|^{2}+\int_{  \Hbar_{\ubar}}|\hnabla^{I}\rho^{F}|^{2}+\int_{  \Hbar_{\ubar}}|\hnabla^{I}\sigma^{F}|^{2}\leq C(\mathcal{F}_{0})+\epsilon C(\mathcal{O}_{0},\mathcal{W},\mathcal{F}),\\
\int_{  \Hbar_{\ubar}}|\hnabla^{I}\bar{\alpha}^{F}|^{2}+\int_{H_{u}}|\hnabla^{I}\rho^{F}|^{2}+\int_{H_{u}}|\hnabla^{I}\sigma^{F}|^{2}\\\nonumber 
\leq C(\mathcal{W}_{0})+ \epsilon^{\frac{1}{2}}C(\mathcal{O}_{0},\mathcal{W},\mathcal{F})+\epsilon C(\mathcal{O}_{0},\mathcal{W},\mathcal{F})+C(\mathcal{O}_{0})\int_{0}^{\ubar}||\hnabla\bar{\alpha}^{F}||^{2}_{L^{2}(  \Hbar)}d\ubar^{'}
\end{eqnarray}
for $0\leq I\leq 2$, $I\in \mathbb{Z}$.
\end{lemma}
\begin{proof} In order to prove this estimate, we use the integration identities (\ref{eq:IBP1}-\ref{eq:IBP2}) once again. First, we identify the pairs $(\alpha^{F},\rho^{F},\sigma^{F})$ and apply the identities (\ref{eq:IBP1}-\ref{eq:IBP2}) to obtain \begin{eqnarray}
\int_{H_{u}}|\hnabla^{I}\alpha^{F}|^{2}+\int_{  \Hbar_{\ubar}}|\hnabla^{I}\rho^{F}|^{2}+\int_{  \Hbar_{\ubar}}|\hnabla^{I}\sigma^{F}|^{2}=\int_{H_{0}}|\hnabla^{I}\alpha^{F}|^{2}+\int_{  \Hbar_{0}}|\hnabla^{I}\rho^{F}|^{2}+\int_{  \Hbar_{0}}|\hnabla^{I}\sigma^{F}|^{2}\\\nonumber 
+\underbrace{\int_{\mathcal{D}_{u,\ubar}}|\hnabla^{I}\alpha^{F}|^{2}(2\omegabar-\frac{1}{2} \tr\chibar)+\int_{\mathcal{D}_{u,\ubar}}|\hnabla^{I}\rho^{F}|^{2}(2\omega-\frac{1}{2} \tr\chi)+\int_{\mathcal{D}_{u,\ubar}}|\hnabla^{I}\sigma^{F}|^{2}(2\omega-\frac{1}{2} \tr\chi)}_{ER15}
\\\nonumber -\underbrace{2\int_{D_{u},\ubar}\left(\langle\hnabla^{I}\alpha^{F},\hnabla_{3}\hnabla^{I}\alpha^{F}\rangle+\langle\hnabla^{I}\rho^{F},\hnabla_{4}\hnabla^{I}\rho^{F}\rangle+\langle\hnabla^{I}\sigma^{F},\hnabla_{4}\hnabla^{I}\sigma^{F}\rangle\right)}_{ER16}.
\end{eqnarray}
We can directly estimate $ER15$ by using the estimates of the connection coefficients (lemma \ref{1}) and the definition of $\mathcal{F}$ since all the Yang-Mills curvature coefficients can be estimated on $H$ by
\begin{eqnarray}
|ER15|\leq \epsilon C(\mathcal{O}_{0},\mathcal{W},\mathcal{F}).
\end{eqnarray}
In order to estimate $ER16$, we use the null Yang-Mills equations 
\begin{eqnarray}
ER16=-2\int_{D_{u},\ubar}\left(\langle\hnabla^{I}\alpha^{F},\hnabla_{3}\hnabla^{I}\alpha^{F}\rangle+\langle\hnabla^{I}\rho^{F},\hnabla_{4}\hnabla^{I}\rho^{F}\rangle+\langle\hnabla^{I}\sigma^{F},\hnabla_{4}\hnabla^{I}\sigma^{F}\rangle\right)\\\nonumber 
=\int_{\mathcal{D}_{u,\ubar}}\left(\langle\nabla^{I}\alpha^{F},-\frac{1}{2}\hnabla^{I}( \tr\chibar\alpha^{F})-\hnabla\hnabla^{I}\rho^{F}+~^{*}\hnabla\hnabla^{I}\sigma^{F}-2\hnabla^{I}(~^{*}\etabar\sigma^{F})+2\hnabla^{I}(\etabar\rho^{F})\right.\\\nonumber\left.+2\hnabla^{I}(\omegabar\alpha^{F})-\hnabla^{I}(\widehat{\chi}\cdot\bar{\alpha}^{F})+[\hnabla_{3},\hnabla^{I}]\alpha^{F}+[\hnabla,\hnabla^{I}]\rho^{F}-[^{*}\hnabla,\hnabla^{I}]\sigma^{F}\rangle\right.\\\nonumber 
\left.+\langle\hnabla^{I}\rho^{F},-\hat{div}\hnabla^{I}\alpha^{F}-\hnabla^{I}( \tr\chi\rho^{F})-\hnabla^{I}((\eta-\etabar)\cdot\alpha^{F})+[\hnabla_{4},\hnabla^{I}]\rho^{F}+[\hat{div},\hnabla^{I}]\alpha^{F}\rangle\right.\\\nonumber 
\left.+\langle\hnabla^{I}\sigma^{F},-\hat{curl}\hnabla^{I}\alpha^{F}-\hnabla^{I}( \tr\chi\sigma^{F})+\hnabla^{I}((\eta-\etabar)\cdot~^{*}\alpha^{F})+[\hnabla_{4},\hnabla^{I}]\sigma^{F}-[\hat{curl},\hnabla^{I}]\alpha^{F}\rangle\right)\\\nonumber 
\sim\int_{\mathcal{D}_{u,\ubar}}\left(\langle\nabla^{I}\alpha^{F},-\frac{1}{2}\hnabla^{I}( \tr\chibar\alpha^{F})-2\hnabla^{I}(~^{*}\etabar\sigma^{F})+2\hnabla^{I}(\etabar\rho^{F})\right.\\\nonumber\left.+2\hnabla^{I}(\omegabar\alpha^{F})-\hnabla^{I}(\widehat{\chi}\cdot\bar{\alpha}^{F})+[\hnabla_{3},\hnabla^{I}]\alpha^{F}+[\hnabla,\hnabla^{I}]\rho^{F}-[^{*}\hnabla,\hnabla^{I}]\sigma^{F}\rangle\right.\\\nonumber 
\left.+\langle\hnabla^{I}\rho^{F},-\hnabla^{I}( \tr\chi\rho^{F})-\hnabla^{I}((\eta-\etabar)\cdot\alpha^{F})+[\hnabla_{4},\hnabla^{I}]\rho^{F}+[\hat{div},\hnabla^{I}]\alpha^{F}\rangle\right.\\\nonumber 
\left.+\langle\hnabla^{I}\sigma^{F},-\hnabla^{I}( \tr\chi\sigma^{F})+\hnabla^{I}((\eta-\etabar)\cdot~^{*}\alpha^{F})+[\hnabla_{4},\hnabla^{I}]\sigma^{F}-[\hat{curl},\hnabla^{I}]\alpha^{F}\rangle\right)\\\nonumber
+\int_{\mathcal{D}_{u,\ubar}}\langle\hnabla^{I}\alpha^{F},((\eta+\etabar)(\hnabla^{I}\rho^{F}+\hnabla^{I}\sigma^{F})\rangle.
\end{eqnarray}
We will sketch how we handle each term. There will be terms of the following type that are estimated easily
\begin{eqnarray}
\int_{\mathcal{D}_{u,\ubar}}\langle\hnabla^{I}\Phi^{F},\varphi\hnabla^{I}\Phi^{F}\rangle\leq \int_{\mathcal{D}_{u,\ubar}}||\hnabla^{I}\Phi^{F}||_{L^{2}(S)}||\varphi||_{L^{\infty}(S)}||\hnabla^{I}\Phi^{F}||_{L^{2}(S)}
\leq \epsilon C(\mathcal{O}_{0},\mathcal{W},\mathcal{F}),\\
\int_{\mathcal{D}_{u,\ubar}}\langle\hnabla^{I}\Phi^{F},\hnabla^{I}\varphi \Phi^{F}\rangle\leq \int_{0}^{u}||\hnabla^{I}\Phi^{F}||_{L^{2}(H)}||\hnabla^{I}\varphi||_{L^{2}(H)}||\Phi^{F}||_{L^{\infty}(S)}
\leq \epsilon C(\mathcal{O}_{0},\mathcal{W},\mathcal{F}),\\\nonumber
\int_{\mathcal{D}_{u,\ubar}}\langle\hnabla^{I}\Phi^{F},\hnabla\varphi \hnabla\Psi\rangle\leq \int_{\mathcal{D}_{u,\ubar}}||\hnabla^{I}\Phi^{F}||_{L^{2}(S)}||\hnabla\varphi||_{L^{4}(S)}||\hnabla\Phi^{F}||_{L^{4}(S)}\leq 
\epsilon C(\mathcal{O}_{0},\mathcal{W},\mathcal{F}),
\end{eqnarray}
where $\Phi^{F}\in (\alpha^{F},\rho^{F},\sigma^{F})$ and therefore are controllable on $H$. For the connection coefficients, we use the lemma (\ref{1}-\ref{7}). Now we need to control the commutator terms 
\begin{eqnarray}
[\hnabla_{3},\hnabla^{2}]\alpha^{F}=[\hnabla_{3},\hnabla]\hnabla\alpha^{F}+\hnabla[\hnabla_{3},\hnabla]\alpha^{F}
\end{eqnarray}
Now 
\begin{eqnarray}
[\hnabla_{3},\hnabla]\Phi^{F}\sim \bar{\beta}\Phi^{F}+\bar{\alpha}^{F}\Phi^{F}+\bar{\alpha}^{F}(\rho^{F}+\sigma^{F})\Phi^{F}\nonumber+(\eta+\etabar)\nabla_{3}\Phi^{F}-\bar{\chi}\nabla\Phi^{F}+\bar{\chi}\eta\Phi^{F},
\end{eqnarray}
and therefore 
\begin{eqnarray}
[\hnabla_{3},\hnabla]\hnabla\alpha^{F}\sim \bar{\beta}\hnabla\alpha^{F}+\bar{\alpha}^{F}\hnabla\alpha^{F}+\bar{\alpha}^{F}(\rho^{F}-\sigma^{F})\hnabla\alpha^{F}+(\eta+\etabar)\hnabla_{3}\hnabla\alpha^{F}\\\nonumber-\bar{\chi}\nabla\hnabla\alpha^{F}+\bar{\chi}\eta\hnabla\alpha^{F}\\\nonumber 
\sim \bar{\beta}\hnabla\alpha^{F}+\bar{\alpha}^{F}\hnabla\alpha^{F}+\bar{\alpha}^{F}(\rho^{F}-\sigma^{F})\hnabla\alpha^{F}+(\eta+\etabar)\hnabla\hnabla_{3}\alpha^{F}\\\nonumber-\bar{\chi}\nabla\hnabla\alpha^{F}+\bar{\chi}\eta\hnabla\alpha^{F}+(\eta+\etabar)[\hnabla_{3},\hnabla]\alpha^{F},
\end{eqnarray}
where $\hnabla\hnabla_{3}\alpha^{F}$ is evaluated by means of the null Yang-Mills equations. Now 
\begin{eqnarray}
\int_{\mathcal{D}_{u,\ubar}}\langle\hnabla^{I}\Phi^{F},[\hnabla_{3},\hnabla^{I}]\Phi^{F}\rangle
\end{eqnarray}
would contain the following terms 
\begin{eqnarray}
|\int_{\mathcal{D}_{u,\ubar}}\nonumber\langle\hnabla^{I}\Phi^{F},\bar{\alpha}^{F}\hnabla\Phi^{F}\rangle|\leq \int_{u,\ubar}||\hnabla^{I}\Phi^{F}||_{L^{2}(S)}||\bar{\alpha}^{F}||_{L^{4}(S)}||\hnabla\Phi^{F}||_{L^{4}(S)}
\leq \epsilon C(\mathcal{O}_{0},\mathcal{F}),\\
|\int_{\mathcal{D}_{u,\ubar}}\langle\hnabla^{I}\Phi^{F},\bar{\alpha}^{F}(\rho^{F}+\sigma^{F})\hnabla\alpha^{F}\rangle|\leq \int_{u,\ubar}||\hnabla^{I}\Phi^{F}||_{L^{2}(S)}\nonumber||\bar{\alpha}^{F}||_{L^{4}(S)}||\hnabla\Phi^{F}||_{L^{4}(S)}||\rho^{F},\sigma^{F})||_{L^{\infty}(S)}\\\nonumber 
\leq \int_{u,\ubar}||\hnabla^{I}\Phi^{F}||_{L^{2}(S)}\nonumber||\bar{\alpha}^{F}||_{L^{4}(S)}||\hnabla\Phi^{F}||_{L^{4}(S)}||\hnabla(\rho^{F},\sigma^{F})||_{L^{4}(S)}\leq \epsilon C(\mathcal{O}_{0},\mathcal{F}),\\
|\int_{\mathcal{D}_{u,\ubar}}\langle\hnabla^{I}\Phi^{F},\varphi\varphi\hnabla\bar{\alpha}^{F}\rangle|\leq \int_{u,\ubar}||\hnabla^{I}\Phi^{F}||_{L^{2}(S)}\nonumber||\varphi||_{L^{\infty}(S)}||\varphi||_{L^{4}(S)}||\hnabla\bar{\alpha}^{F}||_{L^{4}(S)}\nonumber\leq \epsilon C(\mathcal{O}_{0},\mathcal{W},\mathcal{F}),\\
|\int_{\mathcal{D}_{u,\ubar}}\langle\hnabla^{I}\Phi^{F},\varphi\hnabla^{I}\Phi^{F}\rangle|\leq  \int_{u,\ubar}||\hnabla^{I}\Phi^{F}||^{2}_{L^{2}(S)}||\varphi||_{L^{\infty}(S)}\leq \epsilon C(\mathcal{O}_{0},\mathcal{W},\mathcal{F})
,\\
|\int_{\mathcal{D}_{u,\ubar}}\hnabla^{I}\Phi^{F}\bar{\beta}\hnabla\Phi^{F}|\nonumber\leq \int_{u,\ubar}||\hnabla^{I}\Phi^{F}||_{L^{2}(S)}||\bar{\beta}||_{L^{4}(S)}||\hnabla\Phi^{F}||_{L^{4}(S)}\leq \epsilon C(\mathcal{O}_{0},\mathcal{W},\mathcal{F})
\end{eqnarray}
Now note the following commutation relation 
\begin{eqnarray}
[\hnabla_{4},\hnabla^{2}](\rho^{F},\sigma^{F})=[\hnabla_{4},\hnabla]\hnabla(\rho^{F},\sigma^{F})+\hnabla[\hnabla_{4},\hnabla](\rho^{F},\sigma^{F})
\end{eqnarray}
Now 
\begin{eqnarray}
[\hnabla_{4},\hnabla](\rho^{F},\sigma^{F})\sim \beta(\rho^{F},\sigma^{F})+\alpha^{F}(\rho^{F},\sigma^{F})+\alpha^{F}(\rho^{F}+\sigma^{F})(\rho^{F},\sigma^{F})\nonumber+(\eta+\etabar)\nabla_{4}(\rho^{F},\sigma^{F})\\\nonumber-\chi\nabla(\rho^{F},\sigma^{F})+\chi\etabar(\rho^{F},\sigma^{F}),
\end{eqnarray}
and therefore 
\begin{eqnarray}
[\hnabla_{4},\hnabla]\hnabla(\rho^{F},\sigma^{F})\sim \beta\hnabla(\rho^{F},\sigma^{F})+\alpha^{F}\hnabla(\rho^{F},\sigma^{F})\nonumber+\alpha^{F}(\rho^{F}+\sigma^{F})\hnabla(\rho^{F},\sigma^{F})\\\nonumber+(\eta+\etabar)\hnabla_{4}\hnabla(\rho^{F},\sigma^{F}) -\chi\nabla\hnabla(\rho^{F},\sigma^{F})+\bar{\chi}\eta\hnabla(\rho^{F},\sigma^{F})\\\nonumber 
\sim \beta\hnabla(\rho^{F},\sigma^{F})+\alpha^{F}\hnabla(\rho^{F},\sigma^{F})+{\alpha}^{F}(\rho^{F}-\sigma^{F})\hnabla(\rho^{F},\sigma^{F})+(\eta+\etabar)\hnabla\hnabla_{4}(\rho^{F},\sigma^{F})-\bar{\chi}\nabla\hnabla(\rho^{F},\sigma^{F})\\\nonumber +\bar{\chi}\eta\hnabla(\rho^{F},\sigma^{F})+(\eta+\etabar)\left(\beta+\alpha^{F}(\rho^{F}+\sigma^{F}))(\rho^{F},\sigma^{F})+(\eta+\etabar)\hnabla_{4}(\rho^{F},\sigma^{F})\right.\\\nonumber\left.-\chi\hnabla(\rho^{F},\sigma^{F})+\chi\etabar(\rho^{F},\sigma^{F})\right),
\end{eqnarray}
where $\hnabla\hnabla_{4}(\rho^{F},\sigma^{F})$ can be further simplified by means of the null Yang-Mills equations. The resulting terms are easy to estimate 
\begin{eqnarray}
\int_{\mathcal{D}_{u,\ubar}}\langle\hnabla^{I}\Phi^{F},\Phi^{F}\hnabla\Phi^{F}\rangle\leq\nonumber \int_{u,\ubar}||\hnabla^{I}\Phi^{F}||_{L^{2}(S)}||\Phi^{F}||_{L^{4}(S)}||\hnabla\Phi^{F}||_{L^{4}(S)}
\leq \epsilon C(\mathcal{O}_{0},\mathcal{F}),\\
\int_{\mathcal{D}_{u,\ubar}}\langle\hnabla^{I}\Phi^{F},\Phi^{F}(\rho^{F}+\sigma^{F})\hnabla\Phi^{F}\rangle\leq \int_{u,\ubar}||\hnabla^{I}\Phi^{F}||_{L^{2}(S)}\nonumber||\Phi^{F}||_{L^{4}(S)}||\hnabla\Phi^{F}||_{L^{4}(S)}||(\rho^{F},\sigma^{F})||_{L^{\infty}(S)}\\\nonumber 
\leq \int_{u,\ubar}||\hnabla^{I}\Phi^{F}||_{L^{2}(S)}\nonumber||\Phi^{F}||_{L^{4}(S)}||\hnabla\Phi^{F}||_{L^{4}(S)}||\hnabla(\rho^{F},\sigma^{F})||_{L^{4}(S)}\leq \epsilon C(\mathcal{O}_{0},\mathcal{F}),\\
\int_{\mathcal{D}_{u,\ubar}}\langle\hnabla^{I}\Phi^{F},\varphi\varphi\hnabla\Phi^{F}\rangle\leq \int_{u,\ubar}||\hnabla^{I}\Phi^{F}||_{L^{2}(S)}\nonumber||\varphi||_{L^{\infty}(S)}||\varphi||_{L^{4}(S)}||\hnabla\Phi^{F}||_{L^{4}(S)}\leq \epsilon C(\mathcal{O}_{0},\mathcal{W},\mathcal{F}),\\
\int_{\mathcal{D}_{u,\ubar}}\langle\hnabla^{I}\Phi^{F},\varphi\hnabla^{I}\Phi^{F}\rangle\leq\nonumber \epsilon C(\mathcal{O}_{0},\mathcal{W},\mathcal{F}),
\end{eqnarray}
where once again $\Phi^{F}\in (\alpha^{F},\rho^{F},\sigma^{F})$ and $\varphi$ is any connection coefficients belonging to the set $(\widehat{\chi},\chibarhat,  \tr\chi,  \tr\chibar,\etabar,\omega,\eta,\omegabar)$.
Now note the following 
\begin{eqnarray}
[\hnabla,\hnabla^{2}]\Phi^{F}=\hnabla[\hnabla,\hnabla]\Phi^{F}+[\hnabla,\hnabla]\hnabla\Phi^{F}\sim\hnabla((K+\sigma^{F})\Phi^{F})+(K+\sigma^{F})\hnabla\Psi
\end{eqnarray}
and therefore 
\begin{eqnarray}
|\int_{\mathcal{D}_{u,\ubar}}\langle\hnabla^{I}\Phi^{F},[\hnabla,\hnabla^{2}]\Phi^{F}\rangle|\sim|\int_{\mathcal{D}_{u,\ubar}}\langle\hnabla^{I}\Phi^{F},\hnabla((K+\sigma^{F})\Phi^{F})+(K+\sigma^{F})\hnabla\Phi^{F}\rangle|\\\nonumber
\leq \int_{u,\ubar}||\hnabla^{I}\Phi^{F}||_{L^{2}(S)}\left(||\nabla K||_{L^{2}(S)}||\Phi^{F}||_{L^{\infty}(S)}+||K||_{L^{4}(S)}||\hnabla\Phi^{F}||_{L^{4}(S)}+||\hnabla\sigma^{F}||_{L^{4}(S)}||\Phi^{F}||_{L^{4}(S)}\right.\\\nonumber 
\left.+||\sigma^{F}||_{L^{4}(S)}||\hnabla\Phi^{F}||_{L^{4}(S)}\right)
\leq\epsilon C(\mathcal{O}_{0},\mathcal{W},\mathcal{F})
\end{eqnarray}
where $||\nabla K||_{L^{2}(S)},~||K||_{L^{4}(S)}$ are controlled by $\mathcal{W}$ and $\mathcal{F}$ by the virtue of the null Hamiltonian constraint (\ref{eq:4}, corollary 1). Collecting all the terms together, we obtain 
\begin{eqnarray}
|ER16|\leq \epsilon C(\mathcal{O}_{0},\mathcal{W},\mathcal{F})
\end{eqnarray}
and therefore 
\begin{eqnarray}
\int_{H_{u}}|\hnabla^{I}\alpha^{F}|^{2}+\int_{  \Hbar_{\ubar}}|\hnabla^{I}\rho^{F}|^{2}+\int_{  \Hbar_{\ubar}}|\hnabla^{I}\sigma^{F}|^{2}=\int_{H_{0}}|\hnabla^{I}\alpha^{F}|^{2}+\int_{  \Hbar_{0}}|\hnabla^{I}\rho^{F}|^{2}+\int_{  \Hbar_{0}}|\hnabla^{I}\sigma^{F}|^{2}\\\nonumber 
+\epsilon C(\mathcal{O}_{0},\mathcal{W},\mathcal{F})\leq C(\mathcal{W}_{0})+\epsilon C(\mathcal{O}_{0},\mathcal{W},\mathcal{F}).
\end{eqnarray}
Now in order to prove the second estimate of the lemma, we collect the triple $(\bar{\alpha},\rho^{F},\sigma^{F})$ and apply the integration identities (\ref{eq:IBP1}-\ref{eq:IBP2})
\begin{eqnarray}
\label{eq:energyfinal}
\int_{  \Hbar_{\ubar}}|\hnabla^{I}\bar{\alpha}^{F}|^{2}+\int_{H_{u}}|\hnabla^{I}\rho^{F}|^{2}+\int_{H_{u}}|\hnabla^{I}\sigma^{F}|^{2}=\int_{  \Hbar_{0}}|\hnabla^{I}\bar{\alpha}^{F}|^{2}+\int_{H_{0}}|\hnabla^{I}\rho^{F}|^{2}+\int_{H_{0}}|\hnabla^{I}\sigma^{F}|^{2}\\\nonumber
+\underbrace{\int_{\mathcal{D}_{u,\ubar}}|\hnabla^{I}\bar{\alpha}^{F}|^{2}(2\omega-\frac{1}{2} \tr\chi)+\int_{\mathcal{D}_{u,\ubar}}|\hnabla^{I}\rho^{F}|^{2}(2\omegabar-\frac{1}{2} \tr\chibar)+\int_{\mathcal{D}_{u,\ubar}}|\hnabla^{I}\sigma^{F}|^{2}(2\omegabar-\frac{1}{2} \tr\chibar)}_{ER17}\\\nonumber-2\underbrace{\int_{D_{u},\ubar}\left(\langle\hnabla^{I}\bar{\alpha}^{F},\hnabla_{4}\hnabla^{I}\bar{\alpha}^{F}\rangle+\langle\hnabla^{I}\rho^{F},\hnabla_{3}\hnabla^{I}\rho^{F}\rangle+\langle\hnabla^{I}\sigma^{F},\hnabla_{3}\hnabla^{I}\sigma^{F}\rangle\right)}_{ER18}.
\end{eqnarray}
Once again notice that the term involving $|\hnabla^{I}\bar{\alpha}^{F}|^{2}$ contains connection coefficients $(\omega,  \tr\chi)$ which satisfy $\nabla_{3}$ evolution equations and therefore are estimated solely by means of the initial data allowing us to use Gr\"onwall's inequality. $ER17$ is estimated a follows 
\begin{eqnarray}
|ER17|\leq \epsilon C(\mathcal{O}_{0},\mathcal{W},\mathcal{F})+C(\mathcal{O}_{0})\int_{0}^{\ubar}||\hnabla^{I}\bar{\alpha}^{F}||^{2}_{L^{2}(\bar{H)}}d\ubar^{'}.
\end{eqnarray}
Once again $ER18$ is controlled by means of commuted null Yang-Mills equations 
\begin{eqnarray}
\int_{D_{u},\ubar}\left(\langle\hnabla^{I}\bar{\alpha}^{F},\hnabla_{4}\hnabla^{I}\bar{\alpha}^{F}\rangle+\langle\hnabla^{I}\rho^{F},\hnabla_{3}\hnabla^{I}\rho^{F}\rangle+\langle\hnabla^{I}\sigma^{F},\hnabla_{3}\hnabla^{I}\sigma^{F}\rangle\right)\\\nonumber 
\sim\int_{\mathcal{D}_{u,\ubar}}\left(\langle\hnabla^{I}\bar{\alpha}^{F},-\frac{1}{2}\hnabla^{I}( \tr\chi\bar{\alpha}^{F})-2\hnabla^{I}(~^{*}\etabar\cdot\sigma^{F})-2\hnabla^{I}(\etabar\cdot\rho^{F})+2\hnabla^{I}(\omega\bar{\alpha}^{F})-\hnabla^{I}(\chibarhat\cdot\alpha^{F})\rangle\right.\\\nonumber 
\left.+\langle\hnabla^{I}\rho^{F},\hnabla^{I}( \tr\chibar\rho^{F})+\hnabla^{I}((\eta-\etabar)\cdot\bar{\alpha}^{F})\rangle+\langle\hnabla^{I}\sigma^{F},-\hnabla^{I}( \tr\chibar\sigma^{F})+\hnabla^{I}((\eta-\etabar)\cdot~^{*}\bar{\alpha}^{F})\rangle\right)\\
+\int_{\mathcal{D}_{u,\ubar}}\langle\hnabla^{I}\bar{\alpha}^{F},[\hnabla_{4},\hnabla^{I}]\bar{\alpha}^{F}\rangle+\int_{\mathcal{D}_{u,\ubar}}\langle\hnabla^{I}\bar{\alpha}^{F},[\hnabla,\hnabla^{I}](\rho^{F},\sigma^{F})\rangle
+\int_{\mathcal{D}_{u,\ubar}}\langle\hnabla^{I}\rho^{F},[\hnabla_{3},\hnabla^{I}]\rho^{F}\rangle\\\nonumber+\int_{\mathcal{D}_{u,\ubar}}\langle\hnabla^{I}\rho^{F},[\hnabla,\hnabla^{I}]\bar{\alpha}^{F}\rangle
+\int_{\mathcal{D}_{u,\ubar}}\langle\hnabla^{I}\sigma^{F},[\hnabla,\hnabla^{I}]\bar{\alpha}^{F}\rangle +\int_{\mathcal{D}_{u,\ubar}}\langle\hnabla^{I}\alpha^{F},((\eta+\etabar)(\hnabla^{I}\rho^{F}+\hnabla^{I}\sigma^{F})\rangle
\end{eqnarray}
The terms that appear with the top order derivative of $\bar{\alpha}^{F}$ are $ \tr\chi,\omega,$ and $\chibarhat$ which are determined solely by $\mathcal{O}_{0}$. We use the following commutation relation 
\begin{eqnarray}
[\hnabla_{4},\hnabla]\bar{\alpha}^{F}\sim \beta\bar{\alpha}^{F}+\alpha^{F}\bar{\alpha}^{F}+\alpha^{F}(\rho^{F}-\sigma^{F})\bar{\alpha}^{F}\nonumber+(\eta+\etabar)\nabla_{4}\bar{\alpha}^{F}-\bar\chi\nabla\bar{\alpha}^{F}+\chi\etabar\bar{\alpha}^{F}
\end{eqnarray}
to evaluate $[\hnabla_{4},\hnabla^{2}]\bar{\alpha}^{F}$
\begin{eqnarray}
[\hnabla_{4},\hnabla^{2}]\bar{\alpha}^{F}=[\hnabla_{4},\hnabla]\hnabla\bar{\alpha}^{F}+\hnabla[\hnabla_{4},\hnabla]\bar{\alpha}^{F}.
\end{eqnarray}
These will contain terms of the type 
\begin{eqnarray}
\int_{\mathcal{D}_{u,\ubar}}\langle\hnabla^{I}\bar{\alpha}^{F},\hnabla^{I}\varphi \bar{\alpha}^{F}\rangle\leq ||\hnabla^{I}\bar{\alpha}^{F}||_{L^{2}(\mathcal{D}_{u,\ubar})}\sup_{u,\ubar}||\hnabla\bar{\alpha}^{F}||_{L^{4}(S)}\nonumber||\hnabla\varphi||_{L^{2}(\mathcal{D}_{u,\ubar})}\leq \epsilon^{1/2}C(\mathcal{O}_{0},\mathcal{F}),\\
\int_{\mathcal{D}_{u,\ubar}}\langle\hnabla^{I}\bar{\alpha}^{F},\hnabla\varphi \hnabla\bar{\alpha}^{F}\rangle\leq \int_{\mathcal{D}_{u,\ubar}}||\hnabla^{I}\bar{\alpha}^{F}||_{L^{2}(S)}\nonumber||\hnabla\varphi||_{L^{4}(S)}||\hnabla\bar{\alpha}^{F}||_{L^{4}(S)}\leq \epsilon^{1/2}C(\mathcal{O}_{0},\mathcal{F})\\\nonumber 
\int_{\mathcal{D}_{u,\ubar}}\langle\hnabla^{I}\alpha^{F},\nabla K \bar{\alpha}^{F}\rangle\leq \epsilon C(\mathcal{O}_{0},\mathcal{F}),~
\int_{\mathcal{D}_{u,\ubar}}\langle\hnabla^{I}\alpha^{F},K \hnabla\bar{\alpha}^{F}\rangle\leq \epsilon C(\mathcal{O}_{0},\mathcal{F})\\
\int_{\mathcal{D}_{u,\ubar}}\langle\hnabla^{I}\bar{\alpha}^{F},K \hnabla\bar{\alpha}^{F}\rangle\leq \epsilon C(\mathcal{O}_{0},\mathcal{F}),~\nonumber
\int_{\mathcal{D}_{u,\ubar}}\langle\hnabla^{I}\alpha^{F},\nabla K \bar{\alpha}^{F}\rangle\leq \epsilon^{\frac{1}{2}} C(\mathcal{O}_{0},\mathcal{F}),\\\nonumber
\int_{\mathcal{D}_{u,\ubar}}\langle\hnabla^{I}\alpha^{F},\varphi \hnabla\bar{\alpha}^{F}\rangle\leq \epsilon C(\mathcal{O}_{0},\mathcal{F}),~\nonumber
\int_{\mathcal{D}_{u,\ubar}}\langle\hnabla^{I}\bar{\alpha}^{F},\nabla\varphi \bar{\alpha}^{F}\rangle\leq \epsilon^{1/2} C(\mathcal{O}_{0},\mathcal{F})\\\nonumber
\int_{\mathcal{D}_{u,\ubar}}\langle\hnabla^{I}\bar{\alpha}^{F},(\omega, \tr\chi) \hnabla^{I}\bar{\alpha}^{F}\rangle\leq C(\mathcal{O}_{0})\int_{0}^{\ubar}||\hnabla^{I}\bar{\alpha}^{F}||_{L^{2}(  \Hbar)},
\end{eqnarray}
Notice that the last term is controllable by Gr\"onwall's inequality and connection coefficients in the last integral do not contain $\eta$ and $\omegabar$. Therefore, collecting all the terms, we obtain 
\begin{eqnarray}
|ER18|\leq \epsilon^{\frac{1}{2}}C(\mathcal{O}_{0},\mathcal{W},\mathcal{F})+\epsilon C(\mathcal{O}_{0},\mathcal{W},\mathcal{F})+C(\mathcal{O}_{0})\int_{0}^{\ubar}||\hnabla\bar{\alpha}^{F}||^{2}_{L^{2}(  \Hbar)}d\ubar^{'}.
\end{eqnarray}
Substituting this estimate in the main energy identity (\ref{eq:energyfinal}) we obtain 
\begin{eqnarray}
\int_{  \Hbar_{\ubar}}|\hnabla^{I}\bar{\alpha}^{F}|^{2}+\int_{H_{u}}|\hnabla^{I}\rho^{F}|^{2}\nonumber+\int_{H_{u}}|\hnabla^{I}\sigma^{F}|^{2}=\int_{  \Hbar_{0}}|\hnabla^{I}\bar{\alpha}^{F}|^{2}+\int_{H_{0}}|\hnabla^{I}\rho^{F}|^{2}+\int_{H_{0}}|\hnabla^{I}\sigma^{F}|^{2}\\\nonumber 
+\epsilon^{\frac{1}{2}}C(\mathcal{O}_{0},\mathcal{W},\mathcal{F})+\epsilon C(\mathcal{O}_{0},\mathcal{W},\mathcal{F})+C(\mathcal{O}_{0})\int_{0}^{\ubar}||\hnabla\bar{\alpha}^{F}||^{2}_{L^{2}(  \Hbar)}d\ubar^{'}\\\nonumber 
\leq C(\mathcal{F}_{0})+\epsilon^{\frac{1}{2}}C(\mathcal{O}_{0},\mathcal{W},\mathcal{F})+\epsilon C(\mathcal{O}_{0},\mathcal{W},\mathcal{F})+C(\mathcal{O}_{0})\int_{0}^{\ubar}||\hnabla\bar{\alpha}^{F}||^{2}_{L^{2}(  \Hbar)}d\ubar^{'}
\end{eqnarray}
This completes the proof of the lemma.
\end{proof}
\subsection{Energy estimates for $\nabla_{4}\alpha,\nabla_{3}\alphabar$}
\noindent We need to estimate the remaining terms of the curvature energy. \\
\begin{lemma}
\label{74}
\textit{The null derivatives of the null components of the Weyl curvature satisfy the following $L^{2}$ energy estimates
\begin{eqnarray}
\int_{H_{u}}|\nabla_{4}\alpha|^{2}\leq C(\mathcal{O}_{0},\mathcal{W}_{0},\mathcal{F}_{0})+\epsilon C(\mathcal{O}_{0}\mathcal{W},\mathcal{F}),\\
\int_{  \Hbar_{\ubar}}|\nabla_{3} \bar{\alpha}|^{2}\leq C(\mathcal{O}_{0},\mathcal{W}_{0},\mathcal{F}_{0})+\epsilon C(\mathcal{O}_{0}\mathcal{W},\mathcal{F})+
\epsilon^{\frac{1}{2}} C(\mathcal{O}_{0},\mathcal{W},\mathcal{F})\\\nonumber +C(\mathcal{O}_{0})\int_{0}^{\ubar}||\nabla_{3}\bar{\alpha}||^{2}_{L^{2}(  \Hbar)}d\ubar^{'}
\end{eqnarray}}
\end{lemma}
\begin{proof} We proceed exactly the same way as the previous case only we commute the null derivatives with the evolution equations of the Weyl curvature components. First we identify the set of pairs $(\alpha,\beta)$  and $(\bar{\alpha},\bar{\beta})$ and apply the integration identities (\ref{eq:IBP1}-\ref{eq:IBP2})
\begin{eqnarray}
\int_{H_{u}}|\nabla_{4}\alpha|^{2}+2\int_{  \Hbar_{\ubar}}|\nabla_{4}\beta|^{2}=\int_{H_{0}}|\nabla_{4}\alpha|^{2}+2\int_{  \Hbar_{0}}|\nabla_{4}\beta|^{2}\\\nonumber+\underbrace{\int_{\mathcal{D}_{u,\ubar}}|\nabla_{4}\alpha|^{2}(2\omegabar-\frac{1}{2} \tr\chibar)+\int_{\mathcal{D}_{u,\ubar}}2|\nabla_{4}\beta|^{2}(2\omega-\frac{1}{2} \tr\chi)}_{ER19}\\\nonumber -2\underbrace{\int_{\mathcal{D}_{u,\ubar}}\left(\langle\nabla_{4}\alpha,\nabla_{3}\nabla_{4}\alpha\rangle+2\langle\nabla_{4}\beta,\nabla_{4}\nabla_{4}\beta\rangle\right)}_{ER20}
\end{eqnarray}
Now notice that $\nabla_{4}\beta$ contains the term $\nabla\alpha$ which can not be controlled on $  \Hbar$ and therefore we needed to include $||\nabla_{4}\beta||_{L^{2}(  \Hbar_{0}})$ term in the initial data for the Weyl curvature. Fortunately, $\nabla_{4}\beta$ appears in the error term $ER19$ within the bulk integral multiplied by $\omega$ and $ \tr\chi$ and therefore harmless. We first estimate the error terms 
\begin{eqnarray}
|ER19|=|\int_{\mathcal{D}_{u,\ubar}}|\nabla_{4}\alpha|^{2}(2\omegabar-\frac{1}{2} \tr\chibar)+\int_{\mathcal{D}_{u,\ubar}}2|\nabla_{4}\beta|^{2}(2\omega-\frac{1}{2} \tr\chi)|\\\nonumber \leq \epsilon C(\mathcal{O}_{0},\mathcal{W},\mathcal{F})+C(\mathcal{O}_{0})\int_{0}^{\ubar}||\nabla_{4}\beta||^{2}_{L^{2}(  \Hbar)}d\ubar^{'}
\end{eqnarray}
where notice the connection coefficients multiplying $|\nabla_{4}\beta|^{2}$ in the bulk integral is solely determined by the initial data $\mathcal{O}_{0}$. This allows us to utilize Gr\"onwall's estimate. Now for the term $ER20$, we need to use the commuted evolution equations 
\begin{eqnarray}
ER20=-2\int_{\mathcal{D}_{u,\ubar}}\left(\langle\nabla_{4}\alpha,\nabla_{3}\nabla_{4}\alpha\rangle+2\langle\nabla_{4}\beta,\nabla_{4}\nabla_{4}\beta\rangle\right)\\\nonumber 
=-2\int_{\mathcal{D}_{u,\ubar}}\left(\langle\nabla_{4}\alpha,\nabla\hat{\otimes} \nabla_{4}\beta-\frac{1}{2}\nabla_{4}( \tr\chibar\alpha)+4\nabla_{4}(\omegabar\alpha)-3\nabla_{4}(\widehat{\chi}\rho+~^{*}\chi\sigma)+\nabla_{4}((\zeta+4\eta)\hat{\otimes}\beta)\right.\\\nonumber\left. +\frac{1}{2}\nabla_{4}(D_{3}R_{44}-D_{4}R_{33})\gamma+[\nabla_{3},\nabla_{4}]\alpha+[\nabla_{4},\nabla]\hat{\otimes}\beta\rangle+2\langle\nabla_{4}\beta,div(\nabla_{4}\alpha)-2\nabla_{4}( \tr\chi\beta)-2\nabla_{4}(\omega\beta)\right.\\\nonumber 
\left.+\nabla_{4}(\eta\alpha)-\frac{1}{2}\nabla_{4}(D_{b}R_{44}-D_{4}R_{4b})+[\nabla_{4},div]\alpha\rangle\right)\\\nonumber 
\sim\int_{\mathcal{D}_{u,\ubar}}\left(\langle\nabla_{4}\alpha,-\frac{1}{2}\nabla_{4}( \tr\chibar\alpha)+4\nabla_{4}(\omegabar\alpha)-3\nabla_{4}(\widehat{\chi}\rho+~^{*}\chi\sigma)+\nabla_{4}((\eta+\etabar)\hat{\otimes}\beta)\right.\\\nonumber\left. +\frac{1}{2}\nabla_{4}(D_{3}R_{44}-D_{4}R_{33})\gamma+[\nabla_{3},\nabla_{4}]\alpha+[\nabla_{4},\nabla]\hat{\otimes}\beta\rangle+\langle\nabla_{4}\beta,-2\nabla_{4}( \tr\chi\beta)-2\nabla_{4}(\omega\beta)\right.\\\nonumber 
\left.+\nabla_{4}(\eta\alpha)-\frac{1}{2}\nabla_{4}(D_{b}R_{44}-D_{4}R_{4b})+[\nabla_{4},div]\alpha\rangle\right)+\int_{\mathcal{D}_{u,\ubar}}\langle(\eta+\etabar)\nabla_{4}\beta,\nabla_{4}\alpha\rangle.
\end{eqnarray}
Now write down the commutators explicitly 
\begin{eqnarray}
[\nabla_{3},\nabla_{4}]\alpha\sim \sigma \alpha+(\rho^{F}\rho^{F}+\sigma^{F}\sigma^{F})\alpha+\omega\nabla_{3}\alpha+\omegabar\nabla_{4}\alpha+(\eta-\etabar)\nabla\alpha,\\\nonumber 
[\nabla_{4},\nabla]\beta\sim (\beta+\alpha^{F}(\rho^{F}+\sigma^{F}))\beta+(\eta+\etabar)\nabla_{4}\beta-\chi\nabla\beta+\chi\etabar\beta,\\\nonumber 
[\nabla_{4},\nabla]\alpha\sim (\beta+\alpha^{F}(\rho^{F}+\sigma^{F}))\alpha+(\eta+\etabar)\nabla_{4}\alpha-\chi\nabla\beta+\chi\etabar\alpha
\end{eqnarray}
and use the following null evolution equation for $\beta$
\begin{eqnarray}
\nabla_{4}\beta_{a}+2 \tr\chi\beta_{a}=(div\alpha)_{a}-2\omega\beta_{a}+(\eta\cdot\alpha)_{a}-\frac{1}{2}(\mathcal{D}_{a}R_{44}-\mathcal{D}_{4}R_{4a}),
\end{eqnarray}
where the Yang-Mills source term may be evaluated as follows 
\begin{eqnarray}
D_{b}R_{44}-D_{4}R_{4b}\sim \langle\alpha^{F},\hnabla_{b}\alpha^{F}\rangle-\chi_{bc}\mathfrak{T}_{c4}+\eta_{b}\mathfrak{T}_{44}-2\omega \mathfrak{T}_{4b}-(\hnabla_{4}\alpha^{F}_{b}\cdot(\rho^{F}\nonumber+\sigma^{F})\\\nonumber -\alpha^{F}_{b}\cdot(\hnabla_{4}\rho^{F}+\hnabla_{4}\sigma^{F})).
\end{eqnarray}
In addition, we will use the null evolution equations for the connection coefficients whenever they are available. In the error term $ER20$, we note that $\nabla_{4}\etabar$ appears. We can use the estimate from lemma \ref{7} to control this term (this was the whole point of proving lemma \ref{lemma5}-\ref{7}).
Note $\langle\nabla_{4}\alpha,\nabla_{4}(D_{3}R_{44}-D_{4}R_{33})\gamma\rangle=0$ (since $\nabla_{4}$ commutes with $\gamma$). Now we may estimate each term separately 
\begin{eqnarray}
|\int_{\mathcal{D}_{u,\ubar}}\langle\nabla_{4}\alpha,\nabla_{4}\varphi\alpha\rangle| \leq \epsilon \sup_{u}||\nabla_{4}\alpha||_{L^{2}(H)}\sup_{u,\ubar}||\nabla_{4}\varphi||_{L^{4}(S)}||\alpha||_{L^{4}(S)}\leq \epsilon C(\mathcal{O}_{0},\mathcal{W},\mathcal{F}),\\\nonumber
|\int_{\mathcal{D}_{u,\ubar}}\langle\nabla_{4}\alpha,\varphi\nabla_{4}\alpha\rangle|\leq \epsilon \sup_{u}||\nabla_{4}\alpha||^{2}_{L^{2}(H)}\sup_{u,\ubar}||\varphi||_{L^{\infty}(S)}\leq \epsilon C(\mathcal{O},\mathcal{W},\mathcal{F}),\\\nonumber
|\int_{\mathcal{D}_{u,\ubar}}\langle\nabla_{4}\alpha\varphi \nabla_{4}\Psi\rangle|\leq \epsilon \sup_{u}||\nabla_{4}\alpha||_{L^{2}(H)}\sup_{u,\ubar}||\varphi||_{L^{\infty}(S)}||\nabla_{4}\Psi||_{L^{2}(S)}\leq \epsilon C(\mathcal{O}_{0},\mathcal{W},\mathcal{F}),\\\nonumber
|\int_{\mathcal{D}_{u,\ubar}}\langle\nabla_{4}\alpha \nabla_{4}\varphi\Psi\rangle|\leq \epsilon \sup_{u}||\nabla_{4}\alpha||_{L^{2}(H)}\sup_{u,\ubar}||\nabla_{4}\varphi||_{L^{4}(S)}||\Psi||_{L^{4}(S)}\leq \epsilon C(\mathcal{O}_{0},\mathcal{W},\mathcal{F}),\\\nonumber
|\int_{\mathcal{D}_{u,\ubar}}\langle\nabla_{4}\alpha\varphi \nabla\Psi\rangle|\leq \epsilon \sup_{u}||\nabla_{4}\alpha||_{L^{2}(H)}\sup_{u,\ubar}||\varphi||_{L^{\infty}(S)}||\nabla\Psi||_{L^{2}(S)}\leq \epsilon C(\mathcal{O}_{0},\mathcal{W},\mathcal{F}),\\\nonumber
|\int_{\mathcal{D}_{u,\ubar}}\langle\nabla\alpha,\varphi \Phi^{F}\hnabla_{4}\alpha^{F}\rangle|\leq \epsilon \sup_{u}||\nabla\alpha||_{L^{2}(H)}\sup_{u}||\hnabla_{4}\alpha^{F}||_{L^{2}(H)}\sup_{u,\ubar}||\varphi||_{L^{\infty}(S)}||\hnabla\Phi^{F}||_{L^{4}(S)}\\\nonumber \leq \epsilon C(\mathcal{O}_{0},\mathcal{W},\mathcal{F}),\\\nonumber
|\int_{\mathcal{D}_{u,\ubar}}\langle\nabla\alpha,\varphi \alpha^{F}\hnabla_{4}\Phi^{F}\rangle|\leq \epsilon \sup_{u}||\nabla\alpha||_{L^{2}(H)}\sup_{u,\ubar}||\varphi||_{L^{\infty(S)}}||\nabla\alpha^{F}||_{L^{4}(S)}||\hnabla\Phi^{F}||_{L^{2}(S)}\\\nonumber \leq \epsilon C(\mathcal{O}_{0},\mathcal{W},\mathcal{F})\\\nonumber 
|\int_{\mathcal{D}_{u,\ubar}}\langle\nabla\alpha,\varphi\nabla_{4}\alpha\rangle|\leq \epsilon \sup_{u}||\nabla\alpha||_{L^{2}(H)}||\nabla_{4}\alpha||_{L^{2}(H)}\sup_{u,\ubar}||\varphi||_{L^{\infty}(S)}\leq \epsilon C(\mathcal{O}_{0},\mathcal{W},\mathcal{F}),\\\nonumber
|\int_{\mathcal{D}_{u,\ubar}}\langle\nabla_{4}\alpha,\varphi\Phi^{F}\hnabla_{4}\alpha^{F}\rangle|\leq \epsilon \sup_{u}||\nabla_{4}\alpha||_{L^{2}(H)}\sup_{u}||\hnabla_{4}\alpha^{F}||_{L^{2}(H)}\sup_{u,\ubar}||\varphi||_{L^{\infty}(S)}||\hnabla\Phi^{F}||_{L^{4}(S)}\\\nonumber \leq \epsilon C(\mathcal{O}_{0},\mathcal{W},\mathcal{F}),\\\nonumber
|\int_{\mathcal{D}_{u,\ubar}}\langle\nabla_{4}\alpha,\varphi \alpha^{F}\hnabla_{4}\Phi^{F}\rangle|\leq \epsilon \sup_{u}||\nabla_{4}\alpha||_{L^{2}(H)}\sup_{u,\ubar}||\varphi||_{L^{\infty(S)}}||\nabla\alpha^{F}||_{L^{4}(S)}||\hnabla\Phi^{F}||_{L^{2}(S)}\\\nonumber \leq \epsilon C(\mathcal{O}_{0},\mathcal{W},\mathcal{F}),
\end{eqnarray}
where $\Psi:(\alpha,\beta,\rho,\sigma)$ and $\Phi^{F}:=(\alpha^{F},\rho^{F},\sigma^{F})$. We have also used the $||\nabla_{4}\varphi||_{L^{4}(S)}$ estimate from corollary (2). Collecting all the terms, we obtain 
\begin{eqnarray}
|ER20|\leq \epsilon C(\mathcal{O}_{0},\mathcal{W},\mathcal{F})
\end{eqnarray}
and therefore 
\begin{eqnarray}
\int_{H_{u}}|\nabla_{4}\alpha|^{2}+2\int_{  \Hbar_{\ubar}}|\nabla_{4}\beta|^{2}\leq\int_{H_{0}}|\nabla_{4}\alpha|^{2}+2\int_{  \Hbar_{0}}|\nabla_{4}\beta|^{2}\\\nonumber+\epsilon C(\mathcal{O}_{0},\mathcal{W},\mathcal{F})+C(\mathcal{O}_{0})\int_{0}^{\ubar}||\nabla_{4}\beta||^{2}_{L^{2}(  \Hbar)}d\ubar^{'}\leq C(\mathcal{O}_{0},\mathcal{W}_{0},\mathcal{F}_{0})+\epsilon C(\mathcal{O}_{0},\mathcal{W},\mathcal{F})
\end{eqnarray}
by means of Gr\"onwall's inequality. Now we need to prove the second part of the lemma. Once again we identify the pair $(\bar{\alpha},\bar{\beta})$ and apply the integration identities (\ref{eq:IBP1}-\ref{eq:IBP2}) to yield 
\begin{eqnarray}
2\int_{H_{u}}|\nabla_{3}\bar{\beta}|^{2}+\int_{  \Hbar_{\ubar}}|\nabla_{3}\bar{\alpha}|^{2}=2\int_{H_{0}}|\nabla_{3}\bar{\beta}|^{2}\nonumber+\int_{  \Hbar_{0}}|\nabla_{3}\bar{\alpha}|^{2}\\\nonumber +\underbrace{\int_{\mathcal{D}_{u,\ubar}}2|\nabla_{3}\bar{\beta}|^{2}(2\omegabar-\frac{1}{2} \tr\chibar)+\int_{\mathcal{D}_{u,\ubar}}|\nabla_{3}\bar{\alpha}|^{2}(2\omega-\frac{1}{2} \tr\chi)}_{ER21}\\\nonumber 
-2\underbrace{\int_{\mathcal{D}_{u,\ubar}}\left(2\langle\nabla_{3}\bar{\beta},\nabla_{3}\nabla_{3}\bar{\beta}\rangle+\langle\nabla_{3}\bar{\alpha},\nabla_{4}\nabla_{3}\bar{\alpha}\rangle\right)}_{ER22}.
\end{eqnarray}
We first control $ER21$. Notice that $\omegabar$ does not satisfy a $\nabla_{3}$ equation and as a result $||\omegabar||_{L^{\infty}(S)}\leq C(\mathcal{O}_{0},\mathcal{W},\mathcal{F})$. Luckily $|\nabla_{3}\bar{\beta}|^{2}$ is controlled on $H$ and therefore the first term of $ER21$ can be estimated as $\int_{0}^{u}C(\mathcal{O}_{0},\mathcal{W},\mathcal{F})||\nabla_{3}\bar{\beta}||^{2}_{L^{2}(H)}du^{'}$. For the second term, once again $|\nabla_{3}\bar{\alpha}|^{2}$ is multiplied by $\omega$ and $ \tr\chi$ which are solely determined by the norm of the initial data $\mathcal{O}_{0}$. Therefore, we have 
\begin{eqnarray}
|ER21|\leq C(\mathcal{O}_{0},\mathcal{W},\mathcal{F})\int_{0}^{u}||\nabla_{3}\bar{\beta}||^{2}_{L^{2}(H)}du^{'}+C(\mathcal{O}_{0})\int_{0}^{\ubar}||\nabla_{3}\bar{\alpha}||^{2}_{L^{2}(  \Hbar)}d\ubar^{'}\\\nonumber 
\leq \epsilon C(\mathcal{O}_{0},\mathcal{W},\mathcal{F})\sup_{u}||\nabla_{3}\bar{\beta}||^{2}_{L^{2}(H)}+C(\mathcal{O}_{0})\int_{0}^{\ubar}||\nabla_{3}\bar{\alpha}||^{2}_{L^{2}(  \Hbar)}d\ubar^{'}
\end{eqnarray}
Utilizing the null Bianchi equations, $ER22$ may be written as follows 
\begin{eqnarray}
ER22=-2\int_{\mathcal{D}_{u,\ubar}}\left(2\langle\nabla_{3}\bar{\beta},\nabla_{3}\nabla_{3}\bar{\beta}\rangle\nonumber+\langle\nabla_{3}\bar{\alpha},\nabla_{4}\nabla_{3}\bar{\alpha}\rangle\right)\\\nonumber 
=-2\int_{\mathcal{D}_{u,\ubar}}\left(2\langle\nabla_{3}\bar{\beta},-(div\nabla_{3}\bar{\alpha})-2\nabla_{3}( \tr\chibar\bar{\beta})-2\nabla_{3}(\omegabar\bar{\beta})+\nabla_{3}(\etabar\cdot\bar{\alpha})\right.\\\nonumber 
\left.+\frac{1}{2}\nabla_{3}(\mathcal{D}_{a}R_{33}-\mathcal{D}_{3}R_{3a})+[div,\nabla_{3}]\bar{\alpha}\rangle\right.\\\nonumber 
\left.+\langle\nabla_{3}\bar{\alpha},-(\nabla\hat{\otimes}\nabla_{3}\bar{\beta})-\frac{1}{2}\nabla_{3}( \tr\chi\bar{\alpha})+4\nabla_{3}(\omega\bar{\alpha})-3\nabla_{3}(\chibarhat\rho-~^{*}\chibarhat\sigma)+\nabla_{3}((\zeta-4\etabar)\hat{\otimes}\bar{\beta})\right.\\\nonumber 
\left.+\frac{1}{2}\nabla_{3}(\mathcal{D}_{4}R_{33}-\mathcal{D}_{3}R_{34})\gamma+[\nabla_{4},\nabla_{3}]\bar{\alpha}+[\nabla,\nabla_{3}]\hat{\otimes} \bar{\beta} \rangle\right)\\\nonumber 
\sim \int_{\mathcal{D}_{u,\ubar}}\left(\langle\nabla_{3}\bar{\beta},-2\nabla_{3}( \tr\chibar\bar{\beta})-2\nabla_{3}(\omegabar\bar{\beta})+\nabla_{3}(\etabar\cdot\bar{\alpha})+\frac{1}{2}\nabla_{3}(\mathcal{D}_{a}R_{33}-\mathcal{D}_{3}R_{3a})+[div,\nabla_{3}]\bar{\alpha}\rangle\right.\\\nonumber 
\left.+\langle\nabla_{3}\bar{\alpha},-\frac{1}{2}\nabla_{3}( \tr\chi\bar{\alpha})+4\nabla_{3}(\omega\bar{\alpha})-3\nabla_{3}(\chibarhat\rho-~^{*}\chibarhat\sigma)+\nabla_{3}((\zeta-4\etabar)\hat{\otimes}\bar{\beta})\right.\\\nonumber 
\left.+[\nabla_{4},\nabla_{3}]\bar{\alpha}+[\nabla,\nabla_{3}]\hat{\otimes} \bar{\beta} \rangle+\langle(\eta+\etabar)\nabla_{3}\bar{\beta},\nabla_{3}\bar{\alpha}\rangle\right),
\end{eqnarray}
where $\langle\nabla_{3}\bar{\alpha},\frac{1}{2}\nabla_{3}(D_{4}R_{33}-D_{3}R_{44})\gamma\rangle=0$ due to $\gamma-$trace-less property of $\bar{\alpha}$. Now we want to estimate each term separately. First, note the commutators property 
\begin{eqnarray}
[\nabla_{3},\nabla_{4}]\bar{\alpha}\sim \sigma \bar{\alpha}+(\rho^{F}\rho^{F}+\sigma^{F}\sigma^{F})\bar{\alpha}+\omega\nabla_{3}\bar{\alpha}+\omegabar\nabla_{4}\bar{\alpha}+(\eta-\etabar)\nabla\bar{\alpha},\\\nonumber 
[\nabla_{3},\nabla]\bar{\beta}\sim (\bar{\beta}+\bar{\alpha}^{F}(\rho^{F}+\sigma^{F}))\bar{\beta}+(\eta+\etabar)\nabla_{3}\bar{\beta}-\bar{\chi}\nabla\bar{\beta}+\bar{\chi}\eta\bar{\beta},\\\nonumber 
[\nabla_{3},\nabla]\bar{\alpha}\sim (\bar{\beta}+\bar{\alpha}^{F}(\rho^{F}+\sigma^{F}))\bar{\alpha}+(\eta+\etabar)\nabla_{4}\bar{\alpha}-\bar{\chi}\nabla\bar{\alpha}+\bar{\chi}\eta\bar{\alpha}.
\end{eqnarray}
Also, recall the null Bianchi equations that we shall make use of 
\begin{eqnarray}
\nabla_{3}\bar{\beta}_{a}+2 \tr\chibar\bar{\beta}_{a}=-(div\bar{\alpha})_{a}-2\omegabar\bar{\beta}_{a}+(\etabar\cdot\bar{\alpha})_{a}+\frac{1}{2}(\mathcal{D}_{a}R_{33}-\mathcal{D}_{3}R_{3a})\\
\nabla_{4}\bar{\alpha}_{ab}+\frac{1}{2} \tr\chi\bar{\alpha}=-(\nabla\hat{\otimes}\bar{\beta})_{ab}+4\omega\bar{\alpha}_{ab}-3(\chibarhat_{ab}\rho-~^{*}\chibarhat_{ab}\sigma)+((\zeta-4\etabar)\hat{\otimes}\bar{\beta})_{ab}\\\nonumber +\frac{1}{2}(\mathcal{D}_{4}R_{33}-\mathcal{D}_{3}R_{34})\gamma_{ab}+\frac{1}{2}\left(D_a R_{3b} + D_b R_{3a}\right),
\end{eqnarray}
and 
\begin{eqnarray}
D_{b}R_{33}-D_{3}R_{3b}\sim \langle\bar{\alpha}^{F},\hnabla\bar{\alpha}^{F}\rangle-\bar{\chi} \bar{\alpha}^{F}\cdot(\rho^{F}+\sigma^{F})+\etabar|\bar{\alpha}^{F}|^{2}-2\omegabar \bar{\alpha}^{F}\cdot(\rho^{F}+\sigma^{F})\\\nonumber-\hnabla_{3}(\bar{\alpha}^{F}\cdot\rho^{F}+\bar{\alpha}^{F}\cdot\sigma^{F}).
\end{eqnarray}

Note an extremely important point that $||\eta-\etabar||_{L^{\infty}}$ is completely determined by $C(\mathcal{O}_{0})$ since $\eta-\etabar$ satisfies a $\nabla_{3}$ equation
\begin{eqnarray}
\nabla_{3}(\eta-\etabar)\sim-\nabla\omegabar-\bar{\chi}\cdot(\eta+\zeta)+\omegabar(\zeta-\eta)-\frac{1}{2}\bar{\beta}+\frac{1}{2}\sigma^{F}\bar{\alpha}^{F},
\end{eqnarray}
where $\zeta=\frac{1}{2}(\eta-\etabar)$
After commuting with $\nabla$, a direct application of the transport inequality (proposition \ref{transport}), and lemma \ref{1}-\ref{lemma5} yields 
\begin{eqnarray}
||\nabla(\eta-\etabar)||_{L^{2}(S)}\leq C(\mathcal{O}_{0}).
\end{eqnarray}
Application of the trace inequality yields 
\begin{eqnarray}
||\nabla(\eta-\etabar)||_{L^{4}(S)}\leq C\left(||\nabla(\eta-\etabar)||_{L^{4}(S_{0,\ubar})}\nonumber+||\nabla(\eta-\etabar)||^{1/2}_{L^{2}(  \Hbar)}||\hnabla_{3}\nabla(\eta-\etabar)||^{1/4}_{L^{2}(  \Hbar)}(\nabla||\eta-\etabar||_{L^{2}(  \Hbar)}\right.\\\nonumber 
\left.+||\nabla^{2}(\eta-\etabar)||_{L^{2}(  \Hbar)})^{1/4}\right)\\\nonumber 
\leq C(\mathcal{O}_{0})+\epsilon^{\frac{1}{2}}C(\mathcal{O}_{0},\mathcal{W},\mathcal{F},\mathcal{W}(S),\mathcal{F}(S))\\\nonumber 
\leq C(\mathcal{O}_{0})
\end{eqnarray}
after choosing sufficiently small $\epsilon$.
Now use the Sobolev inequality (\ref{eq:sobolev4}) to yield\footnote{So effectively in this setting, we have an improvement $||\eta||_{L^{\infty}(S)}\leq C(\mathcal{O}_{0})$ too.} 
\begin{eqnarray}
||\eta-\etabar||_{L^{\infty}(S)}\leq C(\mathcal{O}_{0}).
\end{eqnarray}
This is extremely important to handle the term $\int_{D_{u,\ubar}}\langle\nabla_{3}\bar{\alpha},(\eta-\etabar)\nabla\bar{\alpha}\rangle$ that arises as a result of commutation. The remaining terms can be estimated similarly to those of the previous cases. We need only take care of the terms where $\nabla_{3}\bar{\alpha}$ appears quadratic and keep track of the associated coefficients.
We estimate each term using lemma \ref{1}-\ref{7} as follows 
\begin{eqnarray}
|\int_{\mathcal{D}_{u,\ubar}}\langle\nabla_{3}\bar{\beta},(\nabla_{3}\varphi)\bar{\beta}|\leq \epsilon\sup_{u}||\nabla_{3}\bar{\beta}||^{2}_{L^{2}(H)}du^{'}+\epsilon C(\mathcal{O}_{0},\mathcal{W},\mathcal{F}),\\
|\int_{\mathcal{D}_{u,\ubar}}\langle\nabla_{3}\bar{\beta},\varphi\nabla_{3}\bar{\beta}\rangle|\leq  \int_{0}^{u}||\nabla_{3}\bar{\beta}||^{2}_{L^{2}(H)}\sup_{u,\ubar}||\varphi||_{L^{\infty}(S)}\\\nonumber \leq \epsilon C(\mathcal{O}_{0},\mathcal{W},\mathcal{F})\sup_{u}||\nabla_{3}\bar{\beta}||^{2}_{L^{2}(H)},\\
|\int_{\mathcal{D}_{u,\ubar}}\langle\nabla_{3}\bar{\beta},\nabla_{3}\varphi \bar{\alpha}\rangle|\leq \epsilon\sup_{u}||\nabla_{3}\bar{\beta}||^{2}_{L^{2}(H)}+\epsilon C(\mathcal{O}_{0},\mathcal{W},\mathcal{F}),\\
|\int_{\mathcal{D}_{u,\ubar}}\langle\nabla_{3}\bar{\beta},\varphi \nabla_{3}\bar{\alpha}\rangle|\leq \epsilon^{\frac{1}{2}}C(\mathcal{O}_{0},\mathcal{W},\mathcal{F})\sup_{u}||\nabla_{3}\bar{\beta}||_{L^{2}(H)},\\
|\int_{\mathcal{D}_{u,\ubar}}\langle\nabla_{3}\bar{\beta}\hnabla^{2}_{3}\bar{\alpha}^{F}\Phi^{F}\rangle|\leq ||\nabla_{3}\bar{\beta}||_{L^{2}(\mathcal{D}_{u,\ubar})}||\hnabla^{2}_{3}\bar{\alpha}^{F}||_{L^{2}(\mathcal{D}_{u,\ubar})}\sup_{u,\ubar}||\Phi^{F}||_{L^{\infty}(S)}\\\nonumber \leq \epsilon^{\frac{1}{2}}C(\mathcal{F}_{0},\mathcal{F})\sup_{u}||\nabla_{3}\bar{\beta}||_{L^{2}(H)},\\
|\int_{\mathcal{D}_{u,\ubar}}\langle\nabla_{3}\bar{\beta}\hnabla_{3}\hnabla\bar{\alpha}^{F}\Phi^{F}\rangle|\leq ||\nabla_{3}\bar{\beta}||_{L^{2}(\mathcal{D}_{u,\ubar})}||\hnabla_{3}\hnabla\bar{\alpha}^{F}||_{L^{2}(\mathcal{D}_{u,\ubar})}\sup_{u,\ubar}||\Phi^{F}||_{L^{\infty}(S)}\\\nonumber\leq \epsilon^{\frac{1}{2}}C(\mathcal{F}_{0},\mathcal{F})\sup_{u}||\nabla_{3}\bar{\beta}||_{L^{2}(H)},\\
|\int_{\mathcal{D}_{u,\ubar}}\langle\nabla_{3}\bar{\beta},\varphi\hnabla_{3}\bar{\alpha}^{F}\Phi^{F}\rangle|\leq \epsilon \sup_{u}||\nabla_{3}\nonumber\bar{\beta}||_{L^{2}(H)}\sup_{u,\ubar}||\varphi||_{L^{\infty}(S)}||\hnabla_{3}\bar{\alpha}^{F}||_{L^{4}(S)}||\Phi^{F}||_{L^{4}(S)}\\\nonumber \leq \epsilon C(\mathcal{O}_{0},\mathcal{W},\mathcal{F})\sup_{u}||\nabla_{3}\bar{\beta}||_{L^{2}(H)}\\\nonumber 
|\int_{\mathcal{D}_{u,\ubar}}\langle\nabla_{3}\bar{\beta},\hnabla_{3}\bar{\alpha}^{F}\hnabla_{3}\Phi^{F}\rangle|\leq \epsilon \sup_{u}||\nabla_{3}\nonumber\bar{\beta}||_{L^{2}(H)}\sup_{u,\ubar}||\hnabla_{3}\bar{\alpha}^{F}||_{L^{4}(S)}||\hnabla_{3}\Phi^{F}||_{L^{4}(S)}\\\nonumber \leq \epsilon C(\mathcal{F})\sup_{u}||\nabla_{3}\nonumber\bar{\beta}||_{L^{2}(H)},\\
|\int_{\mathcal{D}_{u,\ubar}}\langle\nabla_{3}\bar{\alpha}, \tr\chi\nabla_{3}\bar{\alpha}\nonumber+\omega\nabla_{3}\bar{\alpha})\rangle|\leq C(\mathcal{O}_{0})\int_{0}^{\ubar}||\nabla_{3}\bar{\alpha}||^{2}_{  \Hbar}d\ubar^{'},\\
|\int_{\mathcal{D}_{u,\ubar}}\langle\nabla_{3}\bar{\alpha}\varphi\nabla_{3}\Psi\rangle\leq \epsilon^{\frac{1}{2}}\sup_{\ubar}||\nabla_{3}\bar{\alpha}||_{L^{2}(  \Hbar)}\sup_{u}||\nabla\Psi||_{L^{2}(H)}\sup_{u,\ubar}||\varphi||_{L^{\infty(S)}}\\\nonumber \leq \epsilon^{\frac{1}{2}}C(\mathcal{O}_{0},\mathcal{W},\mathcal{F}),\\
|\int_{\mathcal{D}_{u,\ubar}}\langle\nabla_{3}\bar{\alpha},\omega\nabla_{3}\bar{\alpha}\rangle|\leq C(\mathcal{O}_{0}\int_{0}^{\ubar}||\nabla_{3}\bar{\alpha}||^{2}_{L^{2}(  \Hbar)}d\ubar^{'},\\
|\int_{\mathcal{D}_{u,\ubar}}\langle\nabla_{3}\bar{\alpha},(\eta-\etabar)\nabla\bar{\alpha}\rangle|\leq C(\mathcal{O}_{0},\mathcal{W}_{0},\mathcal{F}_{0})\int_{0}^{\ubar}||\nabla_{3}\bar{\alpha}||_{L^{2}(  \Hbar)}d\ubar^{'},\\
|\int_{\mathcal{D}_{u,\ubar}}\langle\nabla_{3}\bar{\beta},\varphi\nabla_{4}\bar{\alpha}\rangle|\leq \epsilon \sup_{u}||\nabla_{3}\bar{\beta}||_{L^{2}(H)}||\nabla_{4}\bar{\alpha}||_{L^{2}(H)}\sup_{u,\ubar}||\varphi||_{L^{\infty}(S)}\\\nonumber \leq \epsilon C(\mathcal{O}_{0},\mathcal{W},\mathcal{F})\sup_{u}||\nabla_{3}\bar{\beta}||_{L^{2}(H)},\\
|\int_{\mathcal{D}_{u,\ubar}}\langle\nabla_{3}\bar{\beta},\varphi\nabla\bar{\alpha}\rangle|\leq \epsilon^{\frac{1}{2}} \sup_{u}||\nabla_{3}\bar{\beta}||_{L^{2}(H)}\sup_{\ubar}||\nabla\bar{\alpha}||_{L^{2}(  \Hbar)}\sup_{u,\ubar}||\varphi||_{L^{\infty}(S)}\\\nonumber \leq \epsilon^{\frac{1}{2}} C(\mathcal{O}_{0},\mathcal{W},\mathcal{F})\sup_{u}||\nabla_{3}\bar{\beta}||_{L^{2}(H)},\\
|\int_{\mathcal{D}_{u,\ubar}}\langle(\eta+\etabar)\nabla_{3}\bar{\beta},\nabla_{3}\bar{\alpha}\rangle|\leq \epsilon^{\frac{1}{2}} C(\mathcal{O}_{0},\mathcal{W},\mathcal{F})\sup_{u}||\nabla_{3}\bar{\beta}||_{L^{2}(H)},
\end{eqnarray}
where we have utilized the available Bianchi equations and null Yang-Mills equations (e.g., we have $\nabla_{4}\bar{\alpha}$ which does not contain derivatives of $\bar{\alpha}$ thereby allowing us to control it on $H$ since we have $L^{4}(S)$ of $\bar{\alpha}$ under control). Once again, we observe that the connection coefficients $\omega$ and $ \tr\chi$ multiplying the most dangerous term $\langle\nabla_{3}\bar{\alpha},\nabla_{3}\bar{\alpha}\rangle$ are determined completely by their initial value. Collecting all the terms, we obtain 
\begin{eqnarray}
2\int_{H_{u}}|\nabla_{3}\bar{\beta}|^{2}+\int_{  \Hbar_{\ubar}}|\nabla_{3}\bar{\alpha}|^{2}\leq 2\int_{H_{0}}|\nabla_{3}\bar{\beta}|^{2}\nonumber+\int_{  \Hbar_{0}}|\nabla_{3}\bar{\alpha}|^{2}+\epsilon C(\mathcal{O}_{0},\mathcal{W},\mathcal{F})\sup_{u}||\nabla_{3}\bar{\beta}||^{2}_{L^{2}(H)}\\\nonumber 
+\epsilon^{\frac{1}{2}} C(\mathcal{O}_{0},\mathcal{W},\mathcal{F})\sup_{u}||\nabla_{3}\bar{\beta}||_{L^{2}(H)}+C(\mathcal{O}_{0})\int_{0}^{\ubar}||\nabla_{3}\bar{\alpha}||^{2}_{L^{2}(  \Hbar)}d\ubar^{'}+C(\mathcal{O}_{0},\mathcal{W}_{0},\mathcal{F}_{0})\int_{0}^{\ubar}||\nabla_{3}\bar{\alpha}||_{L^{2}(  \Hbar)}d\ubar^{'}
\end{eqnarray}
which upon utilizing Gr\"onwall's inequality and smallness of $\epsilon$ yields 
\begin{eqnarray}
2\int_{H_{u}}|\nabla_{3}\bar{\beta}|^{2}+\int_{  \Hbar_{\ubar}}|\nabla_{3}\bar{\alpha}|^{2}\leq C(\mathcal{O}_{0},\mathcal{W}_{0},\mathcal{F}_{0}).
\end{eqnarray}
This concludes the proof of the lemma. 
\end{proof}
\subsection{Energy estimates for $\hnabla_{4}\alpha^{F},\hnabla^{2}_{4}\alpha^{F},\hnabla_{4}\hnabla\alpha^{F},\hnabla_{3}\alphabar^{F},\hnabla^{2}_{3}\alphabar^{F},\hnabla_{3}\hnabla\alphabar^{F}$}
\noindent Now we are left with estimating the similar estimates for the Yang-Mills curvature components. Since the proof goes exactly a similar way as the Weyl curvature case, we will only sketch the proof. \\
\begin{lemma}
\label{75}
The null derivatives of the null components of the Weyl curvature satisfy the following $L^{2}$ energy estimates
\begin{eqnarray}
\int_{H_{u}}|\hnabla^{I}_{4}\alpha^{F}|^{2}\leq C(\mathcal{O}_{0},\mathcal{W}_{0},\mathcal{F}_{0})+\epsilon^{\frac{1}{2}}C(\mathcal{O}_{0},\mathcal{W},\mathcal{F})+\epsilon C(\mathcal{O}_{0},\mathcal{W},\mathcal{F}),\\
\int_{  \Hbar_{\ubar}}|\hnabla^{I}_{3} \bar{\alpha}^{F}|^{2}\leq C(\mathcal{O}_{0},\mathcal{W}_{0},\mathcal{F}_{0})+\epsilon^{\frac{1}{2}}C(\mathcal{O}_{0},\mathcal{W},\mathcal{F})+\epsilon C(\mathcal{O}_{0},\mathcal{W},\mathcal{F})\\\nonumber +C(\mathcal{O}_{0})\int_{0}^{\ubar}||\hnabla^{I}_{3}\bar{\alpha}^{F}||_{L^{2}(  \Hbar)}d\ubar,\\
\int_{H_{u}}|\hnabla_{4}\hnabla\alpha^{F}|^{2}\leq C(\mathcal{O}_{0},\mathcal{W}_{0},\mathcal{F}_{0})+\epsilon^{\frac{1}{2}}C(\mathcal{O}_{0},\mathcal{W},\mathcal{F})+\epsilon C(\mathcal{O}_{0},\mathcal{W},\mathcal{F}),\\
\int_{  \Hbar_{\ubar}}|\hnabla_{3}\hnabla \bar{\alpha}^{F}|^{2}\leq C(\mathcal{O}_{0},\mathcal{W}_{0},\mathcal{F}_{0})+\epsilon^{\frac{1}{2}}C(\mathcal{O}_{0},\mathcal{W},\mathcal{F})+\epsilon C(\mathcal{O}_{0},\mathcal{W},\mathcal{F})\\\nonumber +C(\mathcal{O}_{0})\int_{0}^{\ubar}||\hnabla_{3}\hnabla\bar{\alpha}^{F}||_{L^{2}(  \Hbar)}d\ubar
\end{eqnarray}
for $1\leq I\leq 2, I\in\mathbb{Z}$.
\end{lemma}
\begin{proof} We only sketch the proof for $I=2$ since that is the most non-trivial case for the Yang-Mills fields. Once again we identify the triple $(\alpha^{F},\rho^{F},\sigma^{F})$ and apply the integration identities (\ref{eq:IBP1}-\ref{eq:IBP2})
\begin{eqnarray}
\int_{H_{u}}|\hnabla^{I}_{4}\alpha^{F}|^{2}+\int_{  \Hbar_{\ubar}}|\hnabla^{I}_{4}\rho^{F}|^{2}\nonumber+\int_{  \Hbar_{\ubar}}|\hnabla^{I}_{4}\sigma^{F}|^{2}=\int_{H_{0}}|\hnabla^{I}_{4}\alpha^{F}|^{2}+\int_{  \Hbar_{0}}|\hnabla^{I}_{4}\rho^{F}|^{2}+\int_{  \Hbar_{0}}|\hnabla^{I}_{4}\sigma^{F}|^{2}\\\nonumber
+\underbrace{\int_{\mathcal{D}_{u,\ubar}}|\hnabla^{I}_{4}\alpha^{F}|^{2}(2\omegabar-\frac{1}{2} \tr\chibar)+\int_{\mathcal{D}_{u,\ubar}}|\hnabla^{I}_{4}\rho^{F}|^{2}(2\omega-\frac{1}{2} \tr\chi)+\int_{\mathcal{D}_{u,\ubar}}|\hnabla^{I}_{4}\sigma^{F}|^{2}(2\omega-\frac{1}{2} \tr\chi)}_{ER23}
\\\nonumber -\underbrace{2\int_{D_{u},\ubar}\left(\langle\hnabla^{I}_{4}\alpha^{F},\hnabla_{3}\hnabla^{I}_{4}\alpha^{F}\rangle+\langle\hnabla^{I}_{4}\rho^{F},\hnabla_{4}\hnabla^{I}_{4}\rho^{F}\rangle+\langle\hnabla^{I}_{4}\sigma^{F},\hnabla_{4}\hnabla^{I}_{4}\sigma^{F}\rangle\right)}_{ER24}.
\end{eqnarray}
Using the estimates on the connection coefficients, the first error term $ER_{23}$ is estimated easily 
\begin{eqnarray}
|ER23|=|\int_{\mathcal{D}_{u,\ubar}}|\hnabla^{I}_{4}\alpha^{F}|^{2}(2\omegabar-\nonumber\frac{1}{2} \tr\chibar)+\int_{\mathcal{D}_{u,\ubar}}|\hnabla^{I}_{4}\rho^{F}|^{2}(2\omega-\frac{1}{2} \tr\chi)+\int_{\mathcal{D}_{u,\ubar}}|\hnabla^{I}_{4}\sigma^{F}|^{2}(2\omega-\frac{1}{2} \tr\chi)|\\\nonumber \leq \epsilon C(\mathcal{O}_{0},\mathcal{W},\mathcal{F}) \sup_{u}||\hnabla^{I}_{4}\alpha^{F}||^{2}_{L^{2}(H)}+C(\mathcal{O}_{0})||\hnabla^{I}_{4}\rho^{F}||^{2}_{L^{2}(  \Hbar)}+C(\mathcal{O}_{0})||\hnabla^{I}_{4}\sigma^{F}||^{2}_{L^{2}(  \Hbar)}\\\nonumber 
\leq \epsilon C(\mathcal{O}_{0},\mathcal{W},\mathcal{F})+C(\mathcal{O}_{0})||\hnabla^{I}_{4}\rho^{F}||^{2}_{L^{2}(  \Hbar)}+C(\mathcal{O}_{0})||\hnabla^{I}_{4}\sigma^{F}||^{2}_{L^{2}(  \Hbar)}
\end{eqnarray}
In order to estimate $ER24$, we utilize the null Yang-Mills equations and proceed in an exactly similar way as the previous case
\begin{eqnarray}
\int_{\mathcal{D}_{u,\ubar}}\left(\langle\nabla^{I}_{4}\alpha^{F},\nonumber-\frac{1}{2}\hnabla^{I}_{4}( \tr\chibar\alpha^{F})-\hnabla\hnabla^{I}_{4}\rho^{F}+~^{*}\hnabla\hnabla^{I}_{4}\sigma^{F}-2\hnabla^{I}_{4}(~^{*}\etabar\sigma^{F})+2\hnabla^{I}_{4}(\etabar\rho^{F})\right.\\\nonumber\left.+2\hnabla^{I}_{4}(\omegabar\alpha^{F})-\hnabla^{I}_{4}(\widehat{\chi}\cdot\bar{\alpha}^{F})+[\hnabla_{3},\hnabla^{I}_{4}]\alpha^{F}+[\hnabla,\hnabla^{I}_{4}]\rho^{F}-[^{*}\hnabla,\hnabla^{I}_{4}]\sigma^{F}\rangle\right.\\\nonumber 
\left.+\langle\hnabla^{I}_{4}\rho^{F},-\hat{div}\hnabla^{I}_{4}\alpha^{F}-\hnabla^{I}_{4}( \tr\chi\rho^{F})-\hnabla^{I}_{4}((\eta-\etabar)\cdot\alpha^{F})+[\hnabla_{4},\hnabla^{I}_{4}]\rho^{F}+[\hat{div},\hnabla^{I}_{4}]\alpha^{F}\rangle\right.\\\nonumber 
\left.+\langle\hnabla^{I}_{4}\sigma^{F},-\hat{curl}\hnabla^{I}_{4}\alpha^{F}-\hnabla^{I}_{4}( \tr\chi\sigma^{F})+\hnabla^{I}_{4}((\eta-\etabar)\cdot~^{*}\alpha^{F})+[\hnabla_{4},\hnabla^{I}_{4}]\sigma^{F}-[\hat{curl},\hnabla^{I}_{4}]\alpha^{F}\rangle\right)\\\nonumber 
\sim\int_{\mathcal{D}_{u,\ubar}}\left(\langle\nabla^{I}_{4}\alpha^{F},-\frac{1}{2}\hnabla^{I}_{4}( \tr\chibar\alpha^{F})-2\hnabla^{I}_{4}(~^{*}\etabar\sigma^{F})+2\hnabla^{I}_{4}(\etabar\rho^{F})\right.\\\nonumber\left.+2\hnabla^{I}_{4}(\omegabar\alpha^{F})-\hnabla^{I}_{4}(\widehat{\chi}\cdot\bar{\alpha}^{F})+[\hnabla_{3},\hnabla^{I}_{4}]\alpha^{F}+[\hnabla,\hnabla^{I}_{4}]\rho^{F}-[^{*}\hnabla,\hnabla^{I}_{4}]\sigma^{F}\rangle\right.\\\nonumber 
\left.+\langle\hnabla^{I}_{4}\rho^{F},-\hnabla^{I}_{4}( \tr\chi\rho^{F})-\hnabla^{I}_{4}((\eta-\etabar)\cdot\alpha^{F})+[\hnabla_{4},\hnabla^{I}_{4}]\rho^{F}+[\hat{div},\hnabla^{I}_{4}]\alpha^{F}\rangle\right.\\\nonumber 
\left.+\langle\hnabla^{I}_{4}\sigma^{F},-\hnabla^{I}_{4}( \tr\chi\sigma^{F})+\hnabla^{I}_{4}((\eta-\etabar)\cdot~^{*}\alpha^{F})+[\hnabla_{4},\hnabla^{I}_{4}]\sigma^{F}-[\hat{curl},\hnabla^{I}_{4}]\alpha^{F}\rangle\right)\\\nonumber
+\int_{\mathcal{D}_{u,\ubar}}\langle\hnabla^{I}_{4}\alpha^{F},((\eta+\etabar)(\hnabla^{I}_{4}\rho^{F}+\hnabla^{I}_{4}\sigma^{F})\rangle.
\end{eqnarray}
Now Consider the equation for $\alpha^{F}$ after commuting $\hnabla_{4}$ once 
\begin{eqnarray}
\hnabla_{3}\hnabla_{4}\alpha^{F}=-\frac{1}{2}(\nabla_{4}( \tr\chibar)\alpha^{F}\nonumber+ \tr\chibar\hnabla_{4}\alpha^{F})-\hnabla\hnabla_{4}\rho^{F}+~^{*}\hnabla\hnabla_{4}\sigma^{F}-2~^{*}\nabla_{4}\eta\sigma^{F}\\\nonumber -2\eta\hnabla_{4}\sigma^{F}
+2\nabla_{4}\eta\rho^{F}+2\eta\hnabla_{4}\rho^{F}+2\hnabla_{4}\omegabar\alpha^{F}+2\omegabar\hnabla_{4}\alpha^{F}\nonumber-\hnabla_{4}\widehat{\chi}\cdot \bar{\alpha}^{F}-\widehat{\chi}\cdot \hnabla_{4}\bar{\alpha}^{F}\\\nonumber 
=-\frac{1}{2}(-\frac{1}{2} \tr\chi  \tr\chibar+2\omega  \tr\chibar+\underbrace{2div\etabar}_{II}+2|\etabar|^{2}_{\gamma}+2\rho-\widehat{\chi}\cdot\chibarhat)\alpha^{F}-\frac{1}{2} \tr\chibar\hnabla_{4}\alpha^{F}\\\nonumber 
-\underbrace{\hnabla\hnabla_{4}\rho^{F}+~^{*}\hnabla\hnabla_{4}\sigma^{F}}_{I}-2~^{*}(-\chi\cdot(\eta-\etabar)-\beta-\frac{1}{2}\mathfrak{T}(\cdot,e_{4}))\sigma^{F}\\\nonumber 
-2\eta(-\hat{curl} \alpha^{F}- \tr\chi \sigma^{F}+(\eta-\etabar)\cdot ~^{*}\alpha^{F})+2(-\chi\cdot(\eta-\etabar)-\beta^{W}-\frac{1}{2}\mathfrak{T}(\cdot,e_{4}))\rho^{F}\\\nonumber 
+2\eta (-\hat{div} \alpha^{F}- \tr\chi\rho^{F}-(\eta-\etabar)\cdot\alpha^{F})\\\nonumber
+2(2\omega\omegabar+\frac{3}{4}|\eta-\etabar|^{2}-\frac{1}{4}(\eta-\etabar)\cdot(\eta+\etabar)-\frac{1}{8}|\eta+\etabar|^{2}\nonumber+\frac{1}{2}\rho+\frac{1}{4}\mathfrak{T}_{43})\alpha^{F}\\\nonumber 
+2\omegabar\hnabla_{4}\alpha^{F}-(- \tr\chi \widehat{\chi}-2\omega\widehat{\chi}-\alpha)\\\nonumber 
-\widehat{\chi}\cdot(-\frac{1}{2} \tr\chi\bar{\alpha}^{F}-\hnabla\rho^{F}-~^{*}\hnabla\sigma^{F}-2~^{*}\etabar\cdot\sigma^{F}-2\etabar\cdot\rho^{F}\nonumber+2\omega\bar{\alpha}^{F}-\chibarhat\cdot\alpha^{F})\\\nonumber 
+[\hnabla_{3},\hnabla_{4}]\alpha^{F}-[\hnabla_{4},\hnabla]\rho^{F}+[\hnabla_{4},~^{*}\hnabla]\sigma^{F}.
\end{eqnarray}
Now notice the structure of the previous equation. the principal terms denoted by $I$ are preserved after we take another $\hnabla_{4}$ derivative and therefore are cancelled after integrating by parts. Since $\nabla^{I}_{4}\alpha^{F}$ is controlled over $H$, we can always gain a factor of $\epsilon$. The most problematic terms in the above expression after taking another $\hnabla_{4}$ derivative are 
\begin{eqnarray}
\nabla_{4}\etabar,\nabla_{4}\omega,\nabla_{4}\nabla\etabar.
\end{eqnarray}
We control $\nabla_{4}\etabar$ and $\nabla_{4}\omega$ by their $L^{4}(S)$ norms. $\nabla_{4}\nabla\etabar$ can be controlled by its $||\nabla_{4}\nabla\etabar||_{L^{2}(  \Hbar)}$ norm. Since $\nabla_{4}\nabla\etabar$ is already of top order it appears with innocuous terms that may be estimated as follows 
\begin{eqnarray}
|\int_{\mathcal{D}_{u,\ubar}}\nabla^{I}_{4}\alpha^{F}\nabla_{4}\nabla\etabar\alpha^{F}|\leq \epsilon^{\frac{1}{2}}\sup_{u}||\nabla^{I}_{4}\alpha^{F}||_{L^{2}(H)}\sup_{\ubar}||\nabla_{4}\nabla\etabar||_{L^{2}(  \Hbar)}\sup_{u,\ubar}||\alpha^{F}||_{L^{4}(S)}\leq \epsilon^{\frac{1}{2}}C(\mathcal{O}_{0},\mathcal{W},\mathcal{F}),\\
|\int_{\mathcal{D}_{u,\ubar}}\nabla^{I}_{4}\alpha^{F}\nabla_{4}\varphi\hnabla_{4}\alpha^{F}|\leq \epsilon \sup_{u}||\nabla^{I}_{4}\alpha^{F}||_{L^{2}(H)}\sup_{u,\ubar}\nonumber||\hnabla_{4}\alpha^{F}||_{L^{4}(S)}||\nabla_{4}\varphi||_{L^{4}(S)}\leq \epsilon C(\mathcal{O}_{0},\mathcal{W},\mathcal{F}),
\end{eqnarray}
where we have utilized the lemma \ref{7} to control the $\nabla_{4}\varphi$ in $L^{4}(S)$. Similar terms arise from the equations for $\hnabla^{2}_{4}(\rho^{F},\sigma^{F})$ where most problematic terms are of the above type. Since these estimates involve $\alpha^{F}$, we will always gain a factor of $\epsilon$. By collecting all the terms, we may obtain 
\begin{eqnarray}
\int_{H_{u}}|\hnabla^{I}_{4}\alpha^{F}|^{2}+\int_{  \Hbar_{\ubar}}|\hnabla^{I}_{4}\rho^{F}|^{2}+\int_{  \Hbar_{\ubar}}|\hnabla^{I}_{4}\sigma^{F}|^{2}\leq \int_{H_{0}}|\hnabla^{I}_{4}\alpha^{F}|^{2}\nonumber+\int_{  \Hbar_{0}}|\hnabla^{I}_{4}\rho^{F}|^{2}+\int_{  \Hbar_{0}}|\hnabla^{I}_{4}\sigma^{F}|^{2}\\\nonumber 
+\epsilon^{\frac{1}{2}}C(\mathcal{O}_{0},\mathcal{W},\mathcal{F})+\epsilon C(\mathcal{O}_{0},\mathcal{W},\mathcal{F})+C(\mathcal{O}_{0})||\hnabla^{I}_{4}\rho^{F}||^{2}_{L^{2}(  \Hbar)}+C(\mathcal{O}_{0})||\hnabla^{I}_{4}\sigma^{F}||^{2}_{L^{2}(  \Hbar)}
\end{eqnarray}
which yields 
\begin{eqnarray}
\int_{H_{u}}|\hnabla^{I}_{4}\alpha^{F}|^{2}\leq C(\mathcal{O}_{0},\mathcal{W}_{0},\mathcal{F}_{0})+\epsilon^{\frac{1}{2}}C(\mathcal{O}_{0},\mathcal{W},\mathcal{F})+\epsilon C(\mathcal{O}_{0},\mathcal{W},\mathcal{F})
\end{eqnarray}
through Gr\"onwall estimate and smallness of $\epsilon$. Exact same calculations but commuting the e.o.m with $\hnabla$ in the second time yield the estimate for $\int_{H_{u}}|\hnabla_{4}\hnabla\alpha^{F}|^{2}$. This completes the first part of the lemma. The second part is proved in a similar way. Now we have to pay attention to the terms that are associated with the top derivatives of $\bar{\alpha}^{F}$. Application of the integration identities for the triple $(\bar{\alpha}^{F},\rho^{F},\sigma^{F})$ yields 
\begin{eqnarray}
\int_{  \Hbar_{\ubar}}|\hnabla^{I}_{3}\bar{\alpha}^{F}|^{2}+\int_{H_{u}}|\hnabla^{I}_{3}\rho^{F}|^{2}+\int_{H_{u}}|\hnabla^{I}_{3}\sigma^{F}|^{2}=\nonumber\int_{  \Hbar_{0}}|\hnabla^{I}_{3}\bar{\alpha}^{F}|^{2}+\int_{H_{0}}|\hnabla^{I}_{3}\rho^{F}|^{2}+\int_{H_{0}}|\hnabla^{I}_{3}\sigma^{F}|^{2}\\\nonumber 
+\underbrace{\int_{\mathcal{D}_{u,\ubar}}|\hnabla^{I}_{3}\bar{\alpha}^{F}|^{2}(2\omega-\frac{1}{2} \tr\chi)+\int_{\mathcal{D}_{u,\ubar}}|\hnabla^{I}_{3}\rho^{F}|^{2}(2\omegabar-\frac{1}{2} \tr\chibar)}_{ER25}\\\nonumber-2\underbrace{\int_{D_{u},\ubar}\left(\langle\hnabla^{I}_{3}\bar{\alpha}^{F},\hnabla_{4}\hnabla^{I}_{3}\bar{\alpha}^{F}\rangle+\langle\hnabla^{I}_{3}\rho^{F},\hnabla_{3}\hnabla^{I}_{3}\rho^{F}\rangle+\langle\hnabla^{I}_{3}\sigma^{F},\hnabla_{3}\hnabla^{I}_{3}\sigma^{F}\rangle\right)}_{ER26}.
\end{eqnarray}
As usual notice that the connection coefficients multiplying $|\hnabla^{I}_{3}\bar{\alpha}^{F}|^{2}$ satisfy $\nabla_{3}$ equations and there are completely determined by their initial data. Therefore we obtain 
\begin{eqnarray}
|ER25|\leq C(\mathcal{O}_{0})\int_{0}^{\ubar}||\hnabla^{I}_{3}\bar{\alpha}^{F}||^{2}_{L^{2}(  \Hbar)}+\epsilon C(\mathcal{O}_{0},\mathcal{W},\mathcal{F})\sup_{u}(||\hnabla^{I}_{3}\rho^{F}||_{L^{2}(H)}+||\hnabla^{I}_{3}\sigma^{F}||_{L^{2}(H)}).
\end{eqnarray}
For $ER26$ we utilize the evolution equations
\begin{eqnarray}
\int_{D_{u},\ubar}\left(\langle\hnabla^{I}_{3}\bar{\alpha}^{F},\hnabla_{4}\hnabla^{I}_{3}\bar{\alpha}^{F}\rangle\nonumber+\langle\hnabla^{I}_{3}\rho^{F},\hnabla_{3}\hnabla^{I}_{3}\rho^{F}\rangle+\langle\hnabla^{I}_{3}\sigma^{F},\hnabla_{3}\hnabla^{I}_{3}\sigma^{F}\rangle\right)\\\nonumber
\sim\int_{\mathcal{D}_{u,\ubar}}\left(\langle\hnabla^{I}_{3}\bar{\alpha}^{F},-\frac{1}{2}\hnabla^{I}_{3}( \tr\chi\bar{\alpha}^{F})-2\hnabla^{I}_{3}(~^{*}\etabar\cdot\sigma^{F})-2\hnabla^{I}_{3}(\etabar\cdot\rho^{F})+2\hnabla^{I}_{3}(\omega\bar{\alpha}^{F})-\hnabla^{I}_{3}(\chibarhat\cdot\alpha^{F})\rangle\right.\\\nonumber 
\left.+\langle\hnabla^{I}_{3}\rho^{F},\hnabla^{I}_{3}( \tr\chibar\rho^{F})+\hnabla^{I}_{3}((\eta-\etabar)\cdot\bar{\alpha}^{F})\rangle+\langle\hnabla^{I}_{3}\sigma^{F},-\hnabla^{I}( \tr\chibar\sigma^{F})+\hnabla^{I}_{3}((\eta-\etabar)\cdot~^{*}\bar{\alpha}^{F})\rangle\right)\\
+\int_{\mathcal{D}_{u,\ubar}}\langle\hnabla^{I}_{3}\bar{\alpha}^{F},[\hnabla_{4},\hnabla^{I}_{3}]\bar{\alpha}^{F}\rangle+\int_{\mathcal{D}_{u,\ubar}}\langle\hnabla^{I}_{3}\bar{\alpha}^{F},[\hnabla,\hnabla^{I}_{3}](\rho^{F},\sigma^{F})\rangle
\nonumber+\int_{\mathcal{D}_{u,\ubar}}\langle\hnabla^{I}_{3}\rho^{F},[\hnabla_{3},\hnabla^{I}_{3}]\rho^{F}\rangle\\\nonumber+\int_{\mathcal{D}_{u,\ubar}}\langle\hnabla^{I}_{3}\rho^{F},[\hnabla,\hnabla^{I}]\bar{\alpha}^{F}\rangle
+\int_{\mathcal{D}_{u,\ubar}}\langle\hnabla^{I}_{3}\sigma^{F},[\hnabla,\hnabla^{I}_{3}]\bar{\alpha}^{F}\rangle +\int_{\mathcal{D}_{u,\ubar}}\langle\hnabla^{I}_{3}\alpha^{F},((\eta+\etabar)(\hnabla^{I}_{3}\rho^{F}+\hnabla^{I}_{3}\sigma^{F})\rangle.
\end{eqnarray}
Now we notice some of the key features of the commuted equations. Consider the $\hnabla_{3}$ commuted equation for $\bar{\alpha}^{F}$ 
\begin{eqnarray}
\hnabla_{4}\hnabla_{3}\bar{\alpha}^{F}=-\frac{1}{2}(-\frac{1}{2} \tr\chibar \tr\chi+2\omegabar \tr\chi+2div\eta\nonumber+2|\eta|^{2}+2\rho-\widehat{\chi}\cdot \chibarhat)\bar{\alpha}^{F}\\\nonumber 
-\frac{1}{2} \tr\chi\nabla_{3}\bar{\alpha}^{F}-\hnabla\hnabla_{3}\rho^{F}-~^{*}\hnabla\hnabla_{3}\sigma^{F}-2~^{*}\hnabla_{3}\etabar\cdot\sigma^{F}-2~^{*}\etabar\cdot\hnabla_{3}\sigma^{F}-2\hnabla_{3}\etabar\cdot\rho^{F}\\\nonumber-2\etabar\cdot\hnabla_{3}\rho^{F}
+2\hnabla_{3}\omega\bar{\alpha}^{F}+2\omega\hnabla_{3}\bar{\alpha}^{F}-\hnabla_{3}\chibarhat\cdot\alpha^{F}-\chibarhat\cdot\hnabla_{3}\alpha^{F}\\\nonumber 
=-\frac{1}{2}(-\frac{1}{2} \tr\chibar \tr\chi+2\omegabar \tr\chi+2div\eta\nonumber+2|\eta|^{2}+2\rho-\widehat{\chi}\cdot \chibarhat)\bar{\alpha}^{F}\\\nonumber 
-\frac{1}{2} \tr\chi\nabla_{3}\bar{\alpha}^{F}-\hnabla\hnabla_{3}\rho^{F}-~^{*}\hnabla\hnabla_{3}\sigma^{F}-2~^{*}(-\bar{\chi}\cdot(\etabar-\eta)+\bar{\beta}^{W}+\frac{1}{2}\mathfrak{T}(\cdot,e_{3}))\cdot\sigma^{F}\\\nonumber 
-2~^{*}\etabar\cdot(-\hat{curl} \bar{\alpha}^{F}- \tr\chibar\sigma^{F}+(\eta-\etabar)\cdot~^{*}\bar{\alpha}^{F})-2(-\bar{\chi}\cdot(\etabar-\eta)+\bar{\beta}^{W}+\frac{1}{2}\mathfrak{T}(\cdot,e_{3}))\cdot\rho^{F}\\\nonumber-2\etabar\cdot(-\hat{div} \bar{\alpha}^{F}+ \tr\chibar\rho^{F}+(\eta-\etabar)\cdot\bar{\alpha}^{F})\\
+2(2\omega\omegabar+\frac{3}{4}|\eta-\etabar|^{2}+\frac{1}{4}(\eta-\etabar)\cdot(\eta+\etabar)-\frac{1}{8}|\eta+\etabar|^{2}\nonumber+\frac{1}{2}\rho^{W}+\frac{1}{4}\mathfrak{T}_{43})\bar{\alpha}^{F}\\\nonumber +2\omega\hnabla_{3}\bar{\alpha}^{F}-(- \tr\chibar\chibarhat-2\omegabar\chibarhat-\bar{\alpha}^{W})\cdot\alpha^{F}\\\nonumber -\chibarhat\cdot(-\frac{1}{2} \tr\chibar\alpha^{F}-\hnabla\rho^{F}+~^{*}\widehat{D}\sigma^{F}-2~^{*}\eta\sigma^{F}+2\eta\rho^{F}+2\omegabar\alpha^{F}\nonumber-\widehat{\chi}\cdot \bar{\alpha}^{F})\\\nonumber 
+[\hnabla_{4},\hnabla_{3}]\bar{\alpha}^{F}+[\hnabla,\hnabla_{3}]\rho^{F}+[\hnabla,\hnabla_{3}]\sigma^{F}.
\end{eqnarray}
Similar to the previous case, after another application of $\nabla_{3}$ produces $\nabla\nabla_{3}\eta$. But since this is at the level of top order derivative it contains $\hnabla^{I}\bar{\alpha}^{F}$ in addition to an algebraic term $\bar{\alpha}^{F}$. Now we may control $\bar{\alpha}^{F}$ in $L^{4}(S)$ and $\nabla\nabla_{3}\eta$ in $L^{2}(H)$. This way we gain a factor of $\epsilon$. Similarly, we control $\nabla_{3}\eta$ and $\nabla_{3}\omegabar$ in $L^{4}(S)$ using lemma \ref{7}. The most dangerous terms are estimated as follows 
\begin{eqnarray}
|\int_{\mathcal{D}_{u,\ubar}}\langle\hnabla^{I}_{3}\bar{\alpha}^{F},(\omega, \tr\chi) \hnabla^{I}_{3}\bar{\alpha}^{F}\rangle|\leq C(\mathcal{O}_{0})\int_{0}^{\ubar}||\hnabla^{I}_{3}\bar{\alpha}^{F}||^{2}_{L^{2}(  \Hbar)},\\
|\int_{\mathcal{D}_{u,\ubar}}\nabla^{I}_{3}\bar{\alpha}^{F}\nabla\nabla_{3}\eta\bar{\alpha}^{F}|\leq \epsilon^{\frac{1}{2}}\sup_{u}||\nabla^{I}_{3}\bar{\alpha}^{F}||_{L^{2}(  \Hbar)}\sup_{\ubar}||\nabla\nabla_{3}\eta||_{L^{2}(H)}\sup_{u,\ubar}||\bar{\alpha}^{F}||_{L^{4}(S)}\\\nonumber\leq \epsilon^{\frac{1}{2}}C(\mathcal{O}_{0},\mathcal{W},\mathcal{F}),\\
|\int_{\mathcal{D}_{u,\ubar}}\nabla^{I}_{3}\bar{\alpha}^{F}\nabla_{3}\varphi\hnabla_{3}\bar{\alpha}^{F}|\\\nonumber\leq \epsilon \sup_{u}||\nabla^{I}_{3}\bar{\alpha}^{F}||_{L^{2}(  \Hbar)}\sup_{u,\ubar}||\hnabla_{3}\bar{\alpha}^{F}||_{L^{4}(S)}||\nabla_{3}\varphi||_{L^{4}(S)}\leq \epsilon C(\mathcal{O}_{0},\mathcal{W},\mathcal{F}),
\end{eqnarray}
where $\varphi$ denotes the connection coefficients.
Collecting all the terms, we obtain 
\begin{eqnarray}
|ER26|\leq \epsilon^{\frac{1}{2}}C(\mathcal{O}_{0},\mathcal{W},\mathcal{F})+\epsilon C(\mathcal{O}_{0},\mathcal{W},\mathcal{F})+C(\mathcal{O}_{0})\int_{0}^{\ubar}||\hnabla^{I}_{3}\bar{\alpha}^{F}||^{2}_{L^{2}(  \Hbar)}\\\nonumber+\epsilon^{\frac{1}{2}}C(\mathcal{O}_{0},\mathcal{W},\mathcal{F})||\hnabla^{I}_{3}\rho^{F}||_{L^{2}(H)} 
+\epsilon^{\frac{1}{2}}C(\mathcal{O}_{0},\mathcal{W},\mathcal{F})||\hnabla^{I}_{3}\sigma^{F}||_{L^{2}(H)}
\end{eqnarray}
and therefore 
\begin{eqnarray}
\int_{  \Hbar_{\ubar}}|\hnabla^{I}_{3}\bar{\alpha}^{F}|^{2}+\int_{H_{u}}|\hnabla^{I}_{3}\rho^{F}|^{2}+\int_{H_{u}}|\hnabla^{I}_{3}\sigma^{F}|^{2}\leq\nonumber \int_{  \Hbar_{0}}|\hnabla^{I}_{3}\bar{\alpha}^{F}|^{2}+\int_{H_{0}}|\hnabla^{I}_{3}\rho^{F}|^{2}+\int_{H_{0}}|\hnabla^{I}_{3}\sigma^{F}|^{2}\\\nonumber 
\leq C(\mathcal{O}_{0})\int_{0}^{\ubar}||\hnabla^{I}_{3}\bar{\alpha}^{F}||^{2}_{L^{2}(  \Hbar)}d\ubar^{'}+\epsilon C(\mathcal{O}_{0},\mathcal{W},\mathcal{F})||\hnabla^{I}_{3}\rho^{F}||_{L^{2}(H)}+\epsilon C(\mathcal{O}_{0},\mathcal{W},\mathcal{F})||\hnabla^{I}_{3}\sigma^{F}||_{L^{2}(H)}\\\nonumber 
+\epsilon^{\frac{1}{2}}C(\mathcal{O}_{0},\mathcal{W},\mathcal{F})+\epsilon C(\mathcal{O}_{0},\mathcal{W},\mathcal{F})+\epsilon^{\frac{1}{2}}C(\mathcal{O}_{0},\mathcal{W},\mathcal{F})||\hnabla^{I}_{3}\rho^{F}||_{L^{2}(H)}\\\nonumber 
+\epsilon^{\frac{1}{2}}C(\mathcal{O}_{0},\mathcal{W},\mathcal{F})||\hnabla^{I}_{3}\sigma^{F}||_{L^{2}(H)}.
\end{eqnarray}
Smallness of $\epsilon$ yields 
\begin{eqnarray}
\int_{  \Hbar_{\ubar}}|\hnabla^{I}_{3}\bar{\alpha}^{F}|^{2}\leq C(\mathcal{O}_{0},\mathcal{W}_{0},\mathcal{F}_{0})+\epsilon^{\frac{1}{2}}C(\mathcal{O}_{0},\mathcal{W},\mathcal{F})\nonumber+\epsilon C(\mathcal{O}_{0},\mathcal{W},\mathcal{F})\\\nonumber +C(\mathcal{O}_{0})\int_{0}^{\ubar}||\hnabla^{I}_{3}\bar{\alpha}^{F}||_{L^{2}(  \Hbar)}d\ubar
\end{eqnarray}
Exact same procedure but commuting the e.o.m with $\hnabla$ in the second time yields the estimate for $\int_{  \Hbar_{\ubar}}|\hnabla_{3}\hnabla\bar{\alpha}^{F}|^{2}$. 
This concludes the proof of the lemma and the energy estimates associated with the Weyl and Yang-Mills curvature.
\end{proof}
\begin{corollary} \textit{$\mathcal{W} \leq C(\mathcal{O}_{0},\mathcal{W}_{0},\mathcal{F}_{0}),~\mathcal{F}\leq C(\mathcal{O}_{0},\mathcal{W}_{0},\mathcal{F}_{0})$}.
\end{corollary}
\begin{proof} A direct consequence of lemma \ref{71}- \ref{75}, Gr\"onwall's inequality and the fact that $u\in[0,\epsilon]$ and $\ubar\in[0, J]$.
\end{proof}

\noindent This concludes the proof of the estimates for the Weyl and Yang-Mills curvature throughout $\mathcal{D}_{u,\ubar}$ in terms of the initial data. Once these estimates are derived, we can choose the bootstrap constant $\Delta$ (\ref{eq:bootstrapinitial}) large enough (but finite) to close the argument since the final estimates do not depend on the bootstrap constant $\Delta$. 
\

\section{Concluding Remarks}
\noindent Here we have obtained `semi-global' estimates for the coupled Einstein-Yang-Mills equations in a gauge-invariant way (in the sense of Yang-Mills gauge theory). The null structure of the Einstein-Yang-Mills equations plays a crucial role in achieving this result. This lays the platform to study several problems associated with the coupled Einstein-Yang-Mills dynamics in the immediate future. As we have mentioned, the double null framework is appropriate for the radiation problem associated with gravity or/and Yang-Mills (they both have the same characteristics) equations. Since the Yang-Mills equations are themselves non-linear, they can counterbalance the gravity (e.g., in the case of regular soliton-like solutions). As a consequence, all three possibilities (black hole solution, regular solutions, and naked singular solutions) are open in the context of an evolution problem. Therefore it would be interesting to obtain sharp criteria demarcating all three regimes. 

Another interesting perspective would be to study the Wang-Yau quasi-local energy \cite{Wang1,Wang2} contained in the space-like domain bounded by the membrane $S_{u,\ubar}$. Ultimately, the formation of singularities is associated with the focusing of energy, and since the dynamics of the topological $2-$ spheres $S_{u,\ubar}$ is one of the central parts of the analysis of singularity formation, it is only natural to understand the evolution of the energy contained within it. Therefore, we need a definition of energy in a fully relativistic setting. Luckily as we have mentioned before, \cite{Wang1,Wang2} constructed a notion of quasi-local energy associated with a topological $2-$ sphere $S_{u,\ubar}$. It was proven in \cite{Chen} that such quasi-local energy while evolving along the incoming null direction reproduces the Bel-Robinson energy and the matter stress-energy tensor (at different orders of course) at the limit of approaching vertex. Motivated by this result, it is only natural to study the evolution of this quasi-local energy in the outgoing null direction and observe the behavior in the focusing regime (note an expression of the quasi-local energy was obtained in \cite{puskar} in the presence of a gauge field). We expect an alternate notion of trapped surface formation through the study of this quasi-local energy.

Lastly, we want to mention the fact that our estimates associated with the Yang-Mills curvature components are completely gauge-invariant. In particular, we obtain estimates for the fully gauge covariant angular derivatives. This essentially hints at an apparent similarity with the Maxwell and Yang-Mills theory despite the fact that the latter is a fully non-linear theory because the gauge covariant derivative \textit{hides} the information of the connection and the non-linear coupling shows up only as the commutator of the fully gauge covariant derivatives. This does not cause a problem in the context of obtaining estimates since all the associated inequalities are formulated in terms of the gauge covariant derivatives (and the estimates required for a local existence theory for coupled Yang-Mills equations can be obtained in a gauge-invariant way; note that at the end one ought to choose a gauge and work with the equation for connection. However, our proffered choice of gauge is temporal gauge where the Yang-Mills equations take the form of a symmetric hyperbolic system and the spatial connection can be determined in terms of the gauge invariant norms of the Yang-Mills curvature). There is of course a physical motivation behind this. Since the double null framework in some sense encodes the information about the \textit{physical} nature of the Yang-Mills fields, the choice of the gauge should not matter and one should expect to obtain gauge-invariant estimates as is done in the current context. In this framework of the gauge-invariant estimates, therefore, the exterior stability of the Minkowski space under coupled gravity-Yang-Mills perturbations is expected to hold. In other words, we conjecture ``\textit{Exterior stability of the Minkowski space holds under the coupled gravity-Yang-Mills perturbations}. Since we have developed the framework in this article, we want to pursue such stability conjecture in the near future.

\section{Acknowledgement} \noindent P.M. is supported by the Center of Mathematical Sciences and Applications, Department of Mathematics at Harvard University.

\end{document}